\pdfoutput=1

\newif\ifpenn\pennfalse
\newif\iftwoside\twosidefalse

\iftwoside
\documentclass[%
  final]{memoir}
\else
\documentclass[%
  final,
  oneside]{memoir}
\fi

\ifpenn
\setlrmarginsandblock{1.5in}{1in}{}
\setulmarginsandblock{1in}{1.5in}{*}
\else
\iftwoside\else
\setlrmarginsandblock{1in}{1in}{}
\setulmarginsandblock{1.5in}{1.5in}{*}
\fi
\fi

\setpnumwidth{2em}  

\ifpenn
\setheadfoot{\onelineskip}{0.5in}
\setheaderspaces{1in}{*}{*}
\fi

\iftwoside
\medievalpage[10]
\fi
\checkandfixthelayout

\makepagestyle{mypage}
\ifpenn
\makeevenfoot{mypage}{}{\thepage}{}
\makeoddfoot{mypage}{}{\thepage}{}
\else
\iftwoside
\makeevenhead{mypage}{\leftmark}{}{}
\makeoddhead{mypage}{}{}{\rightmark}
\makeevenfoot{mypage}{}{\thepage}{}
\makeoddfoot{mypage}{}{\thepage}{}
\else
\makeoddhead{mypage}{\rightmark}{}{\thepage}
\fi
\fi

\usepackage{amsmath}
\usepackage{amsthm}
\usepackage{mathtools}
\usepackage{thmtools}
\usepackage{thm-restate}
\usepackage{verbatim}
\usepackage[longnamesfirst,sort,semicolon]{natbib}
\usepackage{stmaryrd}
\usepackage{xspace}
\usepackage{textcomp}

\usepackage{showkeys}
\usepackage[obeyDraft]{todonotes}

\usepackage{microtype}
\usepackage[charter,cal=cmcal]{mathdesign}

\usepackage{xcolor}
\definecolor{DarkGreen}{rgb}{0.1,0.5,0.1}
\definecolor{DarkRed}{rgb}{0.5,0.1,0.1}
\definecolor{DarkBlue}{rgb}{0.1,0.1,0.5}

\usepackage[]{hyperref}
\hypersetup{
    draft=false,            
    unicode=false,          
    pdftoolbar=true,        
    pdfmenubar=true,        
    pdffitwindow=false,     
    pdftitle={},            
    pdfauthor={}
    pdfsubject={},          
    pdfnewwindow=true,      
    pdfkeywords={keywords}, 
    colorlinks=true,        
    linkcolor=DarkRed,      
    citecolor=DarkGreen,    
    filecolor=DarkRed,      
    urlcolor=DarkBlue,      
}
\usepackage[capitalize]{cleveref}

\usepackage{mathpartir}
\makeatletter
\newcommand\inferruleref[3]{%
  \mpr@inferstar@[Left=#1]{#2}{#3}%
  \my@name@inferrule{\rname{#1}}%
}

\newcommand\my@name@inferrule[1]{%
  \def\@currentlabelname{\ensuremath{#1}}%
}
\makeatother

\newtheorem{theorem}{Theorem}[section]
\newtheorem{lemma}[theorem]{Lemma}
\newtheorem{corollary}[theorem]{Corollary}
\newtheorem{proposition}[theorem]{Proposition}

\newtheorem*{theorem*}{Theorem}

\theoremstyle{definition}
\newtheorem{definition}[theorem]{Definition}
\newtheorem{example}[theorem]{Example}

\theoremstyle{remark}
\newtheorem{remark}[theorem]{Remark}

\makeatletter
\providecommand*{\toclevel@definition}{0}
\providecommand*{\toclevel@example}{0}
\providecommand*{\toclevel@proposition}{0}
\providecommand*{\toclevel@theorem}{0}
\providecommand*{\toclevel@remark}{0}
\providecommand*{\toclevel@lemma}{0}
\providecommand*{\toclevel@corollary}{0}
\makeatother

\allowdisplaybreaks

\newcommand{\BB}{\mathbb{B}}

\newcommand{\EE}{\mathbb{E}}

\newcommand{\NN}{\mathbb{N}}

\newcommand{\RR}{\mathbb{R}}

\newcommand{\ZZ}{\mathbb{Z}}

\newcommand{\call}{\mathcal}

\def\cA{{\call A}}
\def\cB{{\call B}}

\def\cD{{\call D}}
\def\cE{{\call E}}

\def\cI{{\call I}}

\def\cN{{\call N}}
\def\cO{{\call O}}
\def\cP{{\call P}}

\def\cR{{\call R}}
\def\cS{{\call S}}
\def\cT{{\call T}}
\def\cU{{\call U}}
\def\cV{{\call V}}
\def\cW{{\call W}}

\DeclareMathOperator{\id}{id}

\newcommand{\Ex}{ \mathop{\EE} }

\newcommand{\exD}[2]{ \Ex_{#1} \left[ #2 \right] }

\newcommand{\supp}{\ensuremath{\text{supp}}}
\newcommand{\tvdist}[2]{d_{\text{tv}}\left(#1, #2\right)}
\newcommand{\epsdist}[3]{d_{#1}\left(#2, #3\right)}
\newcommand{\Dist}{\ensuremath{\mathbf{Distr}}}
\newcommand{\SDist}{\ensuremath{\mathbf{SDistr}}}
\newcommand{\Unif}{\ensuremath{\mathbf{Unif}}}
\newcommand{\Flip}{\ensuremath{\mathbf{Flip}}}
\newcommand{\Lap}[1]{\ensuremath{\mathbf{Lap}_{#1}}}

\newcommand{\kwtrue}{\mathbf{true}}
\newcommand{\kwfalse}{\mathbf{false}}

\newcommand{\Tern}[3]{#1 \mathrel{?} #2 \mathrel{:} #3}
\def\drawplusplus#1#2#3{\hbox to 0pt{\hbox to #1{\hfill\vrule height #3 depth
      0pt width #2\hfill\vrule height #3 depth 0pt width #2\hfill
    }}\vbox to #3{\vfill\hrule height #2 depth 0pt width #1 \vfill}}

\newcommand{\Skip}{\mathbf{skip}}
\newcommand{\Seq}[2]{{#1};\,{#2}}
\newcommand{\Ass}[2]{#1 \leftarrow #2}
\newcommand{\Rand}[2]{#1 \stackrel{\raisebox{-.25ex}[.25ex]%
 {\tiny $\mathdollar$}}{\raisebox{-.2ex}[.2ex]{$\leftarrow$}} #2}
\newcommand{\Cond}[3]{\mathbf{if}\ #1\ \mathbf{then}\ #2\ \mathbf{else}\ #3}
\newcommand{\Condt}[2]{\mathbf{if}\ #1\ \mathbf{then}\ #2}

\newcommand{\WWhile}[2]{\mathbf{while}\ #1\ \mathbf{do}\ #2}

\newcommand{\Expr}{\mathcal{E}}
\newcommand{\DExpr}{\mathcal{DE}}
\newcommand{\Pred}{\mathcal{P}}
\newcommand{\Cmd}{\mathcal{C}}
\newcommand{\Var}{\mathcal{X}}
\newcommand{\LVar}{\mathcal{L}}
\newcommand{\Val}{\mathcal{V}}
\newcommand{\Mem}{\mathbf{State}}
\newcommand{\denot}[1]{\llbracket #1 \rrbracket}

\newcommand{\dunit}{\mathbf{unit}}
\newcommand{\dbind}{\mathbf{bind}}

\newcommand{\MID}{\;\;\mid\;\;}
\newcommand{\subst}[2]{\left\{#2/#1\right\}}
\newcommand{\sidel}{\langle 1\rangle}
\newcommand{\sider}{\langle 2\rangle}
\newcommand{\tagged}{\langle 1/2 \rangle}

\newcommand{\FV}{\text{FV}}
\newcommand{\MV}{\text{MV}}

\newcommand{\liftf}[1]{#1^\sharp}
\newcommand{\lift}[1]{\mathrel{#1^\sharp}}
\newcommand{\alift}[2]{\mathrel{{#1}^{\sharp #2}}}
\newcommand{\symalift}[2]{\mathrel{\overline{#1}^{\sharp #2}}}

\newcommand{\prhl}[4]{#1 \sim #2 : #3 \Longrightarrow #4}
\newcommand{\aprhl}[5]{#1 \sim_{#5} #2 : #3 \Longrightarrow #4}
\newcommand{\symaprhl}[5]{#1 \approx_{#5} #2 : #3 \Longrightarrow #4}
\newcommand{\eprhl}[7]{#1 \sim_{#7} #2 : \{ #3 ; #5 \} \Longrightarrow \{ #4 ; #6 \}}
\newcommand{\xprhl}[5]
           {\left\{ \spacer #3 \right\}
             \mbox{$\begin{array}{c} {{#1}} \\ {{#2}} \end{array}$}
             \left\{ \spacer #4 \right\} \blacktriangleright {{#5}}}
\newcommand{\spacer}{\vphantom{\mbox{$\begin{array}{@{}c@{}} \, \\  \, \end{array}$}}}
\newcommand{\lless}[2]{ #1 \models #2 \text{ lossless}}

\newcommand{\rname}[1]{\textsc{[#1]}}
\newcommand{\sname}[1]{\textsc{#1}}
\newcommand{\Sprhl}{\textsc{pRHL}\xspace}
\newcommand{\Sxprhl}{$\times$\textsc{pRHL}\xspace}
\newcommand{\Saprhl}{\textsc{apRHL}\xspace}
\newcommand{\Seprhl}{$\EE$\textsc{pRHL}\xspace}
\newcommand{\Shoare}{\textsc{HOARe}$^2$\xspace}

\DeclareMathOperator{\NEIGHBORS}{\cN}
\newcommand{\neighbors}[1]{\NEIGHBORS_{#1}}

\DeclareMathOperator{\VALID}{\cV}
\newcommand{\valid}[1]{\VALID_{#1}}

\DeclareMathOperator{\SAFE}{\cS}
\newcommand{\safe}[1]{\SAFE_{#1}}

\newcommand{\evalQ}{\ensuremath{\mathbf{evalQ}}}

\begin{document}
\frontmatter
\ifpenn
\pagestyle{plain}
\else
\copypagestyle{chapter}{empty}
\pagestyle{empty}
\fi


\pagenumbering{roman}

\begin{titlingpage*}
\title{Probabilistic couplings for probabilistic reasoning}
\author{Justin Hsu}

\begin{center}
  {\HUGE PROBABILISTIC COUPLINGS\par FOR PROBABILISTIC REASONING\par}

  \ifpenn\vspace*{4\bigskipamount}\else\vspace*{6\bigskipamount}\fi

  {\Huge Justin Hsu}

  \ifpenn\vspace*{3\bigskipamount}\else\vspace*{5\bigskipamount}\fi

  {\Large A DISSERTATION}

  \bigskip

  in

  \bigskip

  {\Large Computer and Information Science}

  \vspace*{3\bigskipamount}

  Presented to the Faculties of the University of Pennsylvania in Partial

  \ifpenn\medskip\else\smallskip\fi

  Fulfillment of the Requirements for the Degree of Doctor of Philosophy

  \medskip

  2017

\ifpenn

  \vspace*{3\bigskipamount}

  \makebox[0.75\textwidth]{\hrulefill}
  
  Benjamin C. Pierce, Professor of Computer and Information Science

  Co-supervisor of Dissertation

  \vspace*{2.5\bigskipamount}

  \makebox[0.75\textwidth]{\hrulefill}
  
  Aaron Roth, Associate Professor of Computer and Information Science

  Co-supervisor of Dissertation

  \vspace*{2.5\bigskipamount}

  \makebox[0.75\textwidth]{\hrulefill}
  
  Lyle Ungar, Professor of Computer and Information Science

  Graduate Group Chairperson

  \vspace*{3\bigskipamount}

  \centerline{\LARGE Dissertation Committee}

  \vspace*{2\bigskipamount}

  Gilles Barthe, Research Professor, IMDEA Software Institute

  \medskip

  Sampath Kannan, Professor of Computer and Information Science

  \medskip

  Val Tannen, Professor of Computer and Information Science

  \medskip

  Steve Zdancewic, Professor of Computer and Information Science

\else

  \vspace*{4\bigskipamount}
  {\Large Dissertation Supervisors}

  \smallskip

  Benjamin C. Pierce

  Aaron Roth

  \vspace*{2\bigskipamount}

  {\Large Graduate Group Chairperson}

  \smallskip

  Lyle Ungar

  \vspace*{2\bigskipamount}

  {\Large Dissertation Committee}
  \smallskip

  Gilles Barthe

  Sampath Kannan

  Val Tannen

  Steve Zdancewic
\fi
\end{center}

\end{titlingpage*}


\begin{vplace}
  Copyright~\textcopyright~2017~\theauthor

  \smallskip

  Typeset in Charter and Math Design with \LaTeX.

  \smallskip

  Latest version hosted at \url{https://arxiv.org/abs/1710.09951}.
\end{vplace}

\ifpenn\thispagestyle{empty}\fi
\iftwoside\cleardoublepage\else\clearpage\fi

\begin{vplace}
  \begin{center}
    For my family, \\
    for my teachers.
  \end{center}
\end{vplace}

\iftwoside\cleardoublepage\else\clearpage\fi

\chapter*{Acknowledgments}

It takes a village, as they say. I am extremely fortunate to have two amazing
advisers: Benjamin Pierce, who taught me to ponder slowly, and Aaron Roth, who
taught me to think rapidly. This thesis would not have existed without Gilles
Barthe, serving as a third, unofficial adviser. Gilles hosted me during two
highly productive summers at the IMDEA Software Institute in Madrid and remains
an inspiring (and boundless) source of energy and enthusiasm. He and his
longtime collaborators, Pierre-Yves Strub and Benjamin Gr\'egoire, provided a
wealth of technical expertise and a theoretical setting that turned out to be
far richer than anyone had thought.

The strong theory and formal methods faculty at Penn, including Sanjeev Khanna,
Stephanie Weirich, Michael Kearns, Steve Zdancewic, Sampath Kannan, Val Tannen,
Sudipto Guha, and Rajeev Alur, have taught me more than I can hope to remember.
I am also indebted to Marco Gaboardi and Emilio Jes\'us Gallego Arias for
launching me to the research frontiers. Other fellow students and postdocs,
including
Arthur Azevedo de Amorim, 
Rachel Cummings,
Michael Greenberg,
Radoslav Ivanov, 
Shahin Jabbari,
Matthew Joseph,
Jamie Morgenstern,
Christine Rizkallah,
Ryan Rogers,
Nikos Vasilakis,
and Steven Wu,
kept morale at a good level, to say the least.

At Stanford, Greg Lavender's one-off course on Haskell gave me my first glimpse
of programming languages as an intellectual field. Without that formative
experience, I never would have gone on to graduate school in computer science.
Finally, I am deeply grateful for the constant support and encouragement from my
parents Joyce and David, my sister Tammy, and the rest of my family.

\bigskip\bigskip\bigskip

\hfill London, UK

\hfill \today

\iftwoside\cleardoublepage\else\clearpage\fi

\ifpenn
\begin{center}
  ABSTRACT

  \bigskip

  \MakeUppercase{\thetitle}

  \medskip

  \theauthor

  \medskip

  Benjamin C. Pierce and Aaron Roth
\end{center}

\bigskip

\else
\begin{abstract}
\fi

\ifpenn
\begin{DoubleSpace}
\fi
This thesis explores \emph{proofs by coupling} from the perspective of formal
verification. Long employed in probability theory and theoretical computer
science, these proofs construct \emph{couplings} between the output
distributions of two probabilistic processes. Couplings can imply various
\emph{probabilistic relational properties}, guarantees that compare two runs of
a probabilistic computation.

To give a formal account of this clean proof technique, we first show that
proofs in the program logic \Sprhl (probabilistic Relational Hoare Logic)
describe couplings. We formalize couplings that establish various probabilistic
properties, including distribution equivalence, convergence, and stochastic
domination. Then we deepen the connection between couplings and \Sprhl by giving
a proofs-as-programs interpretation: a coupling proof encodes a probabilistic
\emph{product program}, whose properties imply relational properties of the
original two programs. We design the logic \Sxprhl (\emph{product} \Sprhl) to
build the product program, with extensions to model more advanced constructions
including \emph{shift coupling} and \emph{path coupling}.

We then develop an approximate version of probabilistic coupling, based on
\emph{approximate liftings}. It is known that the existence of an approximate
lifting implies \emph{differential privacy}, a relational notion of statistical
privacy. We propose a corresponding proof technique---\emph{proof by approximate
coupling}---inspired by the logic \Saprhl, a version of \Sprhl for building
approximate liftings. Drawing on ideas from existing privacy proofs, we extend
\Saprhl with novel proof rules for constructing new approximate couplings. We
give approximate coupling proofs of privacy for the \emph{Report-noisy-max} and
\emph{Sparse Vector} mechanisms, well-known algorithms from the privacy
literature with notoriously subtle privacy proofs, and produce the first
formalized proof of privacy for these algorithms in \Saprhl.

Finally, we enrich the theory of approximate couplings with several more
sophisticated constructions: a principle for showing accuracy-dependent privacy,
a generalization of the advanced composition theorem from differential privacy,
and an optimal approximate coupling relating two subsets of samples. We also
show equivalences between approximate couplings and other existing definitions.
These ingredients support the first formalized proof of privacy for the
\emph{Between Thresholds} mechanism, an extension of the Sparse Vector
mechanism.

\ifpenn
\end{DoubleSpace}
\else
\end{abstract}
\fi

\iftwoside\cleardoublepage\else\clearpage\fi

\microtypesetup{protrusion=false}

\tableofcontents*

\iftwoside\cleardoublepage\else\clearpage\fi

\listoffigures*

\microtypesetup{protrusion=true}

\mainmatter

\copypagestyle{chapter}{plain}
\pagestyle{mypage}

\pagenumbering{arabic}

\chapter{Introduction} \label{chap:intro}

Randomized algorithms have long stimulated the imagination of computer
scientists. Endowed with the power to draw random samples, these algorithms
provide sophisticated guarantees far beyond the reach of deterministic
computations.  However, their proofs of correctness are often highly intricate,
employing specialized techniques to reason about randomness.

This thesis investigates one such tool---\emph{probabilistic coupling}---for
proving \emph{probabilistic relational properties}, which compare executions of
two randomized algorithms. Couplings are a familiar concept in probability
theory and theoretical computer science, where they support a proof technique
called \emph{proof by coupling}. We explore the reasoning principle behind these
proofs, identifying their theoretical underpinnings, clarifying their structure,
and enabling formal verification.

\section{Challenges in probabilistic reasoning}

While probabilistic programs aren't much harder to express than their
deterministic counterparts, they are significantly more challenging to reason
about. To see why, suppose we want to prove a property about the output of an
algorithm for all inputs. In a deterministic algorithm, each concrete input
produces a single trace through the program. Since different paths correspond to
distinct inputs, we can freely group similar traces together and reason about
each group on its own. The code of the algorithm naturally guides the proof: at
a branching instruction, for instance, we may classify the executions according
to the path they take and then consider each behavior separately. In this way,
we can reason about a complex program by focusing on simpler cases.

For randomized algorithms, this neat picture is considerably more complicated.
A single execution now comprises multiple traces, each with its own probability.
Relations between trace probabilities make it difficult to reason about paths
separately. At a conditional statement, for instance, the execution has some
probability of taking the first branch and some probability of taking the second
branch; in a sense, the computation takes \emph{both} branches.  If we reason
about these two cases in isolation, we must track the probabilities of each
branch in order to join the cases when the paths later merge. This is
challenging even for small programs, as a path's probability can have complex
dependencies on the input and on the probabilities of other possible traces.  If
we instead reason about both behaviors together, we must provide a single
analysis for executions that behave quite differently.

Broadly speaking, then, a central challenge of probabilistic reasoning is to
organize the various execution behaviors into manageable cases while cleanly
tracking the relationship across different groups. To tackle this problem,
researchers in randomized algorithms have crafted a rich array of conceptual
tools to construct their proofs, simplifying arguments by cleverly abstracting
away uninteresting technical details. Also known as \emph{proof techniques},
these instruments can be sophisticated and highly specialized---often tailored
to a single property, as a kind of logical scalpel---but the most useful ones
strike a fine balance: specific enough to pare logical arguments down to just
their essential points, general enough to support proofs for a broad class of
properties. A proof technique is a reusable component for analyzing algorithms,
and is as much of an intellectual contribution as any new proof or algorithm. 

\section{Couplings and relational properties}

In this thesis we explore a proof technique for \emph{probabilistic relational
properties}, guarantees comparing the runs of two randomized algorithms. Such
properties are commonplace in computer science and probability theory. Examples
include:
\begin{itemize}
  \item \textbf{Probabilistic equivalence}: two probabilistic programs produce
    equal distributions.
  \item \textbf{Stochastic domination}: one probabilistic program is more likely
    than another to produce large outputs.
  \item \textbf{Convergence} (also \textbf{mixing}): the output distributions of
    two probabilistic loops approach each other as the loops execute
    more iterations.
  \item \textbf{Indistinguishability} (also \textbf{differential privacy}): the
    output distributions of two probabilistic programs are close together. For
    instance, \emph{differential privacy} requires that two similar
    inputs---say, the real private database and a hypothetical version with one
    individual's data omitted---yield similar output distributions.
  \item \textbf{Truthfulness} (also \textbf{Nash equilibrium}): an agent's
    average utility is larger when reporting an honest value instead of
    deviating to a misleading value.
\end{itemize}
%
%
At first glance, relational properties appear to be even harder to establish
than standard, non-relational properties---instead of analyzing a single
probabilistic computation, we now need to deal with two. (Indeed, any property
of a single program can be viewed as a relational property between the target
program and the trivial, do-nothing program.) However, relational properties
often relate two highly similar programs, even comparing the same program on two
possible inputs. In these cases, we can leverage a powerful abstraction and an
associated proof technique from probability theory---\emph{probabilistic
coupling} and \emph{proof by coupling}.

The fundamental observation is that probabilistic relational properties compare
computations in two different worlds, assuming no particular correlation between
random samples.
\todo[inline]{AAA: I was a bit confused when I read this, because I took this at
first to mean that the two samples were independent, which contradicts what
comes later.}
Accordingly, we may freely assume any correlation we like for the purposes of
the proof---a relational property holds (or doesn't hold) regardless of which
one we pick.
For instance, if two programs generate identical output distributions, this
holds whether they share coin flips or take independent samples; relational
properties don't require that the two programs use separate randomness.
\todo[inline]{AAA: Not very clear what it means at this point to "share a coin
flip"; perhaps this would be easier to illustrate with two example programs.}
By carefully arranging the correlation, we can reason about two executions as if
they were linked in some convenient way.

To take advantage of this freedom, we need some way to design specific
correlations between program executions.
In principle, this can be a highly challenging task.
The two runs may take samples from different distributions, and it is unclear
exactly how they can or should share randomness.
When the two programs have similar shapes, however, we can link two computations
in a step-by-step fashion.
First, correlations between intermediate samples can be described by
\emph{probabilistic couplings}, joint distributions over pairs.
For example, a valid coupling of two fair coin flips could specify that the
draws take opposite values; the correlated distribution would produce ``(heads,
tails)'' and ``(tails, heads)'' with equal probability.
A coupling formalizes what it means to share randomness: a single source of
randomness simulates draws from two distributions.
Since randomness can be shared in different ways, two distributions typically
support a variety of distinct couplings.

A \emph{proof by coupling}, then, describes two correlated executions by piecing
together couplings for corresponding pairs of sampling instructions.
In the course of a proof, we can imagine stepping through the two programs in
parallel, selecting couplings along the way.
For instance, if we apply the opposite coupling to link a coin flip in one
program with a coin flip in the other, we may assume the samples remain opposite
when analyzing the rest of the programs.
By flowing these relations forward from two initial inputs, a proof by coupling
can focus on just pairs of similar executions as it builds up to a coupling
between two output distributions.
This is the main product of the proof: features of the final coupling imply
properties about the output distributions, and hence relational properties about
the original programs.

Working in tandem, couplings and proofs by couplings can simplify probabilistic
reasoning in several ways.
\begin{itemize}
  \item \textbf{Reduce to one source of randomness.}
    By analyzing two runs as if they shared a single source of randomness, we
    can reason about two programs as if they were one.
  \item \textbf{Abstract away probabilities.}
    Proofs by coupling isolate probabilistic reasoning from the
    non-probabilistic parts of the proof, which are more straightforward.
    We only need to think about probabilistic aspects when we select couplings
    at the sampling instructions; throughout the rest of the programs, we can
    reason purely in terms of deterministic relations between the two runs.
  \item \textbf{Enable compositional, structured reasoning.}
    By focusing on each step of an algorithm individually and then smoothly
    combining the results, the coupling proof technique enables a highly modular
    style of reasoning guided by the code of the program.
\end{itemize}
Proofs by coupling are also surprisingly flexible; many probabilistic relational
properties, including the examples listed above, can be proved in this way.
Individual couplings can also be combined in various subtle ways, giving rise to
a rich diversity of coupling proofs.

\section{A formal study of proofs by coupling}

While couplings proofs originate from probability theory as a tool for human
reasoning, \emph{formal verification} will be the setting for our investigation.
Our perspective affords two distinct advantages.
\begin{itemize}
  \item
    The theory of formal verification provides a wealth of concepts to precisely
    describe and analyze proof systems.
    By studying coupling proofs in these terms, we can give a fresh
    understanding of this classical proof technique.
    As a consequence, we can extend proofs by coupling to target new guarantees,
    unifying seemingly unrelated properties and simplifying their proofs.
  \item
    Formal verification systems provide a natural domain to apply our insights.
    First, couplings enable clean proofs for properties that are
    traditionally challenging for computers to verify.
    Existing techniques can also be considered in a new light, clarifying why
    certain features are useful and revealing possible enhancements.
\end{itemize}

The technical chapters of this thesis fall into two parts.
\Cref{chap:exact,chap:products} concern probabilistic couplings, while
\cref{chap:approx,chap:combining} investigate approximate couplings. General
themes and intuitions developed in the first half influence the second half, but
the two parts are largely self-contained and can be read independently.

\Cref{chap:exact} begins our study of probabilistic couplings in formal
verification. We observe that the program logic \Sprhl (\emph{probabilistic
Relational Hoare Logic}), originally proposed by \citet{BartheGZ09} for
verifying proofs of cryptographic security, is in fact a logic for formally
constructing probabilistic couplings. Using this connection, we formalize
classical coupling proofs establishing equivalence, convergence, and stochastic
domination of probabilistic processes.

\Cref{chap:products} deepens our correspondence between couplings and \Sprhl.
First, coupling proofs describe how to meld two probabilistic programs into a
single program; in formal verification, such a construction is known as a
\emph{product program}. Accordingly, we show that \Sprhl proofs encode a novel
kind of product program called the \emph{coupled product}, reflecting the
structure of a coupling proof. This idea recalls a central theme in logic and
computer science: formal proofs can be interpreted as computations, a so-called
\emph{proofs-as-programs} or \emph{Curry-Howard} correspondence. Concretely, we
extend \Sprhl to a logic \Sxprhl (\emph{product} \Sprhl) that constructs the
product program alongside the coupling proof. Then, we design a new loop rule
inspired by \emph{shift coupling}, a way to build couplings asynchronously. As
applications, we formalize rapid mixing for several Markov chains.  Our approach
can also capture a simplified version of the \emph{path coupling} technique
introduced by \citet{bubley1997path}.

\Cref{chap:approx} turns our focus to a generalization of couplings:
\emph{approximate couplings}. These couplings are closely related to
\emph{differential privacy}, a quantitative, relational property modeling
statistical privacy. We begin with a candidate definition of approximate
coupling, refining several existing notions. We then reverse-engineer a
corresponding proof technique called \emph{proof by approximate coupling} from
the program logic \Saprhl, an approximate version of \Sprhl proposed by
\citet{BKOZ13-toplas}. Taking inspiration from this proof technique, we show how
two new approximate couplings of the Laplace distribution and a construction
called \emph{pointwise equality} enable an approximate coupling proof of privacy
for the \emph{Report-noisy-max} and \emph{Sparse Vector} mechanisms. Our proofs
are simpler than existing proofs---which were notoriously difficult to get right
\citep{lyu2016understanding}---and extend to natural variants of the algorithms.
We realize our proof in an extension of \Saprhl, arriving at the first
formalized privacy proofs for these mechanisms.

\Cref{chap:combining} presents a handful of advanced constructions for
approximate couplings: (i) a principle for proving accuracy-dependent privacy;
(ii) a construction for linking two subsets of samples; and (iii) a composition
principle generalizing the advanced composition theorem from differential
privacy. We also clarify the landscape of existing definitions by proving
equivalences between approximate couplings and prior notions of approximate
equivalence. Combining these ingredients, we give a proof by approximate
coupling establishing differential privacy for the \emph{Between Thresholds}
mechanism by \citet{BunSU16}. After extending \Saprhl with several rules
corresponding to our constructions, we achieve the first formalized privacy
proof for this algorithm.

\Cref{chap:future} surveys concurrent work on couplings and formal verification,
outlining promising directions for further developing the theory and application
of proofs by coupling.

\paragraph*{A note about mechanical verification.}

The gold standard in formal verification is \emph{mechanized proof}, where every
step has been fully computer-checked. The logics we will develop are highly
suitable for computer verification, due to their highly structured proofs, but
we do not mechanically verify coupling proofs as part of this thesis.  Instead,
we will describe \emph{formalized proofs} in the logic on paper.  Prototype
implementations in the \sname{EasyCrypt} framework~\citep{BartheDGKSS13} can
machine-check versions of the coupling proofs we will see (see, e.g.,
\citet{BKOZ13-toplas} and \citet{buch:thesis}), but the current implementations
are not precisely aligned with our logics.

\paragraph*{Acknowledgments.}
The technical content of this thesis draws on a fruitful collaboration with
Gilles Barthe, Thomas Espitau, No{\'e}mie Fong, Marco Gaboardi, Benjamin
Gr{\'e}goire, Tetsuya Sato, L{\'e}o Stefanesco and Pierre-Yves Strub.
\Cref{chap:exact} is based on \citet{BEGHSS15}, \cref{chap:products} includes
material from \citet{BGHS16}, \cref{chap:approx} distills results first
appearing in \citet{BGGHS16}, and \cref{chap:combining} presents material from
\citet{BGGHS16c} and \citet{BEHSS17}. The author contributed the bulk of the
work towards the results in this thesis.

\chapter{Couplings \`a la formal verification} \label{chap:exact}

To begin our formal investigation of coupling proofs, we first provide the
necessary mathematical background (\cref{sec:couplings-prelim}), and then draw a
deep connection between coupling proofs and the program logic \Sprhl
(\cref{sec:prhl}); this observation is the principal conceptual contribution of
this chapter and forms the foundation for the entire thesis.  We show how to
formalize several examples of couplings (\cref{sec:prhl-ex}), and discuss
related work on relational program logics and probabilistic liftings
(\cref{sec:prhl-rw}).

\section{Mathematical preliminaries} \label{sec:couplings-prelim}

A discrete probability distribution associates each element of a set with a
number in $[0, 1]$, representing its \emph{probability}. In order to model
programs that may not terminate, we work with a slightly more general notion
called a \emph{sub-distribution}.

\begin{definition}
  A (discrete) \emph{sub-distribution} over a countable set $\cA$ is a map $\mu
  : \cA \to [0, 1]$ taking each element of $\cA$ to a numeric weight such that
  the weights sum to at most $1$:
  \[
    \sum_{a \in \cA} \mu(a) \leq 1 .
  \]
  We write $\SDist(\cA)$ for the set of all sub-distributions over $\cA$. When
  the weights sum to $1$, we call $\mu$ a \emph{proper} distribution; we write
  $\Dist(\cA)$ for the set of all proper distributions over $\cA$.  The
  \emph{empty} or \emph{null sub-distribution} $\bot$ assigns weight $0$ to all
  elements.
\end{definition}

We work with discrete sub-distributions throughout. While this is certainly a
restriction---excluding, for instance, standard distributions over the real
numbers---many interesting coupling proofs can already be expressed in our
setting. Where necessary, we will use discrete versions of standard, continuous
distributions. Our results should mostly carry over to the continuous setting,
as couplings are frequently used on continuous distributions in probability
theory, but the general case introduces measure-theoretic technicalities (e.g.,
working with integrals rather than sums, checking sets are measurable, etc.)
that would distract from our primary focus. We discuss this issue further in
\cref{chap:future}.

We need several concepts and notations related to discrete distributions.
First, the probability of a set $\cS \subseteq \cA$:
\[
  \mu(\cS) \triangleq \sum_{a \in \cS} \mu(a) .
\]
The \emph{support} of a sub-distribution is the set of elements with positive
probability:
\[
  \supp(\mu) \triangleq \{ a \in \cA \mid \mu(a) > 0 \} .
\]
The \emph{weight} of a sub-distribution is the total probability of all
elements:
\[
  |\mu| \triangleq \sum_{a \in \cA} \mu(a) .
\]
Sub-distributions can be ordered pointwise: $\mu_1 \leq \mu_2$ if $\mu_1(a) \leq
\mu_2(a)$ for every element $a \in \cA$. Finally, given a function $f : \cA \to
\cB$ where $\cB$ is numeric (like the integers or the reals), its \emph{expected
value} over a sub-distribution $\mu$ is
\[
  \Ex_\mu[f] \triangleq \Ex_{a \sim \mu} [ f(a) ]
  \triangleq \sum_{a \in \cA} f(a) \cdot \mu(a) .
\]
Under light assumptions, the expected value is guaranteed to exist (for
instance, when $f$ is a bounded function).

To transform sub-distributions, we can lift a function $f : \cA \to \cB$ on sets
to a map $\liftf{f} : \SDist(\cA) \to \SDist(\cB)$ via $\liftf{f}(\mu)(b)
\triangleq \mu( f^{-1}(b) )$. For example, let $p_1 : \cA_1 \times \cA_2 \to
\cA_1$ and $p_2 : \cA_1 \times \cA_2 \to \cA_2$ be the first and second
projections from a pair. The corresponding \emph{probabilistic projections}
$\pi_1 : \SDist(\cA_1 \times \cA_2) \to \SDist(\cA_1)$ and $\pi_2 : \SDist(\cA_1
\times \cA_2) \to \SDist(\cA_2)$ are defined by
\begin{align*}
  \pi_1(\mu)(a_1) &\triangleq \liftf{p_1}(\mu)(a_1) = \sum_{a_2 \in \cA_2} \mu(a_1, a_2) \\
  \pi_2(\mu)(a_2) &\triangleq \liftf{p_2}(\mu)(a_2) = \sum_{a_1 \in \cA_1} \mu(a_1, a_2) .
\end{align*}
We call a sub-distribution $\mu$ over pairs a \emph{joint
sub-distribution}, and the projected sub-distributions $\pi_1(\mu)$ and
$\pi_2(\mu)$ the \emph{first} and \emph{second marginals}, respectively.

\subsection{Probabilistic couplings and liftings}

A probabilistic coupling models two distributions with a single joint
distribution.
\begin{definition}
  Given $\mu_1, \mu_2$ sub-distributions over $\cA_1$ and $\cA_2$, a
  sub-distribution $\mu$ over pairs $\cA_1 \times \cA_2$ is a \emph{coupling}
  for $(\mu_1, \mu_2)$ if $\pi_1(\mu) = \mu_1$ and $\pi_2(\mu) = \mu_2$.
\end{definition}
Generally, couplings are not unique---different witnesses represent different
ways to share randomness between two distributions. To give a few examples, we
first introduce some standard distributions.
\begin{definition}
  Let $\cA$ be a finite, non-empty set. The \emph{uniform distribution} over
  $\cA$, written $\Unif(\cA)$, assigns probability $1/|\cA|$ to each element. We
  write $\Flip$ for the uniform distribution over booleans, the distribution of
  a fair coin flip.
\end{definition}
\begin{example}[Couplings from bijections] \label{ex:bij-couple}
  We can give two distinct couplings of $(\Flip, \Flip)$:
  \begin{description}
    \item[Identity coupling:]
      \[
        \mu_{\id}(a_1, a_2) \triangleq
        \begin{cases}
          1/2 &: a_1 = a_2 \\
          0   &: \text{otherwise} .
        \end{cases}
      \]
    \item[Negation coupling:]
      \[
        \mu_{\neg}(a_1, a_2) \triangleq
        \begin{cases}
          1/2 &: \neg a_1 = a_2 \\
          0   &: \text{otherwise} .
        \end{cases}
      \]
  \end{description}
  More generally, any bijection $f : \cA \to \cA$ yields a coupling of
  $(\Unif(\cA), \Unif(\cA))$:
  \[
    \mu_{f}(a_1, a_2) \triangleq
    \begin{cases}
      1/|\cA| &: f(a_1) = a_2 \\
      0     &: \text{otherwise} .
    \end{cases}
  \]
\end{example}

This coupling matches samples: each sample $a$ from the first distribution is
paired with a corresponding sample $f(a)$ from the second distribution. To take
two correlated samples from this coupling, we can imagine first sampling from
the first distribution, and then applying $f$ to produce a sample for the second
distribution. When $f$ is a bijection, this gives a valid coupling for two
uniform distributions: viewed separately, both the first and second correlated
samples are distributed uniformly.

For more general distributions, if $a_1$ and $a_2$ have different probabilities
under $\mu_1$ and $\mu_2$ then the correlated distribution cannot return $(a_1,
-)$ and $(-, a_2)$ with equal probabilities; for instance, a bijection with
$f(a_1) = a_2$ would not give a valid coupling. However, general distributions
can be coupled in other ways.
\begin{example} \label{ex:id-couple}
  Let $\mu$ be a sub-distribution over $\cA$. The \emph{identity coupling} of
  $(\mu, \mu)$ is
  \[
    \mu_{\id}(a_1, a_2) \triangleq
    \begin{cases}
      \mu(a) &: a_1 = a_2 = a \\
      0   &: \text{otherwise} .
    \end{cases}
  \]
  Sampling from this coupling yields a pair of equal values.
\end{example}

\begin{example} \label{ex:triv-couple}
  Let $\mu_1, \mu_2$ be sub-distributions over $\cA_1$ and $\cA_2$. The
  \emph{independent} or \emph{trivial} coupling is
  \[
    \mu_{\times}(a_1, a_2) \triangleq \mu_1(a_1) \cdot \mu_2(a_2) .
  \]
  This coupling models $\mu_1$ and $\mu_2$ as independent distributions:
  sampling from this coupling is equivalent to first sampling from $\mu_1$ and
  then pairing with a fresh draw from $\mu_2$. The coupled distributions must be
  proper in order to ensure the marginal conditions.
\end{example}

Since any two proper distributions can be coupled by the trivial coupling, the
mere existence of a coupling yields little information. Couplings are more
useful when the joint distribution satisfies additional conditions, for instance
when all elements in the support satisfy some property.

\begin{definition}[Lifting]
  Let $\mu_1, \mu_2$ be sub-distributions over $\cA_1$ and $\cA_2$, and let $\cR
  \subseteq \cA_1 \times \cA_2$ be a relation. A sub-distribution $\mu$ over
  pairs $\cA_1 \times \cA_2$ is a \emph{witness} for the \emph{$\cR$-lifting} of
  $(\mu_1, \mu_2)$ if:
  \begin{enumerate}
    \item $\mu$ is a coupling for $(\mu_1, \mu_2)$, and
    \item $\supp(\mu) \subseteq \cR$.
  \end{enumerate}
  If there exists $\mu$ satisfying these two conditions, we say $\mu_1$ and
  $\mu_2$ are related by the \emph{lifting of $\cR$} and write
  \[
    \mu_1 \lift{\cR} \mu_2 .
  \]
  We typically express $\cR$ using set notation, i.e., 
  \[
    \cR = \{ (a_1, a_2) \in \cA_1 \times \cA_2 \mid \Phi(a_1, a_2) \}
  \]
  where $\Phi$ is a logical formula. When $\Phi$ is a standard mathematical
  relation (e.g., equality), we leave $\cA_1$ and $\cA_2$ implicit and just
  write $\Phi$, sometimes enclosed by parentheses $(\Phi)$ for clarity.
\end{definition}

\begin{example}
  Many of the couplings we saw before are more precisely described as liftings.
  \begin{description}
    \item[Bijection coupling.]
      For a bijection $f : \cA \to \cA$, the coupling in \cref{ex:bij-couple}
      witnesses the lifting
      \[
        \Unif(\cA) \lift{ \{ (a_1, a_2) \mid f(a_1) = a_2 \} } \Unif(\cA) .
      \]
    \item[Identity coupling.]
      The coupling in \cref{ex:id-couple} witnesses the lifting
      \[
        \mu \lift{(=)} \mu .
      \]
    \item[Trivial coupling.]
      The coupling in \cref{ex:triv-couple} witnesses the lifting
      \[
        \mu_1 \lift{\top} \mu_2 .
      \]
      ($\top \triangleq \cA_1 \times \cA_2$ is the trivial relation relating all
      pairs of elements.)
  \end{description}
\end{example}

Liftings were originally introduced in research on \emph{probabilistic
bisimulation}, a technique for verifying equivalence of two probabilistic
transition systems. By viewing liftings as a particular kind of coupling, we can
repurpose verification tools to prove new properties by constructing couplings,
while leveraging ideas from the coupling literature to enrich existing systems.
Before we get to that, let's see how the existence of a coupling can imply
useful probabilistic properties.

\subsection{Useful consequences of couplings and liftings} \label{sec:coupling-conseq}

If there exists a coupling $\mu$ between $(\mu_1, \mu_2)$ satisfying certain
properties, we can deduce probabilistic properties about $\mu_1$ and $\mu_2$.
First of all, two coupled distributions have equal weight.

\begin{proposition}[Equality of weight] \label{fact:couple-wt}
  Suppose $\mu_1$ and $\mu_2$ are sub-distributions over $\cA$ such that there
  exists a coupling $\mu$ of $\mu_1$ and $\mu_2$. Then $|\mu_1| = |\mu_2|$.
\end{proposition}

This follows because $\mu_1$ and $\mu_2$ are both projections of $\mu$, and
projections preserve weight. Couplings can also show that two distributions are
equal.

\begin{proposition}[Equality of distributions] \label{fact:eq-couple}
  Suppose $\mu_1$ and $\mu_2$ are sub-distributions over $\cA$. Then
  $\mu_1 = \mu_2$ if and only if there is a lifting $\mu_1 \lift{(=)} \mu_2$.
\end{proposition}
\begin{proof}
  For the forward direction, define $\mu(a, a) \triangleq \mu_1(a) = \mu_2(a)$
  and $\mu(a_1, a_2) \triangleq 0$ otherwise. Evidently, $\mu$ has support in
  the equality relation $(=)$ and also has the desired marginals: $\pi_1(\mu) =
  \mu_1$ and $\pi_2(\mu) = \mu_2$. Thus $\mu$ is a witness to the desired
  lifting.

  For the reverse direction, let the witness be $\mu$. By the support condition,
  $\pi_1(\mu)(a) = \pi_2(\mu)(a)$ for every $a \in \cA$.  Since the left and
  right sides are equal to $\mu_1(a)$ and $\mu_2(a)$ respectively by the
  marginal conditions, $\mu_1(a) = \mu_2(a)$ for every $a$.  So, $\mu_1$ and
  $\mu_2$ are equal.
\end{proof}

In some cases we can show results in the converse direction: if a property of
two distributions holds, then there exists a particular lifting. To give some
examples, we first introduce a powerful equivalence due to
\citet{strassen1965existence}.

\begin{theorem} \label{thm:strassen}
  Let $\mu_1, \mu_2$ be sub-distributions over $\cA_1$ and $\cA_2$, and let
  $\cR$ be a binary relation over $\cA_1$ and $\cA_2$.  Then the lifting $\mu_1
  \lift{\cR} \mu_2$ implies $\mu_1(\cS_1) \leq \mu_2(\cR(\cS_1))$ for every
  subset $\cS_1 \subseteq \cA_1$, where $\cR(\cS_1) \subseteq \cA_2$ is the
  \emph{image} of $\cS_1$ under $\cR$:
  \[
    \cR(\cS_1) \triangleq \{ a_2 \in \cA_2 \mid \exists a_1 \in \cA_1,\ (a_1, a_2) \in \cR \} .
  \]
  (For instance, if $\cA_1 = \cA_2 = \NN$ and $\cR$ is the relation $\leq$, then
  $\cR(\cS)$ is the set of all natural numbers larger than $\min \cS$.)
  The converse holds if $\mu_1$ and $\mu_2$ have equal weight.
\end{theorem}

Strassen proved \cref{thm:strassen} for continuous (proper) distributions using
deep results from probability theory. In our discrete setting, there is an
elementary proof by the maximum flow-minimum cut theorem; the proof also
establishes a mild generalization to sub-distributions. So as not to interrupt
the flow here, we defer details of the proof to \cref{chap:combining}.  For now,
we use this theorem to illustrate a few more useful consequences of liftings.
For starters, couplings can bound the probability of an event in the first
distribution by the probability of an event in the second distribution.
\begin{proposition} \label{fact:imp-couple}
  Suppose $\mu_1, \mu_2$ are sub-distributions over $\cA_1$ and $\cA_2$
  respectively, and consider two subsets $\cS_1 \subseteq \cA_1$ and $\cS_2
  \subseteq \cA_2$. The lifting
  \[
    \mu_1
    \lift{ \{ (a_1, a_2) \mid a_1 \in \cS_1 \to a_2 \in \cS_2 \} }
    \mu_2
  \]
  implies $\mu_1(\cS_1) \leq \mu_2(\cS_2)$. The converse holds when $\mu_1$ and
  $\mu_2$ have equal weight.
\end{proposition}
\begin{proof}
  Let $\cR$ be the relation $\{ (a_1, a_2) \mid a_1 \in \cS_1 \to a_2 \in \cS_2
  \}$.  The forward direction is immediate by \cref{thm:strassen}, taking the
  subset $\cS_1$. For the reverse direction, consider any non-empty subset
  $\cT_1 \subseteq \cA_1$. If $\cT_1$ is not contained in $\cS_1$, then
  $\cR(\cT_1) = \cA_2$ and $\mu_1(\cT_1) \leq \mu_2(\cR(\cT_1))$ since $\mu_1$
  and $\mu_2$ have equal weight. Otherwise $\cR(\cT_1) = \cS_2$, so
  \[
    \mu_1(\cT_1) \leq \mu_1(\cS_1) \leq \mu_2(\cS_2) = \mu_2(\cR(\cT_1)) .
  \]
  \Cref{thm:strassen} gives the desired lifting:
  \[
    \mu_1
    \lift{ \{ (a_1, a_2) \mid a_1 \in \cS_1 \to a_2 \in \cS_2 \} }
    \mu_2 .
    \qedhere
  \]
\end{proof}
A slightly more subtle consequence is \emph{stochastic domination}, an order on
distributions over an ordered set.
\begin{definition}
  Let $(\cA, \leq_{\cA})$ be an ordered set and suppose $\mu_1, \mu_2$ are
  sub-distributions over $\cA$. We say $\mu_2$ \emph{stochastically
  dominates} $\mu_1$, denoted $\mu_1 \leq_{sd} \mu_2$, if
  \[
    \mu_1 ( \{ a \in \cA \mid k \leq_{\cA} a \} )
    \leq
    \mu_2 ( \{ a \in \cA \mid k \leq_{\cA} a \} )
  \]
  for every $k \in \cA$.
\end{definition}
This order is different from the pointwise order on sub-distributions since it
uses the order on the underlying space. For instance, two proper distributions
satisfy $\mu_1 \leq \mu_2$ exactly when $\mu_1 = \mu_2$, but two unequal
distributions may satisfy $\mu_1 \leq_{sd} \mu_2$; e.g., if we take
distributions over the natural numbers $\NN$ with the usual order and  $\mu_1$
places weight $1$ on $0$ while $\mu_2$ places weight $1$ on $1$.

Stochastic domination is precisely the probabilistic lifting of the order relation.
\begin{proposition} \label{fact:sd-couple}
  Suppose $\mu_1, \mu_2$ are sub-distributions over a set $\cA$ with a reflexive
  order $\leq_\cA$ (i.e., $a \leq_\cA a$). Then $\mu_1 \lift{(\leq_\cA)} \mu_2$
  implies $\mu_1 \leq_{sd} \mu_2$. The converse also holds when $\mu_1$ and
  $\mu_2$ have equal weight, as long as any upwards closed subset of $\cA$
  either contains a minimum element or is the whole set $\cA$ (e.g., $\cA = \NN$
  or $\ZZ$ with the usual order).
\end{proposition}
\begin{proof}
  Let $\cR \triangleq (\leq_\cA)$. For the forward direction,
  \cref{thm:strassen} gives
  \[
    \mu_1(\{ a \in \cA \mid k \leq_{\cA} a \})
    \leq
    \mu_2(\cR(\{ a \in \cA \mid k \leq_{\cA} a \})) .
  \]
  The subset on the right is precisely the set of $a' \in \cA$ such that $a'
  \geq_\cA a$ for some $a \geq_\cA k$; by transitivity and reflexivity, we have
  \[
    \mu_2(\cR(\{ a \in \cA \mid k \leq_{\cA} a \}))
    = 
    \mu_2(\{ a \in \cA \mid k \leq_{\cA} a \}) .
  \]
  This holds for all $k \in \cA$, establishing $\mu_1 \leq_{sd} \mu_2$.

  For the converse, suppose $\mu_1 \leq_{sd} \mu_2$ and $\mu_1$ and $\mu_2$ have
  equal weights, and let $\cS \subseteq \cA$ be any subset. If the upwards
  closure $\cR(\cS)$ is the whole set $\cA$, then $\mu_1(\cS) \leq
  \mu_2(\cR(\cS))$ since $\mu_1$ and $\mu_2$ have equal weights. Otherwise,
  there is a least element $k$ of $\cR(\cS)$ by assumption, and we have
  \[
    \mu_1(\cS) \leq \mu_1(\cR(\cS)) 
    = \mu_1(\{ a \in \cA \mid k \leq_{\cA} a \})
    \leq \mu_2(\{ a \in \cA \mid k \leq_{\cA} a \})
    = \mu_2(\cR(\cS)) ,
  \]
  where the middle inequality is by stochastic domination. \Cref{thm:strassen}
  implies $\mu_1 \lift{(\leq_\cA)} \mu_2$.
\end{proof}
Finally, a typical application of coupling proofs is showing that two
distributions are close together.
\begin{definition}
  Let $\mu_1, \mu_2$ be sub-distributions over $\cA$. The \emph{total variation
    distance} (also known as \emph{TV-distance} or \emph{statistical distance})
  between $\mu_1$ and $\mu_2$ is defined as
  \[
    \tvdist{\mu_1}{\mu_2}
    \triangleq \frac{1}{2} \sum_{a \in \cA} | \mu_1(a) - \mu_2(a) |
    = \max_{\cS \subseteq \cA} | \mu_1(\cS) - \mu_2(\cS) | .
  \]
  In particular, the total variation distance bounds the difference in
  probabilities of any event.
\end{definition}
Couplings are closely related to TV-distance.
\begin{theorem}[see, e.g., \citet{Lindvall02,LevinPW09}] \label{thm:coupling-method}
  Let $\mu_1$ and $\mu_2$ be sub-distributions over $\cA$ and let $\mu$ be a
  coupling. Then
  \[
    \tvdist{\mu_1}{\mu_2} \leq \Pr_{(a_1, a_2) \sim \mu} [ a_1 \neq a_2 ] .
  \]
  In particular, if $\mu$ witnesses the lifting
  \[
    \mu_1 \lift{\{ (a_1, a_2) \in \cA \times \cA \mid (a_1, a_2) \in \cS \to a_1 = a_2 \}} \mu_2 ,
  \]
  then the TV-distance is bounded by the probability
  \[
    \tvdist{\mu_1}{\mu_2} \leq \Pr_{(a_1, a_2) \sim \mu} [ (a_1, a_2) \notin \cS ] .
  \]
\end{theorem}
\Cref{thm:coupling-method} is the fundamental result behind the so-called
\emph{coupling method} \citep{aldous1983random}, a technique to show two
probabilistic processes converge by constructing a coupling that causes the
processes to become equal with high probability.\footnote{%
  The converse of \cref{thm:coupling-method} also holds: there exists a coupling
  $\mu_{max}$, known as the \emph{maximal} or \emph{optimal coupling}, that
  achieves equality (see, e.g., \citet{Lindvall02,LevinPW09}). However, this
result will not be important for our purposes.}
This theorem is usually stated for proper distributions $\mu_1$ and $\mu_2$; the
result on sub-distributions follows as an easy consequence. (If there is a
lifting then $\mu_1$ and $\mu_2$ have equal weights $w$ by
\Cref{fact:couple-wt}, and the inequalities in \cref{thm:coupling-method} are
preserved when $\mu_1$ and $\mu_2$ are scaled by the same constant.  When $w =
0$ the inequality is immediate; otherwise, by scaling up both distributions by
$1/w$, applying the standard theorem to obtain the total variation bound for
proper distributions, then scaling back down by $w$, we recover the total
variation bound for sub-distributions.) Unlike the previous facts, the target
property about $\mu_1$ and $\mu_2$ does not directly follow from the existence
of a lifting---we need more detailed information about the coupling $\mu$.

\subsection{Proof by coupling}

The previous results suggest an indirect approach to proving properties of two
distributions: demonstrate there exists a coupling of a particular form.
However, how are we supposed to find a witness distribution with the desired
properties?  The given distributions may be highly complex, possibly over
infinite sets---it is not clear how to represent, much less construct, the
desired coupling.

To address this challenge, probability theorists have developed a powerful proof
technique called \emph{proof by coupling}. This technique assumes a bit more
information about the distributions: we need concrete descriptions of two
processes \emph{producing} the distributions. Usually, these generating programs
are readily available; indeed, they are often the most natural descriptions of
complex distributions.

Given two programs, a proof by coupling builds a coupling for the output
distributions by coupling intermediate samples. In a bit more detail, we imagine
stepping through the programs in parallel, one instruction at a time, starting
from two inputs. Whenever we reach two corresponding sampling instructions, we
pick a valid coupling for the sampled distributions.  The selected couplings
induce a relation on samples, which we can assume when analyzing the rest of the
programs. For instance, by selecting couplings for earlier samples carefully, we
may be able to assume the coupled programs take the same path at a subsequent
branching statement; in this way, coupling proofs can consider just pairs of
well-behaved executions.

Finding appropriate couplings is the main intellectual challenge when carrying
out a proof by coupling, the steps requiring ingenuity. We close this section
with an example of the proof technique in action.

\begin{example} \label{ex:coupling-pf}
  Consider a probabilistic process that tosses a fair coin $T$ times and
  returns the number of heads. If $\mu_1$, $\mu_2$ are the output distributions
  from running this process for $T = T_1, T_2$ iterations respectively and $T_1
  \leq T_2$, then $\mu_1 \leq_{sd} \mu_2$.
\end{example}
\begin{proof}[Proof by coupling]
  For the first $T_1$ iterations, couple the coin flips to be equal---this
  ensures that after the first $T_1$ iterations, the coupled counts are equal.
  The remaining $T_2 - T_1$ coin flips in the second run can only increase the
  second count, while preserving the first count.  Therefore under the coupling,
  the first count is no more than the second count at termination, establishing
  $\mu_1 \leq_{sd} \mu_2$.
\end{proof}

For readers unfamiliar with these proofs, this argument may appear bewildering.
The coupling is constructed implicitly, and some of the steps are mysterious. To
clarify such proofs, a natural idea is to design a formal logic describing
coupling proofs.  Somewhat surprisingly, the logic we are looking for was
already proposed in the formal verification literature, originally for verifying
security of cryptographic protocols.

\section{A formal logic for coupling proofs} \label{sec:prhl}

We will work with the logic \Sprhl (probabilistic Relational Hoare Logic)
proposed by~\citet{BartheGZ09}. Before detailing its connection to coupling
proofs, we provide a brief introduction to program logics.

\subsection{Program logics: A brief primer}

A \emph{logic} consists of a collection of formulas, also known as
\emph{judgments}, and an interpretation describing what it means---in typical,
standard mathematics---for judgments to be true (\emph{valid}). While it is
possible to prove judgments valid directly by using regular mathematical
arguments, this is often inconvenient as the interpretation may be quite
complicated.  Instead, many logics provide a \emph{proof system}, a set of
\emph{logical rules} describing how to combine known judgments (the
\emph{premises}) to prove a new judgment (the \emph{conclusion}). Each rule
represents a single step in a formal proof. Starting from judgments given by
rules with no premises (\emph{axioms}), we can successively apply rules to prove
new judgments, building a tree-shaped \emph{derivation} culminating in a single
judgment. To ensure that this final judgment is valid, each logical rule should
be \emph{sound}: if the premises are valid, then so is the conclusion.
Soundness is a basic property, typically one of the first results to be proved
about a logic.

\emph{Program logics} were first introduced by \citet{hoare1969axiomatic},
building on earlier ideas by \citet{Floyd67}; they are also called
\emph{Floyd-Hoare logics}. These logics are really two logics in one: the
\emph{assertion logic}, where formulas describe program states, and the program
logic proper, where judgments describe imperative programs. A judgment in the
main program logic consists of three parts: a program $c$ and two
\emph{assertions} $\Phi$ and $\Psi$ from the assertion logic. The
\emph{pre-condition} $\Phi$ describes the initial conditions before executing
$c$ (for instance, assumptions about the input), while the \emph{post-condition}
$\Psi$ describes the final conditions after executing $c$ (for instance,
properties of the output).  \citet{hoare1969axiomatic} proposed the original
logical rules, which construct a judgment for a program by combining judgments
for its sub-programs.  This compositional style of reasoning is a hallmark of
program logics.

By varying the interpretation of judgments, the assertion logic, and the logical
rules, Floyd-Hoare logics can establish a variety of properties about different
kinds of imperative programs. Notable extensions reason about non-determinism
\citep{dijkstra1976discipline}, pointers and memory allocation
\citep{OhRY01,Reynolds01,Reynolds02}, concurrency \citep{OHEARN2007271}, and
more. (Readers should consult a survey for a more comprehensive account of
Floyd-Hoare logic
\citep{Apt:1981:TYH:357146.357150,APT198383,Jones:2003:EST:858595.858602}.)

In this tradition, \citet{BartheGZ09} introduced the logic \Sprhl targeting
security properties in cryptography.  Compared to standard program logics, there
are two twists: each judgment describes \emph{two} programs, and programs can
use random sampling. In short, \Sprhl is a \emph{probabilistic Relational Hoare
Logic}. Judgments encode probabilistic relational properties of two programs,
where a post-condition describes a probabilistic liftings between two output
distributions. More importantly, the proof rules represent different ways to
combine liftings, formalizing various steps in coupling proofs.  Accordingly, we
will interpret \Sprhl as a formal logic for proofs by coupling.

To build up to this connection, we first provide a brief overview of a core
version of \Sprhl, reviewing the programming language, the judgments and their
interpretation, and the logical rules.

\subsection{The logic \Sprhl: the programming language}

Programs in \Sprhl are defined in terms of \emph{expressions} $\Expr$ including
constants, like the integers and booleans, as well as combinations of constants
and variables with primitive operations, like addition and subtraction. We
suppose $\Expr$ also includes terms for basic datatypes, like tuples and lists.
Concretely, $\Expr$ is inductively defined by the following grammar:
\begin{align}
  \Expr\qquad \coloneqq\qquad & \Var
        \MID \LVar
        \tag{variables} \\
        \MID{}& \ZZ
        \MID \Expr + \Expr
        \MID \Expr - \Expr
        \MID \Expr \cdot \Expr
        \tag{numbers} \\
        \MID{}& \BB
        \MID \Expr \land \Expr
        \MID \Expr \lor \Expr
        \MID \neg \Expr
        \MID \Expr = \Expr
        \MID \Expr < \Expr 
        \tag{booleans} \\
        \MID{}& (\Expr, \dots, \Expr)
        \MID \pi_i(\Expr)
        \MID []
        \MID \Expr :: \Expr
        \MID \cO(\Expr)
        \tag{tuples, lists, operations}
\end{align}
Expressions can mention two classes of variables: a countable set $\Var$ of
\emph{program variables}, which can be modified by the program, and a set
$\LVar$ of \emph{logical variables}, which model fixed parameters. Expressions
are typed as numbers, booleans, tuples, or lists, and primitive operations $\cO$
have typed signatures; we consider only well-typed expressions throughout. The
expressions $(\Expr, \dots, \Expr)$ and $\pi_i(\Expr)$ construct and project
from a tuple, respectively; $[]$ is the empty list, and $\Expr :: \Expr$ adds an
element to the head of a list.  We typically use the letter $e$ for expressions,
$x, y, z, \dots$ for program variables, and lower-case Greek letters ($\alpha,
\beta, \dots$) and capital Roman letters ($N, M, \dots$) for logical variables.

We write $\Val$ for the countable set of \emph{values}, including integers,
booleans, tuples, finite lists, etc. We can interpret expressions given maps
from variables and logical variables to values.

\begin{definition}
  Program states are \emph{memories}, maps $\Var \to \Val$; we usually write $m$
  for a memory and $\Mem$ for the set of memories.  \emph{Logical contexts} are
  maps $\LVar \to \Val$; we usually write $\rho$ for a logical context.
\end{definition}

We interpret an expression $e$ as a function $\denot{e}_\rho : \Mem \to \Val$ in
the usual way, for instance:
\[
  \denot{e_1 + e_2}_\rho m \triangleq \denot{e_1}_\rho m + \denot{e_2}_\rho m  .
\]
Likewise, we interpret primitive operations $o$ as functions $\denot{o}_\rho :
\Val \to \Val$, so that
\[
  \denot{o(e)}_\rho m \triangleq \denot{o}_\rho ( \denot{e}_\rho m ) .
\]
We fix a set $\DExpr$ of \emph{distribution expressions} to model primitive
distributions that our programs can sample from. For simplicity, we suppose for
now that each distribution expression $d$ is interpreted as a uniform
distribution over a finite set. So, we have the coin flip and uniform
distributions:
\[
  \DExpr \coloneqq \;\; \Flip \MID \Unif(\Expr)
\]
where $\Expr$ is a list, representing the space of samples.  We will introduce
other primitive distributions as needed. To interpret distribution expressions,
we define $\denot{d}_\rho : \Mem \to \Dist(\Val)$; for instance,
\[
  \denot{\Unif(e)}_\rho m \triangleq \cU( \denot{e}_\rho m )
\]
where $\cU(\cS)$ is the mathematical uniform distribution over a set $\cS$.

Now let's see the programming language.  We work with a standard imperative
language with random sampling. The programs, also called \emph{commands} or
\emph{statements}, are defined inductively:
\begin{align}
  \Cmd\qquad \coloneqq\qquad& \Skip \tag{no-op} \\
                          \MID{}& \Ass{\Var}{\Expr} \tag{assignment} \\
                          \MID{}& \Rand{\Var}{\DExpr} \tag{sampling} \\
                          \MID{}& \Seq{\Cmd}{\Cmd} \tag{sequencing} \\
                          \MID{}& \Cond{\Expr}{\Cmd}{\Cmd} \tag{conditional} \\
                          \MID{}& \WWhile{\Expr}{\Cmd} \tag{loop}
\end{align}
We assume throughout that programs are well-typed; for instance, the guard
expressions in conditionals and loops must be boolean.

We interpret each command as a mathematical function from states to
sub-distributions over output states; this function is known as the
\emph{semantics} of a command. Since the set of program variables and the set of
values are countable, the set of states is also countable so sub-distributions
over states are discrete. To interpret commands, we use two basic constructions
on sub-distributions.
\begin{definition} \label{def:unit-bind}
  The function $\dunit : \cA \to \SDist(\cA)$ maps every element $a \in \cA$ to
  the sub-distribution that places probability $1$ on $a$. The function $\dbind
  : \SDist(\cA) \times (\cA \to \SDist(\cB)) \to \SDist(\cB)$ is defined by
  \[
    \dbind(\mu, f)(b) \triangleq \sum_{a \in \cA} \mu(a) \cdot f(a)(b) .
  \]
  Intuitively, $\dbind$ applies a randomized function on a distribution over
  inputs.
\end{definition}
We use a discrete version of the semantics considered by \citet{Kozen81},
presented in \cref{fig:pwhile-sem}; we write $m[x \mapsto v]$ for the memory $m$
with variable $x$ updated to hold $v$, and $a \mapsto b(a)$ for the function
mapping $a$ to $b(a)$. The most complicated case is for loops. The
sub-distribution $\mu^{(i)}(m)$ models executions that exit after entering the
loop body at most $i$ times, starting from initial memory $m$. For the base case
$i = 0$, the sub-distribution either returns $m$ with probability $1$ when the
guard is false and the loop exits immediately, or returns the null
sub-distribution $\bot$ when the guard is true. The cases $i > 0$ are defined
recursively, by unrolling the loop.

Note that $\mu^{(i)}$ are increasing in $i$: $\mu^{(i)}(m) \leq \mu^{(j)}(m)$
for all $m \in \Mem$ and $i \leq j$. In particular, the weights of the
sub-distributions are increasing. Since the weights are at most $1$, the
approximants converge to a sub-distribution as $i$ tends to infinity by the
monotone convergence theorem (see, e.g., \citet[Theorem 11.28]{RudinPMA}, taking
the discrete (counting) measure over $\Mem$).

\begin{figure}
  \begin{align*}
    \denot{\Skip}_\rho m &\triangleq \dunit(m) \\
    \denot{\Ass{x}{e}}_\rho m &\triangleq \dunit(m[x \mapsto \denot{e}_\rho m]) \\
    \denot{\Rand{x}{d}}_\rho m
    &\triangleq \dbind( \denot{d}_\rho m, v \mapsto \dunit(m[x \mapsto v])) \\
    \denot{\Seq{c}{c'}}_\rho m
    &\triangleq \dbind( \denot{c}_\rho m, \denot{c'}_\rho ) \\
    \denot{\Cond{e}{c}{c'}}_\rho m &\triangleq
    \begin{cases}
      \denot{c}_\rho m &: \denot{e}_\rho m = \kwtrue \\
      \denot{c'}_\rho m &: \denot{e}_\rho m = \kwfalse
    \end{cases} \\
    \denot{\WWhile{e}{c}}_\rho m &\triangleq \lim_{i \to \infty} \mu^{(i)}(m) \\
    \mu^{(i)}(m) &\triangleq
    \begin{cases}
      \bot &: i = 0 \land \denot{e}_\rho m = \kwtrue \\
      \dunit(m) &: i = 0 \land \denot{e}_\rho m = \kwfalse \\
      \dbind(\denot{\Condt{e}{c}}_\rho m, \mu^{(i - 1)}) &: i > 0
    \end{cases}
  \end{align*}
  \caption{Semantics of programs} \label{fig:pwhile-sem}
\end{figure}

\subsection{The logic \Sprhl: judgments and validity}

The program logic \Sprhl features judgments of the following form:
\[
  \prhl{c_1}{c_2}{\Phi}{\Psi}
\]
Here, $c_1$ and $c_2$ are commands and $\Phi$ and $\Psi$ are predicates on pairs
of memories.  To describe the inputs and outputs of $c_1$ and $c_2$, each
predicate can mention two copies $x\sidel$, $x\sider$ of each program variable
$x$; these \emph{tagged} variables refer to the value of $x$ in the executions
of $c_1$ and $c_2$ respectively. 

\begin{definition}
  Let $\Var\sidel$ and $\Var\sider$ be the sets of \emph{tagged variables},
  finite sets of variable names tagged with $\sidel$ or $\sider$ respectively:
  \[
    \Var\sidel \triangleq \{ x\sidel \mid x \in \Var \}
    \quad \text{and} \quad
    \Var\sider \triangleq \{ x\sider \mid x \in \Var \} .
  \]
  Let $\Mem\sidel$ and $\Mem\sider$ be the sets of \emph{tagged memories}, maps
  from tagged variables to values:
  \[
    \Mem\sidel \triangleq \Var\sidel \to \Val
    \quad \text{and} \quad
    \Mem\sider \triangleq \Var\sider \to \Val .
  \]
  Let $\Mem_\times$ be the set of \emph{product memories}, which combine two
  tagged memories:
  \[
    \Mem_\times \triangleq \Var\sidel \uplus \Var\sider \to \Val .
  \]
  For notational convenience we identify $\Mem_\times$ with pairs of memories
  $\Mem\sidel \times \Mem\sider$; for $m_1 \in \Mem\sidel$ and $m_2 \in
  \Mem\sider$, we write $(m_1,m_2)$ for the product memory and we use the usual
  projections on pairs to extract untagged memories from the product memory:
  \[
    p_1(m_1, m_2) \triangleq |m_1|
    \quad \text{and} \quad
    p_2(m_1, m_2) \triangleq |m_2| ,
  \]
  where the memory $|m| \in \Mem$ has all variables in $\Var$.
  For commands $c$ and expressions $e$ with variables in $\Var$, we write
  $c\sidel$, $c\sider$ and $e\sidel$, $e\sider$ for the corresponding
  \emph{tagged commands} and \emph{tagged expressions} with variables in
  $\Var\sidel$ and $\Var\sider$. 
\end{definition}

We consider a set $\Pred$ of \emph{predicates} (\emph{assertions}) from
first-order logic defined by the following grammar:
\begin{align}
  \Pred\qquad \coloneqq\qquad& \Expr\tagged = \Expr\tagged
        \MID \Expr\tagged < \Expr\tagged
        \MID \Expr\tagged \in \Expr\tagged
        \notag \\
        \MID{}& \top \MID \bot
        \MID \cO(\Expr\tagged, \dots, \Expr\tagged)
        \tag{predicates} \\
        \MID{}& \Pred \land \Pred
        \MID \Pred \lor \Pred
        \MID \neg \Pred
        \MID \Pred \to \Pred
        \MID \forall \LVar \in \ZZ,\; \Pred
        \MID \exists \LVar \in \ZZ,\ \Pred
        \tag{first-order formulas}
\end{align}
We typically use capital Greek letters ($\Phi, \Psi, \Theta, \Xi, \dots$) for
predicates. $\Expr\tagged$ denotes an expression where program variables are
tagged with $\sidel$ or $\sider$; tags may be mixed within an expression.  We
consider the usual binary predicates $\{ =, <, \in, \dots \}$ where $e \in e'$
means $e$ is a member of the list $e'$, and we take the always-true and
always-false predicates $\top$ and $\bot$, and a set $\cO$ of other predicates.
Predicates can be combined using the usual connectives $\{ \land, \lor, \neg,
\to \}$ and can quantify over first-order types (e.g., the integers, tuples,
etc.). We will often interpret a boolean expression $e$ as the predicate $e =
\kwtrue$.

Predicates are interpreted as sets of product memories.
\begin{definition}
  Let $\Phi$ be a predicate. Given a logical context $\rho$, $\Phi$ is
  interpreted as a set $\denot{\Phi}_\rho \subseteq \Mem_\times$ in the expected
  way, e.g.,
  \[
    \denot{e_1\sidel < e_2\sider}_\rho
    \triangleq \{ (m_1, m_2) \in \Mem_\times \mid \denot{e_1}_\rho m_1 <
    \denot{e_2}_\rho m_2 \} .
  \]
\end{definition}
We can inject a predicate on single memories into a predicate on product
memories; we call the resulting predicate \emph{one-sided} since it constrains
just one of two memories.
\begin{definition}
  Let $\Phi$ be a predicate on $\Mem$. We define formulas $\Phi\sidel$ and
  $\Phi\sider$ by replacing all program variables $x$ in $\Phi$ with $x\sidel$
  and $x\sider$, respectively, and we define
  \[
    \denot{\Phi\sidel}_\rho \triangleq \{ (m_1, m_2) \mid m_1 \in \denot{\Phi}_\rho \}
    \quad \text{and} \quad
    \denot{\Phi\sider}_\rho \triangleq \{ (m_1, m_2) \mid m_2 \in \denot{\Phi}_\rho \} .
  \]
\end{definition}

Valid judgments in \Sprhl relate two output distributions by lifting the
post-condition.

\begin{definition}[\citet{BartheGZ09}]
  A judgment is \emph{valid} in logical context $\rho$, written $\rho \models
  \prhl{c_1}{c_2}{\Phi}{\Psi}$, if for any two memories $(m_1, m_2) \in
  \denot{\Phi}_\rho$ there exists a lifting of $\Psi$ relating the output
  distributions:
  \[
    \denot{c_1}_\rho m_1  \lift{\denot{\Psi}_\rho} \denot{c_2}_\rho m_2 .
  \]
\end{definition}

For example, a valid judgment
\[
  \models \prhl{c_1}{c_2}{\Phi}{(=)} ,
\]
states that for any two input memories $(m_1, m_2)$ satisfying $\Phi$, the
resulting output distributions from running $c_1$ and $c_2$ are related by
lifted equality; by \cref{fact:eq-couple}, these output distributions must be
equal.

\subsection{The logic \Sprhl: the proof rules}

The logic \Sprhl includes a collection of logical rules to inductively build up
a proof of a new judgment from known judgments. The rules are superficially
similar to those from standard Hoare logic. However, the interpretation of
judgments in terms of liftings means some rules in \Sprhl are not valid in Hoare
logic, and vice versa.

Before describing the rules, we introduce some necessary notation. A system of
logical rules inductively defines a set of \emph{derivable} formulas; we use the
head symbol $\vdash$ to mark such formulas. The premises in each logical rule
are written above the horizontal line, and the single conclusion is written
below the line; for easy reference, the name of each rule is given to the left
of the line.

The main premises are judgments in the program logic, but rules may also use
other \emph{side-conditions}. For instance, many rules require an assertion
logic formula to be valid in all memories. Other side-conditions state that a
program is terminating, or that certain variables are not modified by the
program.  We use the head symbol $\models$ to mark valid side-conditions; while
we could give a separate proof system for these premises, in practice they are
simple enough to check directly.

We also use notation for substitution in assertions. We write $\Phi\subst{x}{e}$
for the formula $\Phi$ with every occurrence of the variable $x$ replaced by
$e$. Similarly, $\Phi\subst{x_1\sidel, x_2\sider}{v_1, v_2}$ is the formula
$\Phi$ where occurrences of the tagged variables $x_1\sidel, x_2\sider$ are
replaced by $v_1, v_2$ respectively.

\medskip

The rules of \Sprhl can be divided into three groups: \emph{two-sided} rules,
\emph{one-sided} rules, and \emph{structural} rules.  All judgments are
parameterized by a logical context $\rho$, but since this context is assumed to
be a fixed assignment of logical variables---constant throughout the proof---we
omit it from the rules.  The two-sided rules in \cref{fig:prhl-two-sided} apply
when the two programs in the conclusion judgment have the same top-level shape.

\begin{figure}
  \begin{mathpar}
  \inferruleref{Skip}
  {~}
  { \vdash \prhl {\Skip}{\Skip} {\Phi}{\Phi} }
  \label{rule:prhl-skip}
  \\
  \inferruleref{Assn}
  {~}
  { \vdash \prhl
    {\Ass{x_1}{e_1}}{\Ass{x_2}{e_2}}
    {\Psi\subst{x_1\sidel,x_2\sider}{e_1\sidel,e_2\sider}}{\Psi} }
  \label{rule:prhl-assn}
  \\
  \inferruleref{Sample}
  { f : \supp(d_1) \to \supp(d_2)~\text{is a bijection} }
  { \vdash \prhl{\Rand{x_1}{d_1}}{\Rand{x_2}{d_2}}
    {\forall v \in \supp(d_1),\; \Psi\subst{x_1\sidel, x_2\sider}{v, f(v)}} {\Psi}}
  \label{rule:prhl-sample}
  \\
  \inferruleref{Seq}
  { \vdash \prhl{c_1}{c_2}{\Phi}{\Psi} \\ \vdash \prhl{c_1'}{c_2'}{\Psi}{\Theta} }
  { \vdash \prhl{c_1;c_1'}{c_2;c_2'}{\Phi}{\Theta} }
  \label{rule:prhl-seq}
  \\
  \inferruleref{Cond}
  { \models \Phi \to e_1\sidel = e_2\sider \\
    \vdash \prhl{c_1}{c_2}{\Phi \land e_1\sidel}{\Psi} \\
    \vdash \prhl{c_1'}{c_2'}{\Phi \land \neg e_1\sidel}{\Psi} }
  { \vdash \prhl{\Cond{e_1}{c_1}{c_1'}}{\Cond{e_2}{c_2}{c_2'}}{\Phi}{\Psi} }
  \label{rule:prhl-cond}
  \\
  \inferruleref{While}
  { \models \Phi \to e_1\sidel = e_2\sider \\
    \vdash \prhl {c_1}{c_2} {\Phi \land e_1\sidel} {\Phi} }
  { \vdash \prhl {\WWhile{e_1}{c_1}}{\WWhile{e_2}{c_2}} {\Phi} {\Phi \land \neg e_1\sidel} }
  \label{rule:prhl-while}
\end{mathpar}
  \caption{Two-sided \Sprhl rules \label{fig:prhl-two-sided}}
\end{figure}

The rule \nameref{rule:prhl-skip} simply states that $\Skip$ instructions
preserve the pre-condition. The rule \nameref{rule:prhl-assn} handles assignment
instructions. It is the usual Hoare-style rule: if $\Psi$ holds initially with
$e_1\sidel$ and $e_2\sider$ substituted for $x_1\sidel$ and $x_2\sider$, then
$\Psi$ holds after the respective assignment instructions.

The rule \nameref{rule:prhl-sample} is more subtle. In some ways it is the key
rule in \Sprhl, allowing us to select a coupling for a pair of sampling
instructions. To gain intuition, the following rule is a special case:
\[
  \inferrule*[Left=Sample*]
  { f : \supp(d) \to \supp(d)~\text{is a bijection} }
  { \vdash \prhl{\Rand{x}{d}}{\Rand{x}{d}} {\top} {f(x\sidel) = x\sider}}
\]
The conclusion states that there exists a coupling of a distribution $d$ with
itself such that each sample $x$ from $d$ is related to $f(x)$. Soundness of
this rule crucially relies on $d$ being \emph{uniform}---as we have seen, any
bijection $f$ induces a coupling of uniform distributions (cf.
\cref{ex:bij-couple}). It is possible to support general distributions at the
cost of a more complicated side-condition,\footnote{%
  Roughly speaking, the probability of any set $S$ under $d$ should be equal to
the probability of $f(S)$ under $d$.}
but we will not need this generality. The full rule \nameref{rule:prhl-sample}
can prove a post-condition of any shape: a post-condition holds after sampling
if it holds before sampling, where $x\sidel$ and $x\sider$ are replaced by any
two coupled samples $(v, f(v))$. 

The rule \nameref{rule:prhl-seq} resembles the normal rule for sequential
composition in Hoare logic, but its reading is more subtle. In particular, note
that the intermediate assertion $\Psi$ is interpreted differently in the two
premises: in the first judgment it is a post-condition and interpreted as a
relation between \emph{distributions over memories} via lifting, while in the
second judgment it is a pre-condition and interpreted as a relation between
\emph{memories}.

The next two rules deal with branching commands. Rule \nameref{rule:prhl-cond}
requires that the guards $e_1\sidel$ and $e_2\sider$ are equal assuming the
pre-condition $\Phi$. The rule is otherwise similar to the standard Hoare logic
rule: if we can prove the post-condition $\Psi$ when the guard is initially true
and when the guard is initially false, then we can prove $\Psi$ as a
post-condition of the conditional.

Rule \nameref{rule:prhl-while} uses a similar idea for loops. We again assume
that the guards are initially equal, and we also assume that they are equal in
the post-condition of the loop body. Since the judgments are interpreted in
terms of couplings, this second condition is a bit subtle. For one thing, the
rule does \emph{not} require $e_1\sidel = e_2\sider$ in all possible executions
of the two programs---this would be a rather severe restriction, for instance
ruling out programs where $e_1\sidel$ and $e_2\sider$ are probabilistic. Rather,
the guards only need to be equal \emph{under the coupling of the two programs
given by the premise}. The upshot is that by selecting appropriate couplings in
the loop body, we can assume the guards are equal when analyzing loops with
probabilistic guards. The rule is otherwise similar to the usual Hoare logic
rule, where $\Phi$ is the loop invariant.

\begin{figure}
  \begin{mathpar}
  \inferruleref{Assn-L}
  {~}
  {\vdash \prhl{\Ass{x_1}{e_1}}{\Skip}{\Psi\subst{x_1\sidel}{e_1\sidel}}{\Psi}}
  \label{rule:prhl-assn-l}
  \\
  \inferruleref{Assn-R}
  {~}
  {\vdash \prhl{\Skip}{\Ass{x_2}{e_2}}{\Psi\subst{x_2\sider}{e_2\sider}}{\Psi}}
  \label{rule:prhl-assn-r}
  \\
  \inferruleref{Sample-L}
  {~}
  {\vdash \prhl{\Rand{x_1}{d_1}}{\Skip}
               {\forall v \in \supp(d_1),\; \Psi\subst{x_1\sidel}{v}}{\Psi}}
  \label{rule:prhl-sample-l}
  \\
  \inferruleref{Sample-R}
  {~}
  {\vdash \prhl{\Skip}{\Rand{x_2}{d_2}}
               {\forall v \in \supp(d_2),\; \Psi\subst{x_2\sider}{v}}{\Psi}}
  \label{rule:prhl-sample-r}
  \\
  \inferruleref{Cond-L}
  {\vdash \prhl{c_1}{c}{\Phi \land e_1\sidel}{\Psi} \\
  \vdash \prhl{c_1'}{c}{\Phi \land \neg e_1\sidel}{\Psi}}
  { \vdash \prhl{\Cond{e_1}{c_1}{c_1'}}{c}{\Phi}{\Psi} }
  \label{rule:prhl-cond-l}
  \\
  \inferruleref{Cond-R}
  {\vdash \prhl{c}{c_2}{\Phi\land e_2\sider}{\Psi} \\
  \vdash \prhl{c}{c_2'}{\Phi \land \neg e_2\sider}{\Psi}}
  { \vdash \prhl{c}{\Cond{e_2}{c_2}{c_2'}}{\Phi}{\Psi} }
  \label{rule:prhl-cond-r}
  \\
  \inferruleref{While-L}{ 
    \vdash \prhl{c_1}{\Skip}{\Phi\land e_1\sidel}{\Phi} 
    \\\\
    \models \Phi \to \Phi_1\sidel \\
    \lless{\Phi_1}{\WWhile{e_1}{c_1}}}
  { \vdash \prhl{\WWhile{e_1}{c_1}}{\Skip}{\Phi}{\Phi \land \neg e_1\sidel} }
  \label{rule:prhl-while-l}
  \\
  \inferruleref{While-R}{ 
    \vdash \prhl{\Skip}{c_2}{\Phi \land e_2\sider}{\Phi} 
    \\\\
    \models \Phi \to \Phi_2\sider \\
    \lless{\Phi_2}{\WWhile{e_2}{c_2}}}
  { \vdash \prhl{\Skip}{\WWhile{e_2}{c_2}}{\Phi}{\Phi \land \neg e_2\sider} }
  \label{rule:prhl-while-r}
\end{mathpar}
  \caption{One-sided \Sprhl rules \label{fig:prhl-one-sided}}
\end{figure}

\medskip

So far, we have seen rules that relate two programs of the same shape. These are
the most commonly used rules in \Sprhl, as relational reasoning is most powerful
when comparing two highly similar (or even the same) programs. However, in some
cases we may need to reason about two programs with different shapes, even if
the two top-level commands are the same. For instance, if we can't guarantee two
executions of a program follow the same path at a conditional statement under a
coupling, we must relate the two different branches. For this kind of reasoning,
we can fall back on the \emph{one-sided} rules in \cref{fig:prhl-one-sided}.
These rules relate a command of a particular shape with $\Skip$ or an arbitrary
command. Each rule comes in a left- and a right-side version.

The assignment rules, \nameref{rule:prhl-assn-l} and \nameref{rule:prhl-assn-r},
relate an assignment instruction to $\Skip$ using the usual Hoare rule for
assignment instructions. The sampling rules, \nameref{rule:prhl-sample-l} and
\nameref{rule:prhl-sample-r}, are similar; they relate a sampling instruction to
$\Skip$ if the post-condition holds for all possible values of the sample. These
rules represent couplings where fresh randomness is used, i.e., where randomness
is not shared between the two programs.

The conditional rules, \nameref{rule:prhl-cond-l} and
\nameref{rule:prhl-cond-r}, are similar to the two-sided conditional rule except
there is no assumption of synchronized guards---the other command $c$ might not
even be a conditional. If we can relate the general command $c$ to the true
branch when the guard is true and relate $c$ to the false branch when the guard
is false, then we can relate $c$ to the whole conditional.

The rules for loops, \nameref{rule:prhl-while-l} and
\nameref{rule:prhl-while-r}, can only relate loops to the $\Skip$; a loop that
executes multiple iterations cannot be directly related to an arbitrary command
that executes only once. These rules mimic the usual loop rule from Hoare logic,
with a critical side-condition: \emph{losslessness}.

\begin{definition} \label{def:ll}
  A command $c$ is \emph{$\Phi$-lossless} if for any memory $m$ satisfying
  $\Phi$ and every logical context $\rho$, the output $\denot{c}_\rho m$ is a
  proper distribution (i.e., it has total probability $1$). We write
  $\Phi$-lossless as the following judgment:
  \[
    \lless{\Phi}{c}
  \]
\end{definition}

Losslessness is needed for soundness: $\Skip$ produces a proper distribution on
any input and liftings can only relate sub-distributions with equal weights
(\cref{fact:couple-wt}), so the loop must also produce a proper distribution to
have any hope of coupling the output distributions. For the examples we will
consider, losslessness is easy to show since loops execute for a finite number
of iterations; when there is no finite bound, proving losslessness may require
more sophisticated techniques (e.g.,
\citet{BEGGHS16,FerrerHermanns15,Chatterjee2016,Chatterjee:2016:AAQ:2837614.2837639,Chatterjee:2017:SIP:3009837.3009873,mciver2016new}).

\todo[inline]{BCP: OK. But why is it sufficient?}

\begin{figure}
  \begin{mathpar}
  \inferruleref{Conseq}
  { \vdash \prhl{c_1}{c_2}{\Phi'}{\Psi'} \\
    \models \Phi \to \Phi'  \\ \models \Psi' \to \Psi }
  { \vdash \prhl{c_1}{c_2}{\Phi}{\Psi} }
  \label{rule:prhl-conseq}
  \\
  \inferruleref{Equiv}
  {\vdash \prhl{c_1'}{c_2'}{\Phi}{\Psi} \\ c_1 \equiv c_1' \\ c_2 \equiv c_2'}
  {\vdash \prhl{c_1}{c_2}{\Phi}{\Psi}}
  \label{rule:prhl-equiv}
  \\
  \inferruleref{Case}{
    \vdash \prhl{c_1}{c_2}{\Phi \land \Theta}{\Psi} \\
    \vdash \prhl{c_1}{c_2}{\Phi \land \neg \Theta}{\Psi} }
  {\vdash \prhl{c_1}{c_2}{\Phi}{\Psi}}
  \label{rule:prhl-case}
  \\
  \inferruleref{Trans}{
    \vdash \prhl{c_1}{c_2}{\Phi}{\Psi} \\
    \vdash \prhl{c_2}{c_3}{\Phi'}{\Psi'} }
  {\vdash \prhl{c_1}{c_3}{\Phi' \circ \Phi}{\Psi' \circ \Psi}}
  \label{rule:prhl-trans}
  \\
  \inferruleref{Frame}
  { \vdash \prhl{c_1}{c_2}{\Phi}{\Psi} \\
    \FV(\Theta) \cap \MV(c_1, c_2) = \varnothing }
  { \vdash \prhl{c_1}{c_2}{\Phi \land \Theta}{\Psi \land \Theta} }
  \label{rule:prhl-frame}
\end{mathpar}
  \caption{Structural \Sprhl rules \label{fig:prhl-structural}}
\end{figure}

\medskip

Finally, \Sprhl includes a handful of \emph{structural} rules which apply to
programs of any shape. The first rule \nameref{rule:prhl-conseq} is the usual
rule of consequence, allowing us to strengthen the pre-condition and weaken the
post-condition---assuming more about the input and proving less about the
output, respectively.

The rule \nameref{rule:prhl-equiv} replaces programs by equivalent programs.
This rule is particularly useful for reasoning about programs of different
shapes. Instead of using one-sided rules, which are often less convenient, we
can sometimes replace a program with an equivalent version and then apply
two-sided rules. For simplicity, we use a strong notion of equivalence:
\[
  c_1 \equiv c_2 \triangleq \denot{c_1}_\rho = \denot{c_2}_\rho
\]
for every logical context $\rho$; more refined notions of equivalence are also
possible, but will not be needed for our purposes.  For our examples, we just
use a handful of basic program equivalences, e.g., $c ; \Skip \equiv c$ and
$\Skip ; c \equiv c$.

The rule \nameref{rule:prhl-case} performs a case analysis on the input.  If we
can prove a judgment when $\Theta$ holds initially and a judgment when $\Theta$
does not hold initially, then we can combine the two judgments provided they
have the same post-condition. 

The rule \nameref{rule:prhl-trans} is the transitivity rule: given a judgment
relating $c_1 \sim c_2$ and a judgment relating $c_2 \sim c_3$, we can glue
these judgments together to relate $c_1 \sim c_3$. The pre- and post-conditions
of the conclusion are given by composing the pre- and post-conditions of the
premises; for binary relations $\cR$ and $\cS$,
relation composition is defined by
\[
  \cR \circ \cS \triangleq \{ (x_1, x_3) \mid \exists x_2.\; (x_1, x_2) \in \cS \land (x_2, x_3) \in \cR \}
  .
\]

The last rule \nameref{rule:prhl-frame} is the frame rule (also called the
\emph{rule of constancy}): it states that an assertion $\Theta$ can be carried
from the pre-condition through to the post-condition as long as the variables
$\MV(c_1, c_2)$ that may be modified by the programs $c_1$ and $c_2$ don't
include any of the variables $\FV(\Theta)$ appearing free in $\Theta$; as usual,
$\MV$ and $\FV$ are defined syntactically by collecting the variables that occur
in programs and assertions.

As expected, the proof system of \Sprhl is sound.

\begin{theorem}[\citet{BartheGZ09}]
  Let $\rho$ be a logical context. If a judgment is derivable
  \[
    \rho \vdash \prhl{c_1}{c_2}{\Phi}{\Psi} ,
  \]
  then it is valid:
  \[
    \rho \models \prhl{c_1}{c_2}{\Phi}{\Psi} .
  \]
\end{theorem}

\subsection{The coupling interpretation}

A valid judgment $\rho \models \prhl{c_1}{c_2}{\Phi}{\Psi}$ implies that for any
two input memories related by $\Phi$, there exists a coupling with support in
$\Psi$ between the two output distributions. By applying the results in
\cref{sec:coupling-conseq}, valid judgments imply relational properties of
programs.

Moreover, by viewing the rules as the discrete steps in a proof, we can identify
common pieces of standard coupling proofs. For instance,
\nameref{rule:prhl-sample} selects a coupling for corresponding sampling
statements; the function $f$ lets us choose among different bijection couplings.
The rule \nameref{rule:prhl-seq} encodes a composition principle for couplings;
when two processes produce samples related by $\Psi$ under a particular
coupling, we can continue to assume this relation when analyzing the remainder
of the program.  The structural rule \nameref{rule:prhl-case} shows we can
select between two possible couplings depending on whether a predicate $\Theta$
holds. In short, not only is \Sprhl a logic for verifying cryptographic
protocols, it is also a formal logic for proofs by coupling.

\section{Constructing couplings, formally} \label{sec:prhl-ex}

Now let's see how to construct coupling proofs in the logic. We give three
examples proving classical probabilistic properties: equivalence,
stochastic domination, and convergence.

\begin{remark}
  There are some inherent challenges in presenting formal proofs on paper.
  Fundamentally, our proofs are branching derivation trees. When such a proof is
  serialized, it may be hard to follow which part of the derivation tree the
  paper proof corresponds to. To help organize the proof, we proceed loosely in
  a top-down fashion, giving proofs and judgments for the most deeply nested
  parts of the program first and then gradually zooming out to consider larger
  and larger parts of the whole program.

  Applications of sequential composition are also natural places to signpost the
  proof; we typically consider the commands in order, unless the second command
  is much more complex than the first. Finally, for space reasons we will gloss
  over applications of the assignment rule \nameref{rule:prhl-assn} and minor
  uses of the rule of consequence \nameref{rule:prhl-conseq}; a completely
  formal proof would also spell out these details.
\end{remark}

\subsection{Probabilistic equivalence}

To warm up, we prove two programs probabilistically equivalent.  Our example
models perhaps the most basic encryption scheme: the XOR cipher.  Given a
boolean $s$ representing the secret message, the XOR cipher flips a fair coin to
draw the secret key $k$ and then returns $k \oplus s$ as the encrypted message.
A receiving party who knows the secret key can decrypt the message by computing
$k \oplus (k \oplus s) = s$.

To prove secrecy of this scheme, we consider the following two programs:
\[
\begin{array}{l}
  \Rand{k}{\Flip}; \\
  \Ass{r}{k \oplus s}
\end{array}
\qquad\qquad\qquad
\vrule
\qquad\qquad\qquad
\begin{array}{l}
  \Rand{k}{\Flip}; \\
  \Ass{r}{k}
\end{array}
\]
The first program $\mathit{xor}_1$ implements the encryption function, storing
the encrypted message into $r$. The second program $\mathit{xor}_2$ simply
stores a random value into $r$. If we can show the distribution of $r$ is the
same in both programs, then the XOR cipher is secure: the distribution on
outputs is completely random, leaking no information about the secret message
$s$. In terms of \Sprhl, it suffices to prove the following judgment:
\[
  \vdash \prhl{\mathit{xor}_1}{\mathit{xor}_2}{\top}{r\sidel = r\sider}
\]
By validity of the logic, this judgment implies that for any two memories $m_1,
m_2$, the output distributions are related by a coupling that always returns
outputs with equal values of $r$; by reasoning similar to \cref{fact:eq-couple},
this implies that the output distributions over $r\sidel$ and $r\sider$ are
equal.\footnote{%
  To be completely precise, \cref{fact:eq-couple} assumes that we have lifted
equality, while here we only have a lifting where the variables $r$ are equal.
An analogous argument shows that the marginal distributions of variable $r$ must
be equal.}

Before proving this judgment in the logic, we sketch the proof by coupling. If
$s\sidel$ is true, then we couple $k$ to take opposite values in the two runs.
If $s\sidel$ is false, then we couple $k$ to be equal in the two runs. In both
cases, we conclude that the results $r\sidel, r\sider$ are equal under the
coupling.

To formalize this argument in \Sprhl, we use the \nameref{rule:prhl-case} rule:
\[
  \inferrule*[Left=Case]{
    \vdash \prhl{\mathit{xor}_1}{\mathit{xor}_2}{s\sidel = \kwtrue}{r\sidel = r\sider}
    \\\\
    \vdash \prhl{\mathit{xor}_1}{\mathit{xor}_2}{s\sidel \neq \kwtrue}{r\sidel = r\sider} }
  {\vdash \prhl{\mathit{xor}_1}{\mathit{xor}_2}{\top}{r\sidel = r\sider}} .
\]
For the first premise we select the negation coupling using the bijection $f =
\neg$ in \nameref{rule:prhl-sample}, apply the assignment rule
\nameref{rule:prhl-assn}, and combine with the sequencing rule
\nameref{rule:prhl-seq}. Concretely, we have
\[
    \inferrule*[Left=Sample]
    { f = \neg }
    { \vdash \prhl{\Rand{k}{\Flip}}{\Rand{k}{\Flip}}
      {s\sidel = \kwtrue}
      {k\sidel = \neg k\sider \land s\sidel = \kwtrue} }
\]
\[
    \inferrule*[Left=Assn]
    {~}
    { \vdash \prhl{\Ass{r}{k \oplus s}}{\Ass{r}{k}}
      {k\sidel = \neg k\sider \land s\sidel = \kwtrue}{r\sidel = r\sider} }
  \]
and we combine the two judgments to give:
\[
  \inferrule*[Left=Seq]
  { \vdash \prhl{\Rand{k}{\Flip}}{\Rand{k}{\Flip}}
    {s\sidel = \kwtrue}
    {k\sidel = \neg k\sider \land s \sidel = \kwtrue}
    \\
    \vdash \prhl{\Ass{r}{k \oplus s}}{\Ass{r}{k}}
    {k\sidel = \neg k\sider \land s\sidel = \kwtrue}{r\sidel = r\sider} }
  { \vdash \prhl{\mathit{xor}_1}{\mathit{xor}_2}{s\sidel = \kwtrue}{r\sidel = r\sider} }
  .
\]
For the other case $s\sidel \neq \kwtrue$, we give the same proof except
with the identity coupling in \nameref{rule:prhl-sample}:
\[
  \inferrule*[Left=Sample]
  { f = \id }
  { \vdash \prhl{\Rand{k}{\Flip}}{\Rand{k}{\Flip}}
    {s\sidel \neq \kwtrue}
    {k\sidel = k\sider \land s\sidel \neq \kwtrue} }
\]
and the assignment rule, we have
\[
  \inferrule*[Left=Assn]
  {~}
  { \vdash \prhl{\Ass{r}{k \oplus s}}{\Ass{r}{k}}
    {k\sidel = k\sider \land s\sidel \neq \kwtrue}{r\sidel = r\sider} } .
\]
Combining the conclusions, we get
\[
  \inferrule*[Left=Seq]
  { \vdash \prhl{\Rand{k}{\Flip}}{\Rand{k}{\Flip}}
    {s\sidel \neq \kwtrue}
    {k\sidel = \neg k\sider \land s\sidel \neq \kwtrue} \\
    \vdash \prhl{\Ass{r}{k \oplus s}}{\Ass{r}{k}}
    {k\sidel = k\sider \land s\sidel \neq \kwtrue}{r\sidel = r\sider} }
  { \vdash \prhl{\mathit{xor}_1}{\mathit{xor}_2}{s\sidel \neq \kwtrue}{r\sidel = r\sider} }
  .
\]
By \nameref{rule:prhl-case}, we conclude the desired post-condition $r\sidel =
r\sider$.

\subsection{Stochastic domination}

For our second example, we revisit \cref{ex:coupling-pf} and replicate the proof
in \Sprhl. The following program $\mathit{sdom}$ flips a coin $T$ times and
returns the number of coin flips that come up true:
\newcommand{\SDcode}{%
\[
  \begin{array}{l}
    \Ass{i}{0};
    \Ass{\mathit{ct}}{0}; \\
    \WWhile{i < T}{} \\
    \quad \Ass{i}{i + 1}; \\
    \quad \Rand{s}{\Flip}; \\
    \quad \Ass{\mathit{ct}}{\Tern{s}{\mathit{ct} + 1}{\mathit{ct}}}
  \end{array}
\]}
\SDcode
(The last line uses the \emph{ternary conditional operator}---$\Tern{s}{\mathit{ct} +
1}{\mathit{ct}}$ is equal to $\mathit{ct} + 1$ if $s$ is true, otherwise equal
to $\mathit{ct}$.)

We consider two runs of this program executing $T_1$ and $T_2$ iterations, where
$T_1 \leq T_2$ are logical variables; call the two programs $\mathit{sdom}_1$
and $\mathit{sdom}_2$. By soundness of the logic and \cref{fact:sd-couple}, the
distribution of $\mathit{ct}$ in the second run stochastically dominates the
distribution of $\mathit{ct}$ in the first run if we can prove the judgment
\[
  \vdash \prhl{\mathit{sdom}_1}{\mathit{sdom}_2}{\top}{\mathit{ct}\sidel \leq \mathit{ct}\sider} .
\]
Encoding the argument from \cref{ex:coupling-pf} in \Sprhl requires a bit of
work. The main obstacle is that the two-sided loop rule in \Sprhl can only
analyze loops in a synchronized fashion, but this is not possible here: when
$T_1 < T_2$ the two loops run for different numbers of iterations, no matter how
we couple the samples. To get around this problem, we use the equivalence rule
\nameref{rule:prhl-equiv} to transform $\mathit{sdom}$ into a more convenient
form using the following equivalence:
\[
  \WWhile{e}{c} \equiv \WWhile{e \land e'}{c}; \WWhile{e}{c}
\]
This transformation, known in the compilers literature as \emph{loop splitting}
\citep{Callahan1988}, separates out the first iterations where $e'$ holds, and
then runs the original loop to completion. We transform $\mathit{sdom}_2$ as
follows:
\[
  \begin{array}{llrr}
  \mathit{sdom}'_{2a}
  &\triangleq
  \left\{
  \begin{array}{l}
    \Ass{i}{0}; \Ass{\mathit{ct}}{0}; \\
    \WWhile{i < T_2 \land i < T_1}{} \\
    \quad \Ass{i}{i + 1}; \\
    \quad \Rand{s}{\Flip}; \\
    \quad \Ass{\mathit{ct}}{\Tern{s}{\mathit{ct} + 1}{\mathit{ct}}} ;
  \end{array}
  \right.
  &
  \left.
  \begin{array}{l}
    \Ass{i}{0}; \Ass{\mathit{ct}}{0}; \\
    \WWhile{i < T_1}{} \\
    \quad \Ass{i}{i + 1}; \\
    \quad \Rand{s}{\Flip}; \\
    \quad \Ass{\mathit{ct}}{\Tern{s}{\mathit{ct} + 1}{\mathit{ct}}} ;
  \end{array}
  \right\}
  &\triangleq
  \mathit{sdom}_1
  \\
  \mathit{sdom}'_{2b}
  &\triangleq
  \left\{
  \begin{array}{l}
    \WWhile{i < T_2}{} \\
    \quad \Ass{i}{i + 1}; \\
    \quad \Rand{s}{\Flip}; \\
    \quad \Ass{\mathit{ct}}{\Tern{s}{\mathit{ct} + 1}{\mathit{ct}}}
  \end{array}
  \right.
  &
\end{array}
\]
We aim to relate $\mathit{sdom}_{2a}'; \mathit{sdom}_{2b}'$ to
$\mathit{sdom}_1$.  First, we apply the two-sided rule \nameref{rule:prhl-while}
to relate $\mathit{sdom}_1$ to $\mathit{sdom}_{2a}'$. Taking the identity
coupling with $f = \id$ in \nameref{rule:prhl-sample}, we relate the sampling in
the loop body via
\[
  \inferrule*[Left=Sample]
  { f = \id }
  { \vdash \prhl{\Rand{s}{\Flip}}{\Rand{s}{\Flip}} {\top} {s\sidel = s\sider}}
\]
and establish the loop invariant
\[
  \Theta \triangleq i\sidel = i\sider \land \mathit{ct}\sidel = \mathit{ct}\sider ,
\]
proving the judgment
\[
  \vdash \prhl{\mathit{sdom}_1}{\mathit{sdom}_{2a}'}{\top}{\Theta} .
\]
Then we use the one-sided rule \nameref{rule:prhl-while-r} for the loop
$\mathit{sdom}_{2b}'$ with loop invariant $\mathit{ct}\sidel \leq
\mathit{ct}\sider$:
\[
  \vdash \prhl{\Skip}{\mathit{sdom}_{2b}'}{\Theta}{\mathit{ct}\sidel \leq
  \mathit{ct}\sider} .
\]
Composing these two judgments with \nameref{rule:prhl-seq} and applying
\nameref{rule:prhl-equiv} gives the desired judgment:
\[
  \inferrule*[Left=Equiv]
  { \vdash \prhl{\mathit{sdom}_1 ; \Skip}{\mathit{sdom}_{2a}'; \mathit{sdom}_{2b}'}
  {\top}{\mathit{ct}\sidel \leq \mathit{ct}\sider} }
  { \vdash \prhl{\mathit{sdom}_1}{\mathit{sdom}_2}{\top}{\mathit{ct}\sidel \leq
  \mathit{ct}\sider} }
\]
using the equivalence $\mathit{sdom}_1; \Skip \equiv \mathit{sdom}_1$. 

\subsection{Probabilistic convergence}

In our final example, we build a coupling witnessing convergence of two
\emph{random walks}. Each process begins at an integer starting point
$\mathit{start}$, and proceeds for $T$ steps. At each step it flips a fair coin.
If true, it increases the current position by $1$; otherwise, it decreases the
position by $1$. Given two random walks starting at different initial locations,
we want to bound the distance between the two resulting output distributions in
terms of $T$. Intuitively, the position distributions spread out as the random
walks proceed, tending towards the uniform distribution on the even integers or
the uniform distribution over the odd integers depending on the parity of the
initial position and the number of steps. If two walks initially have the same
parity (i.e., their starting positions differ by an even integer), then their
distributions after taking the same number of steps $T$ should approach one
another in total variation distance.

We model a single random walk with the following program $\mathit{rwalk}$:
\newcommand{\RWcode}{%
\[
  \begin{array}{l}
    \Ass{\mathit{pos}}{\mathit{start}}; \Ass{i}{0}; \Ass{\mathit{hist}}{[\mathit{start}]}; \\
    \WWhile{i < T}{} \\
    \quad \Ass{i}{i + 1}; \\
    \quad \Rand{r}{\Flip}; \\
    \quad \Ass{\mathit{pos}}{\mathit{pos} + (\Tern{r}{1}{-1})}; \\
    \quad \Ass{\mathit{hist}}{\mathit{pos} :: \mathit{hist}}
  \end{array}
\]}
\RWcode
The last command records the history of the walk in $\mathit{hist}$; this
\emph{ghost variable} does not affect the final output value, but will be
useful for our assertions.

By \cref{thm:coupling-method}, we can bound the TV-distance between the position
distributions by constructing a coupling where the probability of
$\mathit{pos}\sidel \neq \mathit{pos}\sider$ tends to $0$ as $T$ increases. We
don't have the tools yet to reason about this probability (we will revisit this
point in the next chapter), but for now we can build the coupling and prove the
judgment
\[
  \vdash \prhl{\mathit{rwalk}}{\mathit{rwalk}}
  {\mathit{start}\sider - \mathit{start}\sidel = 2K}
  {K + \mathit{start}\sidel \in \mathit{hist}\sidel \to \mathit{pos}\sidel = \mathit{pos}\sider}
\]
where $K$ is an integer logical variable. The pre-condition states that the
initial positions are an even distance apart. To read the post-condition, the
predicate $K + \mathit{start}\sidel \in \mathit{hist}\sidel$ holds if and only
if the first walk has moved to position $K + \mathit{start}\sidel$ at some time
in the past; if this has happened, then the two coupled positions must be equal.

Our coupling mirrors the two walks. Each step, we have the walks make symmetric
moves by arranging opposite samples. Once the walks meet, we have the walks
match each other by coupling the samples to be equal. In this way, if the first
walk reaches $\mathit{start}\sidel + K$, then the second walk must be at
$\mathit{start}\sider - K$ since both walks are coupled to move symmetrically. In this
case, the initial condition $\mathit{start}\sider - \mathit{start}\sidel = 2K$ gives
\[
  \mathit{pos}\sidel = \mathit{start}\sidel + K = \mathit{start}\sider - K = \mathit{pos}\sider
\]
so the walks meet and continue to share the same position thereafter. This
argument requires the starting positions to be an even distance apart so the
positions in the two walks always have the same parity; if the two starting
positions are an odd distance apart, then the two distributions after $T$ steps
have disjoint support and the coupled walks can never meet.

To formalize this argument in \Sprhl, we handle the loop with the two-sided rule
\nameref{rule:prhl-while} and invariant
\[
  \Theta \triangleq \begin{cases}
    |\mathit{hist}\sidel| > 0 \land |\mathit{hist}\sider| > 0 \\
    K + \mathit{start}\sidel \in \mathit{hist}\sidel \to \mathit{pos}\sidel = \mathit{pos}\sider  \\
    K + \mathit{start}\sidel \notin \mathit{hist}\sidel
    \to \mathit{pos}\sider - \mathit{pos}\sidel = 2 (K - (hd(\mathit{hist}\sidel) - \mathit{start}\sidel)) ,
  \end{cases}
\]
where $hd(\mathit{hist})$ is the first element (the \emph{head}) of the
non-empty list $\mathit{hist}$. The last two conditions model the two cases.  If
the first walk has already visited $K + \mathit{start}\sidel$, the walks have
already met under the coupling and they must have the same position.  Otherwise,
the walks have not met. If $d \triangleq hd(\mathit{hist}\sidel) -
\mathit{start}\sidel$ is the (signed) distance the first walk has moved away
from its starting location and the two walks are initially $2K$ apart, then the
current distance between coupled positions must be $2(K - d)$.

To show the invariant is preserved, we perform a case analysis with
\nameref{rule:prhl-case}. If $K + \mathit{start}\sidel \in \mathit{hist}\sidel$
holds then the walks have already met in the past and currently have the same
position (by $\Theta$). So, we select the identity coupling in
\nameref{rule:prhl-sample}:
\[
    \inferrule*[Left=Sample]
    { f = \id }
    { \vdash \prhl{\Rand{r}{\Flip}}{\Rand{r}{\Flip}}
      {K + \mathit{start}\sidel \in \mathit{hist}\sidel}{r\sidel = r\sider} } .
\]
Since $K + \mathit{start}\sidel \in \mathit{hist}\sidel \to \mathit{pos}\sidel =
\mathit{pos}\sider$ holds at the start of the loop, we know $\mathit{pos}\sidel
= \mathit{pos}\sider$ at the end of the loop; since $K + \mathit{start}\sidel \in
\mathit{hist}\sidel$ is preserved by the loop body, the invariant $\Theta$
holds.

Otherwise if $K + \mathit{start}\sidel \notin h \sidel$, then the walks have not
yet met and should be mirrored. So, we select the negation coupling with $f =
\neg$ in \nameref{rule:prhl-sample}:
\[
    \inferrule*[Left=Sample]
    { f = \neg }
    { \vdash \prhl{\Rand{r}{\Flip}}{\Rand{r}{\Flip}}
      {K + \mathit{start}\sidel \notin \mathit{hist}\sidel}{\neg r\sidel = r\sider} }
\]
To show the loop invariant, there are two cases. If $K + \mathit{start}\sidel
\in h \sidel$ holds after the body, the two walks have just met for the first
time and $\mathit{pos}\sidel = \mathit{pos}\sider$ holds.  Otherwise, the walks
remain mirrored: $\mathit{pos}\sidel$ increased by $r\sidel$ and
$\mathit{pos}\sider$ decreased by $r\sidel$, so $\mathit{pos}\sider -
\mathit{pos}\sidel = 2(K + (hd(\mathit{hist}\sidel) - \mathit{start}\sidel))$
and the invariant $\Theta$ is preserved.

Putting it all together, we have the desired judgment:
\[
  \vdash \prhl{\mathit{rwalk}}{\mathit{rwalk}}
  {\mathit{start}\sider - \mathit{start}\sidel = 2K}
  {K + \mathit{start}\sidel \in h \sidel \to \mathit{pos}\sidel = \mathit{pos}\sider}
  .
\]
While this judgment describes a coupling between the position distributions, we
need to analyze finer properties of the coupling distribution to apply
\cref{thm:coupling-method}---namely, we must bound the probability that
$\mathit{pos}\sidel$ is not equal to $\mathit{pos}\sider$. We will consider how
to extract this information in the next chapter.

\section{Related work} \label{sec:prhl-rw}

Relational Hoare logics and probabilistic couplings have been extensively
studied in disparate research communities.

\subsection{Relational Hoare logics}

The logic \Sprhl is a prime example of a \emph{relational} program logic, which
extend standard Floyd-Hoare logics to prove properties about two programs.
\citet{Benton04} first designed a relational version of Hoare logic called
\sname{RHL} to prove equivalence between two (deterministic) programs. Benton
used his logic to verify compiler transformations, showing the original program
is equivalent to the transformed program. Relational versions of other program
logics have also been considered, including an extension of separation logic by
\citet{YANG2007308} to prove relational properties of pointer-manipulating
programs. There is nothing particularly special about relating exactly
\emph{two} programs; recently, \citet{sousa2016cartesian} give a Hoare logic for
proving properties of $k$ executions of the same program for arbitrary $k$.

\citet{BartheGZ09} extended Benton's work to prove relational properties of
probabilistic programs, leading to the logic \Sprhl. As we have seen, the key
technical insight is to interpret the relational post-condition as a
probabilistic lifting between two output distributions.  \citet{BartheGZ09} used
\Sprhl to verify security properties for a variety of cryptographic protocols by
mimicking the so-called \emph{game-hopping} proof technique
\citep{Shoup:2004,Bellare:2006}, where the original program is transformed
step-by-step to an obviously secure version (e.g., a program returning a random
number). Security follows if each transformation approximately preserves the
program semantics. Our analysis of the XOR cipher is a very simple example of
this technique; more sophisticated proofs chain together dozens of
transformations.

\subsection{Probabilistic couplings and liftings}

Couplings are a well-studied tool in probability theory; readers can consult the
lecture notes by \citet{Lindvall02} or the textbooks by \citet{Thorisson00} and
\citet{LevinPW09} for entry points into this vast literature.

Probabilistic liftings were initially proposed in research on
\emph{bisimulation}, techniques for proving equivalence of transition systems.
\citet{LarsenS89} were the first to consider a probabilistic notion of
bisimulation. Roughly speaking, their definition considers an equivalence
relation $E$ on states and requires that any two states in the same equivalence
class have the same probability of stepping to any other equivalence class.  The
construction for arbitrary relations arose soon after, when researchers
generalized probabilistic bisimulation to probabilistic simulation;
\citet[Definition 4.3]{jonsson1991specification} proposes a \emph{satisfaction
relation} using witness distributions, similar to the definition used in \Sprhl.
\citet[Definition 3.6.2]{desharnais1999labelled} and \citet[Definition
12]{DBLP:journals/njc/SegalaL95} give an alternative characterization without
witness distributions, similar to Strassen's theorem
\citep{strassen1965existence}; \citet[Theorem 7.3.4]{desharnais1999labelled}
observed that both definitions are equivalent in the finite case via the max
flow-min cut theorem. Probabilistic (bi)simulation can be characterized
logically, i.e., two systems are (bi)similar if and only if they satisfy the
same formulas in some modal logic
\citep{LarsenS89,DESHARNAIS2002163,DESHARNAIS2003160,fijalkow_et_al:LIPIcs:2017:7368}.
\citet{DengD11} survey logical, metric, and algorithmic characterizations of
these relations.

Probabilistic liftings have proven to be a convenient abstraction for many
styles of formal reasoning beyond bisimulation and program logics.  For
instance, \citet{rfstar} combine probabilistic lifting with a probability monad
to prove relational properties in \sname{RF}$^\star$, a refinement type system
for a probabilistic, functional language.

\chapter{From coupling proofs to product programs} \label{chap:products}

As we have seen, valid judgments in \Sprhl imply a coupling of two output
distributions with a particular support. Some applications of proof by coupling
need more detailed information to conclude a relational property; notable
examples include coupling proofs for convergence, like the random walk example
from the previous chapter. While a valid judgment gives no further information
beyond the support of the coupling, we usually have more information at
hand---often, we have a proof using the logical rules in \Sprhl. Since proof
rules correspond to steps in proofs by coupling, which indirectly construct a
coupling distribution, the structure of \Sprhl proofs should somehow encode the
coupling.

Indeed, this is the case. While we cannot hope to explicitly list the
probabilities of every pair under a coupling---for one thing, there may be
infinitely many---we show that every \Sprhl derivation encodes a probabilistic
program \emph{generating} the witness. Intuitively, a coupling proof describes
how to simulate two probabilistic processes as one, by sharing randomness.
Accordingly, proofs in \Sprhl encode how to combine two programs into one; the
witness of a coupling is just the output distribution of the combined program.
This construction, which we call the \emph{coupled product}, draws a
correspondence between coupling proofs and probabilistic product programs,
recalling a theme in computer science and logic: proofs can be viewed as
programs.

To make our ideas concrete, we design an extension of \Sprhl called \Sxprhl
(\emph{product \Sprhl}), where judgments construct a coupled product program.
Since this program depends on the whole proof derivation and not just the final
judgment, there may be multiple \Sxprhl judgments corresponding to a given
\Sprhl judgment. We first present a core version of \Sxprhl with logical rules
based on \Sprhl (\cref{sec:xprhl-core}), followed by a novel loop rule that
allows asynchronous reasoning (\cref{sec:xprhl-async}).  After establishing
soundness (\cref{sec:xprhl-sound}), we apply our logic to prove convergence and
rapid mixing for probabilistic processes (\cref{sec:xprhl-ex-standard}),
modeling examples of shift couplings (\cref{sec:xprhl-ex-shift}) and path
couplings (\cref{sec:xprhl-ex-path}).  Finally, we compare the coupled product
to prior constructions (\cref{sec:products-rw}).

\section{The core logic \texorpdfstring{\Sxprhl}{x\Sprhl}} \label{sec:xprhl-core}

The logic \Sxprhl extends \Sprhl by pairing each judgment with a product
program.

\subsection{Judgments and validity}

Judgments in \Sxprhl have the following form:
\[
  \xprhl{c_1}{c_2}{\Phi}{\Psi}{c_\times}
\]
Just like in \Sprhl, $c_1$ and $c_2$ are probabilistic programs and the pre- and
post-conditions $\Phi$ and $\Psi$ are assertions on product memories. The new
component is the \emph{coupled product} $c_\times$, which simulates two
correlated executions of $c_1$ and $c_2$. To ensure the two executions do not
interfere with one another, $c_\times$ operates on a product memory with two
copies of each variable, tagged with $\sidel$ and $\sider$.

Semantic validity in \Sxprhl is very similar to validity in \Sprhl: the output
distribution of the product program on two related inputs couples the output
distributions of the two given programs.
\begin{definition} \label{def:xprhl-valid}
  Suppose $c_1, c_2$ have variables in $\Var \cup \LVar$, $\Phi$ and $\Psi$ are
  predicates over $\Var\sidel \cup \Var\sider \cup \LVar$, and $c_\times$ has
  variables in $\Var\sidel \cup \Var\sider \cup \LVar$. An \Sxprhl judgment is
  \emph{valid} in a logical context $\rho$, written
  \[
    \rho \models \xprhl{c_1}{c_2}{\Phi}{\Psi}{c_\times} ,
  \]
  if for every two memories $(m_1, m_2) \in \denot{\Phi}_\rho$ we have
  \begin{enumerate}
    \item $\supp( \denot{c_\times}_\rho (m_1, m_2) ) \subseteq \denot{\Psi}_\rho$;
    \item $\denot{c_1}_\rho m_1 = \pi_1(\denot{c_\times}_\rho (m_1, m_2))$; and
    \item $\denot{c_2}_\rho m_2 = \pi_2(\denot{c_\times}_\rho (m_1, m_2))$.
  \end{enumerate}
  (Recall $\pi_1, \pi_2$ are the first and second projections from
  $\SDist(\Mem_\times)$ to $\SDist(\Mem)$.)
\end{definition}

\subsection{Core proof rules}

Proof rules in \Sxprhl describe how to construct product programs. Like their
\Sprhl counterparts, the core rules of \Sxprhl can be divided into three groups:
two-sided rules, one-sided rules, and structural rules.

\begin{figure}
  \begin{mathpar}
    \inferruleref{Skip}
    {~}
    { \vdash \xprhl{\Skip}{\Skip}{\Phi}{\Phi}{\Skip} }
    \label{rule:xprhl-skip}
  \\
    \inferruleref{Assn}
    {~}
    { \vdash \xprhl
      {\Ass{x_1}{e_1}}{\Ass{x_2}{e_2}}
      {\Psi\subst{x_1\sidel, x_2\sider}{e_1\sidel, e_2\sider}}{\Psi}
      {
        \begin{array}{l}
          \Ass{x_1\sidel}{e_1\sidel}; \\
          \Ass{x_2\sider}{e_2\sider}
        \end{array}
      }
    }
    \label{rule:xprhl-assn}
    \\
    \inferruleref{Sample}
    {  f : \supp(d_1) \to \supp(d_2)~\text{bijection} }
    { \vdash \xprhl
      {\Rand{x_1}{d_1}}{\Rand{x_2}{d_2}}
      {\forall v \in \supp(d_1),\; \Psi\subst{x_1\sidel, x_2\sider}{v, f(v)}}{\Psi}
      {
        \begin{array}{l}
          \Rand{x_1\sidel}{d_1}; \\
          \Ass{x_2\sider}{f(x_1\sidel)}
        \end{array}
      }
    }
    \label{rule:xprhl-sample}
    \\
    \inferruleref{Seq}
    { \vdash \xprhl{c_1}{c_2}{\Phi}{\Psi}{c} \\
      \vdash \xprhl{c'_1}{c'_2}{\Psi}{\Theta}{c'} }
    { \vdash \xprhl{c_1;c'_1}{c_2;c'_2}{\Phi}{\Theta}{c;c'} }
    \label{rule:xprhl-seq}
    \\
    \inferruleref{Cond}
    { \models \Phi \to e_1\sidel = e_2\sider
      \\\\
      \vdash \xprhl{c_1}{c_2}{\Phi \land e_1\sidel}{\Psi}{c} \\
      \vdash \xprhl{c_1'}{c_2'}{\Phi \land \neg e_1\sidel}{\Psi}{c'}}
    { \vdash \xprhl
      {\Cond{e_1}{c_1}{c_1'}}{\Cond{e_2}{c_2}{c_2'}}
      {\Phi}{\Psi}{\Cond{e_1\sidel}{c}{c'}} }
    \label{rule:xprhl-cond}
    \\
    \inferruleref{While}
    { \models \Phi \to e_1\sidel = e_2\sider \\
      \vdash \xprhl{c_1}{c_2}{\Phi \land e_1\sidel}{\Phi}{c} }
    { \vdash \xprhl
      {\WWhile{e_1}{c_1}}{\WWhile{e_2}{c_2}}
      {\Phi}{\Phi \land \neg e_1\sidel}{\WWhile{e_1\sidel}{c}} } 
    \label{rule:xprhl-while}
  \end{mathpar}
  \caption{Two-sided \Sxprhl rules} \label{fig:xprhl-two-sided}
\end{figure}

\medskip

The two-sided rules are presented in \cref{fig:xprhl-two-sided}. For the first
rule \nameref{rule:xprhl-skip}, since the two programs don't have any effect,
the coupled program also has no effect. The next pair of rules handle assignment
and sampling statements.  The rule \nameref{rule:xprhl-assn} relates two
assignment statements; the product program simply performs both operations on
the product memory. The rule \nameref{rule:xprhl-sample} for random sampling is
more interesting. Just like its counterpart in \Sprhl, this rule is
parameterized by a bijection $f$ between the supports of the two distributions.
The product program draws the first sample for $x_1\sidel$ from $d_1$ and then
assigns $x_2\sider$ \emph{deterministically} with $f(x_1\sidel)$---this is the
sample corresponding to $x_1\sidel$ under the coupling. In this way, the product
program simulates two random draws with a single source of randomness.

The sequential composition rule \nameref{rule:xprhl-seq} relates two sequencing
commands. The product program is simply the sequential composition of the
product programs for the first and second commands, highlighting the
compositional nature of couplings.

The final pair of rules relate branching commands. Just like in \Sprhl, the
pre-condition must ensure that the guards are equal. In the rule
\nameref{rule:xprhl-cond}, the premises give two product programs $c$ and $c'$
relating the two true branches and the two false branches, respectively.  The
product program for the conditional first branches on the guard and then
executes the product program for the corresponding branch. In the rule
\nameref{rule:xprhl-while}, the product program for the loop executes the
product program for the body while the guard remains true.

\begin{figure}
  \begin{mathpar}
    \inferruleref{Assn-L}
    {~}
    { \vdash \xprhl
      {\Ass{x_1}{e_1}}{\Skip}
      {\Psi\subst{x_1\sidel}{e_1\sidel}}{\Psi}
      {\Ass{x_1\sidel}{e_1\sidel}} }
    \label{rule:xprhl-assn-l}
    \\
    \inferruleref{Assn-R}
    {~}
    { \vdash \xprhl
      {\Skip}{\Ass{x_2}{e_2}}
      {\Psi\subst{x_2\sider}{e_2\sider}}{\Psi}
      {\Ass{x_2\sider}{e_2\sider}} }
    \label{rule:xprhl-assn-r}
    \\
    \inferruleref{Sample-L}
    {~}
    { \vdash \xprhl
      {\Rand{x_1}{d_1}}{\Skip}
      {\forall v \in \supp(d_1),\; \Psi\subst{x_1\sidel}{v}}{\Psi}{\Rand{x_1\sidel}{d_1}} }
    \label{rule:xprhl-sample-l}
    \\
    \inferruleref{Sample-R}
    {~}
    { \vdash \xprhl
      {\Skip}{\Rand{x_2}{d_2}}
      {\forall v \in \supp(d_2),\; \Psi\subst{x_2\sider}{v}}{\Psi}{\Rand{x_2\sider}{d_2}} }
    \label{rule:xprhl-sample-r}
    \\
    \inferruleref{Cond-L}
    { \vdash \xprhl{c_1}{c_2}{\Phi \land e_1\sidel}{\Psi}{c} \\
      \vdash \xprhl{c_1'}{c_2}{\Phi \land \neg e_1\sidel}{\Psi}{c'}}
    { \vdash \xprhl{\Cond{e_1}{c_1}{c_1'}}{c_2}{\Phi}{\Psi}{\Cond{e_1\sidel}{c}{c'}} }
    \label{rule:xprhl-cond-l}
    \\
    \inferruleref{Cond-R}
    { \vdash \xprhl{c_1}{c_2}{\Phi \land e_2\sider}{\Psi}{c} \\
      \vdash \xprhl{c_1}{c_2'}{\Phi \land \neg e_2\sider}{\Psi}{c'}}
    { \vdash \xprhl{c_1}{\Cond{e_2}{c_2}{c_2'}}{\Phi}{\Psi}{\Cond{e_2\sider}{c}{c'}} }
    \label{rule:xprhl-cond-r}
    \\
    \inferruleref{While-L}
    { \vdash \xprhl{c_1}{\Skip}{\Phi \land e_1\sidel}{\Phi}{c} \\
      \models \Phi \to \Phi_1\sidel \\
      \lless{\Phi_1}{\WWhile{e_1}{c_1}}}
    { \vdash \xprhl{\WWhile{e_1}{c_1}}{\Skip}
      {\Phi}{\Phi\land\neg e_1\sidel}{\WWhile{e_1\sidel}{c}} }
    \label{rule:xprhl-while-l}
    \\
    \inferruleref{While-R}
    { \vdash \xprhl{\Skip}{c_2}{\Phi \land e_2\sider}{\Phi}{c} \\
      \models \Phi \to \Phi_2\sidel \\
      \lless{\Phi_2}{\WWhile{e_2}{c_2}}}
    { \vdash \xprhl{\Skip}{\WWhile{e_2}{c_2}}
      {\Phi}{\Phi\land\neg e_2\sider}{\WWhile{e_2\sider}{c}} }
    \label{rule:xprhl-while-r}
  \end{mathpar}
  \caption{One-sided \Sxprhl rules}\label{fig:xprhl-one-sided}
\end{figure}

\medskip

Next we consider the one-sided proof rules in \cref{fig:xprhl-one-sided}. The
first four rules for assignment and sampling,
\nameref{rule:xprhl-assn-l}/\nameref{rule:xprhl-assn-r} and
\nameref{rule:xprhl-sample-l}/\nameref{rule:xprhl-sample-r}, relate a command
with $\Skip$; the product program simply executes the assignment or sampling
command on the indicated side.

The one-sided rules for conditionals, \nameref{rule:xprhl-cond-l} and
\nameref{rule:xprhl-cond-r}, relate a conditional to an arbitrary command ($c_2$
and $c_1$, respectively). The premises give two product programs relating the
general command with the true and false branches of the conditional. The coupled
product branches on the guard---$e_1\sidel$ or $e_2\sider$---and runs the
product program for the corresponding branch.

The one-sided rules for loops, \nameref{rule:xprhl-while-l} and
\nameref{rule:xprhl-while-r}, are similar.  The premises give a product program
relating the body of the loop to $\Skip$; the resulting product program for the
loop executes the product program for the body while the loop guard is true.
Like the analogous rules in \Sprhl, the loop must be lossless.

\begin{figure}
  \begin{mathpar}
    \inferruleref{Conseq}
    { \vdash \xprhl{c_1}{c_2}{\Phi'}{\Psi'}{c} \\
      \models \Phi \to \Phi' \\
      \models \Psi' \to \Psi }
    { \vdash \xprhl{c_1}{c_2}{\Phi}{\Psi}{c} }
    \label{rule:xprhl-conseq}
    \\
    \inferruleref{Equiv}
    { \vdash \xprhl{c_1'}{c_2'}{\Phi}{\Psi}{c'} \\
      c_1 \equiv c_1' \\
      c_2 \equiv c_2' \\
      c \equiv c' }
    { \vdash \xprhl{c_1}{c_2}{\Phi}{\Psi}{c} }
    \label{rule:xprhl-equiv}
    \\
    \inferruleref{Case}
    { \vdash \xprhl{c_1}{c_2}{\Phi\land e}{\Psi}{c} \\
      \vdash \xprhl{c_1}{c_2}{\Phi\land \neg e}{\Psi}{c'} }
    { \vdash \xprhl{c_1}{c_2}{\Phi}{\Psi}{\Cond{e}{c}{c'}}}
    \label{rule:xprhl-case}
    \\
    \inferruleref{Frame}
    { \vdash \xprhl{c_1}{c_2}{\Phi}{\Psi}{c} \\
      \FV(\Theta) \cap \MV(c) = \varnothing }
    { \vdash \xprhl{c_1}{c_2}{\Phi \land \Theta}{\Psi \land \Theta}{c} }
    \label{rule:xprhl-frame}
  \end{mathpar}
  \caption{Structural \Sxprhl rules} \label{fig:xprhl-structural}
\end{figure}

\medskip

Finally, we come to the structural rules in \cref{fig:xprhl-structural}. The
rules \nameref{rule:xprhl-conseq} and \nameref{rule:xprhl-equiv} are
straightforward: the former rule preserves the product program of the premise,
while the latter rule replaces programs by equivalent programs. The rule
\nameref{rule:xprhl-case} is more interesting; recall that this rule performs a
case analysis on the two input memories. The product programs from the two
logical cases are combined into a final product program that branches on the
predicate and selects the corresponding product program. In this way, a logical
case analysis is realized by a branching statement in the product program. Unlike in
\Sprhl, this rule performs a case analysis on an \emph{expression} $e$ instead
of a general predicate $\Theta$ in the assertion logic; this restriction is
needed to reflect the predicate as a guard expression in the product.\footnote{%
  For instance, there is no boolean expression corresponding to universal or
  existential quantification; such an expression would typically not be
computable.}
Finally, \nameref{rule:xprhl-frame} is the \Sxprhl version of the frame rule.


\begin{remark}
  The careful reader may notice that we do not give an analogous rule for the
  transitivity rule \nameref{rule:prhl-trans} from \Sprhl. Given two product
  programs for the premises, it is not clear how to construct a product program
  for the conclusion; intuitively, we want to somehow interleave the product
  programs together while carefully aligning samples. Finding a \Sxprhl version
  of this rule is an interesting open problem.
\end{remark}

\section{An asynchronous loop rule} \label{sec:xprhl-async}

\begin{figure}
  \begin{mathpar}
    \inferruleref{While-Gen}
    { \models \Phi \to (e_1\sidel \lor e_2\sider) = e \\
      \models \Phi \land e \to p_0 \oplus p_1 \oplus p_2
      \\\\
      \models \Phi \land p_0 \land e \to e_1\sidel = e_2\sider \\
      \models \Phi \land p_1 \land e \to e_1\sidel \land \Phi_1\sidel \\
      \models \Phi \land p_2 \land e \to e_2\sider \land \Phi_2\sider
      \\\\
      \lless{\Phi_1}{\WWhile{e_1 \land p_1}{c_1}} \\
      \lless{\Phi_2}{\WWhile{e_2 \land p_2}{c_2}}
      \\\\
      \vdash \xprhl{(\Condt{e_1}{c_1})^{K_1}}{(\Condt{e_2}{c_2})^{K_2}}
      {\Phi \land e \land p_0}{\Phi}{c_0'}
      \quad \text{with} \quad
      K_1 > 0, K_2 > 0
      \\\\
      \vdash \xprhl{c_1}{\Skip}{\Phi \land e_1\land p_1}{\Phi}{c_1'} \\
      \vdash \xprhl{\Skip}{c_2}{\Phi \land e_2\land p_2}{\Phi}{c_2'} }
    { \vdash \xprhl{\WWhile{e_1}{c_1}}{\WWhile{e_2}{c_2}}
      {\Phi}{\Phi \land \neg e_1\sidel \land \neg e_2\sider}
      {
        \begin{array}{l}
          \WWhile{e}{\Condt{p_0}{c_0'}} \\
          \mathbf{else}~\Cond{p_1}{c_1'}{c_2'}
        \end{array}
      }
    }
    \label{rule:xprhl-while-gen}
  \end{mathpar}
\caption{Asynchronous loop rule \nameref{rule:xprhl-while-gen} for \Sxprhl}\label{fig:xprhl-while}
\end{figure}

The logic \Sxprhl inherits two kinds of loop rules from \Sprhl.  The two-sided
rule relates two loops by relating their bodies, a useful principle since the
loop bodies are often highly similar. However, this rule requires that the two
loops remain synchronized under the coupling. The one-sided loop rules don't
require synchronization, but they are significantly weaker---they can only
relate a loop to the trivial program $\Skip$. Taking a slightly broader view,
each rule captures one way of analyzing loops: (i) relating a block of
iterations in the first with a block of iterations in the second; (ii) relating
one iteration in the first with no iterations in the second; and (iii) relating
one iteration in the second with no iterations in the first.

To support all three kinds of reasoning, we give a new rule
\nameref{rule:xprhl-while-gen} in \cref{fig:xprhl-while}. The three analyses can
be freely intermixed, resulting in a powerful principle for analyzing loops
\emph{asynchronously}. We will step through the premises from top to bottom,
starting with the side-conditions.  First, we specify an expression $e$ in the
product memory that is true if either loop guard is true.  Then we specify three
boolean flags $p_0, p_1, p_2$ indicating which of the three cases to apply;
exactly one of the flags must be true. The second group of premises ensure the
flags and the loop guards are consistent: if $p_0$ is true, then both guards
should be true since we are relating iterations from both loops; if $p_1$ is
true, then the first guard $e_1$ should be true since we want to relate one
iteration in the first loop; if $p_2$ is true, then the second guard $e_2$
should be true to relate one iteration in the second loop.  The
remaining side-conditions guarantee the product programs for the one-sided cases
terminate with probability $1$; these conditions are needed for soundness.
(Intuitively, the one-sided cases can effectively couple $\Skip$ to a loop. This
kind of coupling requires losslessness, as we saw in the one-sided loop rules
and in \cref{fact:couple-wt}.)

The main \Sxprhl premises handle the three cases. We write $c^K$ with a constant
$K$ for
\[
  c^K \triangleq \underbrace{c \mathbin{;} \cdots \mathbin{;} c}_{K \text{ iterations}}
\]
The first \Sxprhl premise handles the first case: $p_0$ is true so we relate
$K_1$ iterations of the first loop with $K_2$ iterations of the second loop,
skipping iterations if either loop terminates early. The second and third
\Sxprhl premises handle the second and third cases: $p_1$ or $p_2$ is true, and
we relate one iteration of the first or second side to $\Skip$. In the
conclusion, the product program interleaves the two original loops depending on
the case---it branches on $p_0, p_1, p_2$, and runs the product program from the
corresponding premise.

While we introduce \nameref{rule:xprhl-while-gen} for \Sxprhl, simply dropping
the product programs recovers a sound loop rule for \Sprhl. Some proofs that
previously required reasoning outside of the program logic, for instance using
program equivalences, can be handled with the extended loop rule. For example,
consider the stochastic domination example we first saw in \cref{ex:coupling-pf}
with the program $\mathit{sdom}$:
\SDcode
and recall we considered two versions of this program, $\mathit{sdom}_1$ and
$\mathit{sdom}_2$, where the number of iterations was $T_1$ and $T_2$
respectively with $T_1 \leq T_2$. When we previously proved the judgment
\[
  \vdash \prhl{\mathit{sdom}_1}{\mathit{sdom}_2}{\top}{\mathit{ct}\sidel \leq \mathit{ct}\sider} ,
\]
showing stochastic domination, we crucially used the program equivalence rule
\nameref{rule:xprhl-equiv} to split the loop in $\mathit{sdom}_2$ into two
pieces, using the two-sided rule \nameref{rule:xprhl-while} to analyze the first
piece and the one-sided rule \nameref{rule:xprhl-while-r} to analyze the second
piece. The \Sprhl version of the general rule \nameref{rule:xprhl-while-gen}
subsumes both loop rules, allowing us to freely switch between two-sided and
one-sided reasoning. As a result, we can prove the desired judgment without
transforming the programs by using \nameref{rule:xprhl-while-gen}, with
parameters
\begin{align*}
  K_1, K_2 &\triangleq 1 \\
  p_0 &\triangleq i\sidel < T_1 \\
  p_1 &\triangleq \kwfalse \\
  p_2 &\triangleq T_1 \leq i\sider < T_2 \\
  \Phi &\triangleq (i\sidel < T_1 \to i\sidel = i\sider)
  \land \mathit{ct}\sidel \leq \mathit{ct}\sider .
\end{align*}
When the first guard  $p_0$ is true, both loops have not terminated and we can
analyze the bodies synchronously. The second guard $p_1$ is always false since
we never want to skip iterations on the second side, while the third guard $p_2$
is true once the first program has terminated---in this case, we advance the
second program alone. We take the couplings from before: the identity coupling
in \nameref{rule:xprhl-sample} when $p_0$ is true, and the one-sided rule
\nameref{rule:xprhl-sample-r} when $p_2$ is true.

\section{Soundness of the logic} \label{sec:xprhl-sound}

The full proof system of \Sxprhl is sound.

\begin{restatable}[Soundness of \Sxprhl]{theorem}{xprhlsound} \label{thm:xprhl-sound}
  Let $\rho$ be a logical context. If a judgment is derivable
  \[
    \rho \vdash \xprhl{c_1}{c_2}{\Phi}{\Psi}{c_\times} ,
  \]
  then it is valid:
  \[
    \rho \models \xprhl{c_1}{c_2}{\Phi}{\Psi}{c_\times} .
  \]
\end{restatable}
\begin{proof}[Proof sketch]
  By induction on the derivation, performing a case analysis on the final rule.
  Most of the cases are straightforward. The most complex case, by far, handles
  the asynchronous rule \nameref{rule:xprhl-while-gen}.  While we can derive the
  other loop rules (the two-sided rule \nameref{rule:xprhl-while} and the
  one-sided rules \nameref{rule:xprhl-while-l}/\nameref{rule:xprhl-while-r})
  from \nameref{rule:xprhl-while-gen} and some basic program equivalences, we
  consider the simpler loop rules as separate cases to decompose the proof for
  \nameref{rule:xprhl-while-gen} as much as possible.  We present the detailed
  proof in \cref{app:sound-xprhl}.
\end{proof}

The natural counterpart to soundness is \emph{completeness}: valid judgments
should be derivable in the proof system. It is possible to show \Sxprhl is
complete in a certain sense for deterministic programs,\footnote{%
  More formally, \emph{relatively complete} for terminating programs given basic
equivalences like $c \equiv c ; \Skip$.}
but currently very little is known about probabilistic programs. We return to
this point in \cref{chap:future}.

\section{Proving probabilistic convergence} \label{sec:xprhl-ex-standard}

The coupled product generates the coupling in a \Sxprhl judgment. By analyzing
the product program, we can bound the probability of specific events in the
coupling distribution to prove quantitative probabilistic relational properties.
To demonstrate, we construct couplings in \Sxprhl for proving convergence bounds
for probabilistic processes, using standard coupling arguments and more advanced
variants like shift coupling and path coupling. In each case, we first build the
coupling as an \Sxprhl judgment and then analyze the coupled product.

Our main goal in this section is to demonstrate the product construction and to
show how it mirrors the corresponding informal proof by coupling.  While
constructing the coupling and generating the coupled product are easily handled
by \Sxprhl, formally reasoning about the product program is more difficult. The
target properties are probabilistic and non-relational, beyond the reach of
\Sxprhl. To keep the exposition as light as possible, we will sketch our proofs
about the coupled product in a standard mathematical style instead of
introducing a separate formal system (e.g., \sname{PPDL}~\citep{Kozen:1985} or
\sname{pGCL} \citep{Morgan:1996}).  General-purpose theorem provers (such as
\sname{Coq} or \sname{Agda}) should also be able to prove the required
properties after formalizing enough of probability theory, but such an approach
would be quite heavy.  Developing more lightweight, easier-to-use techniques for
probabilistic non-relational properties remains a significant open challenge.

We begin by revisiting the simple random walk program $\mathit{rwalk}$ from
\cref{sec:prhl-ex}:
\RWcode
Previously, we proved the following judgment in \Sprhl:
\[
  \vdash \prhl{\mathit{rwalk}}{\mathit{rwalk}}
  {\mathit{start}\sider - \mathit{start}\sidel = 2K}{K + \mathit{start}\sidel \in \mathit{hist}\sidel \to \mathit{pos}\sidel = \mathit{pos}\sider} .
\]
The two walks are initially $2K$ apart and the predicate $K +
\mathit{start}\sidel \in \mathit{hist}\sidel$ is true exactly when the walks
have met under the coupling. Replaying the proof using the corresponding \Sxprhl
rules yields
\begin{equation} \label{eq:xprhl-rw}
  \vdash \xprhl{\mathit{rwalk}}{\mathit{rwalk}}
  {\mathit{start}\sider - \mathit{start}\sidel = 2K}{K + \mathit{start}\sidel \in \mathit{hist}\sidel \to \mathit{pos}\sidel = \mathit{pos}\sider}
  {\mathit{rwalk}_\times} ,
\end{equation}
where $\mathit{rwalk}_\times$ is the following product program:
\[
  \begin{array}{l}
    \Ass{\mathit{pos}\sidel}{\mathit{start}\sidel};
    \Ass{\mathit{pos}\sider}{\mathit{start}\sider}; \\
    \Ass{i\sidel}{0}; \Ass{i\sider}{0}; \\
    \Ass{\mathit{hist}\sidel}{[\mathit{start}\sidel]};
    \Ass{\mathit{hist}\sider}{[\mathit{start}\sider]}; \\
    \WWhile{i\sidel < T}{} \\
    \quad \Ass{i\sidel}{i\sidel + 1}; \Ass{i\sider}{i\sider + 1}; \\
    \quad \Condt{\mathit{pos}\sidel = \mathit{pos}\sider}{} \\
    \quad\quad \Rand{r\sidel}{\Flip};
    \Ass{r\sider}{r\sidel}; \\
    \quad\quad \Ass{\mathit{pos}\sidel}{\mathit{pos}\sidel + (\Tern{r\sidel}{1}{-1})}; \\
    \quad\quad \Ass{\mathit{pos}\sider}{\mathit{pos}\sider + (\Tern{r\sider}{1}{-1})}; \\
    \quad\quad \Ass{\mathit{hist}\sidel}{\mathit{pos}\sidel :: \mathit{hist}\sidel};
    \Ass{\mathit{hist}\sider}{\mathit{pos}\sider :: \mathit{hist}\sider} \\
    \quad \mathbf{else} \\
    \quad\quad \Rand{r\sidel}{\Flip};
    \Ass{r\sider}{\neg r\sidel}; \\
    \quad\quad \Ass{\mathit{pos}\sidel}{\mathit{pos}\sidel + (\Tern{r\sidel}{1}{-1})}; \\
    \quad\quad \Ass{\mathit{pos}\sider}{\mathit{pos}\sider + (\Tern{r\sider}{1}{-1})}; \\
    \quad\quad \Ass{\mathit{hist}\sidel}{\mathit{pos}\sidel :: \mathit{hist}\sidel};
    \Ass{\mathit{hist}\sider}{\mathit{pos}\sider :: \mathit{hist}\sider}
  \end{array}
\]
The structure of the coupled product reflects the coupling proof. For instance,
the loop is introduced by the two-sided rule \nameref{rule:xprhl-while}, and the
conditional statement is introduced by the case analysis
\nameref{rule:xprhl-case}. Intuitively, this program simulates two coupled
random walks. Each iteration, the program branches on whether the positions of
the two walks are equal or not, setting the two samples $r\sidel$ and $r\sider$
to be equal if so, and opposite if not. Thus the positions $\mathit{pos}\sidel$
and $\mathit{pos}\sider$ trace out two mirrored walks when the positions are
different, and a single walk once the positions coincide.  

Now, we can bound the distance between the position distributions in the
original walks by bounding the probability of $K + \mathit{start}\sidel \notin
\mathit{hist}\sidel$ in $\mathit{rwalk}_\times$. We need a basic result from the
theory of random walks.

\begin{theorem}[see, e.g., {\citet[Theorem 2.17]{LevinPW09}}] \label{thm:rw}
  Let $X_0, X_1, \dots$ be the positions of a simple random walk on the integers
  starting at $X_0 = q \in \ZZ$. The probability the walk does not reach
  $0$ within $t$ steps is at most
  \[
    \Pr [ X_0, \dots, X_t \neq 0 ] \leq \frac{12 q}{\sqrt{t}} .
  \]
\end{theorem}

Now we bound the rate of convergence of two random walks.

\begin{theorem} \label{thm:rw-converge}
  Let $m_1, m_2$ be two memories such that $m_2(\mathit{start}) -
  m_1(\mathit{start}) = 2K$ for $K \in
  \ZZ$. Let $\mu_1, \mu_2$ be the final distributions over memories:
  \[
    \mu_1 \triangleq \denot{\mathit{rwalk}} m_1
    \quad \text{and} \quad
    \mu_2 \triangleq \denot{\mathit{rwalk}} m_2 .
  \]
  Let $\eta_1, \eta_2$ be the final distributions over positions:
  \[
    \eta_1 \triangleq \liftf{\denot{\mathit{pos}}}(\mu_1)
    \quad \text{and} \quad
    \eta_2 \triangleq \liftf{\denot{\mathit{pos}}} (\mu_2) .
  \]
  Then the distance between the two output distributions over positions is at
  most
  \[
    \tvdist{\eta_1}{\eta_2} \leq \frac{12 K}{\sqrt{T}} .
  \]
\end{theorem}
\begin{proof}
  The basic idea is to analyze the coupled product in the \Sxprhl judgment
  \cref{eq:xprhl-rw} and then apply the coupling method
  (\cref{thm:coupling-method}), but we need to handle one detail before we can
  string these results together. The coupling method requires a coupling such
  that the two samples are \emph{equal} with high probability, but the coupling
  from the post-condition of \cref{eq:xprhl-rw} only describes when the two
  \emph{positions} are equal---the coupling is a distribution over pairs of
  whole memories, which may be different even if the positions are equal.

  To address this issue, let $\mu_\times$ be the witness in \cref{eq:xprhl-rw}
  generated by the coupled product and let $\eta_\times$ be the projection to
  the positions:
  \[
    \mu_\times \triangleq \denot{\mathit{rwalk}_\times} (m_1, m_2)
    \quad \text{and} \quad
    \eta_\times \triangleq \liftf{\denot{(\mathit{pos}\sidel, \mathit{pos}\sider)}} (\mu_\times) .
  \]
  We directly calculate
  \begin{align*}
    \Pr_{(p_1, p_2) \sim \eta_\times} [ p_1 \neq p_2 ]
    &= \Pr_{(m_1, m_2) \sim \mu_\times} [ m_1(\mathit{pos}\sidel) \neq m_2(\mathit{pos}\sider) ] \\
    &\leq \Pr_{(m_1, m_2) \sim \mu_\times}
    [ (m_1, m_2) \in \denot{K + \mathit{start}\sidel \notin \mathit{hist}\sidel} ] ,
  \end{align*}
  where the inequality follows by the post-condition in \cref{eq:xprhl-rw}:
  pairs of memories satisfying $K + \mathit{start}\sidel \in
  \mathit{hist}\sidel$ must have equal positions.

  So, it suffices to upper bound the probability of $K + \mathit{start}\sidel
  \notin \mathit{hist}\sidel$. Looking at the coupled product
  $\mathit{rwalk}_\times$, as long as the two walks have not met, the
  \emph{distance} between the two coupled walks behaves like a single random
  walk: increasing by $2$ with probability $1/2$, decreasing by $2$ with
  probability $1/2$. This derived random walk starts at $\mathit{start}\sider -
  \mathit{start}\sidel = 2K$; if it reaches $0$ before $T$ steps, then the two
  original walks meet and $K + \mathit{start}\sidel \in \mathit{hist}\sidel$
  holds. Accordingly, \cref{thm:rw} gives
  \[
    \Pr_{(m_1, m_2) \sim \mu_\times} [ (m_1, m_2) \in \denot{K + \mathit{start}\sidel \notin \mathit{hist}\sidel}  ]
    \leq \frac{12 K}{\sqrt{T}}
  \]
  so we can conclude
  \[
    \tvdist{\eta_1}{\eta_2}
    = \tvdist{\pi_1(\eta_\times)}{\pi_2(\eta_\times)}
    \leq \Pr_{(p_1, p_2) \sim \eta_\times} [ p_1 \neq p_2 ]
    \leq \frac{12 K}{\sqrt{T}} ,
  \]
  where the first inequality follows by the coupling method
  (\cref{thm:coupling-method}).
\end{proof}

Hence, the distributions approach one another as the number of timesteps $T$
increases.

\section{Shift couplings} \label{sec:xprhl-ex-shift}

In the previous example, we were able to construct the coupling
\emph{synchronously} because the two coupled walks meet at the same iteration.
This may not be the case in more complex proofs. To demonstrate, we consider an
example of a \emph{shift coupling}---a coupling where the two processes meet at
two random timesteps. To construct this kind of coupling, we cannot use the
synchronous rule \nameref{rule:xprhl-while} since we may need to relate samples
across different iterations. Instead, we will apply our asynchronous rule
\nameref{rule:xprhl-while-gen}.

Our example is called the \emph{Dynkin process}.\footnote{%
  The name comes from a magic trick, known as \emph{Dynkin's card trick} or
\emph{Kruskal's count}.}
This process maintains a position $\mathit{pos} \in \NN$, initialized to
$\mathit{start} \in [0, \dots, 10]$. Each step, it draws a uniformly random
number $r$ from $[1, \dots, 10]$ and increments the position by $r$. The process
stops as soon as $\mathit{pos}$ exceeds $T \in \NN$, returning the final value
as the output. The following code implements the Dynkin process:
\[
  \begin{array}{l}
    \Ass{\mathit{pos}}{\mathit{start}}; \\
    \Ass{\mathit{hist}}{[\mathit{start}]}; \\
    \WWhile{\mathit{pos} < T}{} \\
    \quad \Rand{r}{\Unif([1, \dots, 10])}; \\
    \quad \Ass{\mathit{pos}}{\mathit{pos} + r}; \\
    \quad \Ass{\mathit{hist}}{\mathit{pos} :: \mathit{hist}}
  \end{array}
\]
We call this program $\mathit{dynkin}$ and we write $\mathit{dynbody}$ for the
loop body. We use a ghost variable $\mathit{hist}$ to keep track of the history
of visited positions, just like we did for the random walk. We will
analyze two executions of $\mathit{dynkin}$ starting at different locations and
show the distributions over final positions converge as $T$ increases.

Before seeing the proof in \Sxprhl, let's first sketch the coupling argument. If
the two processes have the same position, then we couple the samplings to return
equal values; this keeps the two positions equal. Otherwise, we sample in the
process that is behind, temporarily pausing the leading process. Since the
sampled process moves at least one step forward in each iteration, the lagging
process will overtake (or land on) the leading process in finitely many steps,
when we will switch to one of the other cases.

We perform this reasoning in \Sxprhl using \nameref{rule:xprhl-while-gen} with
$K_1 = K_2 = 1$. We take the joint guard

\[
  e \triangleq (\mathit{pos}\sidel < T) \lor (\mathit{pos}\sider < T) ,
\]
and flags
\[
  p_0 \triangleq \mathit{pos}\sidel = \mathit{pos}\sider
  \quad \text{and} \quad
  p_1 \triangleq \mathit{pos}\sidel < \mathit{pos}\sider
  \quad \text{and} \quad
  p_2 \triangleq \mathit{pos}\sidel > \mathit{pos}\sider .
\]
These cases are clearly mutually exclusive, and one is always true.
Furthermore, they satisfy the necessary consistency requirements: $\models p_1
\land e \to (\mathit{pos}\sidel < T)$ and $\models p_2 \land e \to
(\mathit{pos}\sider < T)$ both hold. Finally, the loops are clearly lossless:
the position increases by at least $1$ every iteration, so we are in any case
for at most $T$ iterations.

With the side-conditions out of the way, we now turn to the main premises.  We
take the following invariant:
\[
  \Theta \triangleq \begin{cases}
    |\mathit{hist}\sidel| > 0 \land |\mathit{hist}\sider| > 0 \\
    \mathit{hist}\sidel \cap \mathit{hist}\sider \neq \varnothing \to \mathit{pos}\sidel = \mathit{pos}\sider \\
    |\mathit{pos}\sidel - \mathit{pos}\sider| \leq 10 \\
    hd(\mathit{hist}\sidel) = \mathit{pos}\sidel \land hd(\mathit{hist}\sider) = \mathit{pos}\sider \\
    \forall t \in tl(\mathit{hist}\sider),\; \mathit{pos}\sidel > t \land \forall t \in tl(\mathit{hist}\sidel),\; \mathit{pos}\sider > t
  \end{cases}
\]
Reading from the top, the first line states that the history lists are
non-empty. The second conjunct says that if the two processes have visited the
same position at some point in the past, then they currently have the same
position.  The third conjunct states that the coupled positions are at most $10$
apart at all times. The fourth line states that the current position is the
first element in each history list, and the last two conjuncts state that the
position in each process is strictly larger than all the previous positions of
the other process; this holds because we always move the lagging process. (We
write $tl(\mathit{hist})$ for the \emph{tail} of a list $\mathit{hist}$,
consisting of all but the first element.)

We now prove the three main premises in \nameref{rule:xprhl-while-gen}.

\subsection{Premise $p_0$}

When $p_0$ is true, $\mathit{pos}\sidel = \mathit{pos}\sider$ and we need to prove
\[
  \vdash \xprhl
  {\Condt{\mathit{pos} < T}{\mathit{dynbody}}}{\Condt{\mathit{pos} < T}{\mathit{dynbody}}}
  {\Theta \land e \land p_0}{\Theta}{\mathit{dynkin}_{\times 0}} .
\]
Since both guards are true, we use the two-sided rule \nameref{rule:xprhl-cond}.
We use \nameref{rule:xprhl-sample} with $f = \id$ (the identity
coupling), and then the usual assignment rule \nameref{rule:xprhl-assn}. The
invariant is preserved since $p_0$ remains true. So, we have the desired
judgment with product program $\mathit{dynkin}_{\times 0}$:
\[
  \begin{array}{l}
    \Condt{\mathit{pos}\sidel < T}{} \\
    \quad \Rand{r\sidel}{\Unif([1, \dots, 10])}; \\
    \quad \Ass{r\sider}{r\sidel}; \\
    \quad \Ass{\mathit{pos}\sidel}{\mathit{pos}\sidel + r\sidel}; \\
    \quad \Ass{\mathit{pos}\sider}{\mathit{pos}\sider + r\sider}; \\
    \quad \Ass{\mathit{hist}\sidel}{\mathit{pos}\sidel :: \mathit{hist}\sidel}; \\
    \quad \Ass{\mathit{hist}\sider}{\mathit{pos}\sider :: \mathit{hist}\sider}
  \end{array}
\]

\subsection{Premise $p_1$}

When $p_1$ is true, $\mathit{pos}\sidel < \mathit{pos}\sider$ and we need to
prove
\[
  \vdash \xprhl{\Condt{\mathit{pos} < T}{\mathit{dynbody}}}{\Skip}
  {\Theta \land (\mathit{pos}\sidel < T) \land p_1}{\Theta}{\mathit{dynkin}_{\times 1}} .
\]
Since we are relating a program to $\Skip$, we apply the one-sided rules.  To
show we preserve $\Theta$, note that $\mathit{hist}\sidel$ and
$\mathit{hist}\sider$ are both non-empty and $\mathit{hist}\sidel \cap
\mathit{hist}\sider$ is initially empty since $\mathit{pos}\sidel < \mathit{pos}\sider$,
so if $\mathit{hist}\sidel \cap \mathit{hist}\sider \neq \varnothing$ after the
loop body then we must have $\mathit{pos}\sidel \in \mathit{hist}\sider$. The
next conjunct $|\mathit{pos}\sidel - \mathit{pos}\sider| \leq 10$ also holds,
since (i) it holds initially, (ii) $\mathit{pos}\sidel < \mathit{pos}\sider$
initially, and (iii) $\mathit{pos}\sidel$ moves forward by at most $10$.  The
conjuncts involving the head of $\mathit{hist}$ are clear. For the last two
conjuncts, $\mathit{hist}\sider$ is unchanged while $\mathit{pos}\sidel$
increases, so
\[
  \forall t \in tl(\mathit{hist}\sider),\; \mathit{pos}\sidel > t
\]
continues to hold. Similarly, if $\mathit{hist}\sidel$ is initially $q :: ps$
where $q$ is the initial value of $\mathit{pos}\sidel$, then it ends up being
$\mathit{pos}\sidel :: q :: ps$.  Since $\mathit{pos}\sider$ is initially
greater than all elements in $ps$ and also greater than $q$ since $p_1$ holds,
we continue to have \[ \forall t \in tl(\mathit{hist}\sidel),\;
\mathit{pos}\sider > t \] after executing the body. So, we have the desired
judgment with the following product program
$\mathit{dynkin}_{\times 1}$:
\[
  \begin{array}{l}
    \Condt{\mathit{pos}\sidel < T}{} \\
    \quad \Rand{r\sidel}{\Unif([1, \dots, 10])}; \\
    \quad \Ass{\mathit{pos}\sidel}{\mathit{pos}\sidel + r\sidel}; \\
    \quad \Ass{\mathit{hist}\sidel}{\mathit{pos}\sidel :: \mathit{hist}\sidel}
  \end{array}
\]

\subsection{Premise $p_2$}

This case is nearly identical to the previous case, using the right-side
versions instead of left-side versions of the rules. By a symmetric argument, we
have
\[
  \vdash \xprhl{\Skip}{\Condt{\mathit{pos} < T}{\mathit{dynbody}}}
  {\Theta \land (p\sider < T) \land p_2}{\Theta}{\mathit{dynkin}_{\times 2}}
\]
where $\mathit{dynkin}_{\times 2}$ is the following product program:
\[
  \begin{array}{l}
    \Condt{\mathit{pos}\sider < T}{} \\
    \quad \Rand{r\sider}{\Unif([1, \dots, 10])}; \\
    \quad \Ass{\mathit{pos}\sider}{\mathit{pos}\sider + r\sider}; \\
    \quad \Ass{\mathit{hist}\sider}{\mathit{pos}\sider :: \mathit{hist}\sider}
  \end{array}
\]

\subsection{Putting it all together}

Applying \nameref{rule:xprhl-while-gen}, we have the judgment
\begin{equation} \label{eq:xprhl-dynkin}
  \vdash \xprhl{\mathit{dynkin}}{\mathit{dynkin}}
  {\mathit{start}\sidel, \mathit{start}\sider \in [1, \dots, 10]}
  {\mathit{hist}\sidel \cap \mathit{hist}\sider \neq \varnothing \to \mathit{pos}\sidel = \mathit{pos}\sider}{\mathit{dynkin}_\times}
\end{equation}
for the following product program $\mathit{dynkin}_\times$:
\[
  \begin{array}{l}
    \Ass{\mathit{pos}\sidel}{\mathit{start}\sidel};
    \Ass{\mathit{pos}\sider}{\mathit{start}\sider} \\
    \Ass{\mathit{hist}\sidel}{[\mathit{start}\sidel]};
    \Ass{\mathit{hist}\sider}{[\mathit{start}\sider]}; \\
    \WWhile{ (\mathit{pos}\sidel < T) \lor (\mathit{pos}\sider < T) }{} \\
    \quad \Condt{ \mathit{pos}\sidel = \mathit{pos}\sider }{} \\
    \quad\quad \Condt{\mathit{pos}\sidel < T}{} \\
    \quad\quad\quad \Rand{r\sidel}{\Unif([1, \dots, 10])}; \\
    \quad\quad\quad \Ass{r\sider}{r\sidel}; \\
    \quad\quad\quad \Ass{\mathit{pos}\sidel}{\mathit{pos}\sidel + r\sidel}; 
                    \Ass{\mathit{pos}\sider}{\mathit{pos}\sider + r\sider}; \\
    \quad\quad\quad \Ass{\mathit{hist}\sidel}{\mathit{pos}\sidel :: \mathit{hist}\sidel};
                    \Ass{\mathit{hist}\sider}{\mathit{pos}\sider :: \mathit{hist}\sider} \\
    \quad \mathbf{else}~\Condt{ \mathit{pos} \sidel < \mathit{pos}\sider }{} \\
    \quad\quad \Condt{\mathit{pos}\sidel < T}{} \\
    \quad\quad\quad \Rand{r\sidel}{\Unif([1, \dots, 10])}; \\
    \quad\quad\quad \Ass{\mathit{pos}\sidel}{\mathit{pos}\sidel + r\sidel}; \\
    \quad\quad\quad \Ass{\mathit{hist}\sidel}{\mathit{pos}\sidel :: \mathit{hist}\sidel} \\
    \quad \mathbf{else} \\
    \quad\quad \Condt{\mathit{pos}\sider < T}{} \\
    \quad\quad\quad \Rand{r\sider}{\Unif([1, \dots, 10])}; \\
    \quad\quad\quad \Ass{\mathit{pos}\sider}{\mathit{pos}\sider + r\sider}; \\
    \quad\quad\quad \Ass{\mathit{hist}\sider}{\mathit{pos}\sider :: \mathit{hist}\sider}
  \end{array}
\]
This program models the informal coupling proof: if the positions are equal, we
take equal samples and move both processes; otherwise, we move the lagging
process while holding the leading process fixed. We can analyze this program to
show convergence of two Dynkin processes.

\begin{theorem} \label{thm:dynkin-converge}
  Let $m_1, m_2$ be two memories such that $m_1(\mathit{start}),
  m_2(\mathit{start}) \in [0, 10]$. Let
  $\mu_1, \mu_2$ be the final distributions over memories:
  \[
    \mu_1 \triangleq \denot{\mathit{dynkin}} m_1
    \quad \text{and} \quad
    \mu_2 \triangleq \denot{\mathit{dynkin}} m_2 .
  \]
  Let $\eta_1, \eta_2$ be the final distributions over positions:
  \[
    \eta_1 \triangleq \liftf{\denot{\mathit{pos}}}(\mu_1)
    \quad \text{and} \quad
    \eta_2 \triangleq \liftf{\denot{\mathit{pos}}} (\mu_2) .
  \]
  Then the distance between the two position distributions is at most
  \[
    \tvdist{\eta_1}{\eta_2} \leq (9/10)^{\lfloor T/10 \rfloor - 1} .
  \]
\end{theorem}
\begin{proof}
  If $T \leq 10$, the claim is trivial.  Otherwise, let $\mu_\times$ be the
  coupling in \cref{eq:xprhl-dynkin} and let $\eta_\times$ be the coupling
  projected to the two positions:
  \[
    \mu_\times \triangleq \denot{\mathit{dynkin}_\times} (m_1, m_2)
    \quad \text{and} \quad
    \eta_\times \triangleq \liftf{\denot{(\mathit{pos}\sidel, \mathit{pos}\sider)}} (\mu_\times) .
  \]
  We directly calculate
  \begin{align*}
    \Pr_{(p_1, p_2) \sim \eta_\times} [ p_1 \neq p_2 ]
    &= \Pr_{(m_1, m_2) \sim \mu_\times} [ m_1(\mathit{pos}) \neq m_2(\mathit{pos}) ] \\
    &\leq \Pr_{(m_1, m_2) \sim \mu_\times}
    [ (m_1, m_2) \in \denot{\mathit{hist}\sidel \cap \mathit{hist}\sider = \varnothing} ] ,
  \end{align*}
  where the inequality follows by the post-condition of \cref{eq:xprhl-dynkin}:
  pairs of memories where $\mathit{hist}\sidel \cap \mathit{hist}\sider$ is
  non-empty do not have different positions.

  We turn to the product program to bound the last quantity. If the two process
  have not met yet, then $\mathit{hist}\sidel \cap \mathit{hist}\sider =
  \varnothing$. Since the processes are at most $10$ apart, in each iteration of
  the loop there is a $9/10$ chance the lagging process misses the leading
  process, preserving $\mathit{hist}\sidel \cap \mathit{hist}\sider =
  \varnothing$. Since both processes move at most $10$ positions each iteration,
  there are at least $\lfloor T/10 \rfloor - 1$ iterations so
  \[
    \Pr_{(m_1, m_2) \sim \mu_\times}
    [ (m_1, m_2) \in \denot{\mathit{hist}\sidel \cap \mathit{hist}\sider = \varnothing} ]
    \leq (9/10)^{\lfloor T/10 \rfloor - 1} .
  \]
  By the coupling method (\cref{thm:coupling-method}), we conclude
  \[
    \tvdist{\eta_1}{\eta_2}
    \leq \Pr_{(p_1, p_2) \sim \eta_\times} [ p_1 \neq p_2 ]
    \leq (9/10)^{\lfloor T/10 \rfloor - 1} .
    \qedhere
  \]
\end{proof}

\section{Path couplings} \label{sec:xprhl-ex-path}

So far we have used couplings to analyze several \emph{Markov chains}, iterative
processes where the state is a randomized function of the previous state. The
main state space in our examples has been the integers---the position in the
random walk or Dynkin process, or the count of the number of heads in the
stochastic domination example. For more complex state spaces it can be unclear
how to couple the samplings to guide the two states towards one another,
especially if the states are many transitions apart.

To address this issue, \citet{bubley1997path} proposed the \emph{path coupling}
method, a powerful tool to construct couplings. Before describing their idea, we
first set some definitions.

\begin{definition}
  Let $\Omega$ be a finite set of states. We say a metric $d : \Omega \times
  \Omega \to \NN$ is a \emph{path metric} if whenever $d(s, s') > 1$, there
  exists $s'' \neq s, s'$ such that $d(s, s') = d(s, s'') + d(s'', s')$.  We say
  two states $s, s'$ are \emph{adjacent} if $d(s, s') = 1$.   The
  \emph{diameter} $\Delta$ of the state space is the maximum distance between
  any two states.  A \emph{Markov chain} on $\Omega$ is defined by iterating a
  \emph{transition function} $\tau : \Omega \to \Dist(\Omega)$ starting from
  some initial state.
\end{definition}

Then the main theorem of path coupling is as follows.

\begin{theorem}[\citet{bubley1997path}] \label{thm:path-coupling}
  Consider a Markov chain with transition function $\tau$ over a state space
  $\Omega$ with path metric $d$ and diameter at most $\Delta$. Suppose for any
  two adjacent states $s$ and $s'$, there exists a coupling $\mu$ of $\tau(s),
  \tau(s')$ with
  \begin{align*}
    \Ex_{(r, r') \sim \mu}[d(r, r')] \leq \beta .
  \end{align*}
  Let $\mu_1^{(T)}, \mu_2^{(T)}$ be the final distributions from starting in any
  two states $s_1, s_2$ and running $T$ steps of the Markov chain. Then there
  is a coupling $\mu$ of $\mu_1^{(T)}, \mu_2^{(T)}$ with
  \begin{align*}
    \tvdist{\mu_1^{(T)}}{\mu_2^{(T)}} \leq \Pr_{(r, r') \sim \mu} [r \neq r'] \leq \beta^T \Delta .
  \end{align*}
  In particular, the distributions converge in total variation distance
  exponentially quickly if $\beta < 1$. 
\end{theorem}

Intuitively, path coupling can be seen as a transitivity principle for
couplings: if we can couple the distributions after one step from any two
adjacent states, then we can extend to a coupling on distributions from any two
initial states. While we are not able to internalize this principle in \Sxprhl
due to the required bounds on expectations, we can still construct and analyze
the one-step couplings. (We consider how to handle expected distance bounds and
couplings in \cref{chap:future}.) We present two examples from the original
paper by \citet{bubley1997path}.

\subsection{Glauber dynamics: sampling a proper coloring}

The Markov chain in our first example samples approximately uniform graph
colorings. It was first analyzed by \citet{DBLP:journals/rsa/Jerrum95}; we
follow the subsequent, simpler analysis by \citet{bubley1997path} using path
coupling.  Recall that a finite \emph{graph} $G$ consists of a finite set $V$ of
\emph{vertices} and a symmetric binary relation $E$ relating vertices that are
connected by an \emph{edge}; we let $\neighbors{G}(v) \subseteq V$ denote the
\emph{neighbors} of a vertex $v$, i.e., the set of vertices with an edge to $v$.
We write $D$ for the \emph{degree} of $G$, i.e., $|\neighbors{G}(v)| \leq D$ for
all $v$. We write $n \triangleq |V|$ for the number of vertices.

Let $C$ be a finite set of \emph{colors}; we write $k \triangleq |C|$ for the
number of colors. A \emph{coloring} of $G$ is a map $w : V \to C$ assigning a
color to each vertex; the state space of our Markov chain will be the set of
colorings.  Let the path distance $d$ on the state space be the number of
vertices colored differently under two colorings; evidently, the diameter
$\Delta$ of this state space is at most the number of vertices $n$.  A coloring
is \emph{valid} (also called \emph{proper}) if $w(v)$ and $w(v')$ have distinct
colors for all $(v, v') \in E$. The following program models $T$ steps of the
\emph{Glauber dynamics}:
\[
  \begin{array}{l}
    \Ass{i}{0}; \\
    \WWhile{i < T}{}; \\
    \quad \Rand{v}{\Unif(V)}; \\
    \quad \Rand{c}{\Unif(C)}; \\
    \quad \Condt{\valid{G}(w, v, c)}{w \gets w[v \mapsto c]}; \\
    \quad \Ass{i}{i + 1}
  \end{array}
\]
where the guard $\valid{G}(w, v, c)$ holds when $c$ is valid at $v$ in $w$,
namely, when there is no neighbor of $v$ colored with $c$ in $w$.  Informally,
the algorithm starts from a coloring $w$ and iteratively modifies it by
uniformly sampling a vertex $v$ and a color $c$, recoloring $v$ with $c$ if it
is locally valid.  We focus on
the loop body, which encodes the transition function of the Markov chain:
\[
  \begin{array}{l}
    \Rand{v}{\Unif(V)}; \\
    \Rand{c}{\Unif(C)}; \\
    \Condt{\valid{G}(w, v, c)}{w \gets w[v \mapsto c]}
  \end{array}
\]
We call this program $\mathit{glauber}$. To apply path coupling
(\cref{thm:path-coupling}), we must find a coupling where the expected distance
between coupled states is small when $w\sidel$ and $w\sider$ are initially
adjacent.

\begin{theorem} \label{thm:glauber}
  Let $m_1, m_2$ be memories with $m_1(w), m_2(w)$ adjacent colorings. Let
  $\mu_1, \mu_2$ be the distributions over memories after running one step of
  the transition function:
  \[
    \mu_1 \triangleq \denot{\mathit{glauber}} m_1
    \quad \text{and} \quad
    \mu_2 \triangleq \denot{\mathit{glauber}} m_2 .
  \]
  Let $\eta_1, \eta_2$ be the respective distributions over colorings:
  \[
    \eta_1 \triangleq \liftf{\denot{w}}(\mu_1)
    \quad \text{and} \quad
    \eta_2 \triangleq \liftf{\denot{w}}(\mu_2) .
  \]
  Then there is a coupling $\eta_\times$ of $(\eta_1, \eta_2)$ with
  \[
    \Ex_{(w_1, w_2) \sim \eta_\times}[d(w_1, w_2)] \leq 1 - 1/n + 2D/kn .
  \]
  If $\eta_1^{(T)}, \eta_2^{(T)}$ are the distributions over the final colorings
  after $T$ steps starting from any two colorings, then
  \[
    \tvdist{\eta_1^{(T)}}{\eta_2^{(T)}} \leq (1 - 1/n + 2D/kn)^T \cdot n .
  \]
\end{theorem}
\begin{proof}
  Suppose the initial memories contain adjacent colorings $w\sidel$ and
  $w\sider$. First, we couple the sampling from $\Unif(V)$ with
  \nameref{rule:xprhl-sample}, using the identity coupling $f = \id$.

  Now notice that the two initial states $w\sidel$ and $w\sider$ differ in the
  color for a single vertex, call it $v_0$. Letting $a \triangleq w_1(v_0)$ and
  $b \triangleq w_2(v_0)$, we perform a case analysis on the sampled vertex with
  \nameref{rule:xprhl-case}. If $v$ is a neighbor of the differing vertex $v_0$,
  applying \nameref{rule:xprhl-sample} with the transposition bijection
  $\pi^{ab} : C \to C$ defined by
  \[
    \pi^{ab}(x) \triangleq \begin{cases}
      b &: x = a \\
      a &: x = b \\
      x &: \text{otherwise}
    \end{cases}
  \]
  ensures $c\sider = \pi^{ab}(c\sidel)$. Otherwise, \nameref{rule:xprhl-sample}
  with the identity coupling ensures $c\sidel = c\sider$.  By applying the
  one-sided rules for conditionals (\nameref{rule:xprhl-cond-l} and
  \nameref{rule:xprhl-cond-r}) to the left and the right programs, we have
  \begin{equation} \label{eq:xprhl-glauber}
    \vdash \xprhl{\mathit{glauber}}{\mathit{glauber}}
    {d(w\sidel, w\sider) = 1}{d(w\sidel, w\sider) \leq 2}{\mathit{glauber}_\times} ,
  \end{equation}
  where $\mathit{glauber}_\times$ is the following product program:
  \[
    \begin{array}{l}
      \Rand{v\sidel}{\Unif(V)};
      \Ass{v\sider}{v\sidel}; \\
      \Condt{v\sidel \in \neighbors{G}(v_0)}{} \\
      \quad \Rand{c\sidel}{\Unif(C)};
      \Ass{c\sider}{\pi^{ab}(c\sidel)} \\
      \mathbf{else} \\
      \quad \Rand{c\sidel}{\Unif(C)};
      \Ass{c\sider}{c\sidel} \\
      \Condt{\valid{G}(w\sidel, v\sidel, c\sidel)}{} \\
      \quad \Ass{w\sidel}{w\sidel[v\sidel \mapsto c\sidel]} \\
      \Condt{\valid{G}(w\sider, v\sider, c\sider)}{} \\
      \quad \Ass{w\sider}{w\sider[v\sider \mapsto c\sider]}
    \end{array}
  \]
  We analyze this program to bound the expected distance between states under
  the coupling. Let the coupling on memories be $\mu_\times \triangleq
  \denot{\mathit{glauber}_\times} (m_1, m_2)$, and let the coupling on the final
  colorings be $\eta_\times \triangleq
  \liftf{\denot{(w\sidel,w\sider)}}(\mu_\times)$. We have:
  \begin{align*}
    \Ex_{(w_1, w_2) \sim \eta_\times}[d]
    &= 0 \cdot \Pr_{(w_1, w_2) \sim \eta_\times} [ d = 0 ]
    + 1 \cdot \Pr_{(w_1, w_2) \sim \eta_\times} [ d = 1 ]
    + 2 \cdot \Pr_{(w_1, w_2) \sim \eta_\times} [ d = 2 ] \\
    &= 1 - \Pr_{(w_1, w_2) \sim \eta_\times} [ d = 0 ]
    + \Pr_{(w_1, w_2) \sim \eta_\times} [  d = 2 ] \\
    &\leq 1 - \Pr_{(m_1, m_2) \sim \mu_\times} [ m_1(v) = v_0 \land \valid{G}(m_1(w), v_0, m_1(c)) ]  \\
    &\quad+ \Pr_{(m_1, m_2) \sim \mu_\times} [ m_1(v) \in \neighbors{G}(v_0) \land m_1(c) = b ] \\
    &\leq 1 - \frac{1}{n} \left( 1 - \frac{D}{k} \right) + \frac{D}{nk}
    = 1 - \frac{1}{n} + \frac{2 D}{nk} .
  \end{align*}
  The equalities hold because the distance between the resulting colorings is at
  most two by the post-condition of \cref{eq:xprhl-glauber}, so $1 = \Pr[d=0] +
  \Pr[d=1] + \Pr[d=2]$.  The first inequality follows since the distance
  decreases to zero if we select a valid color at $v_0$, and the distance can
  only increase to two if we select a neighbor of $v_0$ and pick the color
  combination $(c\sidel, c\sider) = (b, a)$. The last step follows since each
  vertex has at most $D$ neighbors, so there are at at least $k - D$ valid
  colors at any vertex; in particular, the distance decreases to zero if we
  select $v_0$ (probability $1/n$) and a valid color (probability at least $1 -
  D/k$).

  Thus, we have constructed a coupling $\eta_\times$ such that
  \[
    \Ex_{(w_1, w_2) \sim \eta_\times}[d(w_1, w_2)] \leq 1 - 1/n + 2D/kn .
  \]
  By the path coupling theorem (\cref{thm:path-coupling}), we can bound the
  distance between the $T$-step distributions $\eta_1^{(T)}, \eta_2^{(T)}$ over
  $w$ from any two initial colorings:
  \[
    \tvdist{\eta_1^{(T)}}{\eta_2^{(T)}} \leq (1 - 1/n + 2D/kn)^T \cdot n .
  \]
  When the number of colors $k$ is at least $2D + 1$, the right-hand side tends
  to zero exponentially quickly.
\end{proof}

\begin{remark} \label{rem:glauber}
  \Cref{thm:glauber} bounds how fast the Glauber dynamics converges, started
  from any two colorings. Using basic facts about Markov chains, it is not hard
  to show that the Glauber dynamics has the uniform distribution over valid
  colorings of $G$ as a \emph{stationary distribution} (a distribution $\eta \in
  \Dist(\Omega)$ such that $\dbind(\eta, \tau) = \eta$).\footnote{%
    The Glauber dynamics takes any valid coloring to another valid coloring, and
    the probability of transitioning from $w$ to $w'$ is equal to the
    probability of transitioning from $w'$ to $w$, so the Glauber dynamics is
    \emph{reversible} over the valid colorings and hence the uniform
  distribution is stationary.}
  As a consequence, the Glauber dynamics started in any \emph{valid} coloring
  converges exponentially quickly to the uniform distribution over valid
  colorings when $k \geq 2D + 1$. To see this, let $\mu$ be the distribution on
  colorings after $T$ steps started from some valid coloring. Suppose there are
  $M$ valid colorings on $G$, and let $\mu_1, \dots, \mu_M$ be the corresponding
  distributions over colorings after $T$ steps. Since the uniform distribution
  $\eta$ is stationary, we have
  \[
    \eta = \frac{1}{M} \cdot \mu_1 + \dots + \frac{1}{M} \cdot \mu_M .
  \]
  For every $i$, \cref{thm:glauber} gives
  \[
    \tvdist{\mu}{\mu_i} \leq (1 - 1/n + 2D/kn)^T \cdot n .
  \]
  By linearity of TV distance, the output distribution approaches the uniform
  distribution over valid colorings:
  \[
    \tvdist{\mu}{\eta} \leq (1 - 1/n + 2D/kn)^T \cdot n .
  \]
\end{remark}

\subsection{Condensed hard-core lattice gas: sampling an independent set}

Our second example is a Markov chain from statistical physics modeling the
evolution of a physical system in the \emph{conserved hard-core lattice gas}
(CHLG) model~\citep{bubley1997path}. Suppose we have a finite set $P$ of
\emph{particles}, $s \triangleq |P|$ in total, and we have a finite graph $G =
(V,E)$ with degree at most $D$. A \emph{placement} is a map $w : P \to V$
placing each particle at a vertex of the graph. We wish to set the particles so
that each vertex has at most one particle and no two particles are located at
adjacent vertices; we call such a placement \emph{safe}. (In other words, a safe
placement is an independent set.)

We analyze a Markov chain to sample a uniformly random safe placement.  Take the
state space $\Omega$ to be the set of all placements (not necessarily safe).
The Markov chain starts from an initial placement. Each step, it samples a
particle $p$ from $P$ and a vertex $v$ from $V$ uniformly at random, and tries
to relocate $p$ to $v$. If $p$ is safe at $v$, then the Markov chain updates the
placement; otherwise, it leaves the placement unchanged. We model $T$ steps of
this dynamics with the following program:
\[
  \begin{array}{l}
    \Ass{i}{0}; \\
    \WWhile{i < T}{}; \\
    \quad \Rand{p}{\Unif(P)}; \\
    \quad \Rand{v}{\Unif(V)}; \\
    \quad \Condt{\safe{G}(w, p, v)}{w \gets w[p \mapsto v]}; \\
    \quad \Ass{i}{i + 1}
  \end{array}
\]
where the guard $\safe{G}(w, p, v)$ holds when $p$ is valid at $v$ in $w$, i.e.,
when there is no other particle located at $v$ or its neighbors.  We let the
path metric $d$ be the number of particles with different locations under two
placements; evidently, the diameter of the state space is at most $s$.  To build
a coupling on the one-step distributions from adjacent initial placements, we
analyze the transition function $\mathit{chlg}$ extracted from the loop body:
\[
  \begin{array}{l}
    \Rand{p}{\Unif(P)}; \\
    \Rand{v}{\Unif(V)}; \\
    \Condt{\safe{G}(w, p, v)}{w \gets w[p \mapsto v]}
  \end{array}
\]

\begin{theorem} \label{thm:chlg}
  Let $m_1, m_2$ be memories with $m_1(w), m_2(w)$ adjacent placements. Let
  $\mu_1, \mu_2$ be the respective distributions over memories after running one
  step of the transition function:
  \[
    \mu_1 \triangleq \denot{\mathit{chlg}} m_1
    \quad \text{and} \quad
    \mu_2 \triangleq \denot{\mathit{chlg}} m_2 .
  \]
  Let $\eta_1, \eta_2$ be the respective distributions over placements:
  \[
    \eta_1 \triangleq \liftf{\denot{w}}(\mu_1)
    \quad \text{and} \quad
    \eta_2 \triangleq \liftf{\denot{w}}(\mu_2) .
  \]
  Then there is a coupling $\eta_\times$ of $(\eta_1, \eta_2)$ such that
  \[
    \Ex_{(w_1, w_2) \sim \eta_\times}[d(w_1, w_2)] \leq \beta
    \triangleq \left( 1 - \frac{1}{s} \right) \left( 1 + \frac{ 3 (D + 1) }{ n }
    \right) .
  \]
  If $\eta_1^{(T)}, \eta_2^{(T)}$ are the distributions over final placements
  after $T$ steps starting from any two placements, then
  \[
    \tvdist{\eta_1^{(T)}}{\eta_2^{(T)}} \leq \beta^T \cdot s .
  \]
\end{theorem}
\begin{proof}
  To couple the two runs we use \nameref{rule:xprhl-sample} with $f = \id$
  twice, ensuring $p\sidel = p\sider$ and $v\sidel = v\sider$.  Then we apply
  the one-sided rules for conditionals (\nameref{rule:xprhl-cond-l} and
  \nameref{rule:xprhl-cond-r}) to the left and the right sides to prove
  \begin{equation} \label{eq:xprhl-chlg}
    \vdash \xprhl{\mathit{chlg}}{\mathit{chlg}}
    {d(w\sidel, w\sider) = 1}{d(w\sidel, w\sider) \leq 2}{\mathit{chlg}_\times}
  \end{equation}
  where $\mathit{chlg}_\times$ is the following product program:
  \[
    \begin{array}{l}
      \Rand{p\sidel}{\Unif(P)}; \\
      \Ass{p\sider}{p\sidel}; \\
      \Rand{v\sidel}{\Unif(V)}; \\
      \Ass{v\sider}{v\sidel}; \\
      \Condt{\safe{G}(w\sidel, p\sidel, v\sidel)}{} \\
      \quad \Ass{w\sidel}{w\sidel[p\sidel \mapsto v\sidel]} \\
      \Condt{\safe{G}(w\sider, p\sider, v\sider)}{} \\
      \quad \Ass{w\sider}{w\sider[p\sider \mapsto v\sider]} \\
    \end{array}
  \]
  Now we bound the expected distance between the final placements. The two
  initial placements $w\sidel$ and $w\sider$ differ in the position of a single
  particle $p_0$, located at vertex $a$ and $b$ in $w\sidel$ and $w\sider$
  respectively. Let the coupling on output distributions be $\mu_\times
  \triangleq \denot{\mathit{chlg}_\times} (m_1, m_2)$ and let the coupling on
  placement distributions be $\eta_\times \triangleq
  \liftf{\denot{(w\sidel,w\sider)}}(\mu_\times)$. We have:
  \begin{align*}
    \Ex_{(w_1, w_2) \sim \eta_\times} & [d] \\
    ={}& 1 - \Pr_{(w_1, w_2) \sim \eta_\times} [ d = 0 ]
    + \Pr_{(w_1, w_2) \sim \eta_\times} [ d = 2 ] \\
    ={}& 1 - \Pr_{(m_1, m_2) \sim \mu_\times}
    [ m_1(p) = p_0 \land \safe{G}(m_1(w), m_1(p), m_1(v))] \\
    &+ \Pr_{(m_1, m_2) \sim \mu_\times}
    [ m_1(p) \neq p_0 \land
    (\safe{G}(m_1(w), m_1(p), m_1(v)) \neq \safe{G}(m_2(w), m_2(p), m_2(v)))] \\
    \leq{}& 1 - \Pr_{(m_1, m_2) \sim \mu_\times}
    [ m_1(p) = p_0 \land \safe{G}(m_1(w), m_1(p), m_1(v))] \\
    &+ \Pr_{(m_1, m_2) \sim \mu_\times}
    [ m_1(p) \neq p_0 \land \neg(\safe{G}(m_1(w), m_1(p), m_1(v))
    \land \safe{G}(m_2(w), m_2(p), m_2(v)))]
  \end{align*}
  To bound the first probability, the probability of selecting
  particle $p_0$ is $1/s$ and the selected particle is safe at $v$ if it avoids
  the other $s - 1$ locations and their neighbors (at most $(s - 1) (D + 1)$ bad
  locations). To bound the second probability, the probability of selecting a
  particle not equal to $p_0$ is $1 - 1/s$, and $p$ is safe at $v$ on both sides
  unless we select the position $a$, $b$, or one of their neighbors (at most
  $2(D + 1)$ bad points). Putting everything together, we conclude
  \begin{align*}
    \Ex_{(w_1, w_2) \sim \eta_\times} [d(w_1, w_2)]
    &\leq 1 - \frac{1}{s} \left( 1 - \frac{(s - 1)(D + 1)}{n} \right)
    + \left( 1 - \frac{1}{s} \right) \left( \frac{ 2(D + 1) }{ n } \right) \\
    &= \left(1 - \frac{1}{s} \right) \left( 1 + \frac{3 (D + 1) }{n} \right)
    \triangleq \beta .
  \end{align*}
  By the path coupling theorem (\cref{thm:path-coupling}), we can bound the
  distance between the $T$-step distributions $\eta_1^{(T)}, \eta_2^{(T)}$ over
  final placements from any two initial placements:
  \[
    \tvdist{\eta_1^{(T)}}{\eta_2^{(T)}} \leq \beta^T \cdot s .
  \]
  When $\beta < 1$, the distributions converge exponentially quickly.
\end{proof}

\begin{remark}
  Like the Glauber dynamics, this Markov chain also has the uniform distribution
  over safe placements as a stationary distribution. \Cref{thm:chlg} shows the
  distribution over placements converges exponentially quickly to this
  distribution when $\beta < 1$, starting from any safe placement.

  \citet{bubley1997path} actually proved a stronger version of \cref{thm:chlg}:
  \[
    \Ex_{(w_1, w_2) \sim \eta_\times}[d(w_1, w_2)]
    \leq \left( 1 - \frac{1}{s} \right) \left( 1 + \frac{ 2 (D + 1) }{ n } \right) ,
  \]
  which is sharper than our bound
  \[
    \Ex_{(w_1, w_2) \sim \eta_\times}[d(w_1, w_2)]
    \leq \left( 1 - \frac{1}{s} \right) \left( 1 + \frac{ 3 (D + 1) }{ n } \right) .
  \]
  Their analysis used the maximal coupling to couple the state distributions
  from sampling and updating the placement, giving a tighter bound on the
  expected distance.

  While it is technically possible to extend \Sxprhl with a sampling rule
  modeling the maximal coupling, with the corresponding product program drawing
  correlated samples from the witness distribution, the result would be somewhat
  unnatural. First, we would need to describe the witness distribution
  precisely---the maximal coupling $\mu$ of two distributions $\mu_1, \mu_2$
  satisfies the equation
  \[
    \tvdist{\mu_1}{\mu_2} = \Pr_{(a_1, a_2) \sim \mu}[ a_1 \neq a_2 ]
  \]
  but the probabilities of other events are not specified. In general, there
  could be multiple possible witnesses to the maximal coupling, and it is
  unclear which witness should the canonical choice.

  Furthermore, the maximal coupling is defined in terms of the probability of
  samples being different. This makes the maximal coupling a poor fit for our
  logics, which describe the support of the witness via probabilistic lifting.
  We would only be able to prove the trivial post-condition after applying the
  maximal coupling; the properties of the maximal coupling would then enter as
  axioms when verifying the coupled product.
\end{remark}

\section{Comparison with existing product programs} \label{sec:products-rw}

Product constructions reduce a relational property of two programs to a
non-relational property of a single program, so that more standard techniques
can be brought to bear. We close this chapter by comparing our coupled product
to other existing constructions.

Almost all product constructions were originally designed with non-probabilistic
programs in mind, targeting relational properties like information flow and
correctness of compiler transformations. These approaches include \emph{self
composition}~\citep{BartheDR04}, the \emph{cross product}~\citep{ZaksP08},
\emph{type-directed} product programs~\citep{TerauchiA05}, and
more~\citep{BartheCK11,BartheCK13}. A basic consideration is how to handle
different control flow in the two programs. If the two programs have the same
shape and always take the same branches, the product program can interleave
instructions from the two programs. If the two programs are very different or if
the control flows are not synchronized, an asynchronous construction can combine
the two programs sequentially.

These approaches have different strengths and weaknesses. By placing
corresponding instructions close to one another, synchronized constructions can
better leverage similarity between programs and can often be verified with
simpler invariants and more local reasoning. However, asynchronous products
apply to a wider range of programs. The design of \Sxprhl, and in particular the
asynchronous rule \nameref{rule:xprhl-while-gen}, allows product programs that
are both synchronous and asynchronous.

Probabilistic programs introduce additional challenges for product
constructions. Existing constructions can be blindly applied to randomized
programs, but the results use two independent sources of randomness, and are
difficult to reason about---there is no coordination between the two programs on
sampling instructions, whether the construction has a synchronous structure or
not. A notable exception is the product construction by \citet{BGGHKS14}, which
is specialized to proving differential privacy.  Their construction eliminates
the random sampling statements entirely, yielding a synchronized,
non-probabilistic product. In fact, their product is based on a variant of
probabilistic couplings called \emph{approximate liftings}; we turn to these
couplings in the rest of the thesis.

\chapter{Approximate couplings for \ifpenn\else differential\fi{} privacy} \label{chap:approx}

The first half of this thesis connected proofs by coupling with the logic
\Sprhl, using ideas from the former to enhance the latter. We now explore a
similar connection in reverse, using concepts from program logics to develop a
novel form of probabilistic coupling and a new proof technique.  Our starting
point is \Saprhl, an \emph{approximate} version of \Sprhl proposed by
\citet{BKOZ13-toplas} for verifying \emph{differential privacy}, a statistical
notion of data privacy. This logic was originally based on an approximate
version of probabilistic lifting. By interpreting approximate liftings as a
generalization of probabilistic coupling and reverse-engineering an approximate
version of proof by coupling from \Saprhl, we can give a powerful method to
prove differential privacy.

After briefly reviewing differential privacy (\cref{sec:prelim-dp}), we propose
a new definition of approximate lifting and explore its theoretical properties
(\cref{sec:alift-props}); our approximate liftings are a natural, approximate
version of probabilistic couplings. To build approximate couplings, we review a
core version \Saprhl (\cref{sec:aprhl-core}) and extract a proof technique
inspired by the logic, called \emph{proof by approximate coupling}
(\cref{sec:aprhl-pbac}). We then extend \Saprhl with proof rules modeling new
approximate couplings (\cref{sec:aprhl-lap}) and a principle called
\emph{pointwise equality} for proving differential privacy
(\cref{sec:aprhl-pweq}). As applications, we give new proofs of privacy for the
\emph{Report-noisy-max} and \emph{Sparse Vector} mechanisms
(\cref{sec:aprhl-ex}). Our approximate coupling proofs are significantly cleaner
than existing arguments, and can be formalized in \Saprhl, enabling the first
formal privacy proofs for these mechanisms. Finally, we survey other
verification techniques for differential privacy, and research on approximate
liftings (\cref{sec:aprhl-rw}).

\section{Differential privacy preliminaries} \label{sec:prelim-dp}

\emph{Differential privacy}, proposed by \citet{DMNS06}, is a strong,
probabilistic notion of data privacy that has attracted intensive attention
across computer science and beyond. Differential privacy is a relational
property of probabilistic programs.
\begin{definition} \label{def:dp}
  Let $\varepsilon, \delta$ be non-negative parameters. Consider a set $\cD$
  with a binary \emph{adjacency} relation $\mathit{Adj}$; we sometimes call
  $\cD$ the set of \emph{databases}. Let the \emph{range} $\cR$ be a set of
  possible outputs. A function $M : \cD \to \Dist(\cR)$---often called a
  \emph{mechanism}---is \emph{$(\varepsilon, \delta)$-differentially private} if
  for all pairs of adjacent inputs $(d, d') \in \mathit{Adj}$ and all subsets
  $\cS \subseteq \cR$ of outputs, we have
  \[
    M(d)(\cS) \leq \exp(\varepsilon) \cdot M(d')(\cS) + \delta .
  \]
  When $\delta = 0$, we say $M$ is $\varepsilon$-\emph{differentially
  private}.
\end{definition}

The adjacency relation describes which pairs of databases should lead to
approximately indistinguishable outputs---intuitively, which pairs of databases
differ only in the data of a single person. For instance, if a database is a set
of records belonging to different people, we can consider two databases to be
adjacent if they are identical except for an additional individual's record in
one database. Then under differential privacy, a mechanism's output must be
nearly the same whether any single individual's private data is part of the
input or not.  The degree of similarity---and the strength of the privacy
guarantee---are governed by the parameters $\varepsilon$ and $\delta$: smaller
values give stronger guarantees, while larger values give weaker guarantees.

While typical notions of adjacency are symmetric, much of the theory of
differential privacy applies to arbitrary relations.  However, there are a few
notable results that crucially need a symmetric adjacency relation---we will
highlight these cases as they arise.

\subsection{Standard private mechanisms}

The most basic example of a differentially private mechanism is the
\emph{Laplace mechanism}, which evaluates a numeric query on a database and
adds random noise drawn from the Laplace distribution. For instance, the target
query could compute the average age, or count the number of patients with a
certain disease. While the Laplace distribution is a continuous distribution
over the real numbers, we work with a discrete version to avoid
measure-theoretic technicalities. For concreteness we take the samples to be
integers; our results can be easily adapted to finer discretizations.\footnote{%
  More precisely, any discretization closed under addition.}
\begin{definition} \label{def:lapmech}
  Let $\varepsilon > 0$.  The \emph{(discrete) Laplace distribution} with parameter
  $\varepsilon$, written $\Lap{\varepsilon}$, is the distribution over the integers
  where $v \in \ZZ$ has probability proportional to $\exp( - |v| \cdot \varepsilon )$:
  \[
    \Lap{\varepsilon}(v) \triangleq \frac{\exp( - |v| \cdot \varepsilon )}{W} ,
  \]
  with $W \triangleq \sum_{z \in \ZZ} \exp( - |z| \cdot \varepsilon )$. We write
  $\Lap{\varepsilon}(t)$ for the Laplace distribution with mean $t \in \ZZ$;
  sampling from this distribution is equivalent to sampling from
  $\Lap{\varepsilon}$ and adding $t$.

  Let $q : \cD \to \ZZ$ be an integer-valued query. The \emph{Laplace mechanism}
  with parameter $\varepsilon$ takes a database $d \in \cD$ as input and returns
  a sample from $\Lap{\varepsilon}(q(d))$. This mechanism is also known as the
  $\varepsilon$-\emph{geometric mechanism}~\citep{ghosh2012universally}.
\end{definition}

If the query takes similar values on adjacent databases, the Laplace mechanism
is differentially private. The privacy parameters depend on the
\emph{sensitivity} of the query---the more the answers may differ on adjacent
databases, the weaker the privacy guarantee.

\begin{theorem}[\citet{DMNS06}] \label{thm:lap-priv}
  A query $q : \cD \to \ZZ$ is \emph{$k$-sensitive} if $|q(d) - q(d')| \leq k$
  for every pair of adjacent databases. Releasing a $k$-sensitive query with the
  Laplace mechanism with parameter $\varepsilon$ is $(k \cdot \varepsilon,
  0)$-differentially private.
\end{theorem}

\subsection{Composition theorems}

Differential privacy is closed under several notions of composition, making it
easy to build new private algorithms out of private components. The
\emph{sequential}, or \emph{standard composition theorem} is the most basic
example. When running two private computations in sequence---where the second
computation may use the input database as well as the randomized output from the
first computation---the privacy guarantee should weaken, since we run more
analyses on the data. Indeed, the privacy parameters simply add up.

\begin{theorem}[\citet{DMNS06}] \label{thm:seq-comp}
  Let $M : \cD \to \Dist(\cR)$ be $(\varepsilon, \delta)$-differentially private
  and let $M' : \cR \times \cD \to \Dist(\cR)$ be such that $M'(r,-) : \cD \to
  \Dist(\cR)$ is $(\varepsilon', \delta')$-differentially private for every $r
  \in \cR$. Given a database $d \in \cD$, sampling $r$ from $M(d)$ and then
  returning a sample from $M'(r, d)$ is $(\varepsilon + \varepsilon', \delta +
  \delta')$-differentially private.
\end{theorem}

This useful theorem has two immediate consequences. First, if $M'$ depends only
on its first argument $r$ and ignores its database argument $d$, then $M'(r,-)$
is $(0, 0)$-differentially private. So, transforming the output of a
differentially-private algorithm does not degrade privacy; this property is also
called \emph{closure under post-processing}.

Second, by repeatedly applying the composition theorem, the composition of $n$
separate $(\varepsilon, \delta)$-differentially private mechanisms is
$(n\varepsilon, n\delta)$-differentially private. In certain parameter ranges,
an alternative, \emph{advanced composition} theorem can bound the privacy level
with a smaller $\varepsilon$ at the cost of a slightly larger $\delta$. This
result crucially assumes a symmetric adjacency relation.

\begin{theorem}[\citet{DRV10}] \label{thm:adv-comp}
  Fix a \emph{symmetric} adjacency relation on $\cD$.  Let $f_i : \cR \times \cD
  \to \Dist(\cR)$ be a sequence of $n$ functions such that for every $r \in
  \cR$, the functions $f_i(r, -) : \cD \to \Dist(\cR)$ are $(\varepsilon,
  \delta)$-differentially private. Then for every $\omega \in (0, 1)$, the
  mechanism that executes $f_1, \dots, f_n$ in sequence and returns the final
  output is $(\varepsilon^*, \delta^*)$-differentially private for
  \[
    \varepsilon^* = \varepsilon \sqrt{2 n \ln(1/\omega)} + n \varepsilon(e^\varepsilon - 1)
    \quad \text{and} \quad
    \delta^* = n \delta + \omega .
  \]
  In particular, if we have $\varepsilon' \in (0, 1)$, $\omega \in (0, 1/2)$, and
  \[
    \varepsilon = \frac{\varepsilon'}{2 \sqrt{2 n \ln(1/\omega)}} ,
  \]
  a short calculation\footnote{%
    \label{fn:ac-setting}
    Note $e^\varepsilon - 1 \leq 2 \varepsilon$ for $\varepsilon \in (0, 1)$ by
    convexity of $e^\varepsilon - 2\varepsilon - 1$. Then
    \begin{align*}
      \sqrt{2 n \ln(1/\omega)} \varepsilon + n \varepsilon(e^\varepsilon - 1)
      &\leq \sqrt{2 n \ln(1/\omega)} \varepsilon + 2 n \varepsilon^2 \\
      &= \frac{\varepsilon'}{2} + \frac{\varepsilon'}{2} \cdot \frac{\varepsilon'}{2 \ln(1/\omega)} \\
      &\leq \frac{\varepsilon'}{2} + \frac{\varepsilon'}{2} = \varepsilon',
    \end{align*}
    where the last inequality is because $\omega \in (0, 1/2)$ and $\varepsilon'
    \in (0, 1)$, and the last factor is maximized at $\varepsilon' = 1$ and $\omega
    = 1/2$:
    \[
      \frac{\varepsilon'}{2 \ln(1/\omega)} \leq \frac{1}{2 \ln (2)} < 1 .
    \]}
  shows that the composition is $(\varepsilon', \delta^*)$-differentially private.
\end{theorem}

We omit other standard composition theorems (e.g., parallel composition) as we
will not need them; readers can consult the textbook by \citet{DR14} for more
information.

\begin{remark} \label{rem:dp-cost}
  The sequential composition theorem allows reasoning about differential privacy
  in terms of \emph{privacy costs}. We can imagine tracking an algorithm's
  privacy parameters, initially $(0, 0)$. Every time the algorithm applies an
  $(\varepsilon, \delta)$-private mechanism, we increment the current parameters
  by $(\varepsilon, \delta)$; the final parameters give the privacy level for
  the whole algorithm. In this way, $(\varepsilon, \delta)$ represents the
  \emph{cost} of using a private subroutine.

  While this observation seems to be a restatement of the composition theorems,
  merely a convenient accounting method, the subtlety lies in how the costs are
  computed. The key point is that outputs from previous private mechanisms are
  assumed to be \emph{equal} when computing the cost of subsequent operations.
  Changing the perspective a bit, we can pay cost $(\varepsilon, \delta)$ to
  assume two outputs in related runs of an $(\varepsilon, \delta)$-private
  mechanism are equal. We can begin to see the rough contours of a proof by
  coupling; we will soon make this idea more precise.
\end{remark}

\section{Approximate liftings} \label{sec:alift-props}

Differential privacy is closely related to an approximate version of
probabilistic lifting first proposed by \citet{BKOZ13-toplas} and refined in
later work~\citep{BartheO13,OlmedoThesis}. These liftings are defined in terms
of a distance on distributions.

\begin{definition} \label{def:epsdist-privacy}
  Let $\mu_1, \mu_2$ be sub-distributions over $\cA$. The
  \emph{$\varepsilon$-distance} is defined as
  \[
    \epsdist{\varepsilon}{\mu_1}{\mu_2}
    \triangleq \max_{\cS \subseteq \cA} (\mu_1(\cS) - \exp(\varepsilon) \cdot  \mu_2(\cS)) .
  \]
  This quantity is non-negative since the right-hand side is zero for the empty
  subset $\cS = \varnothing$, but it is not a proper metric---it is not
  symmetric and the triangle inequality does not hold.\footnote{%
    Technically, $\varepsilon$-distance is an $f$-divergence with $f(t) = \max(t
  - \exp(\varepsilon), 0)$}
  If $M : \cD \to \Dist(\cR)$ is a mechanism with
  $\epsdist{\varepsilon}{M(d_1)}{M(d_2)} \leq \delta$ for every pair of adjacent
  $d_1, d_2$, then $M$ is $(\varepsilon, \delta)$-differentially private.
\end{definition}

We are now ready to define approximate liftings.
\begin{definition} \label{def:alift}
  Let $\mu_1, \mu_2$ be sub-distributions over $\cA_1$ and $\cA_2$ respectively
  and let $\cR \subseteq \cA_1 \times \cA_2$ be a relation. Let $\star$ be a
  distinguished element disjoint from $\cA_1$ and $\cA_2$; we write $\cS^\star$
  for the set $\cS \cup \{ \star \}$, and $\cR^\star$ for the relation $\cR \cup
  (\cA_1 \times \{ \star \}) \cup (\{ \star \} \times \cA_2)$ on $\cA_1^\star
  \times \cA_2^\star$.  Two sub-distributions $\mu_L, \mu_R$ over $\cA_1^\star
  \times \cA_2^\star$ are said to be \emph{witnesses} for the
  \emph{$(\varepsilon, \delta)$-approximate $\cR$-lifting} of $(\mu_1, \mu_2)$
  if:
  \begin{enumerate}
    \item $\pi_1(\mu_L) = \mu_1$ and $\pi_2(\mu_R) = \mu_2$;
    \item $\supp(\mu_L) \cup \supp(\mu_R) \subseteq \cR^\star$; and
    \item $\epsdist{\varepsilon}{\mu_L}{\mu_R} \leq \delta$.
  \end{enumerate}
  In the first point $\mu_1$ and $\mu_2$ are implicitly interpreted as
  distributions over $\cA_1^\star$ and $\cA_2^\star$ (i.e., placing zero
  probability on $\star$). We call these conditions the marginal, support, and
  distance conditions, respectively.

  The sub-distributions $\mu_L$ and $\mu_R$ are called \emph{left} and
  \emph{right witnesses} of the lifting.  When the particular witnesses are not
  important, $\mu_1$ and $\mu_2$ are said to be related by the
  \emph{$(\varepsilon, \delta)$-lifting of $\cR$}, denoted
  \[
    \mu_1 \alift{\cR}{(\varepsilon, \delta)} \mu_2 .
  \]
\end{definition}
Our definition generalizes an earlier definition of approximate lifting by
\citet{BartheO13}. The chief novelty is the element $\star$, which ensures each
element in $\cA_1$ and $\cA_2$ can be related to some element under $\cR$
(namely, $\star$). Somewhat paradoxically, the larger space of witnesses lets us
assume more structure on the witness distributions without loss of generality,
making it easier to manipulate and construct approximate liftings.

\subsection{Useful consequences}

The existence of an approximate lifting between two distributions can imply
useful properties about the two distributions. Many of these consequences recall
properties from \cref{sec:coupling-conseq}, with quantitative corrections for
the parameters $(\varepsilon, \delta)$.

\begin{proposition} \label{prop:alift-dp}
  Let $M : \cD \to \Dist(\cR)$ be a randomized algorithm. If for every pair of
  adjacent inputs $(d_1, d_2)$ the output distributions are related by an
  approximate lifting
  \[
    M(d_1) \alift{(=)}{(\varepsilon, \delta)} M(d_2) ,
  \]
  then $M$ is $(\varepsilon, \delta)$-differentially private.
\end{proposition}
\begin{proof}
  Fix a pair of adjacent inputs $(d_1, d_2)$, and let $\mu_L, \mu_R$ be the
  witnesses to the approximate lifting of the output distributions. For any
  subset $\cS \subseteq \cR$ of outputs, we have
  \begin{align}
    M(d_1)(\cS) &= \mu_L(\cS \times \cR^\star)
    \tag{first marginal} \\
    &= \mu_L(\{ (s, s) \mid s \in \cS \} \cup \cS \times \{ \star \})
    \tag{support} \\
    &\leq \exp(\varepsilon) \cdot \mu_R(\{ (s, s) \mid s \in \cS \} \cup \cS \times \{ \star \}) + \delta
    \tag{distance} \\
    &\leq \exp(\varepsilon) \cdot \mu_R( \cR^\star \times \cS^\star ) + \delta
    \tag{support} \\
    &= \exp(\varepsilon) \cdot M(d_2)(\cS) + \delta .
    \tag{second marginal}
  \end{align}
  Thus $M$ is $(\varepsilon, \delta)$-differentially private.
\end{proof}

The approximate lifted version of implication is also useful.

\begin{proposition} \label{prop:alift-impl}
  Let $\mu_1, \mu_2$ be sub-distributions over $\cA_1$ and $\cA_2$, and consider
  subsets $\cS_1 \subseteq \cA_1$, $\cS_2 \subseteq \cA_2$. Suppose we have an
  approximate lifting
  \[
    \mu_1 \alift{ \{ (a_1, a_2) \mid a_1 \in \cS_1 \to a_2 \in \cS_2 \} }{(\varepsilon, \delta)} \mu_2 .
  \]
  Then $\mu_1(\cS_1) \leq \exp(\varepsilon) \cdot \mu_2(\cS_2) + \delta$.
\end{proposition}
\begin{proof}
  Let $\mu_L, \mu_R$ witness the approximate lifting. Then,
  \begin{align}
    \mu_1(\cS_1) &= \mu_L(\cS_1 \times \cR^\star)
    \tag{first marginal} \\
    &= \mu_L( \cS_1 \times \cS_2 \cup \cS_1 \times \{ \star \})
    \tag{support} \\
    &\leq \exp(\varepsilon) \cdot \mu_R( \cS_1 \times \cS_2 \cup \cS_1 \times \{ \star \}) + \delta
    \tag{distance} \\
    &\leq \exp(\varepsilon) \cdot \mu_R( \cR^\star \times \cS_2^\star ) + \delta
    \tag{support} \\
    &= \exp(\varepsilon) \cdot \mu_2(\cS_2) + \delta
    \tag{second marginal}
  \end{align}
  as desired.
\end{proof}

We will see a partial converse in the next chapter (\cref{thm:opt-subset-math}).

\subsection{Structural properties}

Approximate liftings satisfy several natural structural properties. First of
all, they generalize exact liftings.

\begin{proposition} \label{prop:alift-plift}
  Let $\mu_1, \mu_2$ be sub-distributions over $\cA_1$ and $\cA_2$ with equal
  weights. We have the equivalence
  \[
    \mu_1 \lift{\cR} \mu_2 \quad\text{if and only if}\quad \mu_1 \alift{\cR}{(0, 0)} \mu_2 .
  \]
\end{proposition}
\begin{proof}
  The forward direction follows by taking both witnesses of the approximate
  lifting to be the witness of the exact lifting. For the reverse direction, let
  $\mu_L, \mu_R$ witness the approximate lifting. We have
  $\epsdist{0}{\mu_L}{\mu_R} \leq 0$ so $\mu_L(a_1, a_2) \leq \mu_R(a_1, a_2)$
  for every pair $(a_1, a_2) \in \cA_1 \times \cA_2$.  Since $\mu_1$ and $\mu_2$
  have equal weights, the marginal conditions imply $|\mu_L| = |\mu_R|$ and
  hence $\mu_L = \mu_R$. Since $\mu_L( \{ \star \} \times \cA_2) = \mu_R( \cA_1
  \times \{ \star \} ) = 0$, restricting to $\cA_1 \times \cA_2$ gives a witness
  for the exact lifting as desired.
\end{proof}

Second, we may assume witnesses only use pairs in the product of the supports of
the two related distributions.

\begin{proposition} \label{lem:alift-supp}
  Let $\mu_1$ and $\mu_2$ be sub-distributions over $\cA_1$ and $\cA_2$ with an
  approximate lifting
  \[
    \mu_1
    \mathrel{\alift{\cR}{(\varepsilon, \delta)}}
    \mu_2 .
  \]
  Then there are witnesses with support contained in $\supp(\mu_1)^\star \times
  \supp(\mu_2)^\star$.
\end{proposition}

This property is natural---$\mu_1$ and $\mu_2$ are fully defined by their
probabilities on supporting elements, so the witnesses shouldn't need to use
other elements. However, witnesses to an approximate lifting may have positive
mass on points $(a_1, a_2) \notin \supp(\mu_1) \times \supp(\mu_2)$ since the
marginal conditions only constrain one marginal of $\mu_L$ and $\mu_R$; mass can
be distributed arbitrarily along the unconstrained component. In fact, this
support property does not hold for prior definitions of approximate lifting. In
our definition, the $\star$ element serves as a canonical element where mass
outside of $\supp(\mu_1) \times \supp(\mu_2)$ can be located.

\begin{proof}
  Let $\mu_L$ and $\mu_R$ witness the approximate lifting and let $\cS_i
  \triangleq \supp(\mu_i)$ for $i \in \{ 1, 2 \}$.  We construct witnesses
  $\eta_L, \eta_R$ by shifting mass on points outside the support to $\star$,
  while preserving the marginals:
  \[
    \eta_L(a_1, a_2) \triangleq
    \begin{cases}
      \mu_L(a_1, a_2) &: (a_1, a_2) \in \cS_1 \times \cS_2 \\
      \sum_{a_2' \in \cA_2^\star \setminus \cS_2} \mu_L(a_1, a_2') &: a_2 = \star \\
      0 &: \text{otherwise}
    \end{cases}
  \]
  \[
    \eta_R(a_1, a_2) \triangleq
    \begin{cases}
      \mu_R(a_1, a_2) &: (a_1, a_2) \in \cS_1 \times \cS_2 \\
      \sum_{a_1' \in \cA_1^\star \setminus \cS_1} \mu_R(a_1', a_2) &: a_1 = \star \\
      0 &: \text{otherwise}
    \end{cases}
  \]
  It is straightforward to check $\pi_1(\eta_L) = \pi_1(\mu_L) = \mu_1$ and
  $\pi_2(\eta_R) = \pi_2(\mu_R) = \mu_2$, and $\eta_L, \eta_R$ have the
  necessary supports. It only remains to check the distance condition. By the
  distance condition on $\mu_L$ and $\mu_R$, there are non-negative constants
  $\delta(a_1, a_2)$ such that
  \[
    \mu_L(a_1, a_2) \leq \exp(\varepsilon) \cdot \mu_R(a_1, a_2) + \delta(a_1, a_2)
  \]
  for each $(a_1, a_2) \in \cA_1^\star \times \cA_2^\star$, with sum at most
  $\delta$. We define new constants
  \[
    \delta'(a_1, a_2) \triangleq
    \begin{cases}
      \delta(a_1, a_2) &: (a_1, a_2) \in \cS_1 \times \cS_2 \\
      \sum_{a_2' \in \cA_2^\star \setminus \cS_2} \delta(a_1, a_2') &: a_2 = \star \\
      0 &: \text{otherwise}
    \end{cases}
  \]
  and we claim
  \[
    \eta_L(a_1, a_2) \leq \exp(\varepsilon) \cdot \eta_R(a_1, a_2) + \delta'(a_1, a_2) .
  \]
  This is clear on $\cS_1 \times \cS_2$ and also when $a_1 = \star$, since
  $\eta_L(\star, a_2) = 0$. When $a_2 = \star$, unfolding definitions gives
  \begin{align*}
    \eta_L(a_1, \star)
    &= \sum_{a_2' \in \cA_2^\star \setminus \cS_2} \mu_L(a_1, a_2') \\
    &\leq \sum_{a_2' \in \cA_2^\star \setminus \cS_2} \exp(\varepsilon) \cdot \mu_R(a_1, a_2') + \delta(a_1, a_2') \\
    &= \sum_{a_2' \in \cA_2^\star \setminus \cS_2} \delta(a_1, a_2') \\
    &= \exp(\varepsilon) \cdot \eta_R(a_1, \star) + \delta'(a_1, \star)
  \end{align*}
  where the penultimate equality is because $\mu_R(a_1, a_2') = 0$ for $a_2'
  \notin \cS_2$, and the last equality is because $\eta_R(a_1, \star) = 0$ by
  definition. Finally,
  \[
    \sum_{(a_1, a_2) \in \cA_1^\star \times \cA_2^\star} \delta'(a_1, a_2)
    = \sum_{(a_1, a_2) \in \cA_1 \times \cA_2^\star} \delta(a_1, a_2)
    \leq \delta
  \]
  so the distance condition $\epsdist{\varepsilon}{\eta_L}{\eta_R} \leq \delta$
  holds.  Thus $\eta_L$ and $\eta_R$ witness the approximate lifting.
\end{proof}

Approximate liftings are also stable under mappings.

\begin{restatable}{theorem}{aliftextend} \label{thm:alift-extend}
  Let $\mu_1$ and $\mu_2$ be sub-distributions over $\cA_1$ and $\cA_2$. If we
  have functions $f_i : \cA_i \to \cB_i$ for $i \in \{ 1, 2 \}$, and a relation
  $\cR \subseteq \cB_1 \times \cB_2$, then
  \[
    \mu_1
    \mathrel{\alift{ \{ (a_1, a_2) \in \cA_1 \times \cA_2 \mid f_1(a_1) \mathrel{\cR} f_2(a_2) \}}
    {(\varepsilon, \delta)}}
    \mu_2
  \]
  if and only if
  \[
    \liftf{f_1}(\mu_1)
    \mathrel{\alift{ \{ (b_1, b_2) \in \cB_1 \times \cB_2 \mid b_1 \mathrel{\cR} b_2 \}}{(\varepsilon, \delta)}}
    \liftf{f_2}(\mu_2) .
  \]
  (Recall $f : \cA \to \cB$ can be lifted to a map $\liftf{f} : \SDist(\cA) \to
  \SDist(\cB)$ on sub-distributions.)
\end{restatable}

This theorem roughly says that we can change the basis of an approximate
lifting; namely, the ground sets of $\mu_1$ and $\mu_2$ and the ambient space of
the relation.  Several useful consequences follow.  First, if we take $f_1$ and
$f_2$ to inject from $\supp(\mu_1)$ and $\supp(\mu_2)$ into $\cB_1$ and $\cB_2$,
the reverse direction recovers \cref{lem:alift-supp}.  Second, if $\cE$ is a set
of equivalence classes of $\cA$ and $\mu/\cE \in \SDist(\cE)$ is the induced
distribution over equivalence classes, taking $f_1, f_2 : \cA \to \cE$ to map an
element to its equivalence class and $\cR$ to be the equivalence relation
$=_\cE$ recovers a result by \citet[Proposition 8]{BartheO13}:
\[
  \mu_1 \mathrel{\alift{(=_\cE)}{(\varepsilon, \delta)}} \mu_2
  \iff
  \mu_1/\cE \mathrel{\alift{(=)}{(\varepsilon, \delta)}} \mu_2/\cE .
\]
We frequently apply \cref{thm:alift-extend} with $f_1$ and $f_2$ projecting a
memory to the value in variables $x_1$ and $x_2$; by the reverse direction, we
can extend a lifting of the distributions over $x_1$ and $x_2$ to a lifting of
distributions over whole memories.

\begin{proof}
  For the forward direction, take witnesses $\mu_L, \mu_R \in \SDist(\cA_1^\star
  \times \cA_2^\star)$ and define witnesses for the desired approximate lifting
  $\eta_L \triangleq \liftf{(f_1^\star \times f_2^\star)}(\mu_L)$ and $\eta_R
  \triangleq \liftf{(f_1^\star \times f_2^\star)}(\mu_R)$, where $f_1^\star
  \times f_2^\star$ maps $(a_1, a_2) \mapsto (f_1(a_1), f_2(a_2))$ and maps
  $\star$ to $\star$ in both components.  The support condition is clear, the
  marginal requirement is clear, and the distance requirement follows easily:
  for any set $\cS \subseteq \cB_1^\star \times \cB_2^\star$, apply the distance
  condition on $\mu_L, \mu_R$ for the set $(f_1^\star \times
  f_2^\star)^{-1}(\cS)$.

  For the reverse direction, let $\eta_L, \eta_R \in \SDist(\cB_1^\star \times
  \cB_2^\star)$ witness the second approximate lifting. By
  \cref{lem:alift-supp}, without loss of generality $\supp(\eta_L)$ and
  $\supp(\eta_R)$ are contained in
  \begin{equation} \label{eq:supp-map}
    \supp(\liftf{f_1}(\mu_1))^\star \times \supp(\liftf{f_2}(\mu_2))^\star
    \subseteq f_1(\cA_1)^\star \times f_2(\cA_2)^\star .
  \end{equation}
  We construct a pair of witnesses $\mu_L, \mu_R \in \SDist(\cA_1^\star \times
  \cA_2^\star)$ to the first approximate lifting. The basic idea is to define
  $\mu_L$ and $\mu_R$ based on equivalence classes of elements in $\cA_i$ mapping
  to each $b_i \in \cB_i$, smoothing out the probabilities within each class to
  guarantee the distance condition. To begin, for $a_i \in \cA_i$ and $i \in \{ 1, 2
  \}$ we define
  \[
    [a_i]_{f_i} \triangleq f^{-1}_i(f_i(a_i))
    \quad \text{and} \quad
    \alpha_i(a_i) \triangleq \frac{\mu_i( a_i )}{\mu_i( [a_i]_{f_i} )} .
  \]
  We take $\alpha_i(a_i) = 0$ when $\mu_i([a_i]_{f_i}) = 0$, and we let
  $\alpha_i(\star) = 0$. We define $\mu_L$ and $\mu_R$ as
  \begin{align*}
    \mu_L(a_1, a_2) &\triangleq \alpha_L(a_1, a_2) \cdot \eta_L(f_1^\star(a_1), f_2^\star(a_2)) \\
    \mu_R(a_1, a_2) &\triangleq \alpha_R(a_1, a_2) \cdot \eta_R(f_1^\star(a_1), f_2^\star(a_2)) ,
  \end{align*}
  where
  \[
    \alpha_L(a_1, a_2) \triangleq
    \begin{cases}
      \alpha_1(a_1) \cdot \alpha_2(a_2) &: a_2 \neq \star \\
      \alpha_1(a_1) &: a_2 = \star
    \end{cases}
    \quad\text{and}\quad
    \alpha_R(a_1, a_2) \triangleq
    \begin{cases}
      \alpha_1(a_1) \cdot \alpha_2(a_2) &: a_1 \neq \star \\
      \alpha_2(a_2) &: a_1 = \star .
    \end{cases}
  \]
  The support and marginal conditions follow from the corresponding properties
  of $\eta_L$, $\eta_R$, e.g., 
  \begin{align*}
    \pi_1(\mu_L)(a_1)
    &= \sum_{a_2 \in \cA_2^\star} \alpha_L(a_1, a_2) \cdot \eta_L(f_1^\star(a_1), f_2^\star(a_2)) \\
    &= \alpha_1(a_1) \cdot \eta_L(f_1^\star(a_1), \star)
    + \sum_{a_2 \in \cA_2}
    \alpha_1(a_1) \cdot \alpha_2(a_2) \cdot \eta_L(f_1(a_1), f_2(a_2)) \\
    &= \alpha_1(a_1) \left( \eta_L(f_1^\star(a_1), \star)
      + \sum_{b_2 \in f_2(\cA_2)} \eta_L(f_1^\star(a_1), b_2) \sum_{a_2 \in f_2^{-1}(b_2)}
    \alpha_2(a_2) \right) \\
    &= \alpha_1(a_1) \left( \eta_L(f_1^\star(a_1), \star)
      + \sum_{b_2 \in f_2(\cA_2)} \eta_L(f_1(a_1), b_2) \right) \\
    &= \alpha_1(a_1) \sum_{b_2 \in \cB_2^\star} \eta_L(f_1^\star(a_1), b_2)
    = \alpha_1(a_1) \cdot \mu_1([a_1]_{f_1})
    = \mu_1(a_1) .
  \end{align*}
  The last equality replaces the sum over $b_2 \in f_2(\cA_2)^\star$ with a sum
  over $b_2 \in \cB_2^\star$; this holds since the support of
  $\liftf{f_2}(\mu_2)$ is contained in $f_2(\cA_2)$ so we may assume
  $\eta_L(f_1(a_1), b_2) = 0$ for all $b_2$ outside of $f_2(\cA_2)^\star$ by
  \cref{eq:supp-map}. Then we conclude by the marginal condition $\pi_1(\eta_L)
  = \liftf{f_1}(\mu_1)$. The second marginal is similar.
  
  To check the distance condition $\epsdist{\varepsilon}{\mu_L}{\mu_R} \leq
  \delta$, since $\epsdist{\varepsilon}{\eta_L}{\eta_R} \leq \delta$ there
  exists non-negative $\delta(b_1, b_2)$ with
  \[
    \eta_L(b_1, b_2) \leq \exp(\varepsilon) \cdot  \eta_R(b_1, b_2) + \delta(b_1, b_2)
  \]
  and $\sum_{b_1, b_2} \delta(b_1, b_2) \leq \delta$.  We may take
  $\delta(\star, b_2) = 0$ for all $b_2 \in \cB_2$, since $\eta_L(\star, b_2) =
  0$ by the marginal condition. We claim that for any $(a_1, a_2) \in
  \cA_1^\star \times \cA_2^\star$, we have
  $\mu_L(a_1, a_2) \leq \exp(\varepsilon) \cdot  \mu_R(a_1, a_2) + \zeta(a_1,
  a_2)$
  where
  \[
    \zeta(a_1, a_2) \triangleq
     \alpha_L(a_1,a_2) \cdot \delta(f_1^\star(a_1), f_2^\star(a_2)) .
  \]
  Let $b_i \triangleq f_i^\star(a_i)$ and consider $(a_1, a_2) \in \cA_1^\star
  \times \cA_2^\star$. If $a_1 = \star$ we can immediately bound
  \[
    \mu_L(\star, a_2) = 0 \leq \exp(\varepsilon) \cdot  \mu_R(\star, a_2) + \zeta(\star, a_2) .
  \]
  Otherwise $a_1 \neq \star$ and we can bound
  \begin{align*}
    \mu_L(a_1, a_2)
      &= \alpha_L(a_1, a_2) \cdot \eta_L(f_1^\star(a_1), f_2^\star(a_2)) \\
      &\leq \alpha_L(a_1, a_2) \cdot (\exp(\varepsilon) \cdot
    \eta_R(f_1^\star(a_1), f_2^\star(a_2)) + \delta(f_1^\star(a_1), f_2^\star(a_2))) \\
      &= \exp(\varepsilon) \cdot  (\alpha_R(a_1, a_2) \cdot
    \eta_R(f_1^\star(a_1), f_2^\star(a_2))) +
           \alpha_L(a_1, a_2) \cdot \delta(f_1^\star(a_1), f_2^\star(a_2)) \\
      &= \exp(\varepsilon) \cdot  \mu_R(a_1, a_2) + \alpha_L(a_1, a_2) \cdot
           \delta(f_1^\star(a_1), f_2^\star(a_2)) \\
      &= \exp(\varepsilon) \cdot  \mu_R(a_1, a_2) + \zeta(a_1, a_2) .
  \end{align*}
  The third line changes from $\alpha_L$ to $\alpha_R$ in the first term since
  $\alpha_L(a_1, a_2) \neq \alpha_R(a_1, a_2)$ exactly when $a_2 = \star$, when
  $\eta_R(f_1^\star(a_1), f_2^\star(a_2)) = \eta_R(f_1^\star(a_1), \star) = 0$
  as well.

  Now we just need to bound the sum of $\zeta(a_1, a_2)$ to conclude the
  distance bound between $\eta_L$ and $\eta_R$. First, the sum of $\alpha_L$
  within any equivalence class is $1$: for any $(b_1, b_2) \in \cB_1 \times
  \cB_2$, we have
  \begin{align*}
    \sum_{a_1 \in f_1^{-1}(b_1)} \sum_{a_2 \in f_2^{-1}(b_2)} \alpha_L(a_1, a_2)
    = \left( \sum_{a_1 \in f_1^{-1}(b_1)} \alpha_1(a_1) \right)
    \left( \sum_{a_2 \in f_2^{-1}(b_2)} \alpha_2(a_2) \right)
    = 1
  \end{align*}
  by definition. Therefore,
  \begin{align*}
    \sum_{(a_1, a_2) \in \cA_1^\star \times \cA_2^\star} \zeta(a_1, a_2) &=
      \sum_{(b_1, b_2) \in \cB_1^\star \times \cB_2^\star} \delta(b_1, b_2)
        \sum_{a_1 \in f_1^{-1}(b_1)}
        \sum_{a_2 \in f_2^{-1}(b_2)} \alpha_L(a_1, a_2) \\
      &= \sum_{(b_1, b_2) \in \cB_1 \times \cB_2} \delta(b_1, b_2)
      + \sum_{b_1 \in \cB_1} \delta(b_1, \star) \sum_{a_1 \in f_1^{-1}(b_1)} \alpha_1(a_1) \\
      &= \sum_{(b_1, b_2) \in \cB_1 \times \cB_2} \delta(b_1, b_2)
      + \sum_{b_1 \in \cB_1} \delta(b_1, \star) \\
      &= \sum_{(b_1, b_2) \in \cB_1^\star \times \cB_2^\star} \delta(b_1, b_2) \leq \delta.
  \end{align*}
  So for any $\cS \subseteq \cA_1^\star \times \cA_2^\star$ we have $\mu_L(\cS) \leq
  \exp(\varepsilon) \cdot  \mu_R(\cS) + \delta$, showing
  $\epsdist{\varepsilon}{\mu_L}{\mu_R} \leq \delta$ as desired.
\end{proof}

\subsection{From approximate liftings to approximate couplings}

Approximate liftings generalize probabilistic liftings (\cref{prop:alift-plift})
while retaining many features of their exact counterparts: the existence of an
approximate lifting with a certain support implies target properties about the
two related distributions (\cref{prop:alift-dp,prop:alift-impl}), and the
structural properties we saw for approximate liftings
(\cref{lem:alift-supp,thm:alift-extend}) also hold for probabilistic liftings.
Accordingly, we can think of approximate liftings as an approximate
generalization of probabilistic coupling; we will use the term \emph{approximate
coupling} to emphasize this point of view.

Unlike probabilistic coupling, whose definition and key properties have been
refined through decades of research, the proper definition of approximate
coupling is not settled. Other definitions have been proposed, and the relation
between the various notions is somewhat hazy. (See \cref{sec:alift-rw} for a
more detailed comparison.) Nevertheless, we present evidence that our
approximate lifting is the natural approximate counterpart of probabilistic
coupling---or at least, a highly promising candidate---by showing many desirable
properties hold and by exhibiting clean constructions.

However, so far we are still missing a major piece of the puzzle: how do we
construct approximate couplings? In other words, what is the approximate
analogue of proof by coupling? To work out what such a proof technique might
look like, we take inspiration from an existing program logic for approximate
liftings.

\section{The program logic \Saprhl} \label{sec:aprhl-core}

\citet{BKOZ13-toplas} proposed the relational program logic \Saprhl as an
approximate version of \Sprhl, targeting differential privacy. The basic idea is
to use approximate liftings in place of exact liftings, tracking the parameters
$(\varepsilon, \delta)$ in the judgments. We briefly review the language, the
judgments, and the logical rules.

\subsection{The language}

The language of \Saprhl is almost identical to the probabilistic imperative
language we used for \Sprhl. The only difference is instead of the uniform
distribution, we take the Laplace distribution as primitive:
\[
  \DExpr \coloneqq \Lap{\varepsilon}(e) .
\]
The parameter $\varepsilon$ quantifies the spread of the distribution, while the
parameter $e$ represents its mean; we treat $\varepsilon$ as a logical variable.
Similar to how we defined the Laplace mechanism (\cref{def:lapmech}), we
interpret $\Lap{\varepsilon}(e)$ as a discrete distribution over the integers $z
\in \ZZ$:
\[
  (\denot{\Lap{\varepsilon}(e)}_\rho m )(z)
  \triangleq \frac{\exp(- \denot{\varepsilon}_\rho \cdot | z - \denot{e}_\rho m |)}{W}
\]
where $\denot{e}_\rho m$ is an integer and $W$ normalizes the distribution to
have weight $1$:
\[
  W \triangleq \sum_{z \in \ZZ} \exp(- \denot{\varepsilon}_\rho \cdot | z - \denot{e}_\rho m |) .
\]
For example, the Laplace mechanism for a query $q : \cD \to \ZZ$ can be
implemented by sampling:
\[
  \Rand{x}{\Lap{\varepsilon}(q(d))} .
\]

\subsection{Judgments and validity}

Judgments in \Saprhl have the following form:
\[
  \aprhl{c_1}{c_2}{\Phi}{\Psi}{(\varepsilon, \delta)}
\]
Just like in \Sprhl, $\Phi$ and $\Psi$ are assertions on a product memory and
refer to variables tagged with $\sidel$ and $\sider$. The parameters
$\varepsilon, \delta$ are expressions involving constants and logical variables;
in particular, they do not mention program variables and do not depend on the
program state.

Validity for \Saprhl judgments is defined in terms of approximate liftings.

\begin{definition} \label{def:aprhl-valid}
  An \Saprhl judgment is \emph{valid} in logical context $\rho$, written
  \[
    \rho \models \aprhl{c_1}{c_2}{\Phi}{\Psi}{(\varepsilon, \delta)} ,
  \]
  if for any two memories $(m_1, m_2) \in \denot{\Phi}_\rho$ there exists an
  approximate lifting relating the output distributions:
  \[
    \denot{c_1}_\rho m_1
    \alift{\denot{\Psi}_\rho}{(\denot{\varepsilon}_\rho, \denot{\delta}_\rho)}
    \denot{c_2}_\rho m_2 .
  \]
\end{definition}

\subsection{Core proof rules}

Most of the rules in \Saprhl generalize rules from \Sprhl, with special handling
for the $(\varepsilon, \delta)$ parameters. We present the core proof system and
comment on departures from \Sprhl.

\begin{figure}
  \begin{mathpar}
  \inferruleref{Skip}
  {~}
  { \vdash \aprhl {\Skip}{\Skip} {\Phi}{\Phi}{(0, 0)} }
    \label{rule:aprhl-skip}
    \\
  \inferruleref{Assn}
  {~}
  { \vdash \aprhl
    {\Ass{x_1}{e_1}}{\Ass{x_2}{e_2}}
    {\Psi\subst{x_1\sidel,x_2\sider}{e_1\sidel,e_2\sider}}{\Psi}{(0, 0)} }
  \label{rule:aprhl-assn}
  \\
  \inferruleref{Lap}
  { x_1, x_2 \notin \FV(e_1, e_2) \\
    \Phi \triangleq |e_1\sidel - e_2\sider| \leq k \land \forall v \in \ZZ,\; \Psi\subst{x_1\sidel,x_2\sider}{v,v} }
  { \vdash \aprhl
    {\Rand{x_1}{\Lap{\varepsilon}(e_1)}}{\Rand{x_2}{\Lap{\varepsilon}(e_2)}}
    { \Phi } {\Psi}{(k \varepsilon, 0)} }
  \label{rule:aprhl-lap}
  \\
  \inferruleref{Seq}
  { \vdash \aprhl{c_1}{c_2}{\Phi}{\Psi}{(\varepsilon, \delta)} \\
    \vdash \aprhl{c_1'}{c_2'}{\Psi}{\Theta} {(\varepsilon', \delta')}}
  { \vdash \aprhl{c_1;c_1'}{c_2;c_2'}{\Phi}{\Theta}
    {(\varepsilon + \varepsilon', \delta + \delta')}}
  \label{rule:aprhl-seq}
  \\
  \inferruleref{Cond}
  { \models \Phi \to e_1\sidel = e_2\sider \\
    \vdash \aprhl{c_1}{c_2}{\Phi \land e_1\sidel}{\Psi}{(\varepsilon, \delta)} \\
    \vdash \aprhl{c_1'}{c_2'}{\Phi \land \neg e_1\sidel}{\Psi}{(\varepsilon, \delta)}  }
  {\vdash \aprhl{\Cond{e_1}{c_1}{c_1'}}{\Cond{e_2}{c_2}{c_2'}}
    {\Phi}{\Psi}{(\varepsilon, \delta)} }
  \label{rule:aprhl-cond}
  \\
  \inferruleref{While}
  { \models \Phi \land e_v\sidel \leq 0 \to \neg e_1\sidel \\
    \models \Phi \to e\sidel = e\sider \\\\
    \forall K \in \NN,\; \vdash \aprhl
    {c_1}{c_2}
    {\Phi \land e_1\sidel \land e_v\sidel = K}
    {\Phi \land e_v\sidel < K}{(\varepsilon, \delta)}  }
  { \vdash \aprhl
    {\WWhile{e_1}{c_1}}{\WWhile{e_2}{c_2}}
    {\Phi \land e\sidel \leq N}
    {\Phi \land \neg e_1\sidel}{(N\varepsilon, N\delta)}  }
  \label{rule:aprhl-while}
\end{mathpar}
  \caption{Two-sided \Saprhl rules \label{fig:aprhl-two-sided}}
\end{figure}

\medskip

We begin with the two-sided rules in \cref{fig:aprhl-two-sided}. The
\nameref{rule:aprhl-skip} and \nameref{rule:aprhl-assn} rules are lifted from
\Sprhl. To gain intuition for the sampling rule \nameref{rule:aprhl-lap}, we
first consider a special case:
\[
  \inferrule*[Left=Lap*]
  {~}
  { \vdash \aprhl{\Rand{x}{\Lap{\varepsilon}(e)}}{\Rand{x}{\Lap{\varepsilon}(e)}}
    {|e\sidel - e\sider| \leq k} {x\sidel = x\sider}{(k \varepsilon, 0)}}
\]
Since the means $e\sidel$ and $e\sider$ may not be equal, the two distributions
may have different probabilities of sampling the same value and there may be no
exact coupling guaranteeing $x\sidel = x\sider$.  Nevertheless, there is a $(k
\varepsilon, 0)$-approximate coupling when the means differ by at most $k$.
Since approximate lifting of equality models differential privacy, this rule
captures privacy of the Laplace mechanism (\cref{thm:lap-priv}). The full
sampling rule \nameref{rule:aprhl-lap} proves a general post-condition $\Psi$ if
it is true as a pre-condition, assuming the two sampled variables are equal.

The sequencing rule \nameref{rule:aprhl-seq} is similar to the sequencing rule
in \Sprhl, summing up the approximation parameters. This rule reflects a
composition principle for approximate couplings generalizing the sequential
composition theorem from differential privacy (\cref{thm:seq-comp}).

The conditional rule \nameref{rule:aprhl-cond} is similar to its counterpart
from \Sprhl.  Assuming the guards are equal initially, if there is an
$(\varepsilon, \delta)$-coupling of corresponding pairs of branches then there
is an $(\varepsilon, \delta)$-coupling of the two conditionals.  Finally, the
loop rule \nameref{rule:aprhl-while} applies to loops that run at most a finite
number of iterations $N$; this is enforced by the strictly decreasing integer
variant $e_v$. Given an $(\varepsilon, \delta)$-coupling for the loop bodies,
the rule produces a $(N\varepsilon, N\delta)$-coupling of the two loops. Again,
this rule corresponds to a sequential composition principle for approximate
couplings.

\begin{figure}
  \begin{mathpar}
  \inferruleref{Assn-L}
  {~}
  {\vdash \aprhl{\Ass{x_1}{e_1}}{\Skip}{\Psi\subst{x_1\sidel}{e_1\sidel}}{\Psi}{(0, 0)}}
  \label{rule:aprhl-assn-l}
  \\
  \inferruleref{Assn-R}
  {~}
  {\vdash \aprhl{\Skip}{\Ass{x_2}{e_2}}{\Psi\subst{x_2\sider}{e_2\sider}}{\Psi}{(0, 0)}}
  \label{rule:aprhl-assn-r}
  \\
  \inferruleref{Lap-L}
  {~}
  {\vdash \aprhl{\Rand{x_1}{\Lap{\varepsilon}(e_1)}}{\Skip}
    {\forall v \in \ZZ,\; \Psi\subst{x_1\sidel}{v}}{\Psi}{(0, 0)}}
  \label{rule:aprhl-lap-l}
  \\
  \inferruleref{Lap-R}
  {~}
  {\vdash \aprhl{\Skip}{\Rand{x_2}{\Lap{\varepsilon}(e_2)}}
    {\forall v \in \ZZ,\; \Psi\subst{x_2\sider}{v}}{\Psi}{(0, 0)}}
  \label{rule:aprhl-lap-r}
  \\
  \inferruleref{Cond-L}
  {\vdash \aprhl{c_1}{c}{\Phi\land e_1\sidel}{\Psi}{(\varepsilon, \delta)} \\
  \vdash \aprhl{c_1'}{c}{\Phi \land \neg e_1\sidel}{\Psi}{(\varepsilon, \delta)} }
  { \vdash \aprhl{\Cond{e_1}{c_1}{c_1'}}{c}{\Phi}{\Psi}{(\varepsilon, \delta)}  }
  \label{rule:aprhl-cond-l}
  \\
  \inferruleref{Cond-R}
  {\vdash \aprhl{c}{c_2}{\Phi\land e_2\sider}{\Psi}{(\varepsilon, \delta)} \\
  \vdash \aprhl{c}{c_2'}{\Phi \land \neg e_2\sider}{\Psi}{(\varepsilon, \delta)} }
  { \vdash \aprhl{c}{\Cond{e_2}{c_2}{c_2'}}{\Phi}{\Psi}{(\varepsilon, \delta)}  }
  \label{rule:aprhl-cond-r}
  \\
  \inferruleref{While-L}{ 
    \vdash \aprhl{c_1}{\Skip}{\Phi\land e_1\sidel}{\Phi}{(0, 0)} 
    \\\\
    \models \Phi \to \Phi_1\sidel \\
    \lless{\Phi_1}{\WWhile{e_1}{c_1}}}
    { \vdash \aprhl{\WWhile{e_1}{c_1}}{\Skip}{\Phi}{\Phi \land \neg e_1\sidel}{(0, 0)} }
  \label{rule:aprhl-while-l}
  \\
  \inferruleref{While-R}{ 
    \vdash \aprhl{\Skip}{c_2}{\Phi \land e_2\sider}{\Phi}{(0, 0)}
    \\\\
    \models \Phi \to \Phi_2\sider \\
    \lless{\Phi_2}{\WWhile{e_2}{c_2}}}
    { \vdash \aprhl{\Skip}{\WWhile{e_2}{c_2}}{\Phi}{\Phi \land \neg e_2\sider}{(0, 0)} }
  \label{rule:aprhl-while-r}
\end{mathpar}
  \caption{One-sided \Saprhl rules \label{fig:aprhl-one-sided}}
\end{figure}

\begin{figure}
  \begin{mathpar}
  \inferruleref{Conseq}
  { \vdash \aprhl{c_1}{c_2}{\Phi'}{\Psi'}{(\varepsilon', \delta')} \\
    \models \Phi \to \Phi' \\
    \models \Psi' \to \Psi \\
    \models \varepsilon' \leq \varepsilon \\
    \models \delta' \leq \delta
  }
  { \vdash \aprhl{c_1}{c_2}{\Phi}{\Psi}{(\varepsilon, \delta)} }
  \label{rule:aprhl-conseq}
  \\
  \inferruleref{Equiv}
  {\vdash \aprhl{c_1'}{c_2'}{\Phi}{\Psi}{(\varepsilon, \delta)} \\
  c_1 \equiv c_1' \\
  c_2 \equiv c_2'}
  {\vdash \aprhl{c_1}{c_2}{\Phi}{\Psi}{(\varepsilon, \delta)}}
  \label{rule:aprhl-equiv}
  \\
  \inferruleref{Case}{
    \vdash \aprhl{c_1}{c_2}{\Phi \land \Theta}{\Psi}{(\varepsilon, \delta)} \\
    \vdash \aprhl{c_1}{c_2}{\Phi \land \neg \Theta}{\Psi}{(\varepsilon, \delta)} }
  {\vdash \aprhl{c_1}{c_2}{\Phi}{\Psi}{(\varepsilon, \delta)}}
  \label{rule:aprhl-case}
  \\
  \inferruleref{Trans}{
    \vdash \aprhl{c_1}{c_2}{\Phi}{\Psi}{(\varepsilon, \delta)} \\
    \vdash \aprhl{c_2}{c_3}{\Phi'}{\Psi'}{(\varepsilon', \delta')} }
    {\vdash \aprhl{c_1}{c_3}{\Phi' \circ \Phi}{\Psi' \circ \Psi}
    {(\varepsilon + \varepsilon', \exp(\varepsilon') \delta + \delta')}}
  \label{rule:aprhl-trans}
  \\
  \inferruleref{Frame}
  { \vdash \aprhl{c_1}{c_2}{\Phi}{\Psi}{(\varepsilon, \delta)} \\
    \FV(\Theta) \cap \MV(c_1, c_2) = \varnothing }
  { \vdash \aprhl{c_1}{c_2}{\Phi \land \Theta}{\Psi \land \Theta}{(\varepsilon, \delta)} }
  \label{rule:aprhl-frame}
\end{mathpar}
  \caption{Structural \Saprhl rules \label{fig:aprhl-structural}}
\end{figure}

\medskip

The one-sided rules for \Saprhl are presented in \cref{fig:aprhl-one-sided}; the
structural rules, in \cref{fig:aprhl-structural}.  The one-sided sampling rules,
\nameref{rule:aprhl-lap-l} and \nameref{rule:aprhl-lap-r}, give a $(0,
0)$-lifting. The rule of consequence \nameref{rule:aprhl-conseq} allows
increasing the approximate parameters since larger parameters require a looser
bound between the witnesses. The other rules are straightforward generalizations
of their \Sprhl counterparts.

As expected, the logic is sound.
\begin{restatable}[Soundness of \Saprhl]{theorem}{aprhlsound} \label{thm:aprhl-sound}
  Let $\rho$ be a logical context. If a judgment is derivable
  \[
    \rho \vdash \aprhl{c_1}{c_2}{\Phi}{\Psi}{(\varepsilon, \delta)} ,
  \]
  then it is valid:
  \[
    \rho \models \aprhl{c_1}{c_2}{\Phi}{\Psi}{(\varepsilon, \delta)} .
  \]
\end{restatable}
\begin{proof}[Proof sketch]
  By induction on the derivation. The proof is very similar to the proof of
  soundness for \Saprhl by \citet{OlmedoThesis}, with some minor adjustments to
  handle the special element $\star$ in our definition of approximate coupling.
  \Cref{app:sound-aprhl} gives a self-contained proof of soundness for the full
  logic, including the new rules we will soon introduce.
\end{proof}

The natural counterpart to soundness is \emph{completeness}: valid judgments
should be provable by the proof system. \Saprhl is incomplete in at least one
respect: while valid judgments may relate commands that do not always terminate,
derivable judgments can only relate lossless programs. 

\begin{lemma} \label{lem:aprhl-ll}
  If $\rho \vdash \aprhl{c_1}{c_2}{\Phi}{\Psi}{(\varepsilon, \delta)}$ is
  derivable, then $c_1$ and $c_2$ are both $\Phi$-lossless.
\end{lemma}
\begin{proof}
  By induction on the derivation. Since the loop rule \nameref{rule:aprhl-while}
  requires both loops to terminate in at most $n$ iterations and the one-sided
  variants \nameref{rule:aprhl-while-l}/\nameref{rule:aprhl-while-r} assume
  losslessness, $c_1$ and $c_2$ must be lossless under the pre-condition.
\end{proof}

This kind of incompleteness aside, it is not known whether \Saprhl is complete
for terminating programs (or even \emph{relatively complete} in some natural
sense); we discuss this issue further in \cref{chap:future}.

\section{Proof by approximate coupling} \label{sec:aprhl-pbac}

Much like \Sprhl is a logic for formal proofs by coupling, \Saprhl can be viewed
as a logic for formal \emph{proofs by approximate coupling}. With the logical
rules in hand, we can work out an intuitive understanding of these proofs.

First of all, the close resemblance between \Sprhl and \Saprhl indicates that
proofs by approximate couplings are broadly similar to proofs by coupling; the
sampling rule \nameref{rule:aprhl-lap} shows we can choose an approximate
coupling for sampling statements (although for the moment we have just one
choice), the sequencing rule \nameref{rule:aprhl-seq} indicates that we can
sequence approximate couplings together, and the case rule
\nameref{rule:aprhl-case} lets us select an approximate coupling based on the current
state of the coupled executions.

The main difference is we must track the approximation parameters $\varepsilon$
and $\delta$ as we build the coupling. When we apply the sampling rule
\nameref{rule:aprhl-lap}, for instance, we accrue parameters $(k \cdot
\varepsilon, 0)$ where $k$ is an upper bound on the distance between the means.
In the sequencing rule \nameref{rule:aprhl-seq} (and similarly in the loop rule
\nameref{rule:aprhl-while}), we add up the approximate couplings parameters for
the sequenced commands.  The resulting style of analysis blends proof by
coupling with the cost interpretation of differential privacy
(\cref{rem:dp-cost}). For instance, we can think of the rule
\nameref{rule:aprhl-lap} as paying for the privacy cost to couple the
samples to be equal.  Accordingly, proofs by approximate coupling recover proofs
by the standard composition theorem (\cref{thm:seq-comp}). By introducing other
approximate couplings for the Laplace distribution, we can achieve clean and
compositional approximate coupling proofs of privacy even when the standard
composition theorem from differential privacy does not suffice.

\section{New couplings for the Laplace distribution} \label{sec:aprhl-lap}

\begin{figure}
  \begin{mathpar}
    \inferruleref{LapNull}
    { x_1, x_2 \notin \FV(e_1, e_2)
      \\\\
      \Phi \triangleq \forall w_1, w_2 \in \ZZ,\;
      w_1 - w_2 = e_1\sidel - e_2\sider \to
      \Psi\subst{x_1\sidel,x_2\sider}{w_1, w_2} }
    { \vdash \aprhl {\Rand{x_1}{\Lap{\varepsilon}(e_1)}} {\Rand{x_2}{\Lap{\varepsilon}(e_2)}}
      {\Phi} {\Psi} {(0,0)} }
  \label{rule:aprhl-lapnull}
  \\
    \inferruleref{LapGen}
    { x_1, x_2 \notin \FV(e_1, e_2)
      \\\\
      \Phi \triangleq |k + e_1\sidel - e_2\sider| \leq k' \land \forall w_1, w_2 \in \ZZ,\;
      w_1 + k = w_2 \to \Psi\subst{x_1\sidel,x_2\sider}{w_1, w_2} }
    { \vdash \aprhl
      {\Rand{x_1}{\Lap{\varepsilon}(e_1)}}{\Rand{x_2}{\Lap{\varepsilon}(e_2)}}
      {\Phi} {\Psi} {(k' \cdot \varepsilon,0)} }
  \label{rule:aprhl-lapgen}
  \end{mathpar}
  \caption{New Laplace rules for \Saprhl}
  \label{fig:aprhl-lap}
\end{figure}

Unlike the rule \nameref{rule:prhl-sample} in \Sprhl, which can couple two
uniform distributions in different ways by varying the bijection, the Laplace
rule \nameref{rule:aprhl-lap} can only couple samples to be equal. To support
richer proofs, we introduce two new approximate couplings for the Laplace
distribution and build them into \Saprhl rules.

\subsection{Null coupling}

Suppose we want to couple the Laplace distributions $\Lap{\varepsilon}(v_1)$ and
$\Lap{\varepsilon}(v_2)$.  Sampling from these distributions is equivalent to
sampling from $\Lap{\varepsilon}(0)$ and then adding $v_1$ and $v_2$
respectively, so we can couple by using equal draws from $\Lap{\varepsilon}(0)$.
Since equal draws from the same distribution have the same probability, this
approximate coupling is in fact an exact, $(0, 0)$-coupling, an analogue of the
identity coupling (\cref{fact:eq-couple}). More formally, we have the following
result.

\begin{restatable}{proposition}{lapnull} \label{prop:lapnull-math}
  Let $v_1, v_2 \in \ZZ$. Then:
  \[
    \Lap{\varepsilon}(v_1)
    \alift{ \{ (x_1, x_2) \mid x_1 - v_1 = x_2 - v_2 \} }{(0, 0)}
    \Lap{\varepsilon}(v_2) .
  \]
\end{restatable}
\begin{proof}
  We construct witnesses $\mu_L, \mu_R \in \Dist(\ZZ^\star \times \ZZ^\star)$.
  Define the relation
  \[
    \cR \triangleq \{ (x_1, x_2) \in \ZZ \times \ZZ \mid x_1 - v_1 = x_2 - v_2 \}
  \]
  and let $L(r)$ be probability $\Lap{\varepsilon}(0)$ produces $r$.  Define
  witnesses
  \[
    \mu_L(x_1, x_2) = \mu_R(x_1, x_2) \triangleq
    \begin{cases}
      L(x_1 - v_1) &: (x_1, x_2) \in \cR  \\
      0            &: \text{otherwise} .
    \end{cases}
  \]
  Since $\mu_L = \mu_R$, we know $\epsdist{0}{\mu_L}{\mu_R} = 0$. Also,
  $\supp(\mu_L)$ and $\supp(\mu_R)$ lie in $\cR \subseteq \cR^\star$.  So, it
  remains to check the marginal conditions.  Using the support condition, we have
  \[
    \pi_1(\mu_L)(r)
    = \mu_L(r, r - v_1 + v_2) = L(r - v_1)
    = \Lap{\varepsilon}(v_1)(r) .
  \]
  A similar calculation shows $\pi_2(\mu_R) = \Lap{\varepsilon}(v_2)$, so
  $\mu_L$ and $\mu_R$ witness the approximate coupling.
\end{proof}

We can capture this approximate coupling with the rule
\nameref{rule:aprhl-lapnull} in \cref{fig:aprhl-lap}. To gain intuition, the
following rule is a simplified special case:
\[
    \inferrule*[Left=LapNull*]
    { x \notin \FV(e) }
    { \vdash \aprhl {\Rand{x}{\Lap{\varepsilon}(e)}} {\Rand{x}{\Lap{\varepsilon}(e)}}
      {\top} {x\sidel - x\sider = e\sidel - e\sider}{(0,0)} }
\]
The coupling keeps the distance between the samples the same as the distance
between the means. The general rule \nameref{rule:aprhl-lapnull} can prove
post-conditions of any shape.

\begin{theorem} \label{thm:lapnull-sound}
  The rule \nameref{rule:aprhl-lapnull} is sound.
\end{theorem}
\begin{proof}
  We leave the logical context $\rho$ implicit.  Let $V \triangleq \Var
  \setminus \{ x_1, x_2 \}$ be the non-sampled variables; we write $m[V]$ for
  the restriction of a memory $m$ to variables in $V$.  Consider any two
  memories $m_1, m_2$, let the means be $v_1 \triangleq \denot{e_1} m_1$ and
  $v_2 \triangleq \denot{e_2} m_2$, and let the output distributions be
  \[
    \mu_1 \triangleq \denot{ \Rand{x_1}{\Lap{\varepsilon}(e_1)} } m_1
    \quad \text{and} \quad
    \mu_2 \triangleq \denot{ \Rand{x_2}{\Lap{\varepsilon}(e_2)} } m_2 .
  \]
  We construct an approximate coupling between $\mu_1$ and $\mu_2$. By
  \cref{prop:lapnull-math} we have
  \[
    \Lap{\varepsilon}(v_1)
    \alift{ \{ (x_1, x_2) \mid x_1 - v_1 = x_2 - v_2 \} }{(0, 0)}
    \Lap{\varepsilon}(v_2) .
  \]
  By \cref{thm:alift-extend} with maps $\denot{x_1}$ and $\denot{x_2}$, we
  obtain
  \[
    \mu_1
    \alift{ \denot{x_1\sidel - v_1 = x_2\sider - v_2} }{(0, 0)}
    \mu_2 .
  \]
  By the free variable condition, $v_1 = \denot{e_1} m_1'$ and $v_2 =
  \denot{e_2} m_2'$ for every memory $m_1' \in \supp(\mu_1)$ and $m_2' \in
  \supp(\mu_2)$, so we may assume by \cref{lem:alift-supp} that the witnesses
  are supported on such memories, giving witnesses to
  \[
    \mu_1
    \alift{ \denot{x_1\sidel - e_1\sidel = x_2\sider - e_2\sider} }{(0, 0)}
    \mu_2 .
  \]
  Also by the free variable condition, $m_1'[V] = m_1[V]$ and $m_2'[V] = m_2[V]$ so
  \[
    \mu_1
    \alift{ \{ (m_1', m_2') \mid m_1'[V] = m_1[V],\; m_2'[V] = m_2[V],\;
    m_1'(x_1) - \denot{e_1}m_1 = m_2'(x_2) - \denot{e_2}m_2 \} }{(0, 0)}
    \mu_2 .
  \]
  By the pre-condition, $(m_1, m_2)$ satisfy
  \[
     \forall w_1, w_2 \in \ZZ,\; 
      w_1 - w_2 = e_1\sidel - e_2\sider \to
      \Psi\subst{x_1\sidel,x_2\sider}{w_1, w_2}
  \]
  and so
  \[
    \mu_1
    \alift{\Psi}{(0, 0)}
    \mu_2 ,
  \]
  showing \nameref{rule:aprhl-lapnull} is sound.
\end{proof}

\subsection{General coupling}

Our next approximate coupling shifts the samples apart by a constant amount.
Suppose we want to approximately couple the Laplace distributions
$\Lap{\varepsilon}(v_1)$ and $\Lap{\varepsilon}(v_2)$ so that the samples $x_1,
x_2$ satisfy $x_1 + k = x_2$.  Intuitively, the approximation parameters should
depend on the shift $k$ and the distance $|v_1 - v_2|$ between the
means---larger shifts and larger distances imply that we match samples with
greater difference in probabilities.  More formally, we have the following
result.

\begin{restatable}{proposition}{lapgen} \label{prop:lapgen-math}
  Let $k, k', v_1, v_2 \in \ZZ$, and suppose $|k + v_1 - v_2| \leq k'$.
  Then:
  \[
    \Lap{\varepsilon}(v_1)
    \alift{ \{ (x_1, x_2) \mid x_1 + k = x_2 \} }{(k' \varepsilon, 0)}
    \Lap{\varepsilon}(v_2) .
  \]
\end{restatable}
\begin{proof}
  We need two witnesses $\mu_L, \mu_R \in \Dist(\ZZ^\star \times \ZZ^\star)$. Define the
  relation
  \[
    \cR \triangleq \{ (x_1, x_2) \in \ZZ \times \ZZ \mid x_1 + k = x_2 \}
  \]
  and let $L(r)$ be the probability $\Lap{\varepsilon}(0)$ produces $r$.
  Define witnesses
  \[
    \mu_L(x_1, x_2) \triangleq
    \begin{cases}
      L(x_1 - v_1) : (x_1, x_2) \in \cR  \\
      0            : \text{otherwise}
    \end{cases}
  \quad\text{and}\quad
    \mu_R(x_1, x_2) \triangleq
    \begin{cases}
      L(x_2 - v_2) : (x_1, x_2) \in \cR  \\
      0            : \text{otherwise} .
    \end{cases}
  \]
  By definition, $\supp(\mu_L)$ and $\supp(\mu_R)$ lie in $\cR \subseteq \cR^\star$.
  The marginal conditions are easy to check. So, it remains to check the
  distance condition. It suffices to show
  \[
    \mu_L(x_1, x_2) \leq \exp(k' \varepsilon) \cdot \mu_R(x_1, x_2)
  \]
  for every $(x_1, x_2) \in \ZZ^\star \times \ZZ^\star$, since summing over any
  set $\cS \subseteq \ZZ^\star \times \ZZ^\star$ gives $\mu_L(\cS) \leq
  \exp(k' \varepsilon) \cdot \mu_R(\cS)$.
 
  Clearly the claim is true for $(x_1, x_2) \notin \cR$; note that $\mu_L$ and
  $\mu_R$ are both zero when $x_1$ or $x_2$ is $\star$. Otherwise we just need
  to bound
  \[
    L(x_1 - v_1) \leq \exp(k' \varepsilon) \cdot L(x_2 - v_2)
  \]
  where $x_1 + k = x_2$. Unfolding definitions, it suffices to bound
  \[
    \exp( -|x_1 - v_1| \varepsilon )
    \leq \exp(k' \varepsilon) \cdot \exp( -|x_1 + k - v_2| \varepsilon ) ,
  \]
  which follows by assumption:
  \[
    |x_1 + k - v_2| - |x_1 - v_1| \leq |k - v_2 + v_1| \leq k' .
  \]
  So, $\epsdist{k' \varepsilon}{\mu_L}{\mu_R} \leq 0$ and $\mu_L, \mu_R$
  witness the approximate coupling.
\end{proof}

\todo[inline]{SAZ: Are there any pictures you can use to illustrate these definitions?}

This approximate coupling is modeled by the rule \nameref{rule:aprhl-lapgen}, in
\cref{fig:aprhl-lap}. Note that $k$ and $k'$ must be logical expressions, independent of
the program state. To gain intuition, the following rule is a simplified special
case:
\[
    \inferrule*[Left=LapGen*]
    { }
    { \vdash \aprhl
      {\Rand{x}{\Lap{\varepsilon}(e)}}{\Rand{x}{\Lap{\varepsilon}(e)}}
      {|k + e\sidel - e\sider| \leq k'}
      {x \sidel + k = x \sider}
      {(k' \cdot \varepsilon,0)} }
\]
As expected, the post-condition ensures that the coupled samples are shifted
apart by $k$. The approximation parameters scale as $k'$; this measures the
difference between the $k$-shifted means. As a sanity check, when $k' = 0$ the
distribution means are shifted by $k$ and we have an exact, $(0, 0)$-coupling.
The general rule \nameref{rule:aprhl-lapgen} can prove post-conditions of any
shape.

\begin{theorem} \label{thm:lapgen-sound}
  The rule \nameref{rule:aprhl-lapgen} is sound.
\end{theorem}
\begin{proof}
  We leave the logical context $\rho$ implicit.  Let $V \triangleq \Var
  \setminus \{ x_1, x_2 \}$ be the non-sampled variables; we write $m[V]$ for
  the restriction of a memory $m$ to variables in $V$.  Consider any two
  memories $m_1, m_2$, let the means be $v_1 \triangleq \denot{e_1} m_1$ and
  $v_2 \triangleq \denot{e_2} m_2$ such that $|k + v_1 - v_2| \leq k'$, and let
  the output distributions be
  \[
    \mu_1 \triangleq \denot{ \Rand{x_1}{\Lap{\varepsilon}(e_1)} } m_1
    \quad \text{and} \quad
    \mu_2 \triangleq \denot{ \Rand{x_2}{\Lap{\varepsilon}(e_2)} } m_2 .
  \]
  We construct an approximate coupling between $\mu_1$ and $\mu_2$. By
  \cref{prop:lapgen-math}, we have
  \[
    \Lap{\varepsilon}(v_1)
    \alift{ \{ (x_1, x_2) \mid x_1 + k = x_2 \} }{(k' \varepsilon, 0)}
    \Lap{\varepsilon}(v_2) .
  \]
  By \cref{thm:alift-extend} with maps $\denot{x_1}$ and $\denot{x_2}$, we
  obtain
  \[
    \mu_1
    \alift{ \denot{x_1\sidel + k = x_2\sider}_\rho }{(k' \varepsilon, 0)}
    \mu_2 .
  \]
  By the free variable condition, $m_1'[V] = m_1[V]$ and $m_2'[V] = m_2[V]$ for all
  memories $m_1' \in \supp(\mu_1)$ and $m_2' \in \supp(\mu_2)$, so we may assume
  by \cref{lem:alift-supp} that the witnesses are supported on such memories.
  Hence, we have the following lifting:
  \[
    \mu_1
    \alift{ \{ (m_1', m_2') \mid m_1'[V] = m_1[V],\; m_2'[V] = m_2[V],\;
    m_1'(x_1) + k = m_2'(x_2) \in \denot{x_1\sidel + k = x_2\sider} \} }{(k' \varepsilon, 0)}
    \mu_2
  \]
  By the pre-condition, $(m_1, m_2)$ satisfy
  \[
    \forall w_1, w_2 \in \ZZ,\; 
    w_1 + k = w_2 \to \Psi\subst{x_1\sidel,x_2\sider}{w_1, w_2}
  \]
  and so
  \[
    \mu_1
    \alift{\Psi}{(k' \varepsilon, 0)}
    \mu_2 ,
  \]
  showing \nameref{rule:aprhl-lapgen} is sound.
\end{proof}

\begin{remark}
  If we could take $k' \triangleq 0$ and $k \triangleq e_2\sider - e_1\sidel$ in
  \nameref{rule:aprhl-lapgen}, we would recover \nameref{rule:aprhl-lapnull}.
  However, $k$ must be a constant or logical variable. (We will discuss possible
  ways to lift this requirement in \cref{sec:concurrent}.)
\end{remark}

\section{Pointwise privacy} \label{sec:aprhl-pweq}

In sophisticated privacy proofs, it is often convenient to focus on a single
output at a time. We call this pattern \emph{pointwise equality} and formalize
it as the following property of approximate couplings.

\begin{restatable}{proposition}{pweq} \label{prop:pw-eq-math}
  Let $\mu_1, \mu_2$ be sub-distributions over $\cR$ and suppose for every $i
  \in \cR$, we have
  \[
    \mu_1
    \alift{ \{ (r_1, r_2) \mid r_1 = i \to r_2 = i \} }{(\varepsilon, \delta_i)}
    \mu_2
  \]
  for non-negative $\varepsilon$ and $\{ \delta_i \}_{i \in \cR}$.  Then we have
  \[
    \mu_1
    \alift{ (=) }{(\varepsilon, \delta)}
    \mu_2
  \]
  where $\delta = \sum_{i \in \cR} \delta_i$.
\end{restatable}
\begin{proof}
  By \cref{prop:alift-impl} we know for every $i \in \cR$,
  \[
    \mu_1(i) \leq \exp(\varepsilon) \cdot \mu_2(i) + \delta_i .
  \]
  So for any subset $\cS \subseteq \cR$, summing over $i \in \cS$ yields
  \[
    \mu_1(\cS)
    \leq \exp(\varepsilon) \cdot \mu_2(\cS) + \sum_{i \in \cS} \delta_i
    \leq \exp(\varepsilon) \cdot \mu_2(\cS) + \delta
  \]
  since $\delta_i \geq 0$. We define witnesses $\mu_L(r, r) \triangleq \mu_1(r)$
  and $\mu_R(r, r) \triangleq \mu_2(r)$ for $r \neq \star$, and zero otherwise.
  The support and marginal conditions are easy to check. For the distance
  condition, consider any set $\cT \subseteq \cR^\star \times \cR^\star$ and let
  \[
    \cT' \triangleq \cT \cap \{ (r_1, r_2) \in \cR \times \cR \mid r_1 = r_2 \} .
  \]
  We know $\mu_L(\cT) = \mu_L(\cT')$ and $\mu_R(\cT) = \mu_R(\cT')$.
  Letting $\cS' = \{ r \in \cR \mid (r, r) \in \cT' \}$, we have
  \[
    \mu_L(\cT')
    = \mu_1(\cS')
    \leq \exp(\varepsilon) \cdot \mu_2(\cS') + \delta
    \leq \exp(\varepsilon) \cdot \mu_R(\cT') + \delta
  \]
  so $\epsdist{\varepsilon}{\mu_L}{\mu_R} \leq \delta$ and we have witnesses as
  desired.
\end{proof}

\begin{figure}
  \begin{mathpar}
    \inferruleref{PW-Eq}
    { \models \sum_{i \in \cR} \delta' \subst{\gamma}{i} \leq \delta \\
      \gamma \notin \FV(\Phi,c_1,c_2,e_1,e_2,\varepsilon,\delta)
      \\\\
      \forall \gamma \in \cR,\;
      \vdash \aprhl{c_1}{c_2}
      {\Phi}{e_1\sidel = \gamma \to e_2\sider = \gamma}{(\varepsilon, \delta')} }
    { \vdash \aprhl{c_1}{c_2}
      {\Phi}{e_1\sidel = e_2\sider}{(\varepsilon, \delta)} }
      \label{rule:aprhl-pw-eq}
  \end{mathpar}
  \caption{Pointwise equality rule \nameref{rule:aprhl-pw-eq} for \Saprhl}
  \label{fig:aprhl-pw-eq}
\end{figure}

Pointwise equality simplifies coupling proofs of differential privacy: rather
than proving differential privacy in one shot, we can give a separate proof for
each possible output and then combine the results.  We can internalize pointwise
equality as the \Saprhl rule \nameref{rule:aprhl-pw-eq} in
\cref{fig:aprhl-pw-eq}. In the premise, the pointwise judgment is indexed by a
logical variable $\gamma$. The conclusion gives an approximate lifting of
equality in the post-condition.

\begin{theorem} \label{thm:pw-eq-sound}
  The rule \nameref{rule:aprhl-pw-eq} is sound. 
\end{theorem}
\begin{proof}
  Let $\rho$ be the logical context.  The proof follows essentially by
  \cref{prop:pw-eq-math}, handling the logical variables carefully. Consider two
  memories $(m_1, m_2) \in \denot{\Phi}_\rho$ and output distributions
  \[
    \mu_1 \triangleq \denot{c_1}_\rho m_1
    \quad \text{and} \quad
    \mu_2 \triangleq \denot{c_2}_\rho m_2 .
  \]
  We construct an approximate lifting relating $\mu_1$ and $\mu_2$. By the free
  variable condition, $(m_1, m_2) \in \denot{\Phi}_{\rho \cup \gamma \mapsto i}$
  for any $i$ and so by validity of the premises, we can use the forward
  direction of \cref{thm:alift-extend} to project the liftings in the premises
  to the expressions $e_1$ and $e_2$:
  \[
    \liftf{(\denot{e_1}_{\rho \cup \gamma \mapsto i})} (\mu_1)
    \alift{ \{ (a_1, a_2) \in \cR \times \cR \mid a_1 = i \to a_2 = i \} }{(\varepsilon, \delta')}
    \liftf{(\denot{e_2}_{\rho \cup \gamma \mapsto i})} (\mu_2)
  \]
  for each $i \in \cR$. (Technically $\varepsilon$ and $\delta'$ are also
  interpreted in the logical context $\rho \cup \gamma \mapsto i$; we elide
  this.) By the free variable condition, the two projected distributions are in
  fact the same for all $i$, and everything besides $\delta'$ can be interpreted
  in the original context $\rho$. \Cref{prop:pw-eq-math} with $\delta_i
  \triangleq \denot{\delta'}_{\rho \cup \gamma \mapsto i}$ gives
  \[
    \liftf{\denot{e_1}_{\rho}} (\mu_1)
    \alift{ \{ (a_1, a_2) \mid a_1 = a_2 \} }{(\varepsilon, \delta)}
    \liftf{\denot{e_2}_{\rho}} (\mu_2) ,
  \]
  and the reverse direction of \cref{thm:alift-extend} with maps
  $\denot{e_1}_\rho$ and $\denot{e_2}_\rho$ gives
  \[
    \mu_1
    \alift{ (\denot{e_1\sidel = e_2\sider}_\rho) }{(\varepsilon, \delta)}
    \mu_2 .
  \]
  Thus, \nameref{rule:aprhl-pw-eq} is sound.
\end{proof}

\begin{remark}
From a logical perspective, pointwise equality resembles the \emph{Leibniz
equality} principle:
\[
  \models (\forall i,\; x = i \to y = i) \to x = y .
\]
Indeed, if \Saprhl had a structural rule to convert an external universal
quantifier into an internal universal quantifier, e.g., something like
\[
  \inferrule*[Left=Forall]
  { \forall i,\; \vdash \aprhl{c_1}{c_2}{\Phi}{\Psi_i}{(\varepsilon, \delta)} }
  { \vdash \aprhl{c_1}{c_2}{\Phi}{\forall i,\; \Psi_i}{(\varepsilon, \delta)} }
\]
\nameref{rule:aprhl-pw-eq} could be derived using the rule of consequence with
Leibniz equality. Unfortunately this rule is not sound, not even in \Sprhl. In
fact, if we have just two judgments with post-conditions $\Psi_1$ and $\Psi_2$,
it is not sound in general to combine them into a single judgment with
post-condition $\Psi_1 \land \Psi_2$: the underlying witnesses may be different.
The rule \nameref{rule:aprhl-pw-eq} is a special case where we may safely
combine post-conditions across different judgments.
\end{remark}

\begin{remark}
  On a more practical note, the post-condition in \nameref{rule:aprhl-pw-eq} is
  highly specific---the assertion must be equality. In \cref{chap:combining} we
  will see some ways to partially mitigate this limitation, for instance by
  incorporating one-sided conjuncts (\cref{sec:aprhl-utb}).
\end{remark}

\section{Coupling proofs of privacy} \label{sec:aprhl-ex}

Approximate coupling proofs are a convenient and compositional tool for proving
differential privacy.  Starting from two adjacent inputs, we select an
approximate coupling for each pair of corresponding sampling instructions such
that (i) the total cost does not exceed the target privacy parameters
$(\varepsilon, \delta)$, and (ii) the outputs on the two executions are equal
under the approximate coupling.  By pointwise equality, we can sometimes
establish point (ii) by building an approximate coupling separately for each
possible output value $i$, ensuring that if the first output is equal to $i$
then the second output is also equal to $i$. We will apply the asymmetric
approximate couplings from \cref{sec:aprhl-lap} to induce this kind of
asymmetric relation on outputs.

Compared to existing proofs of privacy, approximate coupling proofs are simpler
and more concise, abstracting away reasoning about conditional probabilities.
To demonstrate, we prove differential privacy for two algorithms from the
privacy literature. We present each proof twice: first as an approximate
coupling proof, then as a formal proof in \Saprhl.

\subsection{The Report-noisy-max mechanism}

Our first example is called \emph{Report-noisy-max}~\citep{DR14}. Given a list
of numeric queries $q_1, \dots, q_N : \cD \to \ZZ$ and a database $d \in \cD$,
this mechanism computes $q_i(d)$ for each $i$ and adds fresh Laplace noise to
each answer, releasing the index $i$ with the highest noisy answer.
Alternatively, we can think of each query as indexed by an element $r$ in some
finite range $\cR$, where $q_r$ computes the \emph{score} for $r$ given private
data $d$.  Then Report-noisy-max is a close cousin to the well-known
\emph{Exponential mechanism} of \citet{MT07}, which finds an element with a high
score while preserving privacy.

Suppose $\evalQ(i, d)$ returns $q_i(d)$.  We implement Report-noisy-max as the
following program $\mathit{rnm}$:
\[
  \begin{array}{l}
    \Ass{\mathit{maxA}}{0}; \\
    \Ass{\mathit{maxI}}{0}; \\
    \Ass{i}{1}; \\
    \WWhile{i \leq N}{} \\
    \quad \Rand{a}{\Lap{\varepsilon/2}(\evalQ(i, d))}; \\
    \quad \Condt{\mathit{maxI} = 0 \lor a > \mathit{maxA}}{} \\
    \quad\quad \Ass{\mathit{maxA}}{a}; \Ass{\mathit{maxI}}{i}; \\
    \quad \Ass{i}{i + 1}
  \end{array}
\]
The variable $\mathit{maxI}$ stores the output of the mechanism; we assume it
ranges over $\NN$.
\begin{theorem}
  Suppose each query $q_i$ is $1$-sensitive: $|q_i(d) - q_i(d')| \leq 1$
  for adjacent databases $d, d'$. Then executing $\mathit{rnm}$ and returning
  $\mathit{maxI}$ is $\varepsilon$-differentially private.
\end{theorem}

While we could prove privacy with the sequential composition theorem
(\cref{thm:seq-comp}), we would get an overly conservative bound of
$(N\varepsilon, 0)$-privacy since we must pay for each Laplace sampling.
Report-noisy-max is an example of a mechanism where the precise analysis showing
$(\varepsilon, 0)$-privacy previously required an ad hoc proof. However, since
approximate couplings satisfy a more general composition principle, we can prove
privacy for this mechanism compositionally.

\begin{proof}[Proof by approximate coupling]
  Consider a possible output $j \in \NN$. We construct an $(\varepsilon,
  0)$-approximate coupling such that if $\mathit{maxI}\sidel = j$, then
  $\mathit{maxI}\sider = j$. If $j = 0$ this is easy since the only way
  $\mathit{maxI} = 0$ is if $N = 0$ and the loops terminate immediately. If $j >
  N$ this is also easy, as $\mathit{maxI}$ cannot exceed $N$.

  So suppose $j \in [1, N]$. In iterations $i \neq j$, we couple the samplings
  so both runs use the same amount of noise:
  \[
    a\sidel - \evalQ(i\sidel, d\sidel) = a\sider - \evalQ(i\sider, d\sider) .
  \]
  In particular, $a\sider \leq a\sidel + 1$. This is a $(0, 0)$-approximate
  coupling for each iteration. For iteration $i = j$, we couple so the noisy
  answer in the second run is one larger than the noisy answer in the first run:
  \[
    a\sidel + 1 = a\sider .
  \]
  The true answers $\evalQ(i\sidel, d\sidel)$ and $\evalQ(i\sider, d\sider)$ are
  at most $1$ apart and the shift is $1$. Since we use Laplace noise with
  parameter $\varepsilon/2$, this is a $(2 \cdot \varepsilon/2,0) =
  (\varepsilon, 0)$-coupling.

  Now if the noisy answer on iteration $j$ is the highest noisy answer in the
  first run, then it must also be the highest noisy answer in the second run: by
  the coupling, $a\sidel + 1 = a\sider$ for iteration $j$ and $a\sider \leq
  a\sidel + 1$ for all other iterations. The total cost is $(\varepsilon, 0)$,
  establishing $(\varepsilon, 0)$-differential privacy.
\end{proof}

\begin{remark}
  Earlier versions of Report-noisy-max also returned the noisy answer
  $\mathit{maxA}$ in addition to the index $\mathit{maxI}$. However, subtle
  errors in the privacy proof were later discovered; a correct proof of privacy
  is currently not known. Attempting a proof by approximate coupling immediately
  runs into a problem: we have coupled $a\sidel + 1 = a\sider$ for the critical
  iteration, but we need $a\sidel = a\sider$ if we are to safely return the
  noisy answer too.
\end{remark}

In order to perform this proof in \Saprhl, the main complication is to only pay
for coupling the critical iteration $j$.  Directly applying the loop rule would
give an overly conservative guarantee of $(N \varepsilon, 0)$-privacy since
\nameref{rule:aprhl-while} assumes each iteration has the same cost.  To get
around this problem, we first use the program equivalence rule to split the loop
into three separate pieces: iterations before $j$, iteration $j$, and iterations
after $j$. Then we arrange a $(0, 0)$-coupling for each iteration in the first
and last loops, and an $(\varepsilon, 0)$-coupling for the middle loop
consisting of just the critical iteration.
\begin{theorem}
  Suppose each query $q_i$ is $1$-sensitive: $|q_i(d) - q_i(d')| \leq 1$
  for adjacent databases $d, d'$. Then the following judgment is derivable in
  \Saprhl:
  \[
    \vdash \aprhl{\mathit{rnm}}{\mathit{rnm}}{\mathit{Adj}(d\sidel, d\sider)}
    {\mathit{maxI}\sidel = \mathit{maxI}\sider}{(\varepsilon, 0)}
  \]
\end{theorem}

\begin{proof}
  We verify an equivalent program, dividing the loop in three:
  \[
    \begin{array}{l}
      \Ass{\mathit{maxA}}{0}; \\
      \Ass{\mathit{maxI}}{0}; \\
      \Ass{i}{1}; \\
      \WWhile{i \leq N \land i < j}{} \\
      \quad \Rand{a}{\Lap{\varepsilon/2}(\evalQ(i, d))}; \\
      \quad \Condt{\mathit{maxI} = 0 \lor a > \mathit{maxA}}{} \\
      \quad\quad {\Ass{\mathit{maxA}}{a}; \Ass{\mathit{maxI}}{i}}; \\
      \quad \Ass{i}{i + 1}; \\
      \WWhile{i \leq N \land i = j}{} \\
      \quad \Rand{a}{\Lap{\varepsilon/2}(\evalQ(i, d))}; \\
      \quad \Condt{\mathit{maxI} = 0 \lor a > \mathit{maxA}}{} \\
      \quad\quad {\Ass{\mathit{maxA}}{a}; \Ass{\mathit{maxI}}{i}}; \\
      \quad \Ass{i}{i + 1}; \\
      \WWhile{i \leq N}{} \\
      \quad \Rand{a}{\Lap{\varepsilon/2}(\evalQ(i, d))}; \\
      \quad \Condt{\mathit{maxI} = 0 \lor a > \mathit{maxA}}{} \\
      \quad\quad {\Ass{\mathit{maxA}}{a}; \Ass{\mathit{maxI}}{i}}; \\
      \quad \Ass{i}{i + 1}
    \end{array}
  \]
  We call this program $\mathit{rnm}'$. Our goal is to prove the pointwise judgment
  \[
    \vdash \aprhl{\mathit{rnm}'}{\mathit{rnm}'}{\mathit{Adj}(d\sidel, d\sider)}
    {\mathit{maxI}\sidel = j \to \mathit{maxI}\sider = j}{(\varepsilon, 0)}
  \]
  for every $j \in \NN$. When $j = 0$ or $j > N$, the proof is
  straightforward---in the first case $N = 0$, and in the second case
  $\mathit{maxI}\sidel = j$ must be false.  So we focus on the more interesting
  cases, $j \in [1, N]$.  The initial assignment statements can be handled with
  \nameref{rule:aprhl-assn}. Let the three loops be $w_<$, $w_=$, and $w_>$, and
  let $\mathit{body}$ be the common loop body. Define the following invariants:
  \begin{align*}
    \Theta_< &\triangleq \begin{cases}
      |\mathit{maxA}\sidel - \mathit{maxA}\sider| \leq 1 \\
      \mathit{maxI}\sidel < i\sidel \land \mathit{maxI}\sider < i\sidel \\
      \neg (i\sidel \leq N \land i\sidel < j) \to i\sidel = j
    \end{cases} \\
    \Theta_= &\triangleq \begin{cases}
      |\mathit{maxA}\sidel - \mathit{maxA}\sider| \leq 1 \\
      \mathit{maxI}\sidel = j \to (\mathit{maxI}\sider = j \land \mathit{maxA}\sidel + 1 = \mathit{maxA}\sider) \\
      \neg (i\sidel \leq N \land i\sidel = j) \to i\sidel = j + 1
    \end{cases} \\
    \Theta_> &\triangleq \begin{cases}
      i\sidel > j \\
      \mathit{maxI}\sidel = j \to (\mathit{maxI}\sider = j \land \mathit{maxA}\sidel + 1 = \mathit{maxA}\sider)
    \end{cases}
  \end{align*}
  We leave the invariant $\mathit{Adj}(d\sidel, d\sider) \land i\sidel =
  i\sider$ implicit and we prove three judgments corresponding to the three
  cases. First, we have
  \[
    \vdash \aprhl{\mathit{body}}{\mathit{body}}{\Theta_<}{\Theta_<}{(0, 0)}
  \]
  by coupling the Laplace samplings using \nameref{rule:aprhl-lapnull}, ensuring
  $|\mathit{maxA}\sidel - \mathit{maxA}\sider| \leq 1$. Thus, we have the
  following judgment for the first loop by \nameref{rule:aprhl-while}:
  \[
    \vdash \aprhl{w_<}{w_<}{\Theta_<}{\Theta_< \land \neg(i\sidel \leq N \land i\sidel < j)}{(0, 0)}
    .
  \]
  For the next loop body, we have
  \[
    \vdash \aprhl{\mathit{body}}{\mathit{body}}{\Theta_=}{\Theta_=}{(\varepsilon, 0)}
  \]
  by coupling the Laplace samplings using \nameref{rule:aprhl-lapgen} with $k =
  1, k' = 2$, ensuring $a\sidel + 1 = a\sider$. Combined with the pre-condition
  $\Theta_=$, if the first run exceeds $\mathit{maxA}\sidel$ then the second run
  also exceeds $\mathit{maxA}\sider$.  By the rule \nameref{rule:aprhl-lapgen},
  this coupling costs $(\varepsilon, 0)$. Since this loop runs for just one
  iteration, we have a judgment for the second loop by
  \nameref{rule:aprhl-while}:
  \[
    \vdash \aprhl{w_=}{w_=}{\Theta_=}{\Theta_= \land \neg(i\sidel \leq N \land i\sidel = j)}{(\varepsilon, 0)} .
  \]
  Finally for the last loop, we have
  \[
    \vdash \aprhl{\mathit{body}}{\mathit{body}}{\Theta_>}{\Theta_>}{(0, 0)}
  \]
  by coupling the samplings using \nameref{rule:aprhl-lapnull}. Applying
  \nameref{rule:aprhl-while} gives a similar judgment for the last loop:
  \[
    \vdash \aprhl{w_>}{w_>}{\Theta_>}{\Theta_>}{(0, 0)}
  \]
  We can combine the three loop judgments while summing the approximation parameters
  with \nameref{rule:aprhl-seq}, applying the rule of consequence with
  implications
  \begin{align*}
    &\models \Theta_< \land \neg(i\sidel \leq N \land i\sidel < j) \to \Theta_= \\
    &\models \Theta_= \land \neg(i\sidel \leq N \land i\sidel = j) \to \Theta_>
  \end{align*}
  to establish
  \[
    \vdash \aprhl{\mathit{rnm}'}{\mathit{rnm}'}{\mathit{Adj}(d\sidel, d\sider)}
    {\mathit{maxI}\sidel = j \to \mathit{maxI}\sider = j}{(\varepsilon, 0)} .
  \]
  We conclude differential privacy by applying \nameref{rule:aprhl-pw-eq} and
  \nameref{rule:aprhl-equiv}:
  \[
    \vdash \aprhl{\mathit{rnm}}{\mathit{rnm}}{\mathit{Adj}(d\sidel, d\sider)}
    {\mathit{maxI}\sidel = \mathit{maxI}\sider}{(\varepsilon, 0)} .
    \qedhere
  \]
\end{proof}

\begin{remark}
  Report-noisy-max draws noise from the Laplace distribution. A slight variant
  of this algorithm uses the \emph{one-sided Laplace} distribution, also called
  the \emph{exponential} distribution, to achieve similar results. This variant
  is closely related to the Exponential mechanism of \citet{MT07}; for instance,
  if we add noise from the continuous exponential distribution, Report-noisy-max
  is equivalent to the Exponential mechanism \citep[Theorem 3.13]{DR14}.

  Replacing the Laplace distribution with the one-sided Laplace distribution
  makes the privacy proof only a bit more difficult. While privacy still does
  not follow from the standard composition theorem---in fact, there is now
  nothing to compose because sampling from the one-sided Laplace distribution
  isn't differentially private---we can give a similar proof by approximate
  coupling. The main difference is in the rule \nameref{rule:aprhl-lapgen}: the
  analogous rule for the one-sided Laplace distribution has a slightly stronger
  pre-condition, with $0 \leq k + e_1\sidel - e_2\sider \leq k'$ in place of $|k
  + e_1\sidel - e_2\sider| \leq k'$. Our coupling proof is otherwise unchanged.
\end{remark}

\subsection{The Sparse Vector mechanism} \label{subsec:sv}

Our second example is the \emph{Sparse Vector mechanism}, a well-known algorithm
with a challenging privacy proof. At least six variants were thought to be
proved private, only for subtle mistakes to later surface in four of them
\citep{lyu2016understanding}. Perhaps the canonical (correct) version of the
algorithm can be found in the textbook by \citet{DR14}, where it is called
\textsc{NumericSparse}.  This mechanism takes a numeric \emph{threshold} $T \in
\ZZ$, a cutoff $C \in \NN$, a list of numeric queries $q_1, \dots, q_N : \cD \to
\ZZ$, and a database $d \in \cD$ as input. Sparse Vector releases the indices of
the first $C$ queries that have answer approximately above the threshold, along
with noisy answers for each of these queries.  The mechanism adds Laplace noise
to the threshold and Laplace noise to each query answer, returning the queries
with noisy answer above the noisy threshold.  Again, the challenge in the
privacy analysis is to only pay for above-threshold queries, rather than all
queries.\footnote{%
  A precursor of this algorithm was designed to release a private version of a
  vector of numbers where most of the entries are known to be zero, i.e., a
sparse vector.}

\begin{figure}
\[
  \begin{array}{l}
    \Ass{i}{1};
    \Ass{\mathit{out}}{[]}; \\
    \Rand{t}{\Lap{\varepsilon/4}(T)}; \\
    \WWhile{i \leq N \land |\mathit{out}| < C}{} \\
    \quad \Ass{\mathit{ans}}{(0, 0)};
    \Ass{\mathit{go}}{\kwtrue}; \\
    \quad \WWhile{i \leq N \land \mathit{go}}{} \\
    \quad\quad \Rand{a}{\Lap{\varepsilon/8C}(\evalQ(i, d))}; \\
    \quad\quad \Condt{a > t}{} \\
    \quad\quad\quad \Rand{\mathit{noisy}}{\Lap{\varepsilon/4C}(\evalQ(i, d))}; \\
    \quad\quad\quad \Ass{\mathit{ans}}{(i, \mathit{noisy})}; \\
    \quad\quad\quad \Ass{\mathit{out}}{\mathit{ans} :: \mathit{out}}; \\
    \quad\quad\quad \Ass{\mathit{go}}{\kwfalse}; \\
    \quad\quad \Ass{i}{i + 1}
  \end{array}
\]
\caption{Sparse Vector}
\label{fig:code-sv}
\end{figure}

\Cref{fig:code-sv} presents the code for the Sparse Vector algorithm.  The
variable $\mathit{out}$ stores a list of pairs of an index and a noisy answer
for each query that is approximately above-threshold; the list is initially
empty and pairs are added to the head. The algorithm stops when it answers $C$
queries or when it processes all $N$ queries. The code is structured in a
slightly artificial way---the queries are broken into chunks, where each
iteration of the outer loop corresponds to exactly one above-threshold query. In
their presentation, \citet{DR14} first prove privacy for a subroutine called
\textsc{AboveThreshold}---which randomizes the threshold and executes one
iteration of the outer loop---by carefully manipulating conditional
probabilities. They then verify the whole mechanism \textsc{NumericSparse} by
composing calls to \textsc{AboveThreshold} and applying the standard composition
theorem (\cref{thm:seq-comp}).

Rather than re-randomize the threshold after every answered query, we add noise
to the threshold just at the beginning of the algorithm; this variant was
independently proposed by \citet{lyu2016understanding}. Accordingly, it is no
longer possible to analyze the outer loop iterations via standard privacy
composition since each iteration of the outer loop is not differentially
private. While using the same noise for the threshold does not affect the
asymptotic accuracy of Sparse Vector, practical applications may benefit.

\begin{theorem}
  Suppose each query $q_i$ is $1$-sensitive: $|q_i(d) - q_i(d')| \leq 1$ for
  adjacent databases $d, d'$, and the threshold $T$ is the same for both runs.
  Then Sparse Vector is $\varepsilon$-differentially private.
\end{theorem}

\begin{proof}[Proof by approximate coupling]
  We first couple the threshold sampling so $t\sidel + 1 = t\sider$. The means
  are $0$ apart, the coupled samples are $1$ apart, and the noise is from the
  Laplace distribution with parameter $\varepsilon/4$, so this is an $(1 \cdot
  \varepsilon/4, 0) = (\varepsilon/4, 0)$ approximate coupling.  Assuming this
  coupling, we argue that the two executions of the inner loop can be
  approximately coupled to satisfy $\mathit{ans}\sidel = \mathit{ans}\sider$.
  We consider the inner loop and construct an approximate coupling such that if
  $\mathit{ans}\sidel = (j, v)$ then $\mathit{ans}\sider = (j, v)$ as well.

  Just like we did in the proof of Report-noisy-max, we use different couplings
  depending on where we are in the loop relative to iteration $j$. In iterations
  before $j$, we use the null coupling when sampling $a$ and $\mathit{noisy}$ to
  give a $(0, 0)$-approximate coupling such that $|a\sidel - a\sider| \leq 1$;
  this ensures that if we don't go above threshold in the first execution before
  $j$, then we also don't go above threshold in the second execution before $j$.
  We take the same $(0, 0)$-coupling for iterations after $j$. In the critical
  iteration $j$, we couple the samplings for $a$ to ensure $a\sidel + 1 =
  a\sider$ and we couple $\mathit{noisy}\sidel = \mathit{noisy}\sider$ if
  necessary.  Combined with the threshold coupling $t\sidel + 1 = t\sider$, this
  ensures that if we go above threshold in iteration $j$ in the first execution
  then we also go above threshold in iteration $j$ in the second execution, and
  the noisy answers for above-threshold queries are equal.
  
  To compute the approximation parameters, the coupling for $a$ is an
  $(\varepsilon/4C, 0)$-approximate coupling: the distance between the coupled
  samples is at most $2$ greater than the distance between the means, and the
  noise is drawn from $\Lap{\varepsilon/8C}$. The coupling for $\mathit{noisy}$
  is the standard coupling for the Laplace mechanism; it is an $(\varepsilon/4C,
  0)$-approximate coupling since the true answers are at most $1$ apart and the
  noise is drawn from $\Lap{\varepsilon/4C}$. So, iteration $j$ uses an
  $(\varepsilon/4C + \varepsilon/4C, 0) = (\varepsilon/2C, 0)$-approximate
  coupling and the loop coupling ensures that if $\mathit{ans}\sidel = (j, v)$
  then $\mathit{ans}\sider = (j, v)$. This gives an $(\varepsilon/2C,
  0)$-approximate coupling for the inner loop ensuring $\mathit{ans}\sidel =
  \mathit{ans}\sider$ and preserving $\mathit{out}\sidel = \mathit{out}\sider$.
  
  Since there are at most $C$ iterations of the outer loop, we have an
  $(\varepsilon/2, 0)$-approximate coupling ensuring $\mathit{out}\sidel =
  \mathit{out}\sider$ at the end of the algorithm. Combined with the
  $(\varepsilon/2,0)$-coupling for the threshold, this shows that Sparse Vector
  is $(\varepsilon,0)$-differentially private.
\end{proof}

\begin{remark}
  Earlier versions of Sparse Vector returned the noisy answers without adding
  fresh noise. These variants are now known to be incorrect:
  \citet{lyu2016understanding} show they are not $\varepsilon$-differentially
  private for any finite $\varepsilon$. If we attempt a proof by approximate
  coupling we immediately hit a problem: after coupling $a\sidel + 1= a\sider$,
  returning the noisy answer $a$ is not differentially private.  By itself, this
  obstacle doesn't show the algorithm is not differentially private.  However,
  it suggests that something may be amiss.
\end{remark}

We can also give a more formal proof of privacy in \Saprhl. Like we did for
Report-noisy-max, we transform the loops in order to apply couplings with
different costs in different iterations. Sparse Vector also introduces an
additional complication: under the we will build coupling, we don't know the
inner loops remain synchronized. So, we work with the following, equivalent
implementation:
\[
  \begin{array}{l}
    \Ass{i}{1};
    \Ass{\mathit{out}}{[]}; \\
    \Rand{t}{\Lap{\varepsilon/4}(T)}; \\
    \WWhile{i \leq N \land |\mathit{out}| < C}{} \\
    \quad \Ass{\mathit{ans}}{(0, 0)};
    \Ass{\mathit{go}}{\kwtrue}; \\
    \quad \WWhile{i \leq N}{} \\
    \quad\quad \Rand{a}{\Lap{\varepsilon/8C}(\evalQ(i, d))}; \\
    \quad\quad \Condt{a > t \land \mathit{go}}{} \\
    \quad\quad\quad \Rand{\mathit{noisy}}{\Lap{\varepsilon/4C}(\evalQ(i, d))}; \\
    \quad\quad\quad \Ass{\mathit{ans}}{(i, \mathit{noisy})}; \\
    \quad\quad\quad \Ass{\mathit{go}}{\kwfalse}; \\
    \quad\quad \Ass{i}{i + 1}; \\
    \quad \Condt{p_1(\mathit{ans}) \neq 0} \\
    \quad\quad \Ass{\mathit{out}}{\mathit{ans} :: \mathit{out}};
    \Ass{i}{p_1(\mathit{ans}) + 1} \\
  \end{array}
\]

Compared to the more straightforward implementation in \cref{fig:code-sv}, the
main difference is that the inner loop passes over all the queries. Once the
inner loop finds an above threshold query, the algorithm sets the flag
$\mathit{go}$ and the inner loop skips over all remaining queries. Then if an
above-threshold query was found in the inner loop, the index in $\mathit{ans}$
must be non-zero.  In this case, the outer loop records the answer and resets
the counter to the query after the most recent above-threshold query (recall
$p_1$ returns the first element of a pair).  By running through all the queries,
the inner loops can be analyzed in a synchronized fashion. We call the inner
loop $\mathit{aboveT}$, and the whole program program $\mathit{sparseV}$.

\begin{theorem}
  Suppose each query $q_i$ is $1$-sensitive: $|q_i(d) - q_i(d')| \leq 1$
  for adjacent databases $d, d'$ and the threshold $T$ is the same for both
  runs. Then the following judgment is derivable in \Saprhl:
  \[
    \vdash \aprhl{\mathit{sparseV}}{\mathit{sparseV}}
    {\mathit{Adj}(d\sidel, d\sider)}{\mathit{out}\sidel = \mathit{out}\sider}{(\varepsilon, 0)}
  \]
\end{theorem}

\begin{proof}
  We elide the adjacency assertion $\mathit{Adj}(d\sidel, d\sider)$ and
  synchronization assertion $i\sidel = i\sider$ since they are preserved
  throughout the proof. Let's first consider the inner loop $\mathit{aboveT}$.
  We prove the following judgment for every pair $(j, v) \in \NN \times \ZZ$:
  \[
    \vdash \aprhl{\mathit{aboveT}}{\mathit{aboveT}}{t\sidel + 1 = t\sider}
    {\mathit{ans}\sidel = (j, v) \to \mathit{ans}\sider = (j, v)}
    {(\varepsilon/2C, 0)}
  \]
  The cases $j = 0$ and $j > N$ are trivial, so we consider $j \in [1, N]$. Much
  like we did for Report-noisy-max, we split the loop into three pieces:
  iterations before $j$, iteration $j$, and iterations after $j$.
  \[
    \begin{array}{l}
      \WWhile{i \leq N \land i < j}{} \\
      \quad \Rand{a}{\Lap{\varepsilon/8C}(\evalQ(i, d))}; \\
      \quad \Condt{a > t \land \mathit{go}}{} \\
      \quad\quad \Rand{\mathit{noisy}}{\Lap{\varepsilon/4C}(\evalQ(i, d))}; \\
      \quad\quad \Ass{\mathit{ans}}{(i, \mathit{noisy})};
      \Ass{\mathit{go}}{\kwfalse}; \\
      \quad \Ass{i}{i + 1}; \\
      \WWhile{i \leq N \land i = j}{} \\
      \quad \Rand{a}{\Lap{\varepsilon/8C}(\evalQ(i, d))}; \\
      \quad \Condt{a > t \land \mathit{go}}{} \\
      \quad\quad \Rand{\mathit{noisy}}{\Lap{\varepsilon/4C}(\evalQ(i, d))}; \\
      \quad\quad \Ass{\mathit{ans}}{(i, \mathit{noisy})};
      \Ass{\mathit{go}}{\kwfalse}; \\
      \quad \Ass{i}{i + 1}; \\
      \WWhile{i \leq N}{} \\
      \quad \Rand{a}{\Lap{\varepsilon/8C}(\evalQ(i, d))}; \\
      \quad \Condt{a > t \land \mathit{go}}{} \\
      \quad\quad \Rand{\mathit{noisy}}{\Lap{\varepsilon/4C}(\evalQ(i, d))}; \\
      \quad\quad \Ass{\mathit{ans}}{(i, \mathit{noisy})};
      \Ass{\mathit{go}}{\kwfalse}; \\
      \quad \Ass{i}{i + 1}
    \end{array}
  \]
  We call this program $\mathit{aboveT}'$, the loops $w_<$, $w_=$, and $w_>$,
  and body of the loop $\mathit{body}$. We take invariants:
  \begin{align*}
    \Theta_< &\triangleq \begin{cases}
      t\sidel + 1 = t\sider \\
      \mathit{go}\sidel \to \mathit{go}\sider \\
      \neg(i\sidel \leq N \land i\sidel < j) \to i\sidel = j
    \end{cases} \\
    \Theta_= &\triangleq \begin{cases}
      t\sidel + 1 = t\sider \\
      \mathit{go}\sidel \to \mathit{go}\sider \\
      \mathit{ans}\sidel = (j, v) \to \mathit{ans}\sider = (j, v) \\
      \neg(i\sidel \leq N \land i\sidel = j) \to i\sidel = j + 1
    \end{cases} \\
    \Theta_> &\triangleq \begin{cases}
      t\sidel + 1 = t\sider \\
      i\sidel > j \\
      \mathit{ans}\sidel = (j, v) \to \mathit{ans}\sider = (j, v)
    \end{cases}
  \end{align*}
  We begin with the first loop. To show
  \[
    \vdash \aprhl{\mathit{body}}{\mathit{body}}
    {i\sidel \leq N \land i\sidel < j \land \Theta_<}{\Theta_<}{(0, 0)} ,
  \]
  we couple the sampling for $a$ with the null coupling \nameref{rule:aprhl-lapnull} so that
  \[
    |a\sidel - a\sider| = |\evalQ(i\sidel, d\sidel) - \evalQ(i\sider, d,\sider)| \leq 1 .
  \]
  For the conditional statements we use the one-sided rules
  \nameref{rule:aprhl-cond-l} and \nameref{rule:aprhl-cond-r}, giving four
  possible cases for the guard $a > t \land \mathit{go}$ in the two executions:
  \begin{description}
    \item[(True, True)]
      We use \nameref{rule:aprhl-lapnull} to couple the samplings for
      $\mathit{noisy}$ and establish $\neg \mathit{go}\sidel$.
    \item[(True, False)]
      We use \nameref{rule:aprhl-lap-l} for sampling
      $\mathit{noisy}\sidel$ to establish $\neg \mathit{go}\sidel$.
    \item[(False, True)]
      If $\mathit{go}\sidel$ is false, then we use \nameref{rule:aprhl-lap-r}
      for sampling $\mathit{noisy}\sider$ and conclude $\mathit{go}\sidel \to
      \mathit{go}\sider$.

      If $\mathit{go} \sidel$ is true, then $a\sidel$ must be below threshold
      but this case is impossible: $a\sider$ must be above threshold but the
      thresholds are coupled so that $t\sidel + 1 = t\sider$ and $|a\sidel -
      a\sider| \leq 1$.
    \item[(False, False)]
      We use \nameref{rule:aprhl-skip}, preserving $\mathit{go}\sidel \to \mathit{go}\sider$.
  \end{description}
  Putting the cases together, we have
  \[
    \vdash \aprhl{\mathit{body}}{\mathit{body}}{\Theta_<}{\Theta_<}{(0, 0)} .
  \]
  Since the loops are synchronized we apply \nameref{rule:aprhl-while} to get
  \[
    \vdash \aprhl{w_<}{w_<}
    {\Theta_<}{\Theta_< \land \neg(i\sidel \leq N \land i\sidel < j)}{(0, 0)} .
  \]
  Next, we turn to the second loop. We couple the samplings for $a$ so that
  \[
    a\sidel + 1 = a\sider
  \]
  with \nameref{rule:aprhl-lapgen},  taking $k = 1$, $k' = 2$. Since the
  parameter for the Laplace sampling is $\varepsilon/8C$, this is a $(2 \cdot
  \varepsilon / 8C, 0) = (\varepsilon/4C, 0)$-approximate coupling. Like for the
  first loop, we have four cases when analyzing the conditional. The most
  interesting case is when both guards are true, when we couple the samplings
  for $\mathit{noisy}$ with the standard Laplace rule \nameref{rule:aprhl-lap}
  so that $\mathit{noisy}\sidel = \mathit{noisy}\sider$; this is an
  $(\varepsilon/4C, 0)$-approximate coupling since the queries are
  $1$-sensitive. We wind up with $\neg \mathit{go}\sidel$ and $\neg \mathit{go}
  \sider$, establishing the post-condition $\mathit{go} \sidel \to
  \mathit{go}\sider$. Moreover,
  \[
    \mathit{ans}\sidel = (j, v) \to \mathit{ans}\sider = (j, v)
  \]
  under the coupling. This suffices to establish the invariant $\Theta_=$ when
  both guards are true. We use a similar argument for the other three cases,
  proving
  \[
    \vdash \aprhl{\mathit{body}}{\mathit{body}}{\Theta_=}{\Theta_=}{(\varepsilon/2C, 0)} .
  \]
  Since there is exactly one iteration, \nameref{rule:aprhl-while} gives
  \[
    \vdash \aprhl{w_=}{w_=}
    {\Theta_=}{\Theta_= \land \neg(i\sidel \leq N \land i\sidel = j)}{(\varepsilon/2C, 0)} .
  \]
  In the last loop, we couple the samplings for $a$ with
  \nameref{rule:aprhl-lapnull} and the samplings for $\mathit{noisy}$ with
  \nameref{rule:aprhl-lapnull} or the one-sided rules \nameref{rule:aprhl-lap-l}
  or \nameref{rule:aprhl-lap-r}, depending on whether the guards are true or
  not. This gives
  \[
    \vdash \aprhl{w_>}{w_>}
    {\Theta_>}{\Theta_> \land \neg(i\sidel \leq N)}{(0, 0)} .
  \]
  After using the rule of consequence with implications
  \begin{align*}
    &\models \Theta_< \land \neg(i\sidel \leq N \land i\sidel < j) \to \Theta_= \\
    &\models \Theta_= \land \neg(i\sidel \leq N \land i\sidel = j) \to \Theta_> ,
  \end{align*}
  we apply \nameref{rule:aprhl-seq} to combine the loop judgments and sum the
  approximation parameters:
  \[
    \vdash \aprhl{\mathit{aboveT}'}{\mathit{aboveT}'}
    {t\sidel + 1 = t\sider}
    {\mathit{ans}\sidel = (j, v) \to \mathit{ans}\sider = (j, v)}{(\varepsilon/2C, 0)} .
  \]
  By applying pointwise equality \nameref{rule:aprhl-pw-eq} and then the frame
  rule \nameref{rule:aprhl-frame} to preserve the threshold coupling, we
  establish the desired judgment for the inner loop:
  \[
    \vdash \aprhl{\mathit{aboveT}'}{\mathit{aboveT}'}
    {t\sidel + 1 = t\sider}
    {\mathit{ans}\sidel = \mathit{ans}\sider \land t\sidel + 1 = t\sider}{(\varepsilon/2C, 0)} .
  \]
  Now we turn to the outer loop $w$ of $\mathit{sparseV}$. At the end of each
  iteration, we know
  \[
    i\sidel = i\sider
    \land \mathit{out}\sidel = \mathit{out}\sider
    \land t\sidel + 1 = t\sider
  \]
  since the inner loop guarantees $\mathit{ans}\sidel = \mathit{ans}\sider$.
  Applying \nameref{rule:aprhl-while} with decreasing variant
  \[
    e_v \triangleq \Tern{(i = N)}{0}{C - |\mathit{out}|} ,
  \]
  there at most $C$ iterations and each iteration is related by an
  $(\varepsilon/2C, 0)$-coupling. So we have the following judgment for the
  outer loop:
  \[
    \vdash \aprhl{w}{w}
    {\mathit{out}\sidel = \mathit{out}\sider \land t\sidel + 1 = t\sider}{\mathit{out}\sidel = \mathit{out}\sider}
    {(\varepsilon/2, 0)} .
  \]
  Finally, we ensure the loop pre-condition $t\sidel + 1 = t\sider$ by coupling
  the sampling instructions for $t$ with \nameref{rule:aprhl-lapgen}, taking $k,
  k' \triangleq 1$. Since the Laplace distribution has parameter
  $\varepsilon/2$, this is an $(\varepsilon/2, 0)$-approximate coupling.
  Putting everything together we have
  \[
    \vdash \aprhl{\mathit{sparseV}}{\mathit{sparseV}}
    {\mathit{Adj}(d\sidel, d\sider)}{\mathit{out}\sidel = \mathit{out}\sider}{(\varepsilon, 0)} ,
  \]
  showing that Sparse Vector is $\varepsilon$-differentially private.
\end{proof}

\begin{remark}
  It would be a bit more natural to use the guard $\mathit{go} = \kwfalse$ in
  the final conditional, but showing $\mathit{go}\sidel = \mathit{go}\sider$
  after the inner loop is not so easy: our proof can only establish
  $\mathit{go}\sidel \to \mathit{go}\sider$. In order to verify the program with
  guard $\mathit{go} = \kwfalse$, we would need the one-sided invariant
  \[
    p_1(\mathit{ans}) \neq 0 \leftrightarrow \mathit{go} = \kwfalse 
  \]
  on both sides. While this invariant does hold, here we hit a limitation of the
  pointwise equality rule \nameref{rule:aprhl-pw-eq}: the post-condition is
  narrowly restricted and we cannot show the above invariant in the
  post-condition of the inner loop. Later in \cref{chap:combining} we will see
  how to leverage these one-sided invariants (cf.  rules
  \nameref{rule:aprhl-and-l} and \nameref{rule:aprhl-and-r}).
\end{remark}

\section{Discussion} \label{sec:aprhl-rw}

To close this chapter, we briefly survey related systems for formally verifying
differential privacy and discuss other applications of approximate couplings.

\subsection{Formal verification of differential privacy}

Due to its rich composition properties and compelling motivations, differential
privacy is an attractive target for formal verification. Researchers have
considered a broad array of techniques including linear
types~\citep{ReedPierce10,AGGH14,ACGHK16,adafuzz}, sized types~\citep{GHHNP13},
product programs~\citep{BGGHKS14}, refinement types~\citep{BGGHRS15}, and more
\citep{pinq,Tschantz201161,ebadi2016featherweight,DBLP:conf/popl/EbadiSS15,DBLP:conf/pst/EbadiAS16,proserpio2014calibrating,palamidessi:hal-00760688}.
(Readers can consult the recent survey by \citet{Murawski:2016:2893582} for a
more comprehensive overview.)

Most existing techniques cannot verify privacy proofs beyond composition, such
as the two examples we presented in this chapter. One notable exception is the
\sname{LightDP} system proposed by \citet{zhang2016autopriv}, which combines a
relational, dependent type system with a product program construction. This
system can prove privacy for the Sparse Vector mechanism with a high degree of
automation by using a novel type inference algorithm and a MAXSAT solver to
optimize the privacy cost.

The key theoretical idea behind \sname{LightDP} is \emph{randomness alignment},
which specifies an injection from one sample space to another while recording
the difference in probabilities.  Randomness alignments are similar to the
approximate couplings we saw for the Laplace mechanism (e.g., in the rules
\nameref{rule:aprhl-lapnull} and \nameref{rule:aprhl-lapgen}). One important
novelty in \sname{LightDP} is that alignments can be selected lazily, based on
the result of the sample in the first execution. In this way, \sname{LightDP}
can sometimes construct a privacy proof in one shot where \Saprhl would need to
reason about each output separately with \nameref{rule:aprhl-pw-eq}. In the
Sparse Vector mechanism, for instance, \sname{LightDP} can select the shift
coupling when the first iteration goes above threshold, and use the null
coupling when it does not. (This approach does not work for Report-noisy-max, as
the iteration with the highest noisy score is not known until the program has
finished executing.) This lazy choice of alignment can be modeled by an
approximate coupling that selects between two couplings, depending on a
predicate on the first sample. If the predicate and two couplings satisfy a
technical non-overlapping condition, the result is again an approximate
coupling.

\begin{theorem}[Choice coupling] \label{thm:choice-couple}
  Let $\mu_1, \mu_2$ be sub-distributions over $\cA_1$ and $\cA_2$.  Suppose we
  have a predicate $\cP \subseteq \cA_1$ and two approximate couplings
  \[
    \mu_1 \alift{\cR}{(\varepsilon, \delta)} \mu_2
    \quad \text{and} \quad
    \mu_1 \alift{\cS}{(\varepsilon, \delta)} \mu_2
  \]
  such that the following \emph{non-overlapping} condition holds:
  \[
    \cR(\cP) \cap \cS(\cA_1 \setminus \cP) = \varnothing ,
  \]
  where $\cR(\cP)$ is the set of elements in $\cA_2$ related to something in
  $\cP$ under $\cR$, and $\cS(\cA_1 \setminus \cP)$ is the set of elements in
  $\cA_2$ related to something outside of $\cP$ under $\cS$. Then there is an
  approximate coupling
  \[
    \mu_1
    \alift{ \cT  }
    {(\varepsilon, 2\delta)}
    \mu_2
  \]
  where $\cT$ is the relation
  \[
    \cT \triangleq \{ (a_1, a_2)
    \mid (a_1 \in \cP \to (a_1, a_2) \in \cR) \land (a_1 \notin \cP \to (a_1, a_2) \in \cS) \} .
  \]
\end{theorem}
\begin{proof}
  Let $\rho_L, \rho_R$ and $\sigma_L, \sigma_R$ witness the two approximate
  couplings. Define witnesses
  \[
    \mu_L(a_1, a_2) \triangleq \begin{cases}
      \rho_L(a_1, a_2) &: a_1 \in \cP \\
      \sigma_L(a_1, a_2) &: a_1 \notin \cP \\
      0 &: a_1 = \star
    \end{cases}
    \quad \text{and} \quad
    \mu_R(a_1, a_2) \triangleq \begin{cases}
      \rho_R(a_1, a_2) &: a_1 \in \cP \\
      \sigma_R(a_1, a_2) &: a_1 \notin \cP \\
      \mu_2(a_2) - \sum_{a_1' \in \cA_1} \mu_R(a_1', a_2) &: a_1 = \star .
    \end{cases}
  \]
  The support and marginal conditions are immediate. The main thing to show is
  that $\mu_R(\star, a_2)$ is non-negative; it suffices to show $\sum_{a_1' \in
  \cA_1} \mu_R(a_1', a_2) \leq \mu_2(a_2)$.  There are three cases: either $a_2
  \in \cR(\cP)$, $a_2 \in \cS(\cA_1 \setminus \cP)$, or none of the above; by
  the non-overlapping condition, these cases are mutually exclusive. In the
  first case, we have
  \[
    \sum_{a_1' \in \cA_1} \mu_R(a_1', a_2)
    = \sum_{a_1' \in \cP} \rho_R(a_1', a_2) + \sum_{a_1' \in \cA \setminus \cP} \sigma_R(a_1', a_2)
    = \sum_{a_1' \in \cP} \rho_R(a_1', a_2)
    \leq \mu_2(a_2) .
  \]
  The second case is similar:
  \[
    \sum_{a_1' \in \cA_1} \mu_R(a_1', a_2)
    = \sum_{a_1' \in \cP} \rho_R(a_1', a_2) + \sum_{a_1' \in \cA \setminus \cP} \sigma_R(a_1', a_2)
    = \sum_{a_1' \in \cA \setminus \cP} \sigma_R(a_1', a_2)
    \leq \mu_2(a_2) .
  \]
  In the third case the inequality clearly holds, as the sum is equal to $0$.
  
  It only remains to check the distance condition
  $\epsdist{\varepsilon}{\mu_L}{\mu_R} \leq 2\delta$. By the distance conditions
  on the given witnesses, there are non-negative constants $\zeta(a_1, a_2),
  \xi(a_1, a_2)$ such that
  \[
    \rho_L(a_1, a_2) \leq \exp(\varepsilon) \cdot \rho_R(a_1, a_2) + \zeta(a_1, a_2)
    \quad \text{and} \quad
    \sigma_L(a_1, a_2) \leq \exp(\varepsilon) \cdot \sigma_R(a_1, a_2) + \xi(a_1, a_2)
  \]
  with bounded sums:
  \[
    \sum_{a_1, a_2} \zeta(a_1, a_2) \leq \delta
    \qquad\text{and}\qquad
    \sum_{a_1, a_2} \xi(a_1, a_2) \leq \delta .
  \]
  By definition, we have
  \[
    \mu_L(a_1, a_2) \leq \exp(\varepsilon) \cdot \mu_R(a_1, a_2)
    + \max(\zeta(a_1, a_2), \xi(a_1, a_2))
  \]
  for all $a_1, a_2 \neq \star$; it is easy to check
  \[
    \mu_L(a_1, a_2) \leq \exp(\varepsilon) \cdot \mu_R(a_1, a_2)
  \]
  when $a_1 = \star$ or $a_2 = \star$. We can bound the sums
  \[
    \sum_{a_1, a_2} \max(\zeta(a_1, a_2), \xi(a_1, a_2))
    \leq \sum_{a_1, a_2} \zeta(a_1, a_2) + \xi(a_1, a_2)
    \leq 2 \delta
  \]
  to give the claimed distance condition. Thus $\mu_L, \mu_R$ witness the
  desired approximate coupling.
\end{proof}

However, this coupling is not quite precise enough: its cost is greater than the
\emph{maximum} cost of the two couplings. Taking the example of Sparse Vector
again, the shift coupling \nameref{rule:aprhl-lapgen} has a non-zero cost while
the null coupling \nameref{rule:aprhl-lapnull} has zero cost. If we are
selecting between these two couplings, we do not want to pay for the (more
expensive) \nameref{rule:aprhl-lapgen} coupling on every iteration, but only on
the single iteration where the first execution goes above threshold.

\sname{LightDP} features a more fine-grained analysis where the cost can depend
on which choice was taken. Since the choice depends on whether the first sample
satisfies a predicate (e.g., goes above threshold), this analysis involves a
\emph{randomized} notion of privacy cost; \sname{LightDP} uses a product
construction as a secondary analysis to bound the parameters in all possible
executions. In contrast, \Saprhl requires the approximation parameters to be
constant at each stage, though a more general form of approximate coupling
allowing variable costs for different samples enables \sname{LightDP}-style
privacy proofs.  (See \cref{chap:future} for further discussion.)

\subsection{Approximate couplings in formal verification}

Approximate liftings are a flexible abstraction for reasoning about differential
privacy. While we have focused on program logics, approximate liftings have
played a central role in other formal verification settings.

\citet{BGGHKS14} show how to verify differential privacy by first interleaving
two programs into a single program and then analyzing the result, a so-called
``synchronized product'' approach. Their construction replaces every pair of
corresponding random sampling commands with a single, non-deterministic
assignment of a pair, along with a specification of the relation between the
returned values. In this way, they can verify differential privacy by
constructing proofs in non-deterministic Hoare logic. Their technique is based
on approximate liftings and roughly corresponds to the fragment of \Saprhl where
all conditionals are synchronized under the coupling, so only pairs of identical
programs are related.

Approximate liftings can also play a useful role in type systems.
\citet{BGGHRS15} propose a relational refinement type system for a functional
language \Shoare. To handle relational reasoning for distributions, their system
features a probability monad over a relation $\cR$ on the base type, indexed by
approximation parameters. This monad is then interpreted as an approximate
lifting with support contained in $\cR$. In their typing rule for monadic bind
with initial distributions related by a $\cR$-lifting, the body is typed under
the assumption that the samples are related by $\cR$, giving a clean way to use
information about distributions when reasoning about samples. This principle can
be seen in the \Saprhl rule \nameref{rule:aprhl-seq} or more abstractly, as a
monadic composition principle for approximate liftings.

\citet{BGGHRS15} also explore an interesting application of approximate
liftings: given sub-distributions $\mu_1, \mu_2$ over the unit interval $[0,
1]$,\footnote{%
  More precisely, a discrete version of the unit interval $[0, 1]$.}
the approximate lifting
\[
  \mu_1 \alift{(\leq)}{(\varepsilon, \delta)} \mu_2
\]
implies a bound on expected values: $\exD{x_1 \sim \mu_1}{x_1} \leq
\exp(\varepsilon) \cdot \exD{x_2 \sim \mu_2}{x_2} + \delta$; this can be seen as
a consequence of approximate stochastic domination. \citet{BGGHRS15} use this
observation to prove relational properties involving expectations for algorithms
at the intersection of mechanism design and differential privacy, where the
mechanisms are randomized and the incentive properties follow from differential
privacy.  \citet{HKM-verif16} use similar ideas to verify more sophisticated
incentive properties.

\chapter{Advanced approximate couplings} \label{chap:combining}

In the previous chapter, we saw how approximate couplings of the Laplace
distribution and the pointwise equality principle support new proofs of privacy
by approximate coupling. To enhance the power of this proof technique, we
develop the theory of approximate couplings further in this chapter, giving a
potpourri of new constructions and showing equivalences with other notions of
approximate lifting.  Our results enable richer proofs by approximate coupling,
capable of modeling more advanced proofs of privacy.

To begin, we show that approximate couplings are a discrete version of the
approximate lifting recently proposed by \citet{Sato16}.  This equivalence gives
a highly convenient method for constructing approximate couplings and extends a
classical result by \citet{strassen1965existence} for probabilistic couplings
(\cref{sec:apx-strassen}). Then, we consider two new constructions:
\emph{up-to-bad} approximate coupling (\cref{sec:aprhl-utb}) and \emph{optimal
subset} coupling (\cref{sec:aprhl-subset}). To follow, we identify a symmetric
version of approximate coupling that supports an advanced composition principle
generalizing the advanced composition theorem of differential privacy
(\cref{sec:aprhl-ac}). To make our constructions concrete, we introduce new
\Saprhl proof rules and prove differential privacy for the \emph{Between
Thresholds} mechanism, recently proposed by~\citet{BunSU16}
(\cref{sec:aprhl-bt}). Finally, we show approximate couplings unify several
previously proposed notions (\cref{sec:alift-rw}).  Taken together, our
equivalences and constructions serve as strong evidence that we have arrived at
a natural, approximate generalization of probabilistic coupling.

\section{Equivalence with Sato's approximate lifting} \label{sec:apx-strassen}

So far, we have considered approximate couplings for discrete distributions. In
recent work, \citet{Sato16} develops a version of \Saprhl where programs can
sample from continuous distributions, like the Laplace and Gaussian
distributions. Intriguingly, Sato takes a significantly different definition of
approximate lifting as the foundation of his logic. In the discrete case, his
definition is as follows.

\begin{definition}[\citet{Sato16}]
  Let $\mu_1$ and $\mu_2$ be sub-distributions over countable sets $\cA_1$ and
  $\cA_2$, and let $\cR \subseteq \cA_1 \times \cA_2$ be a
  relation.  There is an \emph{$(\varepsilon, \delta)$-approximate
  $\cR$-lifting} of $(\mu_1, \mu_2)$ if for every subset $\cS_1 \subseteq
  \cA_1$, the following inequality holds:
  \[
    \mu_1(\cS_1) \leq \exp(\varepsilon) \cdot \mu_2(\cR(\cS_1)) + \delta .
  \]
  (Recall $\cR(\cS_1)$ is the subset of $\cA_2$ that is related to some element
  in $\cS_1$ under $\cR$.)
\end{definition}

This definition is interesting for several reasons. First, rather than requiring
the existence of witness distributions, Sato's definition quantifies over all
subsets of samples. Second, Sato shows that his definition generalizes the prior
definition of approximate lifting by \citet{BartheO13} and \citet{OlmedoThesis},
leaving open the question of whether they are equivalent; in fact, they are not.
However, we show our definition of approximate lifting (\cref{def:alift}) is
equivalent to Sato's definition in the discrete case. Our result can be seen as
an approximate version of Strassen's theorem (\cref{thm:strassen}); it also
implies Strassen's theorem for discrete sub-distributions.

One direction of the equivalence is not hard to show.

\begin{theorem}[Approximate lifting implies Sato's lifting] \label{thm:lift-to-sato}
  Let $\mu_1$ and $\mu_2$ be sub-distributions over $\cA_1$ and $\cA_2$, and let
  $\cR \subseteq \cA_1 \times \cA_2$ be a binary relation.  Suppose there
  exists an approximate lifting
  \[
    \mu_1 \alift{\cR}{(\varepsilon, \delta)} \mu_2 .
  \]
  Then
  $\mu_1(\cS_1) \leq \exp(\varepsilon) \cdot \mu_2(\cR(\cS_1)) + \delta$
  for every subset $\cS_1 \subseteq \cA_1$.
\end{theorem}
\begin{proof}
  Let $\mu_L, \mu_R$ witness the approximate lifting. By the distance,
  support, and marginal conditions,
  \begin{align*}
    \mu_1(\cS_1) &= \mu_L(\cS_1 \times \cA_2^\star) \\
    &\leq \exp(\varepsilon) \cdot \mu_R(\cS_1 \times \cA_2^\star) + \delta \\
    &= \exp(\varepsilon) \cdot \mu_R(\cS_1 \times \cR(\cS_1)) + \delta \\
    &\leq \exp(\varepsilon) \cdot \mu_R(\cA_1^\star \times \cR(\cS_1)) + \delta \\
    &= \exp(\varepsilon) \cdot \mu_R(\cR(\cS_1)) + \delta .
    \qedhere
  \end{align*}
\end{proof}

The other direction---showing Sato's approximate lifting implies our approximate
lifting---is a bit more involved. We proceed in two steps. First, we prove the
implication for distributions over finite sets. Then we generalize to
distributions over countable sets by a limiting argument.

\subsection{The finite case}

The finite case follows from the max flow-min cut theorem. Roughly speaking,
Sato's condition ensures that in a certain graph, the minimum cut is not too
small so the maximum flow must be large. This will imply we can build witnesses
to the approximate lifting from the maximum flow. First, we recall the classical
max flow-min cut theorem (see any standard textbook on algorithms, e.g.,
\citet{Kleinberg:2005:AD:1051910}).

\begin{theorem}[Max flow-min cut] \label{prop:maxflow-mincut}
  Let $G$ be a finite graph with vertices $V$ and directed edges $E$.  Let $s
  \in V$ be the \emph{source} node (i.e., there are no edges $(a, s) \in E$) and
  let $t \in V$ be the \emph{sink} node (i.e., there are no edges $(t, b) \in
  E$); we assume $s$ and $t$ are unique. We suppose each edge has
  \emph{capacity} $c(a, b) \in \RR \cup \{ \infty \}$. A \emph{flow} from $s$ to
  $t$ is a map $f : E \to \RR^+$ such that (i) the flow is conserved at each
  internal node:
  \[
    \sum_{(a, v) \in E} f(a, v) = \sum_{(v, b) \in E} f(v, b)
  \]
  for every node $v \neq s, t$, and (ii) the flow respects the capacity
  constraints: $f(a, b) \leq c(a, b)$. The \emph{weight} of a flow $|f|$ is the
  amount of flow leaving $s$; by conservation, this is equal to the total flow
  entering $t$. A \emph{cut} $C$ is a partition of the vertices into two sets
  $(V_1, V_2)$. The \emph{capacity} of a cut $|C|$ is the total capacity of all
  edges $(a, b)$ crossing $(V_1, V_2)$, i.e., with $a \in V_1$ and $b \in V_2$.

  The weight of the largest flow equals the minimum capacity of a cut $(V_1,
  V_2)$ with $s \in V_1$ and $t \in V_2$.
\end{theorem}

\begin{theorem} \label{lem:alift:domfin}
  Let $\mu_1$ and $\mu_2$ be sub-distributions with finite support over sets
  $\cA_1$ and $\cA_2$, and let $\cR \subseteq \cA_1 \times \cA_2$ be a binary
  relation such that
  $\mu_1(\cS_1) \leq \exp(\varepsilon) \cdot \mu_2(\cR(\cS_1)) + \delta$
  for every $\cS_1 \subseteq \cA_1$.  Then there exists an approximate lifting
  \[
    \mu_1 \alift{\cR}{(\varepsilon, \delta)} \mu_2 .
  \]
\end{theorem}
\begin{proof}
  Without loss of generality, by \cref{thm:alift-extend} we may take $\cA_1$ and
  $\cA_2$ to be the supports of $\mu_1$ and $\mu_2$ respectively; these are
  finite by assumption. We define a finite graph with vertices $\cA_1^\star
  \cup \cA_2^\star \cup \{ \top, \bot \}$. Note that we take two distinct
  vertices $\star_1, \star_2$ corresponding to the $\star$ elements in
  $\cA_1^\star$ and $\cA_2^\star$. We connect the source $\top$ to every element
  of $\cA_1^\star$ with capacities
  \begin{align*}
    c(\top, a_1) &\triangleq \mu_1(a_1) \cdot \exp(-\varepsilon) \\
    c(\top, \star_1) &\triangleq \omega - \exp(-\varepsilon) \cdot \mu_1(\cA_1) ,
  \end{align*}
  where $\omega \triangleq \mu_2(\cA_2) + \exp(-\varepsilon) \cdot \delta$. Now
  $c(\top, \star_1) \geq 0$ since by assumption,
  \[
    \mu_1(\cA_1)
    \leq \exp(\varepsilon) \cdot \mu_2(\cR(\cA_1)) + \delta
    \leq \exp(\varepsilon) \cdot \mu_2(\cA_2) .
  \]
  We connect every element of $\cA_2^\star$ to the sink $\bot$, with capacities
  \begin{align*}
    c(a_2, \bot) &\triangleq \mu_2(a_2) \\
    c(\star_2, \bot) &\triangleq \exp(-\varepsilon) \cdot \delta .
  \end{align*}
  For the internal nodes, we connect $(a_1, a_2) \in \cA_1 \times \cA_2$ for all
  $(a_1, a_2) \in \cR$ and $(a_1, \star_2)$,  $(\star_1, a_2)$ for all $a_1,
  a_2$, all with capacity $\infty$.

  Note that $(\{ s \}, V \setminus \{ s \})$ and $(V \setminus \{ t \}, \{ t
  \})$ are both cuts with capacity $\omega$. We show that these are minimal cuts
  in the graph. Consider any other cut $C = (V_1, V_2)$ with edges $E(C)$
  crossing the cut. If there is any internal edge $(a, b) \in E(C)$ with $a, b
  \neq \top, \bot$ then $C$ has infinite capacity and is not a minimal cut. So,
  we may suppose $E(C)$ contains only edges of the form $(\top, a_1)$ and $(a_2,
  \bot)$.
  
  Now if $E(C)$ does not contain $(\top, \star_1)$, then it must cut all edges
  leading into $\bot$; similarly, if $E(C)$ does not contain $(\star_2, \bot)$,
  then it must cut all edges leading from $\top$. Either way, its capacity is at
  least $\omega$.

  Finally, suppose $E(C)$ contains no internal edges and contains both $(\top,
  \star_1)$ and $(\star_2, \bot)$. Let $\cS_2 \subseteq \cA_2$ be the set of all
  nodes $a_2 \in \cA_2$ with $(a_2, \bot) \in E(C)$, and let $\cS_1 \subseteq
  \cA_1$ be the set of all nodes $a_1 \in \cA_1$ with $(\top, a_1) \in E(C)$.
  Since $C$ separates $\top$ and $\bot$, we have
  \[
    \cR(\cA_1 \setminus \cS_1) \subseteq \cS_2 .
  \]
  We can now lower-bound the capacity:
  \begin{align*}
    |C| &= c(\top, \cS_1) + c(\cS_2, \bot) + c(\top, \star_1) + c(\star_2, \bot) \\
    &= \exp(-\varepsilon) \cdot \mu_1(\cS_1) + c(\cS_2, \bot)
    + \omega - \exp(-\varepsilon) \cdot \mu_1(\cA_1) + \exp(-\varepsilon) \cdot \delta \\
    &\geq \exp(-\varepsilon) \cdot \mu_1(\cS_1) + c(\cR(\cA_1 \setminus \cS_1), \bot)
    + \omega - \exp(-\varepsilon) \cdot \mu_1(\cA_1) + \exp(-\varepsilon) \cdot \delta \\
    &\geq \exp(-\varepsilon) \cdot \mu_1(\cS_1)
    + \exp(-\varepsilon) \cdot \mu_1(\cA_1 \setminus \cS_1)
    - \exp(-\varepsilon) \cdot \delta
    + \omega - \exp(-\varepsilon) \cdot \mu_1(\cA_1)
    + \exp(-\varepsilon) \cdot \delta \\
    &= \omega
  \end{align*}
  The final inequality is by Sato's condition applied to the set $\cA_1
  \setminus \cS_1$.  So every cut in this graph has capacity at least $\omega$,
  and there is a cut achieving capacity $\omega$. By \cref{prop:maxflow-mincut},
  there is a maximum flow $f$ with weight $\omega$. We define witnesses
  \begin{align*}
    \mu_L(a_1, a_2) &\triangleq \exp(\varepsilon) \cdot f(a_1, a_2)
    &: \text{if } (a_1, a_2) \in \cR \text{ or } a_2 = \star_2 \\
    \mu_R(a_1, a_2) &\triangleq f(a_1, a_2)
    &: \text{if } (a_1, a_2) \in \cR \text{ or } a_1 = \star_1
  \end{align*}
  and zero otherwise. The support condition is clear. Since $f$ has weight
  $\omega$, it must saturate all edges exiting $\top$ and entering $\bot$ and so
  the marginal conditions are also clear.

  The only thing to check is the distance condition
  $\epsdist{\varepsilon}{\mu_L}{\mu_R} \leq \delta$. It suffices to show this
  condition pointwise, by finding non-negative $\zeta(a_1, a_2)$ such that
  $\mu_L(a_1, a_2) \leq \exp(\varepsilon) \cdot \mu_R(a_1, a_2) + \zeta(a_1,
  a_2)$ and $\sum_{(a_1, a_2)} \zeta(a_1, a_2) \leq \delta$. For all $a_1 \in
  \cA_1^\star$ and all $a_2 \neq \star_2$, we take $\zeta(a_1, a_2) = 0$.  When
  $a_2 = \star_2$ we know
  \[
    \mu_L(a_1, \star_2) = \exp(\varepsilon) \cdot f(a_1, \star_2)
    \qquad\text{and}\qquad
    \mu_R(a_1, \star_2) = 0 ,
  \]
  so we may take $\zeta(a_1, \star_2) = \exp(\varepsilon) \cdot f(a_1,
  \star_2)$.  Conservation of flow yields
  \[
    \sum_{(a_1, a_2) \in \cA_1^\star \times \cA_2^\star} \zeta(a_1, a_2)
    = \sum_{a_1 \in \cA_1} \exp(\varepsilon) \cdot f(a_1, \star_2)
    = \exp(\varepsilon) \cdot f(\star_2, \bot) = \delta ,
  \]
  establishing the desired distance condition
  $\epsdist{\varepsilon}{\mu_L}{\mu_R} \leq \delta$.
\end{proof}

\subsection{The countable case}

There are several possible approaches to generalize \cref{lem:alift:domfin} to
countable distributions. Perhaps the most straightforward is to apply a version
of the max flow-min cut theorem for countable graphs
\citep{DBLP:journals/jct/AharoniBGPS11}. Instead, we will give a more elementary
proof.  Besides being self-contained, our proof also establishes limit and
compactness properties of approximate couplings and their witnesses, which may
be of independent interest.

We first show that given a convergent sequence of pairs of distributions with an
approximate lifting for each pair, there is a sub-sequence of witnesses
converging to witnesses of an approximate lifting for the limits.
We then generalize \cref{lem:alift:domfin} to countable domains by viewing a
distribution over a countable set as the pointwise limit of distributions with
finite support, using the finite case to build approximate liftings (and
witnesses) for each pair of finite restrictions

We will need a generalized version of the dominated convergence theorem.

\begin{theorem}[see, e.g., {\citet[Chapter 4, Theorem 19]{royden-analysis}}] \label{lem:dct:modif}
  Let $\Omega$ be a measurable space with measure $\mu$. Let $\{ f_n \}$ and $\{
  g_n \}$ be two sequences of measurable functions $\Omega \to \RR$ such that
  there exist functions $f, g : \Omega \to \RR$ with
  \begin{enumerate}
  \item $\lim_{n \to \infty} f_n = f$ pointwise;
    \item $|f_n| \leq g_n$; and
    \item $\lim_{n \to \infty} \int g_n~d\mu = \int g~d\mu < \infty$.
  \end{enumerate}
  Then we have
  \[
    \lim_{n \to \infty} \int f_n~d\mu = \int f~d\mu .
  \]
\end{theorem}

Since we work with countable spaces, we take $\mu$ to be the discrete measure.
In this case, the integrals are simply plain sums. We will also need a lemma
about witnesses to approximate liftings---roughly speaking, we may assume the
witnesses are within a purely multiplicative factor of each other except on
pairs with $\star$.

\begin{lemma} \label{lem:alift:normalform}
  Suppose $\mu_1, \mu_2$ are sub-distributions over $\cA_1$ and $\cA_2$ such
  that
  \[
    \mu_1 \alift{\cR}{(\varepsilon, \delta)} \mu_2 .
  \]
  Then there exists $(\eta_L, \eta_R)$ witnessing the approximate lifting with
  \[
    \eta_R(a_1, a_2) \leq \eta_L(a_1, a_2) \leq \exp(\varepsilon) \cdot  \eta_R(a_1, a_2)
  \]
  for all $a_1, a_2 \neq \star$.
\end{lemma}
\begin{proof}
  Let $\mu_L, \mu_R$ be witnesses. Define witnesses
  \begin{align*}
    \eta_L(a_1, a_2) &\triangleq
    \begin{cases}
      \min(\mu_L(a_1, a_2), \exp(\varepsilon) \cdot \mu_R(a_1, a_2))
      &: a_1 \neq \star, a_2 \neq \star \\
      \mu_1(a_1) - \sum_{a_2' \in \cA_2} \eta_L(a_1, a_2')
      &: a_1 \neq \star, a_2 = \star \\
      0 &: \text{otherwise;}
    \end{cases}
    \\
    \eta_R(a_1, a_2) &\triangleq
    \begin{cases}
      \min(\mu_L(a_1, a_2), \mu_R(a_1, a_2))
      &: a_1 \neq \star, a_2 \neq \star \\
      \mu_2(a_2) - \sum_{a_1' \in \cA_1} \eta_R(a_1', a_2)
      &: a_1 = \star, a_2 \neq \star \\
      0 &: \text{otherwise.}
    \end{cases}
  \end{align*}
  The marginal and support conditions follow from the respective conditions for
  $(\mu_L, \mu_R)$. Note that $\eta_L$ and $\eta_R$ are non-negative by the
  marginal conditions for $\mu_L$ and $\mu_R$. Furthermore for all $(a_1, a_2)
  \in \cA_1 \times \cA_2$, we have
  \[
    \eta_R(a_1, a_2) \leq \eta_L(a_1, a_2) \leq \exp(\varepsilon) \cdot \eta_R(a_1, a_2) .
  \]
  It only remains to check the distance condition. Define non-negative constants
  \[
    \zeta(a_1, a_2) \triangleq \max(\mu_L(a_1, a_2) - \exp(\varepsilon) \cdot \mu_R(a_1, a_2), 0) .
  \]
  Since $\epsdist{\varepsilon}{\mu_L}{\mu_R} \leq \delta$, we know $\mu_L(a_1,
  a_2) \leq \exp(\varepsilon) \cdot \mu_R(a_1, a_2) + \zeta(a_1, a_2)$ with
  equality when $\zeta(a_1, a_2) > 0$, and $\sum_{a_1, a_2 \in \cA_1^\star
  \times \cA_2^\star} \zeta(a_1, a_2) \leq \delta$. Thus, $\eta_L(a_1, a_2) =
  \mu_L(a_1, a_2) - \zeta(a_1, a_2)$ for every $a_1, a_2 \neq \star$. Also, we
  know $\eta_L(a_1, a_2) \leq \exp(\varepsilon) \cdot \eta_R(a_1, a_2)$. Thus
  for any subset $\cS \subseteq \cA_1^\star \times \cA_2^\star$, we have
  \begin{align*}
    \eta_L(\cS) &\leq \exp(\varepsilon) \cdot  \eta_R(\cS \cap (\cA_1 \times \cA_2))
    + \eta_L(\cS \cap (\cA_1 \times \{ \star \})) \\
    &\leq \exp(\varepsilon) \cdot  \eta_R(\cS \cap (\cA_1 \times \cA_2))
    + \eta_L(\cA_1 \times \{ \star \}) \\
    &= \exp(\varepsilon) \cdot  \eta_R(\cS \cap (\cA_1 \times \cA_2))
    + \sum_{a_1 \in \cA_1} \left( \mu_1(a_1) - \sum_{a_2 \in \cA_2} \mu_L(a_1, a_2) - \zeta(a_1, a_2) \right) \\
    &= \exp(\varepsilon) \cdot  \eta_R(\cS \cap (\cA_1 \times \cA_2))
    + \sum_{a_1 \in \cA_1} \mu_L(a_1, \star) + \sum_{(a_1, a_2) \in \cA_1 \times \cA_2} \zeta(a_1, a_2) \\
    &= \exp(\varepsilon) \cdot  \eta_R(\cS \cap (\cA_1 \times \cA_2))
    + \sum_{(a_1, a_2) \in \cA_1 \times \cA_2^\star} \zeta(a_1, a_2) \\
    &\leq \exp(\varepsilon) \cdot  \eta_R(\cS) + \delta .
    \qedhere
  \end{align*}
\end{proof}

We are now ready to prove that a converging sequence of pairs of distributions
related by approximate liftings implies an approximate lifting for the limit
distributions.

\begin{lemma} \label{lem:bw}
  Let $\cR$ be a binary relation between countable sets $\cA_1,\cA_2$.  Consider
  a sequence $\{ (\mu_1^{(n)}, \mu_2^{(n)}) \}_{n \in \NN}$ with $\mu_1^{(n)}
  \in \SDist(\cA_1)$ and $\mu_2^{(n)} \in \SDist(\cA_2)$ such that there exists
  an approximate lifting for each $n$:
  \[
    \mu_1^{(n)} \alift{\cR}{(\varepsilon_n, \delta_n)} \mu_2^{(n)} .
  \]
  Suppose $\lim_{n \to \infty} (\varepsilon_n, \delta_n) = (\varepsilon,
  \delta)$ and $\{ \mu_1^{(n)} \}_n , \{ \mu_2^{(n)} \}_n$ converge to
  $\mu_1, \mu_2$ under the $L^1$ norm:
  \[
    \lim_{n \to \infty} \sum_{a_i \in \cA_i} \left| \mu_i^{(n)}(a_i) - \mu_i(a_i) \right| = 0
  \]
  for $i = 1, 2$. Then there exists an approximate
  lifting of the limit sub-distributions:
  \[
    \mu_1 \alift{\cR}{(\varepsilon, \delta)} \mu_2 .
  \]
\end{lemma}

\begin{proof}
  Let $(\eta_L^{(n)}, \eta_R^{(n)})$ witness the approximate lifting of
  $\mu_1^{(n)}$ and $\mu_2^{(n)}$, satisfying \cref{lem:alift:normalform}. Each
  witness can be viewed as a map $\eta_L^{(n)}, \eta_R^{(n)} : \cA_1^\star
  \times \cA_2^\star \to [0, 1]$. Since $\cA_1$ and $\cA_2$ are countable and
  $[0, 1]$ is compact, $\cA_1^\star \times \cA_2^\star \to [0, 1]$ is the
  countable product of compact sets and is itself (sequentially) compact. Hence,
  there exists a sub-sequence of indices $\{ \omega_n \}_n$ such that
  $\eta_L^{(\omega_n)}, \eta_R^{(\omega_n)}$ both converge pointwise to
  sub-distributions $(\eta_L, \eta_R)$. (See any real analysis textbook, e.g.,
  \citet{royden-analysis} for a discussion about sequential compactness.)
 
  We claim these limit sub-distributions are the desired witnesses. It is
  clear that $\supp(\eta_L)$ and $\supp(\eta_R)$ are contained in $\cR$.  The
  marginal conditions are a bit trickier.  Let $a_1 \in \cA_1$ (the marginal for
  $a_1 = \star$ is clear), and let $\varepsilon_{max}$ be an upper bound of the
  sequence $\{ \varepsilon_n \}_n$; since the sequence converges to
  $\varepsilon$, we may assume $\varepsilon_{max}$ is finite.  By
  \cref{lem:alift:normalform} and the marginal condition on
  $\mu_2^{(\omega_n)}$, the sequence $\{ \eta_L^{(\omega_n)} (a_1, -) \}_{n \in
  \NN}$ is bounded by $\beta_L^{(\omega_n)} : \cA_2^\star \to \RR$,
  where
  \[
    \beta_L^{(\omega_n)}(a_2) \triangleq
    \begin{cases}
      e^{\varepsilon_{max}} \mu_2^{(\omega_n)}(a_2) &: \text{if } a_2 \neq \star \\
      1 &: \text{if } a_2 = \star .
    \end{cases}
  \]
  The sequence $\{ \beta_L^{(\omega_n)} \}_n$ converges under the $L^1$ norm to
  $\beta_L : \cA_2^\star \to \RR$, where
  \[
    \beta_L(a_2) \triangleq
    \begin{cases}
      e^{\varepsilon_{max}} \mu_2(a_2) &: \text{if } a_2 \neq \star \\
      1 &: \text{if } a_2 = \star .
    \end{cases}
  \]
  Evidently $\sum_{a_2 \in \cA_2^\star} \beta_L(a_2)$ exists and is at most $1
  + e^{\varepsilon_{max}}$. Now for the first marginal,
  \begin{align*}
    \pi_1(\eta_L)(a_1)
      &= \sum_{a_2 \in \cA_2^\star} \eta_L(a_1, a_2)
       = \sum_{a_2 \in \cA_2^\star} \lim_{n \to \infty} \eta_L^{(\omega_n)}(a_1, a_2) \\
      &= \lim_{n \to \infty} \sum_{a_2 \in \cA_2^\star} \eta_L^{(\omega_n)}(a_1, a_2)
       = \lim_{n \to \infty} \pi_1(\eta_L^{(\omega_n)})(a_1) \\
      &= \lim_{n \to \infty} \mu_1^{(\omega_n)}(a_1)
       = \mu_1(a_1).
  \end{align*}
  We can interchange the sum and the limit by the dominated convergence theorem
  with bounding functions $\beta_L^{(\omega_n)}$ (\cref{lem:dct:modif}).

  For the second marginal, let $a_2 \in \cA_2$ (the marginal for $a_2 = \star$
  is clear). By \cref{lem:alift:normalform} and the marginal condition on
  $\mu_1^{(\omega_n)}$, the sequence $\{ \eta_R^{(\omega_n)} (-, a_2) \}_{n \in
  \NN}$ is bounded by $\beta_R^{(\omega_n)} : \cA_1^\star \to \RR$, where
  \[
    \beta_R^{(\omega_n)}(a_1) \triangleq
    \begin{cases}
      \mu_1^{(\omega_n)}(a_1) &: \text{if } a_1 \neq \star \\
      1 &: \text{if } a_1 = \star .
    \end{cases}
  \]
  The sequence $\{ \beta_R^{(\omega_n)} \}_n$ converges under the $L^1$ norm to
  $\beta_R : \cA_1^\star \to \RR$, where
  \[
    \beta_R(a_1) \triangleq
    \begin{cases}
      \mu_1(a_1) &: \text{if } a_1 \neq \star \\
      1 &: \text{if } a_1 = \star .
    \end{cases}
  \]
  Evidently $\sum_{a_1 \in \cA_1^\star} \beta_R(a_1)$ exists and is at most $2$.
  For the second marginal,
  \begin{align*}
    \pi_2(\eta_R)(a_2)
      &= \sum_{a_1 \in \cA_1^\star} \eta_R(a_1, a_2)
       = \sum_{a_1 \in \cA_1^\star} \lim_{n \to \infty} \eta_R^{(\omega_n)}(a_1, a_2) \\
      &= \lim_{n \to \infty} \sum_{a_1 \in \cA_1^\star} \eta_R^{(\omega_n)}(a_1, a_2)
       = \lim_{n \to \infty} \pi_2(\eta_R^{(\omega_n)})(a_2) \\
      &= \lim_{n \to \infty} \mu_2^{(\omega_n)}(a_2)
       = \mu_2(a_2).
  \end{align*}
  As before, to interchange the sum and the limit we apply the dominated
  convergence theorem with bounding functions $\beta_R^{(\omega_n)}$
  (\cref{lem:dct:modif}).
  
  The distance condition now follows by taking limits. For any subset $\cS
  \subseteq \cA_1^\star \times \cA_2^\star$, we have
  \begin{align*}
      \eta_L(\cS) - \exp(\varepsilon) \cdot \eta_R(\cS)
      &= \lim_{n \to \infty} \eta_L^{(\omega_n)}(\cS) -
      \lim_{n \to \infty} \exp(\varepsilon_{\omega_n}) \cdot \lim_{n \to \infty} \eta_R^{(\omega_n)}(\cS) \\
      &= \lim_{n \to \infty} \left( \eta_L^{(\omega_n)}(\cS) - \exp(\varepsilon_{\omega_n}) \cdot \eta_R^{(\omega_n)}(\cS) \right) \\
      &\leq \lim_{n \to \infty} \delta_{\omega_n} \\
      &= \delta. \qedhere
  \end{align*}
\end{proof}

Finally, we obtain the countable version of \cref{lem:alift:domfin}.

\begin{theorem} \label{thm:sato-to-lift}
  Let $\mu_1$ and $\mu_2$ be sub-distributions over countable sets $\cA_1$ and
  $\cA_2$, and let $\cR \subseteq \cA_1 \times \cA_2$ be a binary relation such
  that
  $\mu_1(\cS_1) \leq \exp(\varepsilon) \cdot \mu_2(\cR(\cS_1)) + \delta$
  for every $\cS_1 \subseteq \cA_1$.  Then there exists an approximate lifting
  \[
    \mu_1 \alift{\cR}{(\varepsilon, \delta)} \mu_2 .
  \]
\end{theorem}

\begin{proof}
  Since $\cA_1$ and $\cA_2$ are countable, there are finite subsets $\cI_1^{(n)}
  \subseteq \cA_1, \cI_2^{(n)} \subseteq \cA_2$ such that $\{ \cI_1^{(n)} \}_n$
  and $\{ \cI_2^{(n)} \}_n$ are increasing with $\cup_n \cI_1^{(n)} = \cA_1$ and
  $\cup_n \cI_2^{(n)} = \cA_2$. Consider the sequences of restricted
  sub-distributions
  \[
    \mu_1^{(n)}(a_1) \triangleq
    \begin{cases}
      \mu_1(a_1) &: a_1 \in \cI_1^{(n)} \\
      0 &: \text{otherwise}
    \end{cases}
    \qquad\qquad
    \mu_2^{(n)}(a_2) \triangleq
    \begin{cases}
      \mu_2(a_2) &: a_2 \in \cI_2^{(n)} \\
      0 &: \text{otherwise} .
    \end{cases}
  \]
  For any subset $\cS_1 \subseteq \cA_1$, by assumption
  \[
    \mu_1(\cS_1) \leq \exp(\varepsilon) \cdot  \mu_2(\cR(\cS_1)) + \delta .
  \]
  On the restricted sub-distributions, we have
  \begin{align*}
    \mu_1^{(n)}(\cS_1) \leq \mu_1(\cS_1)
    &\leq \exp(\varepsilon)  \cdot \mu_2(\cR(\cS_1) \cap \cI_2^{(n)}) 
    + \exp(\varepsilon)  \cdot \mu_2(\cA_2 \setminus \cI_2^{(n)})
    + \delta \\
    &\triangleq \exp(\varepsilon)  \cdot \mu_2^{(n)}(\cR(\cS_1)) + \delta_n.
  \end{align*}
  Evidently $\lim_{n \to \infty} \delta_n = \delta$. Since $\mu_1^{(n)}$ and
  $\mu_2^{(n)}$ have finite support contained in $\cI_1^{(n)}$ and
  $\cI_2^{(n)}$, \cref{lem:alift:domfin} gives an approximate lifting for each
  finite restriction:
  \[
    \mu_1^{(n)} \alift{\cR}{(\varepsilon, \delta_n)} \mu_2^{(n)} .
  \]
  Since $\mu_1^{(n)}$ and $\mu_2^{(n)}$ converge in $L^1$ to $\mu_1$ and
  $\mu_2$, we can conclude by \cref{lem:bw}.
\end{proof}

\subsection{Alternative proofs of coupling constructions}

The equivalence from \cref{thm:lift-to-sato,thm:sato-to-lift} gives a convenient way
to construct approximate couplings. For instance, we can easily prove a
transitivity principle.

\begin{lemma} \label{lem:alift-trans}
  Let $\mu_1, \mu_2, \mu_3$ be sub-distributions over $\cA_1, \cA_2, \cA_3$
  respectively, and let $\cR \subseteq \cA_1 \times \cA_2$ and $\cS \subseteq
  \cA_2 \times \cA_3$ be binary relations. If we have
  \[
    \mu_1 \alift{\cR}{(\varepsilon, \delta)} \mu_2
    \qquad\text{and}\qquad
    \mu_2 \alift{\cS}{(\varepsilon', \delta')} \mu_3 ,
  \]
  then we also have
  \[
    \mu_1 \alift{(\cS \circ \cR)}{(\varepsilon + \varepsilon', \exp(\varepsilon') \delta + \delta')} \mu_3 .
  \]
\end{lemma}
\begin{proof}
  Let $\cT_1 \subseteq \cA_1$ be any subset. By \cref{thm:lift-to-sato} we have
  \begin{align*}
    \mu_1(\cT_1) &\leq \exp(\varepsilon) \mu_2(\cR(\cT_1)) + \delta \\
    \mu_2(\cR(\cT_1)) &\leq \exp(\varepsilon') \mu_3((\cS \circ \cR)(\cT_1)) + \delta' .
  \end{align*}
  Chaining the inequalities and applying \cref{thm:sato-to-lift} yields the
  desired approximate lifting.
\end{proof}

We can also give alternative proofs for the couplings from \cref{chap:approx}.

\aliftextend*
\begin{proof}[Proof (alternative)]
  For the forward direction, let $\cT_1 \subseteq \cB_1$ be any subset. Then
  \begin{align}
    \liftf{f}_1(\mu_1)(\cT_1) &= \mu_1(f_1^{-1}(\cT_1)) \notag \\
    &\leq \exp(\varepsilon) \cdot \mu_2(f_2^{-1}(\cR(\cT_1))) + \delta
    \tag{\cref{thm:lift-to-sato}} \\
    &= \exp(\varepsilon) \cdot \liftf{f}_2(\mu_2)(\cR(\cT_1)) + \delta , \notag
  \end{align}
  so we conclude by \cref{thm:sato-to-lift}.
  For the reverse direction, let $\cS_1 \subseteq \cA_1$ be any subset. Then
  \begin{align}
    \mu_1(\cS_1) &\leq \mu_1(f_1^{-1}(f_1(\cS_1))) \notag \\
    &= \liftf{f}_1(\mu_1)(f_1(\cS_1)) \notag \\
    &\leq \exp(\varepsilon) \cdot \liftf{f}_2(\mu_2)(\cR(f_1(\cS_1))) + \delta
    \tag{\cref{thm:lift-to-sato}} \\
    &= \exp(\varepsilon) \cdot \mu_2(f_2^{-1}(\cR(f_1(\cS_1)))) + \delta . \notag
  \end{align}
  Since $f_1(x_1) \mathrel{\cR} f_2(x_2)$ precisely when $x_1 \mathrel{(f_2^{-1}
    \circ \cR \circ f_1)} x_2$, we conclude by \cref{thm:sato-to-lift}.
\end{proof}

\lapnull*
\begin{proof}[Proof (alternative)]
  Let $\cS \subseteq \ZZ$ be any subset and let $\cS'$ be the set $\{ s - v_1 +
  v_2 \mid s \in \cS \}$. Noting $\Lap{\varepsilon}(v_1)(s) =
  \Lap{\varepsilon}(v_2)(s - v_1 + v_2)$ for every $s$ and summing over all $s
  \in \cS$, we have
  \[
    \Lap{\varepsilon}(v_1)(\cS) = \Lap{\varepsilon}(v_2)(\cS') .
  \]
  \Cref{thm:sato-to-lift} gives the desired approximate coupling.
\end{proof}

\lapgen*
\begin{proof}[Proof (alternative)]
  Let $\cS \subseteq \ZZ$ be any subset and let $\cS'$ be the set $\{ s + k \mid s
  \in \cS \}$. Noting
  \begin{align*}
    \Lap{\varepsilon}(v_1)(s)
    &= \Lap{\varepsilon}(v_2)(s - v_1 + v_2) \\
    &\leq \exp(|k - v_2 + v_1| \varepsilon) \cdot \Lap{\varepsilon}(v_2)(s + k) \\
    &\leq \exp(k' \varepsilon) \cdot \Lap{\varepsilon}(v_2)(s + k)
  \end{align*}
  for every $s$ and summing over all $s \in \cS$, we have
  \[
    \Lap{\varepsilon}(v_1)(\cS) = \exp(k'\varepsilon) \cdot \Lap{\varepsilon}(v_2)(\cS') .
  \]
  \Cref{thm:sato-to-lift} gives the desired approximate coupling.
\end{proof}

\pweq*
\begin{proof}[Proof (alternative)]
  By \cref{thm:lift-to-sato} we have $\mu_1(i) \leq \exp(\varepsilon) \cdot
  \mu_2(i) + \delta_i$ for every $i \in \cR$. Hence for any set $\cS \subseteq
  \cR$, summing over $i \in \cS$ gives
  \[
    \mu_1(\cS) \leq \exp(\varepsilon) \cdot \mu_2(\cS) + \sum_{i \in \cS} \delta_i
    \leq \exp(\varepsilon) \cdot \mu_2(\cS) + \delta .
  \]
  \Cref{thm:sato-to-lift} gives the desired approximate coupling.
\end{proof}

For the couplings we introduce in the rest of this chapter, we will give each
construction in two ways: first as a consequence of Sato's definition, then in
terms of two explicit witness distributions.

\section{Accuracy-dependent approximate couplings} \label{sec:aprhl-utb}

A common technique in proofs for cryptographic protocols is \emph{up-to-bad}
reasoning. Roughly, two versions of a protocol---say, one that operates on the
true secret information and one that operates on random noise---are said to be
\emph{equivalent up-to-bad} if they have the same distribution over outputs
assuming some probabilistic event, the so-called \emph{bad event}, does not
happen. If the bad event has small probability, up-to-bad equivalence implies
that the output distributions of the two programs are close.  This principle can
be seen as a property about exact couplings, a consequence of the coupling
method (\cref{thm:coupling-method}).

\begin{proposition} \label{prop:uptobad}
  Let $\mu_1, \mu_2$ be sub-distributions over $\cA$ and let $\cP \subseteq \cA$
  be a subset. If for $i \in \{ 1, 2 \}$ we have an exact lifting
  \[
    \mu_1 \lift{ \{ (x_1, x_2) \mid x_i \in \cP \to x_1 = x_2 \} } \mu_2 ,
  \]
  then $\tvdist{\mu_1}{\mu_2} \leq \mu_i(\cA \setminus \cP)$.
\end{proposition}
\begin{proof}
  Let $\mu$ be the witness. We have
  \[
    \tvdist{\mu_1}{\mu_2} \leq \Pr_{(x_1, x_2) \sim \mu}[ x_1 \neq x_2 ]
    = \Pr_{(x_1, x_2) \sim \mu}[ x_1 \neq x_2 \land x_i \notin \cP ]
    \leq \mu_i(\cA \setminus \cP) ,
  \]
  by \cref{thm:coupling-method}, the support condition, and the first marginal
  condition.
\end{proof}

\subsection{Up-to-bad approximate couplings}

The $\delta$ parameter of an approximate coupling is closely related to
TV-distance. For example, the distance bound $\epsdist{0}{\mu_1}{\mu_2} \leq
\delta$ is equivalent to $\tvdist{\mu_1}{\mu_2} \leq \delta$ for proper
distributions. This observation suggests we can generalize \cref{prop:uptobad}
to approximate couplings. We introduce two constructions, which we call
\emph{up-to-bad approximate couplings}.

\begin{proposition} \label{prop:alift-utb-math}
  Let $\mu_1, \mu_2$ be sub-distributions over $\cA_1$ and $\cA_2$, and let
  $\cP_1, \cP_2$ be subsets of $\cA_1$ and $\cA_2$. Consider any binary relation
  $\cR \subseteq \cA_1 \times \cA_2$.
  \begin{enumerate}
    \item If $\mu_1(\cA_1 \setminus \cP_1) \leq \delta'$, then
      \[
        \mu_1 \alift{\{ (a_1, a_2) \mid a_1 \in \cP_1 \to (a_1, a_2) \in \cR \}} {(\varepsilon, \delta)} \mu_2
      \qquad
      \text{implies}
      \qquad
        \mu_1 \alift{\cR}{(\varepsilon, \delta + \delta')} \mu_2 .
      \]
    \item If $\mu_2(\cA_2 \setminus \cP_2) \leq \delta'$, then
      \[
        \mu_1 \alift{\{ (a_1, a_2) \mid a_2 \in \cP_2 \to (a_1, a_2) \in \cR \}}{(\varepsilon, \delta)} \mu_2
      \qquad
      \text{implies}
      \qquad
        \mu_1 \alift{\cR}{(\varepsilon, \delta + \exp(\varepsilon) \cdot  \delta')} \mu_2 .
      \]
  \end{enumerate}
\end{proposition}

The slight difference between the two versions is due to our asymmetric
definition of approximate coupling; bad events in $\mu_1$ are not treated the
same as bad events in $\mu_2$.

\begin{proof}
  We first introduce some notation for binary relations and sets. First, we will
  interpret $\cP_1$ and $\cP_2$ as subsets of $\cA_1 \times \cA_2$ via $\cP_1
  \times \cA_2$ and $\cA_1 \times \cP_2$. If $\cR$ is a binary relation over
  $\cB_1 \times \cB_2$, we write $\neg \cR$ for the binary relation $\cB_1
  \times \cB_2 \setminus \cR$. Finally, we write $\cA \to \cB$ for the binary
  relation $\neg \cB \cup \cA$.

  To prove the first point, let $\cS_1 \subseteq \cA_1$ be any subset.  By assumption
  and \cref{thm:lift-to-sato},
  \[
    \mu_1(\cS_1 \cap \cP_1) \leq \exp(\varepsilon) \cdot \mu_2((\cP_1 \to \cR)(\cS_1 \cap \cP_1)) + \delta
    = \exp(\varepsilon) \cdot \mu_2(\cR(\cS_1 \cap \cP_1)) + \delta .
  \]
  Since $\mu_1(\neg \cP_1) \leq \delta'$, we also have
  \[
    \mu_1(\cS_1) \leq \mu_1(\cS_1 \cap \cP_1) + \delta'
    \leq \exp(\varepsilon) \cdot \mu_2(\cR(\cS_1 \cap \cP_1)) + \delta + \delta'
    \leq \exp(\varepsilon) \cdot \mu_2(\cR(\cS_1)) + \delta + \delta'
  \]
  and hence \cref{thm:sato-to-lift} gives the desired approximate coupling.

  The second point is similar. Let $\cS_1 \subseteq \cA_1$ be any subset.
  By assumption and \cref{thm:lift-to-sato},
  \begin{align*}
    \mu_1(\cS_1) &\leq \exp(\varepsilon) \cdot \mu_2((\cP_2 \to \cR)(\cS_1)) + \delta \\
    &\leq \exp(\varepsilon) \cdot \mu_2(\neg \cP_2) + \exp(\varepsilon) \cdot \mu_2(\cR(\cS_1)) + \delta \\
    &\leq \exp(\varepsilon) \cdot \mu_2(\cR(\cS_1)) + \delta + \exp(\varepsilon) \cdot  \delta'
  \end{align*}
  and hence \cref{thm:sato-to-lift} gives the desired approximate coupling.
  
  \medskip

  To give witnesses for the first point, let $\mu_L, \mu_R$ witness the
  approximate lifting of $\cP_1 \to \cR$.  We define two witnesses $\eta_L,
  \eta_R \in \SDist(\cA_1^\star \times \cA_2^\star)$ for the approximate lifting
  of $\cR$:
  \begin{align*}
    \eta_L(a_1, a_2) &\triangleq
      \begin{cases}
        \mu_L(a_1, a_2) &: (a_1, a_2) \in \cR \\
        \mu_L(a_1, \star) + \sum_{a_2 \in \cA_2 : (a_1, a_2) \notin \cR} \mu_L(a_1, a_2)
                        &: a_2 = \star \\
        0               &: \text{otherwise.}
      \end{cases}
      \\
    \eta_R(a_1, a_2) &\triangleq
      \begin{cases}
        \mu_R(a_1, a_2) &: (a_1, a_2) \in \cR \\
        \mu_R(\star, a_2) + \sum_{a_1' \in \cA_1 : (a_1', a_2) \notin \cR} \mu_R(a_1', a_2)
                        &: a_1 = \star \\
        0               &: \text{otherwise.}
      \end{cases}
  \end{align*}
  By construction, $\supp(\eta_L) \cup \supp(\eta_R) \subseteq \cR^\star$. We
  can check the first marginal condition:
  \begin{align*}
    \pi_1(\eta_L)(a_1)
    &= \sum_{a_2 \in \cA_2^\star} \eta_L(a_1, a_2) \\
    &= \eta_L(a_1, \star) + \sum_{a_2 \in \cA_2 : (a_1, a_2) \in \cR} \eta_L(a_1, a_2) \\
    &= \mu_L(a_1, \star)
    + \sum_{a_2 \in \cA_2 : (a_1, a_2) \notin \cR} \mu_L(a_1, a_2)
    + \sum_{a_2 \in \cA_2 : (a_1, a_2) \in \cR} \mu_L(a_1, a_2) \\
    &= \sum_{a_2 \in \cA_2^\star} \mu_L(a_1, a_2) = \pi_1(\mu_L)(a_1) .
  \end{align*}
  The second marginal is similar:
  \begin{align*}
    \pi_2(\eta_R)(a_2)
    &= \sum_{a_1 \in \cA_1^\star} \eta_R(a_1, a_2) \\
    &= \eta_R(\star, a_2)
    + \sum_{a_1 \in \cA_1 : (a_1, a_2) \in \cR} \eta_R(a_1, a_2) \\
    &= \mu_R(\star, a_2)
    + \sum_{a_1 \in \cA_1 : (a_1, a_2) \notin \cR} \mu_R(a_1, a_2)
    + \sum_{a_1 \in \cA_1 : (a_1, a_2) \in \cR} \mu_R(a_1, a_2) \\
    &= \sum_{a_1 \in \cA_1^\star} \mu_R(a_1, a_2) = \pi_2(\mu_R)(a_2) .
  \end{align*}
  It remains to check the distance condition. Compared to the old witnesses, the
  new witnesses have larger mass on subsets satisfying $\cR^\star$: for all
  subsets $\cS \subseteq \cR^\star$, we have $\mu_L(\cS) \leq \eta_L(\cS)$ and
  $\mu_R(\cS) \leq \eta_R(\cS)$. For any set $\cS \subseteq \cA_1^\star \times
  \cA_2^\star$, we can also bound $\eta_L(\cS)$ from above:
  \begin{align*}
    \eta_L(\cS)
    &= \sum_{(a_1, a_2) \in \cS \cap \cR^\star} \eta_L(a_1, a_2) \\
    &= \sum_{(a_1, a_2) \in \cS \cap \cR} \eta_L(a_1, a_2)
    + \sum_{(a_1, \star) \in \cS} \eta_L(a_1, \star) \\
    &= \sum_{(a_1, a_2) \in \cS \cap \cR} \mu_L(a_1, a_2)
    + \sum_{(a_1, \star) \in \cS} \left( \mu_L(a_1, \star)
    + \sum_{a_2 \in \cA_2 : (a_1, a_2) \notin \cR} \mu_L(a_1, a_2) \right) \\
    &\leq \sum_{(a_1, a_2) \in \cS} \mu_L(a_1, a_2)
    + \sum_{a_1 \in \cA_1 \setminus \cP_1} \sum_{a_2 \in \cA_2^\star} \mu_L(a_1, a_2) \\
    &= \sum_{(a_1, a_2) \in \cS} \mu_L(a_1, a_2) + \mu_1(\neg \cP_1) \\
    &\leq \mu_L(\cS) + \delta' .
  \end{align*}
  The first inequality uses the support of $\mu_L$; the final inequality is
  by assumption. Finally, we chain these bounds:
  \begin{align*}
    \eta_L(\cS)
    &= \eta_L(\cS \cap \cR^\star) \\
    &\leq \mu_L(\cS \cap \cR^\star) + \delta' \\
    &\leq \exp(\varepsilon) \cdot \mu_R(\cS \cap \cR^\star) + \delta + \delta' \\
    &\leq \exp(\varepsilon) \cdot \eta_R(\cS \cap \cR^\star) + \delta + \delta' \\
    &= \exp(\varepsilon) \cdot \eta_R(\cS) + \delta + \delta' .
  \end{align*}
  This implies $\epsdist{\varepsilon}{\eta_L}{\eta_R} \leq \delta + \delta'$,
  so $\eta_L$ and $\eta_R$ witness the approximate lifting.

  To give witnesses for the second point, let $\eta_L, \eta_R$ be defined as
  above and consider any subset $\cS \subseteq \cA_1^\star \times \cA_2^\star$.
  The marginal and support conditions follow as before. To check the distance
  condition, we first bound $\eta_L$ in terms of $\mu_L$:
  \begin{align*}
    \eta_L(\cS)
    &= \sum_{(a_1, a_2) \in \cS} \eta_L(a_1, a_2) \\
    &\leq \sum_{(a_1, a_2) \in \cS \cap \cR^\star} \mu_L(a_1, a_2)
    + \sum_{(a_1, a_2) \notin \cR} \mu_L(a_1, a_2) \\
    &= \mu_L(\cS \cap \cR^\star) + \mu_L(\neg \cR) \\
    &= \mu_L( (\cS \cap \cR^\star) \cup \neg \cR )
  \end{align*}
  The last equality is because the two events are disjoint. We then complete the
  calculation as before:
  \begin{align*}
    \eta_L(\cS) 
    &\leq \mu_L( (\cS \cap \cR^\star) \cup \neg \cR ) \\
    &\leq \exp(\varepsilon) \cdot \mu_R( (\cS \cap \cR^\star) \cup \neg \cR ) + \delta \\
    &\leq \exp(\varepsilon) (\mu_R(\cS \cap \cR^\star) + \mu_R( \neg \cP_2 )) + \delta \\
    &= \exp(\varepsilon) (\mu_R(\cS \cap \cR^\star) + \mu_R( \neg \cP_2 )) + \delta \\
    &\leq \exp(\varepsilon) (\mu_R(\cS \cap \cR^\star) + \delta') + \delta \\
    &\leq \exp(\varepsilon) (\eta_R(\cS \cap \cR^\star) + \delta') + \delta \\
    &= \exp(\varepsilon) \cdot \eta_R(\cS) + \delta + \exp(\varepsilon) \cdot  \delta' .
  \end{align*}
  Thus $\epsdist{\varepsilon}{\eta_L}{\eta_R} \leq \delta + \exp(\varepsilon)
  \cdot  \delta'$, so $(\eta_L, \eta_R)$ witness the desired approximate
  coupling.
\end{proof}

\begin{figure}
  \begin{mathpar}
    \inferruleref{UtB-L}
    { \vdash \aprhl{c_1}{c_2}{\Phi}{ \Theta\sidel \to \Psi }{(\varepsilon, \delta)} \\
      \forall m,\; \models \Pr_{ \denot{c_1} m }[ \neg \Theta ] \leq \delta'
    }
    { \vdash \aprhl{c_1}{c_2}{\Phi}{\Psi}{(\varepsilon, \delta + \delta')} }
    \label{rule:aprhl-utb-l}
    \\
    \inferruleref{UtB-R}
    { \vdash \aprhl{c_1}{c_2}{\Phi}{ \Theta\sider \to \Psi }{(\varepsilon, \delta)} \\
      \forall m,\; \models \Pr_{ \denot{c_2} m }[ \neg \Theta ] \leq \delta'
    }
    { \vdash \aprhl{c_1}{c_2}{\Phi}{\Psi}{(\varepsilon, \delta + \exp(\varepsilon) \cdot \delta')} }
    \label{rule:aprhl-utb-r}
  \end{mathpar}
  \caption{Up-to-bad rules for \Saprhl}
  \label{fig:aprhl-utb}
\end{figure}

We realize these couplings in \Saprhl with the up-to-bad rules in
\cref{fig:aprhl-utb}. In both rules, $\Theta$ is a predicate on $\Mem$;
$\Theta\sidel$ and $\Theta\sider$ are the associated predicates on the product
memories $\Mem_\times$; syntactically, where all variables in $\Theta$ are
tagged with $\sidel$ or $\sider$ respectively.

\begin{theorem} \label{thm:utb-sound}
  The rules \nameref{rule:aprhl-utb-l} and \nameref{rule:aprhl-utb-r} are sound.
\end{theorem}
\begin{proof}
  By validity of the premises and \cref{prop:alift-utb-math}.
\end{proof}

\Cref{fig:aprhl-oneand} presents two useful variants of the up-to-bad rules that
are restricted versions of the rule of conjunction from Hoare logic.  As we
discussed before, the general conjunction rule is not sound in \Sprhl, nor in
\Saprhl. However if one of the conjuncts mentions only one side, we can recover
a version of the conjunction rule.

\begin{figure}
  \begin{mathpar}
    \inferruleref{And-L}
    { \vdash \aprhl{c_1}{c_2}{\Phi}{ \Psi }{(\varepsilon, \delta)} \\
      \forall m,\; \models \Pr_{ \denot{c_1} m }[ \neg \Theta ] \leq \delta'
    }
    { \vdash \aprhl{c_1}{c_2}{\Phi}{\Theta\sidel \land \Psi}{(\varepsilon, \delta + \delta')} }
    \label{rule:aprhl-and-l}
    \\
    \inferruleref{And-R}
    { \vdash \aprhl{c_1}{c_2}{\Phi}{ \Psi }{(\varepsilon, \delta)} \\
      \forall m,\; \models \Pr_{ \denot{c_2} m }[ \neg \Theta ] \leq \delta'
    }
    { \vdash \aprhl{c_1}{c_2}{\Phi}{\Theta\sider \land \Psi}{(\varepsilon, \delta + \exp(\varepsilon) \cdot  \delta')} }
    \label{rule:aprhl-and-r}
  \end{mathpar}
  \caption{One-sided conjunction rules for \Saprhl}
  \label{fig:aprhl-oneand}
\end{figure}

\begin{corollary} \label{cor:and-sound}
  The rules \nameref{rule:aprhl-and-l} and \nameref{rule:aprhl-and-r} are sound. 
\end{corollary}
\begin{proof}
  From the premise of \nameref{rule:aprhl-and-l}, the rule of consequence gives
  \[
    \vdash \aprhl{c_1}{c_2}{\Phi}{ \Theta\sidel \to \Theta\sidel \land \Psi }{(\varepsilon, \delta)} 
  \]
  and hence we can conclude by applying \nameref{rule:aprhl-utb-l}:
  \[
    \vdash \aprhl{c_1}{c_2}{\Phi}{ \Theta\sidel \land \Psi }{(\varepsilon, \delta + \delta')}  .
  \]
  Similarly, we can derive \nameref{rule:aprhl-and-r} from \nameref{rule:aprhl-utb-r}.
\end{proof}

When $\delta' = 0$, the rules \nameref{rule:aprhl-and-l} and
\nameref{rule:aprhl-and-r} can add one-sided support assertions to the
post-condition of any \Saprhl rule. This can be useful to work around the narrow
post-conditions in certain \Saprhl rules (e.g., \nameref{rule:aprhl-pw-eq}). We
can also use these rules to introduce accuracy bounds. As an example, we give a
basic tail bound for the discrete Laplace distribution.
\begin{proposition} \label{prop:lap-tail}
  Let $\varepsilon, \beta > 0$ and let $t \in \ZZ$. Then we can bound the
  probability of samples from the Laplace distribution being far from the mean:
  \[
    \Pr_{x \sim \Lap{\varepsilon}(t)}
    \left[ |x - t| > \frac{1}{\varepsilon} \ln \frac{1}{\beta} \right] 
    \leq \beta .
  \]
\end{proposition}
This bound gives the two rules in \cref{fig:aprhl-laptail}.
\begin{figure}
  \begin{mathpar}
    \inferruleref{LapAcc-L}
    { \vdash \aprhl{\Rand{x_1}{\Lap{\varepsilon}(e_1)}}{c_2}
      {\Phi}{\Psi}{(\varepsilon, \delta)} \\
      x_1 \notin \FV(e_1) }
    { \vdash \aprhl{\Rand{x_1}{\Lap{\varepsilon}(e_1)}}{c_2}
      {\Phi}{|x_1\sidel - e_1\sidel| \leq  \frac{1}{\varepsilon} \ln \frac{1}{\beta} \land \Psi}
      {(\varepsilon, \delta + \beta)} }
    \label{rule:aprhl-lapacc-l}
    \\
    \inferruleref{LapAcc-R}
    { \vdash \aprhl{c_1}{\Rand{x_2}{\Lap{\varepsilon}(e_2)}}
      {\Phi}{\Psi}{(\varepsilon, \delta)} \\
      x_2 \notin \FV(e_2) }
    { \vdash \aprhl{c_1}{\Rand{x_2}{\Lap{\varepsilon}(e_2)}}
      {\Phi}{|x_2\sider - e_2\sider| \leq  \frac{1}{\varepsilon} \ln \frac{1}{\beta} \land \Psi}
      {(\varepsilon, \delta + \exp(\varepsilon) \cdot \beta)} }
    \label{rule:aprhl-lapacc-r}
  \end{mathpar}
  \caption{Laplace accuracy bounds in \Saprhl}
  \label{fig:aprhl-laptail}
\end{figure}
\begin{corollary} \label{cor:acc-sound}
  The rules \nameref{rule:aprhl-lapacc-l} and \nameref{rule:aprhl-lapacc-r} are
  sound. 
\end{corollary}
\begin{proof}
  By the rules \nameref{rule:aprhl-and-l}, \nameref{rule:aprhl-and-r}, and
  \cref{prop:lap-tail}.
\end{proof}

\section{Optimal subset coupling} \label{sec:aprhl-subset}

By \cref{prop:alift-impl}, approximate lifted implication
\[
  \mu_1
  \alift{\{ (a_1, a_2) \mid a_1 \in \cS_1 \to a_2 \in \cS_2 \}}{(\varepsilon, \delta)}
  \mu_2
\]
ensures $\mu_1(\cS_1) \leq \exp(\varepsilon) \cdot \mu_2(\cS_2) + \delta$. In this
section, we explore a partial converse.

\begin{theorem}[Optimal subset coupling] \label{thm:opt-subset-math}
  Let $\alpha \geq 1$ and $\delta \geq 0$. Let $\mu_1$ and $\mu_2$ be
  sub-distributions over $\cA_1$ and $\cA_2$ with equal weight, and consider
  subsets $\cS_1 \subseteq \cA_1, \cS_2 \subseteq \cA_2$.  Then $\mu_1(\cS_1)
  \leq \alpha \mu_2(\cS_2) + \delta$ and $\mu_1(\cA_1 \setminus \cS_1) \leq
  \alpha \mu_2(\cA_2 \setminus \cS_2) + \delta$ if and only if
  \[
    \mu_1
    \alift{ \{ (a_1, a_2) \mid a_1 \in \cS_1 \leftrightarrow a_2 \in \cS_2 \} }
    {(\ln \alpha, \delta)}
    \mu_2 .
  \]
\end{theorem}

The equivalence shows that approximate couplings can capture the bounds
$\mu_1(\cS_1) \leq \alpha \mu_2(\cS_2) + \delta$ and $\mu_1(\cA_1 \setminus
\cS_1) \leq \alpha \mu_2(\cA_2 \setminus \cS_2) + \delta$ with the most precise
approximation parameters, much like the maximal coupling can precisely model the
TV-distance between two distributions.

\begin{proof}
  The reverse direction follows by \cref{thm:lift-to-sato}. For the forward
  implication, take any set $\cT_1 \subseteq \cA_1$ and write $\cR$ for the
  relation $\{ (a_1, a_2) \mid a_1 \in \cS_1 \leftrightarrow a_2 \in \cS_2 \}$.
  If $\cT_1 \cap \cS_1$ and $\cT_1 \cap (\cA_1 \setminus \cS_1)$ are both
  non-empty, then $\cR(\cT_1) = \cA_2$ and then clearly $\mu_1(\cT_1) \leq
  \alpha \mu_2(\cR(\cT_1)) + \delta$ as $\mu_1$ and $\mu_2$ have equal weights.
  Otherwise $\cT_1$ is contained in $\cS_1$ or in $\cA_1 \setminus \cS_1$.  In
  the first case, $\cR(\cT_1) = \cS_2$ and so
  \[
    \mu_1(\cT_1) \leq \mu_1(\cS_1) \leq \alpha \mu_2(\cS_2) + \delta = \alpha
    \mu_2(\cR(\cT_1)) + \delta 
  \]
  by assumption. In the second case, $\cR(\cT_1) = \cA_2 \setminus \cS_2$ and we again
  have $\mu_1(\cT_1) \leq \mu_2(\cR(\cT_1)) + \delta$. Hence we have the desired
  approximate coupling by \cref{thm:sato-to-lift}.
  
  Alternatively, we can directly construct two witnesses. For simplicity we
  consider just the case $\delta = 0$.
  \todo{Generalize?}
  Define:
  \begin{align*}
    \mu_L(a_1, a_2) &\triangleq
    \begin{cases}
      \frac{\mu_1(a_1) \cdot \mu_2(a_2)}{\mu_2(\cS_2)}
      &: \text{if } a_1 \in \cS_1 \text{ and } a_2 \in \cS_2 \\
      \frac{\mu_1(a_1) \cdot \mu_2(a_2)}{\mu_2(\cA_2 \setminus \cS_2)}
      &: \text{if } a_1 \notin \cS_1^\star \text{ and } a_2 \notin \cS_2^\star \\
      0
      &: \text{otherwise} .
    \end{cases}
    \\
    \mu_R(a_1, a_2) &\triangleq
    \begin{cases}
      \frac{\mu_1(a_1) \cdot \mu_2(a_2)}{\mu_1(\cS_1)}
      &: \text{if } a_1 \in \cS_1 \text{ and } a_2 \in \cS_2 \\
      \frac{\mu_1(a_1) \cdot \mu_2(a_2)}{\mu_1(\cA_1 \setminus \cS_1)}
      &: \text{if } a_1 \notin \cS_1^\star \text{ and } a_2 \notin \cS_2^\star \\
      \mu_2(a_2) - \sum_{a_1' \in \cA_1} \mu_R(a_1', a_2)
      &: \text{if } a_1 = \star \\
      0
      &: \text{otherwise} .
    \end{cases}
  \end{align*}
  When any denominator is zero, we treat the fraction as zero. It is not hard to
  see that the support conditions are satisfied. To show the marginal
  conditions, there are a few cases. Consider the first marginal
  $\pi_1(\mu_L)(a_1)$. For $a_1 \in \cS_1$, if $\mu_2(\cS_2) = 0$ then
  $\mu_1(\cS_1) = 0$ by assumption; if $\mu_1(\cS_2) > 0$ then the marginal is
  clear.  Likewise for $a_1 \notin \cS_1^\star$, if $\mu_2(\cA_2 \setminus
  \cS_2) = 0$ then $\mu_1(\cA_1 \setminus \cS_1) = 0$ by assumption and
  $\pi_1(\mu_L)(a_1) = 0$; if $\mu_2(\cA_2 \setminus \cS_2) > 0$ then the
  marginal is clear. The second marginal $\pi_2(\mu_R) = \mu_2$ holds by
  construction, after checking $\mu_R(\star, a_2) \geq 0$.
  
  Finally for the distance condition, $\mu_L(a_1, a_2) \leq \alpha \mu_R(a_1,
  a_2)$ by the first assumption when $(a_1, a_2) \in \cS_1 \times \cS_2$; by the
  second assumption when $(a_1, a_2) \in (\cA_1 \setminus \cS_1) \times (\cA_2
  \setminus \cS_2)$; and trivially in all other cases since $\mu_L(a_1, a_2) =
  \mu_R(a_1, a_2) = 0$. Hence we have a $(\ln \alpha, 0)$-approximate coupling.
\end{proof}

A useful special case is when the distributions are equal and the subsets are
nested. 

\begin{corollary}[Optimal subset coupling] \label{cor:opt-subset-math}
  Let $\mu$ be a sub-distribution over $\cA$ and consider nested sets $\cS_2
  \subseteq \cS_1 \subseteq \cA$. Then $\mu(\cS_1) \leq \alpha \mu(\cS_2) +
  \delta$ if and only if
  \[
    \mu
    \alift{ \{ (a_1, a_2) \mid a_1 \in \cS_1 \leftrightarrow a_2 \in \cS_2 \} }
    {(\ln \alpha, \delta)}
    \mu .
  \]
\end{corollary}
\begin{proof}
  By \cref{thm:opt-subset-math}; the requirement $\mu(\cA \setminus \cS_1) \leq
  \alpha \mu(\cA \setminus \cS_2) + \delta$ is automatic since $\cS_2 \subseteq
  \cS_1$.
\end{proof}

As an application, we give a subset coupling for the Laplace distribution.
First, we prove a bound relating the probabilities of two nested intervals for
the Laplace distribution. A similar bound for the continuous Laplace
distribution was originally proved by \citet{BunSU16}; we adapt their proof to
the discrete case.

\begin{proposition} \label{prop:lapint-bound}
  Let $a, a', b, b' \in \ZZ$ be such that $a < b$ and $[a, b] \subseteq
  [a', b']$. Then
  \[
    \Pr_{r \sim \Lap{\varepsilon}} [ r \in [ a', b' ] ]
    \leq \alpha \Pr_{r \sim \Lap{\varepsilon}} [ r \in [ a, b ] ]
  \]
  with constants
  \[
    \alpha \triangleq \frac{\exp( \eta \varepsilon )}{1 - \exp(-(b - a + 2)\varepsilon/2)}
    \quad \text{and} \quad
    \eta \triangleq (b' - a') - (b - a) .
  \]
\end{proposition}
\begin{proof}
  Let $W$ be the total mass of the Laplace distribution before normalization. By
  a calculation,
  \[
    W
    = \sum_{r = -\infty}^{+\infty} \exp(- |r| \varepsilon)
    = \frac{e^\varepsilon + 1}{e^\varepsilon - 1} .
  \]
  Let $L(x, y)$ be the mass of the Laplace distribution in $[x, y]$.  We want to
  bound $L(a', b') \leq \alpha L(a, b)$.  There are four cases: $a < b \leq 0$,
  $a < 0 < b$ with $|a| \leq |b|$, $0 \leq a < b$, and $a < 0 < b$ with $|a|
  \geq |b|$. By symmetry of the Laplace distribution, it suffices to consider
  the first two cases.

  For the first case, $a < b \leq 0$. By direct calculation, we have
  \begin{align*}
    L(a', b')
    &\leq L(a, b) + \frac{1}{W} \sum_{r = b + 1}^{b + \eta} e^{r \varepsilon} \\
    &= \frac{ e^{ (b + 1 + \eta)\varepsilon }  - e^{ a \varepsilon } }{ e^{\varepsilon} + 1 }
    \\
    &= \frac{1}{e^{\varepsilon} + 1} (e^{ (b + 1)\varepsilon }  - e^{ a \varepsilon })
    \left( \frac{ e^{\eta \varepsilon} - e^{- (b - a + 1) \varepsilon } }{ 1 - e^{ - (b
          - a + 1) \varepsilon }} \right)
    \\
    &= \left( \frac{ e^{ \eta \varepsilon } - e^{ - (b - a + 1) \varepsilon } }{ 1 -
        e^{ - (b - a + 1) \varepsilon }} \right) L(a, b)
    \leq \alpha L(a, b) .
  \end{align*}

  For the second case, $a < 0 < b$ with $|a| \leq |b|$. We can bound
  \begin{align*}
    L(a', b')
    &\leq L(a, b) + \eta L(a, a)
    = L(a, b) + \eta \left( \frac{e^{ \varepsilon } - 1}{e^{ \varepsilon } + 1} \right)
    e^{ a \varepsilon }
    \\
    &\leq \left( 1 + \eta \left( \frac{e^{ \varepsilon } - 1}{e^{ \varepsilon } + 1}
      \right)  \frac{e^{ a\varepsilon }}{L(0, b)} \right) L(a, b)
    \\
    &= \left( 1 + \frac{\eta (e^{ \varepsilon } - 1) e^{ a \varepsilon }}{ e^{ \varepsilon } -
        e^{ - b \varepsilon } } \right) L(a, b)
    \\
    &= \left( \frac{1 - e^{ -(b + 1)\varepsilon } + \eta(e^{ \varepsilon } - 1)
        e^{ (a - 1) \varepsilon } }{ 1 - e^{ -(b + 1)\varepsilon } } \right) L(a, b)
    \\
    &\leq \left( \frac{ 1 + \eta (e^{ \varepsilon } - 1)}{1 - e^{ - (b + 1) \varepsilon }} \right) L(a, b)
    \\
    &\leq \left( \frac{ e^{ 2 \eta \varepsilon } }{1 - e^{ - (b - a + 2)\varepsilon/2 }} \right) L(a, b)
    \leq \alpha L(a, b) .
  \end{align*}
  The last line is because $(b + 1) \geq (b - a + 2)/2$, and because $1 + \eta
  (e^{ \varepsilon } - 1) \leq e^{ \eta \varepsilon }$ for $\eta \in \NN$ and
  $\varepsilon \geq 0$; to see this, note that equality holds at $\eta = 0$ and
  \[
    \frac{1 + (\eta + 1) (e^{ \varepsilon } - 1)}{1 + \eta (e^{ \varepsilon } - 1)}
    \leq \frac{e^{ (\eta + 1) \varepsilon }}{e^{ \eta \varepsilon }} = e^{ \varepsilon }
  \]
  for $\varepsilon \geq 0$, so the inequality is preserved as we increase $\eta$.
\end{proof}

As a corollary, we have a subset coupling for the Laplace distribution.

\begin{lemma} \label{cor:lapint-math}
  Let $a, a', b, b' \in \ZZ$ be such that $a < b$ and $[a, b] \subseteq
  [a', b']$. We have an approximate lifting
  \[
    \Lap{\varepsilon}
    \alift{ \{ (r_1, r_2) \mid r_1 \in [a', b'] \leftrightarrow r_2 \in [a, b] \} }
    {(\ln \alpha, 0)}
    \Lap{\varepsilon}
  \]
  with constants
  \[
    \alpha \triangleq \frac{\exp( \eta \varepsilon )}{1 - \exp(-(b - a + 2)\varepsilon/2)}
    \quad \text{and} \quad
    \eta \triangleq (b' - a') - (b - a) .
  \]
\end{lemma}
\begin{proof}
  Immediate by the forward direction of \cref{cor:opt-subset-math} and
  \cref{prop:lapint-bound}.
\end{proof}

\begin{figure}
  \begin{mathpar}
  \inferruleref{LapInt}
  { \varepsilon' \triangleq
    \ln \left( \frac{\exp( \eta \varepsilon )}{1 - \exp(-\sigma\varepsilon/2)} \right) \\
    x_1, x_2 \notin \FV(p, q, r, s, e_1, e_2) \\\\
    \Phi \triangleq
    {\begin{cases}
      |e_1\sidel - e_2\sider| \leq k \\
      p + k \leq r < s \leq q - k \land
      (q - p) - (s - r) \leq \eta \land
      0 < \sigma \leq (s - r) + 2 \\
      \forall w_1, w_2 \in \ZZ,\;
      (w_1 \in [p, q] \leftrightarrow w_2 \in [r, s])
      \to \Psi\subst{x_1\sidel,x_2\sider}{w_1, w_2}
    \end{cases}}
  }
  { \vdash \aprhl
    {\Rand{x_1}{\Lap{\varepsilon}(e_1)}}{\Rand{x_2}{\Lap{\varepsilon}(e_2)}}
    { \Phi }{ \Psi } {(\varepsilon',0) } }
  \label{rule:aprhl-lapint}
\end{mathpar}
\caption{Interval coupling rule \nameref{rule:aprhl-lapint} for \Saprhl}
\label{fig:aprhl-lapint}
\end{figure}

To use this coupling in \Saprhl, we introduce the rule \nameref{rule:aprhl-lapint} in
\cref{fig:aprhl-lapint}. To gain intuition, the following rule is a simplified
special case:
\[
  \inferrule*[Left=LapInt*]
  { \varepsilon' \triangleq
    \ln \left( \frac{\exp( \eta \varepsilon )}{1 - \exp(-\sigma\varepsilon/2)} \right) \\
    x \notin \FV(p, q, r, s) \\
    \models \Phi \to
    {\begin{cases}
      |e\sidel - e\sider| \leq k \\
      p + k \leq r < s \leq q - k \\
      (q - p) - (s - r) \leq \eta \\
      0 < \sigma \leq (s - r) + 2
    \end{cases}}
  }
  { \vdash \aprhl
    {\Rand{x}{\Lap{\varepsilon}(e)}}{\Rand{x}{\Lap{\varepsilon}(e)}}
    { \Phi }{ x\sidel \in [p, q] \leftrightarrow x\sider \in [r, s] } {(\varepsilon',0) } }
\]
Ignoring the technical side-conditions, this rule gives an approximate coupling
relating the samples in $[p, q]$ in the first distribution with the samples in
$[r, s]$ in the second distribution. The general rule
\nameref{rule:aprhl-lapint} can prove post-conditions of any shape.

\begin{theorem} \label{thm:lapint-sound}
  The rule \nameref{rule:aprhl-lapint} is sound.
\end{theorem}
\begin{proof}
  We leave the logical context $\rho$ implicit. Let $V \triangleq \Var \setminus
  \{ x_1, x_2 \}$ be the non-sampled variables; we write $m[V]$ for the
  restriction of a memory $m$ to variables in $V$.  Consider two memories $m_1,
  m_2$ and let the means $v_1 \triangleq \denot{e_1} m_1$ and $v_2 \triangleq
  \denot{e_2} m_2$ satisfy $|v_1 - v_2| \leq k$. By the free variable condition,
  the expressions $p, q, r, s$ are preserved by the command so we will abuse
  notation and treat $p, q, r, s$ as integer constants satisfying the
  pre-condition $\Phi$.  Let the output distributions be
  \[
    \mu_1 \triangleq \denot{ \Rand{x_1}{\Lap{\varepsilon}(e_1)} } m_1
    \quad \text{and} \quad
    \mu_2 \triangleq \denot{ \Rand{x_2}{\Lap{\varepsilon}(e_2)} } m_2 .
  \]
  We construct an approximate coupling of $\mu_1$ and $\mu_2$.  Define the
  intervals
  \[
    \cI_1 \triangleq [p - v_1, q - v_1]
    \quad \text{and} \quad 
    \cI_2 \triangleq [r - v_2, s - v_2] .
  \]
  Since $p + k \leq r$ and $s \leq q - k$ and $|v_1 - v_2| \leq k$, we know $\cI_2
  \subseteq \cI_1$. \Cref{cor:lapint-math} gives
  \[
    \Lap{\varepsilon}
    \alift{ \{ (r_1, r_2) \mid r_1 \in [p - v_1, q - v_1]
      \leftrightarrow r_2 \in [r - v_2, s - v_2] \} }
    {(\ln \alpha, 0)}
    \Lap{\varepsilon}
  \]
  with constants
  \[
    \alpha \triangleq \frac{\exp( \eta \varepsilon )}{1 - \exp(-(s - r + 2)\varepsilon/2)}
    \quad \text{and} \quad
    \eta \triangleq (q - p) - (s - r) .
  \]
  Since $0 < \sigma \leq (s - r + 2)$, we have $\ln \alpha \leq \varepsilon'$ for
  \[
    \varepsilon' \triangleq
    \ln \left( \frac{\exp( \eta \varepsilon )}{1 - \exp(-\sigma\varepsilon/2)} \right) .
  \]
  \Cref{prop:lapint-bound} yields an approximate coupling
  \[
    \Lap{\varepsilon}
    \alift{ \{ (r_1, r_2) \mid r_1 \in [p - v_1, q - v_1]
      \leftrightarrow r_2 \in [r - v_2, s - v_2] \} }
    {(\varepsilon', 0)}
    \Lap{\varepsilon} .
  \]
  Rearranging, this is equivalent to
  \[
    \Lap{\varepsilon}
    \alift{ \{ (r_1, r_2) \mid r_1 + v_1 \in [p, q] \leftrightarrow r_2 + v_2 \in [r, s] \} }
    {(\varepsilon', 0)}
    \Lap{\varepsilon} .
  \]
  Applying \cref{thm:alift-extend} with $f_1, f_2$ mapping $r$ to $r + v_1$, $r
  + v_2$ respectively, we obtain
  \[
    \liftf{f_1}(\Lap{\varepsilon})
    \alift{ \{ (w_1, w_2) \mid w_1 \in [p, q] \leftrightarrow w_2 \in [r, s] \} }
    {(\varepsilon', 0)}
    \liftf{f_2}(\Lap{\varepsilon}) .
  \]
  Now since $\liftf{f_1}(\Lap{\varepsilon}) = \Lap{\varepsilon}(v_1)$ and
  $\liftf{f_2}(\Lap{\varepsilon}) = \Lap{\varepsilon}(v_2)$, we have
  \[
    \Lap{\varepsilon}(v_1)
    \alift{ \{ (w_1, w_2) \mid w_1 \in [p, q] \leftrightarrow w_2 \in [r, s] \} }
    {(\varepsilon', 0)}
    \Lap{\varepsilon}(v_2) .
  \]
  Applying \cref{thm:alift-extend} with maps $\denot{x_1}$ and $\denot{x_2}$, we
  get
  \[
    \mu_1
    \alift{ \denot{x_1\sidel \in [p, q] \leftrightarrow x_2\sider \in [r, s]} }
    {(\varepsilon', 0)}
    \mu_2 .
  \]
  By the free variable condition, $m_1'[V] = m_1[V]$ and $m_2'[V] = m_2[V]$ for
  all memories $m_1' \in \supp(\mu_1)$ and $m_2' \in \supp(\mu_2)$, so we may
  assume by \cref{lem:alift-supp} that the witnesses are supported on such
  memories.  Hence, we have witnesses to
  \[
    \mu_1
    \alift{ \{ (m_1', m_2') \mid m_1'[V] = m_1[V],\; m_2'[V] = m_2[V],\;
    m_1'(x_1) \in [p, q] \leftrightarrow m_2'(x_2) \in [r, s] \} }{(\varepsilon', 0)}
    \mu_2 .
  \]
  By the pre-condition, $(m_1, m_2)$ satisfy
  \[
    \forall w_1, w_2 \in \ZZ,\; 
    w_1 \in [p, q] \leftrightarrow w_2 \in [r, s] \to
    \Psi\subst{x_1\sidel,x_2\sider}{w_1, w_2}
  \]
  and so
  \[
    \mu_1
    \alift{\Psi}{(\varepsilon', 0)}
    \mu_2 ,
  \]
  showing \nameref{rule:aprhl-lapint} is sound.
\end{proof}

\section{Advanced coupling composition} \label{sec:aprhl-ac}

The sequencing rule \nameref{rule:aprhl-seq} in \Saprhl composes two approximate
couplings while summing the approximation parameters; this rule is a
generalization of the standard composition theorem of differential privacy
(\cref{thm:seq-comp}). In this section we extend the \emph{advanced} composition
theorem of differential privacy, \cref{thm:adv-comp}, which allows trading off
the $\varepsilon$ and $\delta$ parameters when analyzing a composition of
private mechanisms.

While the proof of the sequential composition theorem is fairly straightforward,
the advanced composition theorem follows from a more technical argument using
Azuma's inequality. It is not obvious how to extend the proof to approximate
liftings, but fortunately we don't need to. The key observation is that the
$\varepsilon$-distance condition on witnesses ensures differential privacy
generalized to distributions over \emph{pairs} of outputs.  Therefore, we can
directly generalize the advanced composition theorem to liftings by viewing the
function mapping a pair of inputs to the left/right witness as itself
differentially private.

However, there is an important catch: the advanced composition theorem assumes a
\emph{symmetric} adjacency relation. In particular, the witnesses must satisfy a
two-sided, symmetric distance bound to compose, but approximate lifting only
gives a one-sided bound for witnesses.  So, we first introduce a symmetric
version of approximate lifting where the witnesses satisfy the bound in both
directions. Then we develop an advanced composition theorem for symmetric
liftings in two stages. First we prove an advanced composition theorem for
$\varepsilon$-distance, showing how to control the distance between the output
distributions of two compositions if we can bound the symmetric distance between
the output distributions of each step. Then, we give an advanced composition
theorem given a symmetric approximate lifting at each step of a composition.  To
apply this principle in \Saprhl, we introduce a symmetric judgment in \Saprhl
and show how to prove it from standard \Saprhl judgments, and we internalize
advanced composition in a loop rule for symmetric judgments.

\begin{remark}
  The advanced composition theorem from differential privacy implicitly assumes
  that all mechanisms terminate with probability $1$, so in this section we
  assume all commands are lossless; this is not a serious restriction as
  derivable judgments in \Saprhl only relate lossless programs
  (\cref{lem:aprhl-ll}).
\end{remark}

\begin{remark}
  While we focus on the advanced composition theorem, our technique provides a
  simple route to generalize other sequential composition theorems, like the
  optimal composition theorem and the heterogeneous composition
  theorem~\citep{DBLP:journals/corr/OhV13}, and composition theorems where the
  parameters can be selected adaptively~\citep{rogers2016privacy}.
\end{remark}

\subsection{Symmetric approximate liftings}

We first introduce a symmetric version of approximate lifting.
\begin{definition} \label{def:symalift}
  Let $\mu_1, \mu_2$ be sub-distributions over $\cA_1$ and $\cA_2$,
  and let $\cR \subseteq \cA_1 \times \cA_2$ be a relation. Let $\star$ be an
  element disjoint from $\cA_1$ and $\cA_2$. Two sub-distributions $\mu_L,
  \mu_R$ over pairs $\cA_1^\star \times \cA_2^\star$ are \emph{witnesses} for the
  \emph{symmetric $(\varepsilon, \delta)$-approximate $\cR$-lifting} of $(\mu_1,
  \mu_2)$ if:
  \begin{enumerate}
    \item $\pi_1(\mu_L) = \mu_1$ and $\pi_2(\mu_R) = \mu_2$;
    \item $\supp(\mu_L) \cup \supp(\mu_R) \subseteq \cR^\star$; and
    \item $\epsdist{\varepsilon}{\mu_L}{\mu_R} \leq \delta$ and
      $\epsdist{\varepsilon}{\mu_R}{\mu_L} \leq \delta$.
  \end{enumerate}
  (Recall $\cS^\star$ is the set $\cS \cup \{ \star \}$, and $\cR^\star$ is the
  relation $\cR \cup (\cA_1 \times \{ \star \}) \cup (\{ \star \} \times
  \cA_2)$.)
  When the particular witnesses are not important, we say $\mu_1$ and $\mu_2$
  are related by the \emph{symmetric $(\varepsilon, \delta)$-lifting of $\cR$},
  denoted
  \[
    \mu_1 \symalift{\cR}{(\varepsilon, \delta)} \mu_2 .
  \]
  $\cR$ need not be symmetric---in fact, $\cA_1$ and $\cA_2$ may be different
  sets.
\end{definition}
%


This definition is nearly identical to standard approximate liftings
(\cref{def:alift}) except it requires the distance bound in both directions.
The two-sided bound in a symmetric lifting implies two standard approximate
liftings: if $\mu_1 \symalift{\cR}{(\varepsilon, \delta)} \mu_2$ holds, then
$\mu_1 \alift{\cR}{(\varepsilon, \delta)} \mu_2$ and $\mu_2
\alift{(\cR^{-1})}{(\varepsilon, \delta)} \mu_1$ both hold by taking witnesses
$(\mu_L, \mu_R)$ and $(\mu_R^\top, \mu_L^\top)$ respectively, since
$\epsdist{\varepsilon}{\mu_R^\top}{\mu_L^\top} =
\epsdist{\varepsilon}{\mu_R}{\mu_L}$. In general, the converse may not be true.
However when the relation $\cR$ is of a particular form, we can construct a
symmetric approximate lifting by giving two approximate liftings.

\begin{lemma} \label{lem:aprhl-to-sym}
  Suppose $\cS_1, \cS_2$ are subsets of $\cA_1, \cA_2$ respectively, and we have maps
  $f_1 : \cA_1 \to \cB$ and $f_2 : \cA_2 \to \cB$. Define a relation $\cR$ on $\cA_1
  \times \cA_2$ by
  \[
    a_1 \mathrel{\cR} a_2 \iff a_1 \in \cS_1 \land a_2 \in \cS_2 \land f_1(a_1) = f_2(a_2) .
  \]
  Let $\mu_1, \mu_2$ be sub-distributions over $\cA_1$ and $\cA_2$.
  The approximate liftings
  \[
    \mu_1 \alift{\cR}{(\varepsilon, \delta)} \mu_2
    \qquad \text{and} \qquad
    \mu_2 \alift{(\cR^{-1})}{(\varepsilon, \delta)} \mu_1 ,
  \]
  imply the symmetric approximate lifting
  \[
    \mu_1 \symalift{\cR}{(\varepsilon, \delta)} \mu_2 .
  \]
\end{lemma}
\begin{proof}
  Let $(\mu_L, \mu_R)$ witness $\mu_1 \alift{\cR}{(\varepsilon, \delta)} \mu_2$
  and let $(\nu_L, \nu_R)$ witness $\mu_2 \alift{(\cR^{-1})}{(\varepsilon,
  \delta)} \mu_1$. For every $b \in \cB$, define subsets $[b]_{\cA_1} \triangleq
  f_1^{-1}(b) \subseteq \cA_1$ and $[b]_{\cA_2} \triangleq f_2^{-1}(b) \subseteq
  \cA_2$ partitioning $\cA_1$ and $\cA_2$. First, we have
  \begin{align*}
    \mu_1([b]_{\cA_1}) &= \mu_L([b]_{\cA_1} \times \cA_2^\star) \\
    &\leq \exp(\varepsilon) \cdot \mu_R([b]_{\cA_1} \times \cA_2^\star) + \delta \\
    &= \exp(\varepsilon) \cdot \mu_R([b]_{\cA_1} \times [b]_{\cA_2}) + \delta \\
    &\leq \exp(\varepsilon) \cdot \mu_R(\cA_1^\star \times [b]_{\cA_2}) + \delta \\
    &= \exp(\varepsilon) \cdot \mu_2([b]_{\cA_2}) + \delta .
  \end{align*}
  Define non-negative constants:
  \[
    \rho(b) \triangleq \max(\mu_1([b]_{\cA_1}) - \exp(\varepsilon) \cdot \mu_2([b]_{\cA_2}), 0) .
  \]
  Then
  \[
    \mu_1([b]_{\cA_1}) \leq \exp(\varepsilon) \cdot \mu_2([b]_{\cA_2}) + \rho(b) ,
  \]
  with equality if $\rho(b) > 0$. It is not hard to show $\sum_{b \in \cB}
  \rho(b) \leq \delta$; let $\cB' \triangleq \{ b \in\cB \mid \rho(b) > 0 \}$. Then
  \[
    \mu_1(\cup_{b \in \cB'} [b]_{\cA_1}) = \exp(\varepsilon) \cdot \mu_2(\cup_{b \in \cB'} [b]_{\cA_2}) + \sum_{b \in \cB'} \rho(b) ,
  \]
  but \cref{thm:lift-to-sato} bounds the left side:
  \[
    \mu_1(\cup_{b \in \cB'} [b]_{\cA_1}) \leq \exp(\varepsilon) \cdot \mu_2(\cup_{b \in \cB'} [b]_{\cA_2}) + \delta .
  \]

  By a similar calculation with $(\nu_L, \nu_R)$ in place of $(\mu_L, \mu_R)$,
  we have a symmetric bound $\mu_2([b]_{\cA_2}) \leq \exp(\varepsilon) \cdot
  \mu_1([b]_{\cA_1}) + \sigma(b)$ for minimal non-negative constants $\sigma(b)$
  such that $\sum_{b \in \cB} \sigma(b) \leq \delta$. Note that $\rho(b)$ and
  $\sigma(b)$ can't both be strictly positive, by minimality.  We define
  witnesses
  \begin{align*}
    \eta_L(a_1, a_2) &\triangleq
    \begin{cases}
      \frac{\mu_1(a_1) \cdot \mu_2(a_2)}{\mu_2([b]_{\cA_2})}
      \left( 1 - \frac{\rho(b)}{\mu_1([b]_{\cA_1})} \right)
      &: f_1(a_1) = f_2(a_2) = b \\
      \frac{\mu_1(a_1) \cdot \rho(b)}{\mu_1([b]_{\cA_1})} &: a_2 = \star \\
      0 &: \text{otherwise.}
    \end{cases}
    \\
    \eta_R(a_1, a_2) &\triangleq
    \begin{cases}
      \frac{\mu_1(a_1) \cdot \mu_2(a_2)}{\mu_1([b]_{\cA_1})}
      \left( 1 - \frac{\sigma(b)}{\mu_2([b]_{\cA_2})} \right)
      &: f_1(a_1) = f_2(a_2) = b \\
      \frac{\mu_2(a_2) \cdot \sigma(b)}{\mu_2([b]_{\cA_2})} &: a_1 = \star \\
      0 &: \text{otherwise.}
    \end{cases}
  \end{align*}
  Throughout, if a denominator is $0$ we take the fraction to be $0$ as well.
  Since $\supp(\mu_1) \subseteq \cS_1$ and $\supp(\mu_2) \subseteq \cS_2$ by the
  marginal and support conditions of the two asymmetric liftings,
  $\supp(\eta_L)$ and $\supp(\eta_R)$ are contained in $\cR^\star$.
  
  For the first marginal $\pi_1(\eta_L)(a_1)$, if $\mu_1([f_1(a_1)]_{\cA_1})$ is
  zero then $\rho(f_1(a_1)) = 0$ by minimality and $\mu_1(a_1) = 0$, so
  $\eta_L(a_1,a_2) = 0$ for all $a_2 \in \cA_2$. Otherwise if
  $\mu_2([f_1(a_1)]_{\cA_2}) = 0$ then $\rho(f_1(a_1)) =
  \mu_1([f_1(a_1)]_{\cA_1})$ by minimality, and $\eta_L(a_1, a_2) = \mu_1(a_1)$
  for $a_2 = \star$ and zero for $a_2 \in \cA_2$. By a symmetric argument, the
  second marginal is similar.
  
  To check the symmetric distance conditions, take any set $\cW \subseteq
  \cA_1^\star \times \cA_2^\star$. We want to compare
  \[
    \eta_L(\cW)
    = \sum_{(a_1, a_2) \in \cW_0} \eta_L(a_1, a_2)
    + \sum_{(a_1, \star) \in \cW} \eta_L(a_1, \star)
  \]
  with
  \[
    \eta_R(\cW) =
    \sum_{(a_1, a_2) \in \cW_0} \eta_R(a_1, a_2)
    + \sum_{(\star, a_2) \in \cW} \eta_R(\star, a_2) ,
  \]
  where $\cW_0 \triangleq \cW \cap (\cA_1 \times \cA_2)$.  We claim (i) $\eta_L(a_1, a_2)
  \leq \exp(\varepsilon) \cdot \eta_R(a_1, a_2)$ for all $(a_1, a_2) \in \cA_1
  \times \cA_2$, and (ii) $\sum_{(a_1, \star) \in \cW} \eta_L(a_1, \star) \leq
  \delta$.  Without loss of generality, we assume $\cW$ is contained in
  $\cR^\star$.

  To show (i), let $b \triangleq f_1(a_1) = f_2(a_2)$. If either
  $\mu_1([b]_{\cA_1})$ or $\mu_2([b]_{\cA_2})$ are zero then the relevant
  probabilities in $\eta_L$ and $\eta_R$ are zero as well. Otherwise there are
  three cases. If both $\rho(b)$ and $\sigma(b)$ are both zero, then
  \[
    \frac{\eta_L(a_1, a_2)}{\eta_R(a_1, a_2)}
    = \frac{\mu_1([b]_{\cA_1})}{\mu_2([b]_{\cA_2})}
    \leq \exp(\varepsilon) .
  \]
  If $\rho(b) > 0$, then $\sigma(b) = 0$ and $\mu_1([b]_{\cA_1}) > 0$. If
  $\mu_2([b]_{\cA_2}) = 0$ then the claim is immediate; otherwise,
  \[
    \frac{\eta_L(a_1, a_2)}{\eta_R(a_1, a_2)}
    = \frac{\mu_1([b]_{\cA_1})}{\mu_2([b]_{\cA_2})}
      \left( 1 - \frac{\rho(b)}{\mu_1([b]_{\cA_1})} \right)
    = \frac{\mu_1([b]_{\cA_1}) - \rho(b)}{\mu_2([b]_{\cA_2})}
    = \exp(\varepsilon)
  \]
  where the final equality is by minimality of $\rho(b)$.  Similarly if
  $\sigma(b) > 0$, then $\rho(b) = 0$ and $\mu_2([b]_{\cA_2}) > 0$ so
  \[
    \frac{\eta_L(a_1, a_2)}{\eta_R(a_1, a_2)}
    = \frac{\mu_1([b]_{\cA_1})}{\mu_2([b]_{\cA_2})}
      \left( \frac{\mu_2([b]_{\cA_2})}{\mu_2([b]_{\cA_2}) - \sigma(b)} \right)
    = \frac{\mu_1([b]_{\cA_1})}{\mu_2([b]_{\cA_2}) - \sigma(b)}
    = \frac{\mu_1([b]_{\cA_1})}{\exp(\varepsilon) \cdot \mu_1([b]_{\cA_1})}
    \leq \exp(\varepsilon) ,
  \]
  where the final equality is by minimality of $\sigma(b)$; note that if
  $\mu_2([b]_{\cA_2}) = \sigma(b)$, then $\mu_1([b]_{\cA_1}) = 0$, $\eta_L(a_1,
  a_2)$, and $\eta_R(a_1, a_2)$ are all zero. This establishes (i).

  Showing (ii) is more straightforward:
  \[
    \sum_{(a_1, \star) \in \cW} \eta_L(a_1, \star)
    \leq \sum_{a_1 \in \cA_1} \eta_L(a_1, \star)
    = \sum_{b \in \cB} \rho(b)
    \leq \delta .
  \]
  Hence we have
  \begin{align*}
    \eta_L(\cW)
    &= \sum_{(a_1, a_2) \in \cW_0} \eta_L(a_1, a_2)
    + \sum_{(a_1, \star) \in \cW} \eta_L(a_1, \star) \\
    &\leq \exp(\varepsilon) \sum_{(a_1, a_2) \in \cW_0} \eta_R(a_1, a_2) + \delta \\
    &\leq \exp(\varepsilon) \cdot \eta_R(\cW) + \delta ,
  \end{align*}
  giving the distance bound $\epsdist{\varepsilon}{\eta_L}{\eta_R} \leq
  \delta$. A similar calculation yields the symmetric bound
  $\epsdist{\varepsilon}{\eta_R}{\eta_L} \leq \delta$, so $(\eta_L, \eta_R)$
  witness the desired symmetric approximate lifting.
\end{proof}

\subsection{Advanced composition of symmetric $\varepsilon$-distance}

Building up to advanced composition for symmetric approximate liftings, we
first show advanced composition for symmetric $\varepsilon$-distance. Suppose we
have two sequences of $n$ functions $\{ f_i \}_{i \in [n]}, \{ g_i \}_{i \in
[n]}$ where $f_i, g_i : \cA \to \Dist(\cA)$ are such that for any $a \in \cA$,
we can bound the $\varepsilon$-distance between $f_i(a)$ and $g_i(a)$. Then we
will bound the $\varepsilon$-distance between the output distributions from the
$n$-fold compositions.

We use notation for the sequential composition of algorithms.  Given a
sequence of functions $\{ h_i \}_{i \in [k]}$ where $h_i : \cA \to \Dist(\cA)$,
we write $h^k : \cA \to \Dist(\cA)$ for the composition of $\{ h_i \}$.
Formally, we define
\[
  h^k(a) \triangleq 
  \begin{cases}
    \dunit(a)                 &: k = 0 \\
    \dbind(h^{k - 1}(a), h_k) &: k > 0 .
  \end{cases}
\]
(Recall $\dunit : \cA \to \Dist(\cA)$ and $\dbind : \Dist(\cA) \times (\cA
\to \Dist(\cB)) \to \Dist(\cB)$ are the monadic operations for distributions
from \cref{def:unit-bind}.) We use the same notation for functions of type $h_i
: \cD \times \cA \to \Dist(\cA)$, defining $h^k : \cD \times \cA \to \Dist(\cA)$
as
\[
  h^k(d, a) \triangleq 
  \begin{cases}
    \dunit(a)                          &: k = 0 \\
    \dbind(h^{k - 1}(d, a), h_k(d, -)) &: k > 0 .
  \end{cases}
\]

\begin{proposition} \label{prop:advcompdist}
Let $f_i,g_i: \cA \to \Dist(\cA)$ satisfy $\epsdist{\varepsilon}{f_i(a)}{g_i(a)}
\leq \delta$ and $\epsdist{\varepsilon}{g_i(a)}{f_i(a)} \leq \delta$ for every
$i \in [n]$ and $a \in \cA$. For any $\omega \in (0, 1)$, let
\[
  \varepsilon^* \triangleq \varepsilon \sqrt{2 n \ln(1/\omega)} + n \varepsilon(e^\varepsilon - 1)
  \quad \text{and} \quad
  \delta^* \triangleq n \delta + \omega .
\]
Then for every $n \in \NN$ and $a \in \cA$, we have
$\epsdist{\varepsilon^*}{f^n(a)}{g^n(a)} \leq \delta^*$ and
$\epsdist{\varepsilon^*}{g^n(a)}{f^n(a)} \leq \delta^*$.
\end{proposition}

\begin{proof}
  Let $\BB$ be the booleans and define $h_i: \BB \times \cA \to \Dist(\cA)$ as
  \[
    h_i(\kwtrue, a) \triangleq f_i(a)
    \quad \text{and} \quad
    h_i(\kwfalse, a) \triangleq g_i(a)
  \]
  for every $a\in \cA$. Then $\epsdist{\varepsilon}{f_i(a)}{g_i(a)} \leq \delta$
  and $\epsdist{\varepsilon}{g_i(a)}{f_i(a)} \leq \delta$ imply $h_i(a, -) : \BB
  \to \Dist(\cA)$ is $(\varepsilon, \delta)$-differentially private for every $a
  \in \cA$, where we view $\BB$ as the set of databases with the full adjacency
  relation relating all pairs of booleans; in particular, this is a symmetric
  relation. Applying the advanced composition theorem of differential privacy
  (\cref{thm:adv-comp}), $h^n(-, a) : \BB \to \Dist(\cA)$ is
  $(\varepsilon^*,\delta^*)$-differentially private for every $a \in \cA$.  By
  \cref{def:epsdist-privacy} we have
  \[
    \epsdist{\varepsilon^*}{h^n(\kwtrue, a)}{h^n(\kwfalse, a)} \leq \delta^*
    \qquad \text{and} \qquad
    \epsdist{\varepsilon^*}{h^n(\kwfalse, a)}{h^n(\kwtrue, a)} \leq \delta^*
  \]
  for every $a \in \cA$. Since $h^n(\kwtrue, a) = f^n(a)$ and $h^n(\kwfalse, a)
  = g^n(a)$ by definition, we conclude
  \[
    \epsdist{\varepsilon^*}{f^n(a)}{g^n(a)} \leq \delta^*
    \qquad \text{and} \qquad
    \epsdist{\varepsilon^*}{g^n(a)}{f^n(a)} \leq \delta^* .
    \qedhere
  \]
\end{proof}

\subsection{Advanced composition of symmetric approximate liftings}

Next, we extend \cref{prop:advcompdist} to symmetric approximate liftings;
roughly speaking, we will apply the proposition to the functions mapping related
inputs to the left or right witness distributions. We need a lemma about how
witnesses are transformed under composition.

\begin{lemma} \label{lem:witness-comp}
  Consider two sequences of functions $\{ f_i \}_{i \in [n]}, \{ g_i \}_{i \in
  [n]}$ with $f_i : \cA_1 \to \Dist(\cA_1)$ and $g_i : \cA_2 \to \Dist(\cA_2)$,
  and a sequence of binary relations $\{ \Phi_i \}_{i \in \{0, \dots, n\}}$ on
  $\cA_1 \times \cA_2$.

  Suppose we have two sequences of functions $\{ l_i \}_{i \in [n]}, \{ r_i
  \}_{i \in [n]}$ with $l_i, r_i : \cA_1^\star \times \cA_2^\star \to
  \Dist(\cA_1^\star \times \cA_2^\star)$ producing witnesses to an approximate
  lifting of $\Phi_i$:
  \begin{enumerate}
    \item $\pi_1 (l_i(a_1, a_2)) = f_i(a_1)$ and $\pi_2 (r_i(a_1, a_2)) = g_i(a_2)$ for $(a_1,
      a_2) \in \Phi_{i - 1}$;
    \item $\pi_1 (l_i(a_1, \star)) = f_i(a_1)$ and $\pi_2 (r_i(\star, a_2)) = g_i(a_2)$; and
    \item $\supp(l_i(a_1, a_2)) \cup \supp(r_i(a_1, a_2)) \subseteq \Phi_i^\star$ for
      $(a_1, a_2) \in \Phi_{i - 1}^\star$
  \end{enumerate}
  for every $i \in [n]$. Then $l^n$ and $r^n$ generate witnesses for an
  approximate lifting relating the $n$-fold compositions:
  \begin{enumerate}
    \item $\pi_1 (l^n(a_1, a_2)) = f^n(a_1)$ and $\pi_2 (r^n(a_1, a_2)) = g^n(a_2)$ for $(a_1,
      a_2) \in \Phi_0$;
    \item $\pi_1 (l^n(a_1, \star)) = f^n(a_1)$ and $\pi_2 (r^n(\star, a_2)) = g^n(a_2)$; and
    \item $\supp(l^n(a_1, a_2)) \cup \supp(r^n(a_1, a_2)) \subseteq \Phi_n^\star$
      for every $(a_1, a_2) \in \Phi_0^\star$.
  \end{enumerate}
\end{lemma}
\begin{proof}
  By induction on $n$. The base case $n = 0$ is trivial. When $n > 0$, the
  support condition follows by induction; the marginal conditions follow by a
  direct computation (\cref{lem:proj-bind}).
\end{proof}

We are now ready to prove advanced composition for symmetric liftings.

\begin{theorem}\label{thm:advcomplift}
  Let $\omega \in (0, 1)$. Consider two sequences of functions $\{ f_i \}_{i
  \in [n]}$ and $\{ g_i \}_{i \in [n]}$ with $f_i : \cA_1 \to \Dist(\cA_1)$ and
  $g_i : \cA_2 \to \Dist(\cA_2)$, and a sequence of binary relations $\{ \Phi_i
  \}_{i \in [n]}$ on $\cA_1 \times \cA_2$ and $\Phi_0 \subseteq \cA_1 \times
  \cA_2$.  Suppose for every $i \in [n]$ and $(a_1, a_2) \in \Phi_{i - 1}$,
  there is a symmetric approximate lifting:
  \[
    f_i(a_1) \symalift{\Phi_i}{(\varepsilon, \delta)} g_i(a_2) .
  \]
  Then for every $(a_1,a_2) \in \Phi_0$, we have a symmetric lifting
  \[
    f^n(a_1) \symalift{\Phi_n}{(\varepsilon^*, \delta^*)} g^n(a_2)
  \]
  where $\varepsilon^* \triangleq \varepsilon \sqrt{2 n \ln(1/\omega)} + n
  \varepsilon(e^\varepsilon - 1)$ and $\delta^* \triangleq n \delta + \omega$.
\end{theorem}
\begin{proof}
  For $(a_1, a_2) \in \Phi_{i - 1}$, let $(\mu_L^{(i)}(a_1, a_2),
  \mu_R^{(i)}(a_1, a_2))$ witness the approximate lifting of $\Phi_i$ relating
  $f_i(a_1)$ and $g_i(a_2)$.  Define functions $\{ l_i \}_{i \in [n]}, \{ r_i
  \}_{i \in [n]}$ of type $l_i, r_i : \cA_1^\star \times \cA_2^\star \to
  \Dist(\cA_1^\star \times \cA_2^\star)$ as follows:
  \begin{align*}
    l_i(a_1, a_2) &\triangleq
    \begin{cases}
      \mu_L^{(i)}(a_1, a_2) &: (a_1, a_2) \in \Phi_{i - 1} \\
      \dunit(\star) \times g_i(a_2) &: a_1 = \star, a_2 \neq \star \\
      f_i(a_1) \times \dunit(\star) &: a_1 \neq \star, a_2 = \star \\
      \dunit(\star, \star) &: a_1 = a_2 = \star
    \end{cases}
    \\
    r_i(a_1, a_2) &\triangleq
    \begin{cases}
      \mu_R^{(i)}(a_1, a_2) &: (a_1, a_2) \in \Phi_{i - 1} \\
      \dunit(\star) \times g_i(a_2) &: a_1 = \star, a_2 \neq \star \\
      f_i(a_1) \times \dunit(\star) &: a_1 \neq \star, a_2 = \star \\
      \dunit(\star, \star) &: a_1 = a_2 = \star
    \end{cases}
  \end{align*}
  Given distributions $\eta_1$ and $\eta_2$ over $\cB_1$ and $\cB_2$
  respectively, $\eta_1 \times \eta_2 \in \Dist(\cB_1 \times \cB_2)$ denotes the
  product distribution defined in the expected way:
  \[
    (\eta_1 \times \eta_2) (b_1, b_2) \triangleq \eta_1(b_1) \cdot \eta_2(b_2) .
  \]
  Now by assumption on $(\mu_L^{(i)}(a_1, a_2), \mu_R^{(i)}(a_1, a_2))$ and by
  definition when $a_1 = \star$ or $a_2 = \star$, we have
  \[
    \epsdist{\varepsilon}{l_i(a_1, a_2)}{r_i(a_1, a_2)} \leq \delta
    \qquad \text{and} \qquad
    \epsdist{\varepsilon}{r_i(a_1, a_2)}{l_i(a_1, a_2)} \leq \delta
  \]
  for all $(a_1, a_2) \in \Phi_{i - 1}^\star$, and we have the marginal conditions
  required by \cref{prop:advcompdist}. Now take any $(a_1, a_2) \in \Phi_0$. By
  \cref{prop:advcompdist}, we have
  \[
    \epsdist{\varepsilon^*}{l^n(a_1, a_2)}{r^n(a_1, a_2)} \leq \delta^*
    \qquad \text{and} \qquad
    \epsdist{\varepsilon^*}{r^n(a_1, a_2)}{l^n(a_1, a_2)} \leq \delta^*.
  \]
  \Cref{lem:witness-comp} gives the marginal conditions $\pi_1(l^n(a_1, a_2)) =
  f^n(a_1)$ and $\pi_2(r^n(a_1, a_2)) = g^n(a_2)$ and shows that $\supp(l^n(a_1,
  a_2)), \supp(r^n(a_1, a_2))$ are contained in $\Phi_n^\star$, so $l^n(a_1,
  a_2)$ and $r^n(a_1, a_2)$ witness the desired symmetric approximate lifting
  \[
    f^n(a_1) \symalift{\Phi_n}{(\varepsilon^*, \delta^*)} g^n(a_2) .
    \qedhere
  \]
\end{proof}

\subsection{Symmetric judgments in \Saprhl}

In order to use advanced composition in \Saprhl, we extend the logic with a new
judgment modeling symmetric approximate liftings.  We call such judgments
\emph{symmetric} judgments.
\begin{definition} \label{def:symaprhl-valid}
  A symmetric \Saprhl judgment is \emph{valid} in logical context $\rho$, written
  \[
    \rho \models \symaprhl{c_1}{c_2}{\Phi}{\Psi}{(\varepsilon, \delta)} ,
  \]
  if for any two inputs $(m_1, m_2) \in \denot{\Phi}_\rho$ there exists an
  symmetric approximate lifting relating the outputs:
  \[
    \denot{c_1}_\rho m_1
    \symalift{\denot{\Psi}_\rho}{(\denot{\varepsilon}_\rho, \denot{\delta}_\rho)}
    \denot{c_2}_\rho m_2 .
  \]
\end{definition}
To prove these judgments, we extend \Saprhl with a few proof rules. To keep our
proof system as simple as possible, we introduce rules for symmetric judgments
only where absolutely needed---namely, for advanced composition---and use the
conversion rules in \cref{fig:symaprhl-conv} to move between symmetric and
standard, asymmetric judgments. The inverse relation $\Phi^{-1}$ can be defined
syntactically by simply interchanging the tags $\sidel$ and $\sider$ in a
formula $\Phi$.
\begin{figure}
  \begin{mathpar}
    \inferruleref{SymIntro}
    { \Psi \triangleq e_1\sidel = e_2\sider \land \Psi_1\sidel \land \Psi_2\sider
      \\\\
      \vdash \aprhl{c_1}{c_2}{\Phi}{\Psi}{(\varepsilon, \delta)} \\
      \vdash \aprhl{c_2}{c_1}{\Phi^{-1}}{\Psi^{-1}}{(\varepsilon, \delta)} }
    { \vdash \symaprhl{c_1}{c_2}{\Phi}{\Psi}{(\varepsilon, \delta)} }
    \label{rule:aprhl-symintro}
    \\
    \inferruleref{SymElim-L}
    { \vdash \symaprhl{c_1}{c_2}{\Phi}{\Psi}{(\varepsilon, \delta)} }
    { \vdash \aprhl{c_1}{c_2}{\Phi}{\Psi}{(\varepsilon, \delta)} }
    \label{rule:aprhl-symelim-l}
    \\
    \inferruleref{SymElim-R}
    { \vdash \symaprhl{c_1}{c_2}{\Phi}{\Psi}{(\varepsilon, \delta)} }
    { \vdash \aprhl{c_2}{c_1}{\Phi^{-1}}{\Psi^{-1}}{(\varepsilon, \delta)} }
    \label{rule:aprhl-symelim-r}
  \end{mathpar}
  \caption{Conversion rules between symmetric and standard judgments for \Saprhl}
  \label{fig:symaprhl-conv}
\end{figure}
Soundness of these rules is straightforward.
\begin{theorem} \label{thm:symaprhl-sound}
  The rules \nameref{rule:aprhl-symintro}, \nameref{rule:aprhl-symelim-l}, and
  \nameref{rule:aprhl-symelim-r} are sound.
\end{theorem}
\begin{proof}
  Soundness of \nameref{rule:aprhl-symintro} follows by \cref{lem:aprhl-to-sym}. Soundness of
  \nameref{rule:aprhl-symelim-l} and
  \ifpenn \\ \fi 
  \nameref{rule:aprhl-symelim-r} follow by definition of symmetric approximate lifting.
\end{proof}

\subsection{An advanced composition rule for \Saprhl}

\begin{figure}
  \begin{mathpar}
    \inferruleref{While-AC}
    { \varepsilon^* \triangleq
        \varepsilon \sqrt{2 N \ln(1/\omega)} + N \varepsilon(e^\varepsilon - 1) \\
      \delta^* \triangleq N \delta + \omega \\
      \omega \in (0, 1)
      \\\\
      \models \Phi \to e_v\sidel \leq 0 \to \neg e_1\sidel \\
      \models \Phi \to e_1\sidel = e_2\sider
      \\\\
      \forall K \in \NN,\; \vdash \symaprhl{c_1}{c_2}
      {\Phi \land e_1\sidel \land e_v\sidel = K}
      {\Phi \land e_v\sidel < K}
      {(\varepsilon,\delta)}
    }
    { \vdash \symaprhl
      {\WWhile{e_1}{c_1}}{\WWhile{e_2}{c_2}}
      {\Phi \land e_v\sidel \leq N}
      {\Phi \land \neg e_1\sidel}
      {(\varepsilon^*,\delta^*)} }
    \label{rule:aprhl-while-ac}
  \end{mathpar}
  \caption{Advanced composition rule \nameref{rule:aprhl-while-ac} for \Saprhl}
  \label{fig:aprhl-ac}
\end{figure}

Finally, we internalize advanced composition of liftings as the loop rule
\nameref{rule:aprhl-while-ac} in \cref{fig:aprhl-ac}. Like the usual rule
\nameref{rule:aprhl-while}, the guards must be synchronized and the loops run at
most $N$ iterations. An $(\varepsilon, \delta)$-approximate coupling of the loop
bodies gives an $(\varepsilon^*, \delta^*)$-approximate coupling of the two
loops, where $\varepsilon^*$ and $\delta^*$ are from the advanced composition
theorem of differential privacy (\cref{thm:adv-comp}).

\begin{theorem} \label{thm:while-ac-sound}
  The rule \nameref{rule:aprhl-while-ac} is sound.
\end{theorem}
\begin{proof}
  The proof follows essentially by \cref{thm:advcomplift}.  As usual, we will
  leave the logical context $\rho$ implicit. Consider two memories $(m_1, m_2)
  \in \denot{\Phi \land e_v\sidel \leq N}$ and two output distributions
  \[
    \mu_1 \triangleq \denot{\WWhile{e_1}{c_1}} m_1
    \quad \text{and} \quad
    \mu_2 \triangleq \denot{\WWhile{e_2}{c_2}} m_2 .
  \]
  We construct a symmetric approximate lifting relating $\mu_1$ and $\mu_2$. The
  value of $N$ is given by the logical context $\rho$; we treat it as a
  constant. We unroll the loop $N$ times and define
  \[
    \mu_1' \triangleq \denot{(\Condt{e_1}{c_1})^N} m_1
    \quad \text{and} \quad
    \mu_2' \triangleq \denot{(\Condt{e_2}{c_2})^N} m_2 .
  \]
  We claim $\denot{e_1} m_1' = \denot{e_2} m_2' = \kwfalse$ for all $m_1' \in
  \supp(\mu_1')$ and $m_2' \in \supp(\mu_2')$. We can use the valid symmetric
  \Saprhl judgment in the premise and symmetric versions of the rules
  \nameref{rule:aprhl-seq} and \nameref{rule:aprhl-cond} to construct a
  symmetric approximate lifting
  \[
    \mu_1'
    \symalift{ \Phi \land e_v\sidel \leq 0 }
    {(N \varepsilon, N \delta)}
    \mu_2' .
  \]
  Since $\models \Phi \land e_v\sidel \leq 0 \to \neg e_1\sidel$, we have
  \[
    \mu_1'
    \symalift{ \neg e_1\sidel \land \neg e_2\sider }
    {(N \varepsilon, N \delta)}
    \mu_2' .
  \]
  Let $\mu_L', \mu_R'$ be the corresponding witnesses.  We know
  $\pi_1(\mu_L') = \mu_1'$ and $\pi_2(\mu_R') = \mu_2'$, and also
  \[
    \supp(\mu_L') \cup \supp(\mu_R')
    \subseteq \denot{ \neg e_1\sidel \land \neg e_2\sider } ,
  \]
  so $\denot{e_1} m_1' = \denot{e_2} m_2' = \kwfalse$ for all $m_1',
  m_2'$ in the support of $\mu_1', \mu_2'$ respectively.  By the equivalences
  \begin{align*}
    \WWhile{e_1}{c_1} &\equiv (\Condt{e_1}{c_1})^N; \WWhile{e_1}{c_1} \\
    \WWhile{e_2}{c_2} &\equiv (\Condt{e_2}{c_2})^N; \WWhile{e_2}{c_2} ,
  \end{align*}
  we know
  \[
    \mu_1 = \denot{(\Condt{e_1}{c_1})^N} m_1
    \quad \text{and} \quad
    \mu_2 = \denot{(\Condt{e_2}{c_2})^N} m_2 .
  \]
  Defining a family of relations
  \[
    \Phi_i \triangleq \Phi \land (e_v\sidel \leq N - i \lor \neg e_1\sidel) ,
  \]
  we have
  \[
    \models \symaprhl{\Condt{e_1}{c_1}}{\Condt{e_2}{c_2}}
    {\Phi_i}{\Phi_{i + 1}}{(\varepsilon, \delta)}
  \]
  for every $i$ using the premise, since $\Phi_i$ ensures the guards
  $e_1$ and $e_2$ are equal in the initial memories. By validity, for any pair
  of memories satisfying $\Phi_i$ there is a symmetric approximate lifting of
  $\Phi_{i + 1}$ relating the two output distributions. We can apply
  \cref{thm:advcomplift} with $\cA_1 = \cA_2 = \Mem$, functions $f_i \triangleq
  \denot{\Condt{e_1}{c_1}}$ and $g_i \triangleq \denot{\Condt{e_2}{c_2}}$, and
  relations $\Phi_i$ to get the symmetric approximate lifting
  \[
    \mu_1
    \symalift{\Phi \land (e_v\sidel \leq 0 \lor \neg e_1\sidel)}
    {(\varepsilon^*, \delta^*)}
    \mu_2 .
  \]
  Since $\models \Phi \land e_v\sidel \leq 0 \to \neg e_1\sidel$, we conclude
  \[
    \mu_1
    \symalift{\Phi \land \neg e_1\sidel}
    {(\varepsilon^*, \delta^*)}
    \mu_2
  \]
  so \nameref{rule:aprhl-while-ac} is sound.
\end{proof}

\begin{remark}
  Our approach narrowly limits the scope of symmetric judgments: they can be
  used in \nameref{rule:aprhl-while-ac} or eliminated to a standard judgment.
  There are at least two other choices. One option would be to define a full
  proof system based on symmetric judgments. Almost all the basic proof rules
  from \Saprhl would directly generalize, including the standard rules for
  program commands and the Laplace rules. However, it is not clear how to
  generalize the more advanced rules, including \nameref{rule:aprhl-pw-eq} and
  \nameref{rule:aprhl-utb-l}/\nameref{rule:aprhl-utb-r}. The optimal subset
  coupling (\cref{thm:opt-subset-math}) also does not directly generalize to
  symmetric liftings; this poses a problem for a symmetric version of
  \nameref{rule:aprhl-lapint}.

  For another option, we could avoid symmetric judgments entirely by fusing
  \nameref{rule:aprhl-symintro}, \nameref{rule:aprhl-while-ac}, and
  \nameref{rule:aprhl-symelim-l} together into a single rule. While this would
  suffice for our examples, it is conceptually clearer to separate symmetric and
  asymmetric judgments. Our design choice leaves room for other rules specific
  to symmetric approximate liftings, and clearly identifies the main bottleneck
  in converting from standard approximate liftings to symmetric liftings in the
  rule \nameref{rule:aprhl-symintro}.
\end{remark}

\section{Proving privacy for Between Thresholds} \label{sec:aprhl-bt}

To draw everything together, we prove differential privacy for the \emph{Between
Thresholds} mechanism proposed by \citet{BunSU16}, a more advanced version of
the Sparse Vector mechanism. The input is again a stream of numeric queries, but
now there are \emph{two} numeric thresholds $A$ and $B$ with $A < B$. The
original mechanism outputs \textsc{Left} if the answer is approximately below
$A$, \textsc{Right} if the answer is approximately above $B$, and \textsc{Halt}
if the answer is approximately between $A$ and $B$.

\begin{figure}
\[
  \begin{array}{l}
    \Ass{i}{1}; \\
    \Ass{\mathit{out}}{[]}; \\
    \WWhile{i \leq N \land |\mathit{out}| < C}{} \\
    \quad \Rand{u}{\Lap{\varepsilon'}(0)}; \\
    \quad \Ass{a}{A - u}; \\
    \quad \Ass{b}{B + u}; \\

    \quad \Ass{\mathit{go}}{\kwtrue}; \\
    \quad \Ass{\mathit{ans}}{(0, 0)}; \\

    \quad \WWhile{i \leq N \land \mathit{go}}{} \\
    \quad\quad \Rand{v}{\Lap{\varepsilon'/3}(\evalQ(i, d))}; \\
    \quad\quad \Condt{a < v < b}{} \\
    \quad\quad\quad \Rand{\mathit{noisy}}{\Lap{\varepsilon'}(\evalQ(i, d))}; \\
    \quad\quad\quad \Ass{\mathit{ans}}{(i, \mathit{noisy})}; \\
    \quad\quad\quad \Ass{\mathit{out}}{\mathit{ans} :: \mathit{out}}; \\
    \quad\quad\quad \Ass{\mathit{go}}{\kwfalse}; \\
    \quad\quad \Ass{i}{i + 1}
  \end{array}
\]
\caption{Between Thresholds}
\label{fig:code-bt}
\end{figure}

We analyze a variant of Between Thresholds that releases the index and
approximate answer of the first $C$ queries between the thresholds;
\cref{fig:code-bt} presents the code of the algorithm.  The variables $a$ and
$b$ contain the noisy thresholds. Unlike Sparse Vector, we resample the noise
$u$ when computing $a$ and $b$ after each between-threshold query---this is
needed to analyze the outer loop by advanced composition.  Also, the noise $u$
is added in \emph{opposite} directions to the two thresholds.  Otherwise, the
code is largely the same as Sparse Vector.

The privacy analysis of this algorithm is more complex than for Sparse Vector.
First, privacy fails if the noisy thresholds $a$ and $b$ are too close together.
Even if the exact thresholds $A$ and $B$ are far apart, there is always some
small, non-zero probability that the noise $u$ may be very large.  Therefore the
best we can hope for is $(\varepsilon, \delta)$-differential privacy, where
$\delta$ bounds the probability that the threshold noise is too large. Second,
while the proof strategy for the inner loop remains broadly the same, in the
critical iteration we must ensure that if one execution is between thresholds,
then so is the other; we use the subset coupling for this purpose.  Finally, we
apply the advanced composition theorem to analyze the outer loop.

It will be useful to have a simpler bound on the approximation parameter for the
subset coupling.

\begin{lemma} \label{lem:lapint-simple}
  Let $\lambda \in (0, 1/2)$.  Suppose we have $r, s \in \ZZ$ such that
  \[
    s - r \geq \frac{6}{\lambda} \ln \frac{4}{\lambda} - 2 ,
  \]
  and suppose we have two means $v_1, v_2 \in \ZZ$ with $|v_1 - v_2| \leq 1$.
  Then we have an approximate lifting
  \[
    \Lap{\lambda/3}(v_1)
    \alift{ \{ (x_1, x_2) \mid
      x_1 \in [r - 1, s + 1] \leftrightarrow x_2 \in [r, s] \} }
    {(\lambda, 0)}
    \Lap{\lambda/3}(v_2) .
  \]
\end{lemma}
\begin{proof}
  By the soundness of \nameref{rule:aprhl-lapint} (\cref{thm:lapint-sound}), we have an
  approximate lifting
  \[
    \Lap{\lambda/3}(v_1)
    \alift{ \{ (x_1, x_2) \mid
      x_1 \in [r - 1, s + 1] \leftrightarrow x_2 \in [r, s] \} }
    {(\kappa, 0)}
    \Lap{\lambda/3}(v_2)
  \]
  where
  \[
    \kappa \triangleq
    \ln \left( \frac{\exp( 2\lambda/3 )}{1 - \exp(-\sigma\lambda/6)} \right)
    \quad \text{and} \quad
    \sigma \triangleq (s - r) + 2 .
  \]
  We check $\kappa \leq \lambda$ assuming $\sigma \geq \frac{6}{\lambda}
  \ln \frac{4}{\lambda}$. Substituting, it suffices to show
  \[
    \frac{\exp(2 \lambda/3)}{1 - \lambda/4} \leq \exp(\lambda)
  \]
  which is equivalent to
  \[
    \lambda/4 + \exp(-\lambda/3) - 1 \leq 0 .
  \]
  Since the left side is convex in $\lambda$, the maximum occurs on the boundary
  of the domain. We can directly check the inequality at the endpoints $\lambda
  = \{ 0, 1/2 \}$.
\end{proof}

We are now ready to prove privacy for Between Thresholds.  As we did for Sparse
Vector, we start with an informal proof by approximate coupling. 

\begin{theorem} \label{thm:bt-priv-coupling}
  Let $\varepsilon, \delta \in (0, 1)$ and let $q_1, \dots, q_N : \cD \to \ZZ$ be a list of
  $1$-sensitive queries. If we set
  \[
    \varepsilon' \triangleq \frac{\varepsilon}{6 \sqrt{2 C \ln(2/\delta)}}
  \]
  and the thresholds $A, B$ are equal across both runs and satisfy
  \[
    B - A \geq \frac{6}{\varepsilon'} \ln(4/\varepsilon') 
          + \frac{2}{\varepsilon'} \ln (2/\delta C) ,
  \]
  then the Between Thresholds algorithm (\cref{fig:code-bt}) is $(\varepsilon,
  \delta)$-differentially private.
\end{theorem}
\begin{proof}[Proof by approximate coupling]
  Consider the outer loop body.  We have $|u\sidel| \leq (1/\varepsilon') \ln
  (2/\delta C)$ in the first process except with probability $\delta/2C$, and we
  couple $u\sidel$ and $u\sider$ so $u\sidel - 1 = u\sider$; this is an
  $(\varepsilon', \delta/2C)$-approximate coupling since the noise is drawn from
  $\Lap{\varepsilon'}(0)$. The coupling ensures the noisy thresholds satisfy
  \begin{equation} \label{eq:bt-thresh-couple}
    a\sidel + 1 = a\sider
    \quad \text{and} \quad
    b\sidel = b\sider + 1 .
  \end{equation}
  Next, consider the inner loop. Each iteration, we approximately couple the
  processes so $\mathit{ans}\sidel = \mathit{ans}\sider$. For any pair $(j, y)$
  with $j \in \NN$ and $y \in \ZZ$, we construct an approximate coupling of the
  inner loops such that if $\mathit{ans}$ on the first side is equal to $(j, y)$,
  then so is $\mathit{ans}$ on the second side; by pointwise equality, this will
  imply an approximate coupling with $\mathit{ans}\sidel = \mathit{ans}\sider$.

  As before, if $j \notin [1, N]$ the proof is trivial. Otherwise, we handle the
  inner iterations in one of two ways. On iterations $i \neq j$ we couple the
  samplings for $v$ and $\mathit{noisy}$ with the null coupling, ensuring
  $|v\sidel - v\sider| \leq 1$. This guarantees that before iteration $j$, if
  the first side is outside the thresholds, then so is the second side (by the
  coupling of the thresholds, \cref{eq:bt-thresh-couple}).  We use $(0,
  0)$-approximate couplings for these iterations.

  On the critical iteration $i = j$, we use the optimal subset coupling when
  sampling $v$ so that
  \begin{equation} \label{eq:subset-couple}
    v\sidel \in [a\sidel, b\sidel]
    \leftrightarrow v\sider \in [a\sider, b\sider] .
  \end{equation}
  Given our accuracy bound on $|u\sidel|$, the inner interval
  $[a\sider, b\sider]$ satisfies
  \[
    b\sider - a\sider \geq \frac{6}{\varepsilon'} \ln(4/\varepsilon') - 2
  \]
  under the threshold coupling, so \cref{eq:subset-couple} is an $(\varepsilon',
  0)$-approximate coupling (\cref{lem:lapint-simple}).  This coupling ensures
  the two processes behave the same at the conditional. If both processes are
  between thresholds, we apply the standard coupling for the Laplace mechanism
  so $\mathit{noisy}\sidel = \mathit{noisy}\sider$; this is an $(\varepsilon',
  0)$-approximate coupling.  If both processes are not between thresholds then
  we don't sample $\mathit{noisy}$.  So, we have an $(2\varepsilon',
  0)$-approximate coupling for the inner loop such that if $\mathit{ans}$ is
  equal to $(j, y)$ on the first run, then $\mathit{ans}$ is equal to $(j, y)$
  on the second run. By pointwise equality, this implies an $(2\varepsilon',
  0)$-approximate coupling for the inner loop with $\mathit{ans}\sidel =
  \mathit{ans}\sider$ as long as the threshold noises satisfy $u\sidel - 1 =
  u\sider$ and the accuracy bound.
  
  Combined with the initial $(\varepsilon', \delta/2C)$-approximate coupling for
  $u$, we have an $(2\varepsilon' + \varepsilon', \delta/2C + 0) =
  (3\varepsilon', \delta/2C)$-approximate coupling  ensuring $\mathit{ans}\sidel
  = \mathit{ans}\sider$ for the body of the outer loop.  The outer loop executes
  at most $C$ iterations, so by the advanced composition theorem (using the
  parameter setting from \cref{fn:ac-setting}) we have an $(\varepsilon,
  \delta)$-approximate coupling of the outer loops with $\mathit{out}\sidel =
  \mathit{out}\sider$, establishing $(\varepsilon, \delta)$-differential
  privacy.
\end{proof}

We can give a more formal proof of privacy in \Saprhl. We work with the
following, equivalent version of Between Thresholds:
\[
  \begin{array}{l}
    \Ass{i}{1}; \\
    \Ass{\mathit{out}}{[]}; \\
    \WWhile{i \leq N \land |\mathit{out}| < C}{} \\
    \quad \Rand{u}{\Lap{\varepsilon'}(0)}; \\
    \quad \Ass{a}{A - u}; \\
    \quad \Ass{b}{B + u}; \\
    \quad \Ass{\mathit{go}}{\kwtrue}; \\
    \quad \Ass{\mathit{ans}}{(0, 0)}; \\
    \quad \WWhile{i \leq N}{} \\
    \quad\quad \Rand{v}{\Lap{\varepsilon'/3}(\evalQ(i, d))}; \\
    \quad\quad \Condt{a < v < b \land \mathit{go}}{} \\
    \quad\quad\quad \Rand{\mathit{noisy}}{\Lap{\varepsilon'}(\evalQ(i, d))}; \\
    \quad\quad\quad \Ass{\mathit{ans}}{(i, \mathit{noisy})}; \\
    \quad\quad\quad \Ass{\mathit{go}}{\kwfalse}; \\
    \quad\quad \Ass{i}{i + 1}; \\
    \quad \Condt{p_1(\mathit{ans}) \neq 0} \\
    \quad\quad \Ass{i}{p_1(\mathit{ans}) + 1}; \\
    \quad\quad \Ass{\mathit{out}}{\mathit{ans} :: \mathit{out}}
  \end{array}
\]
We call this program $\mathit{BT}$ and the inner loop $\mathit{in}$.
Compared to the algorithm in \cref{fig:code-bt}, the main difference is in the
inner loop: each execution of $\mathit{in}$ runs through all the queries,
skipping the check once we have found a between-threshold query. More precisely,
the flag $\mathit{go}$, which indicates we have not yet found a
between-threshold query, is in the inner loop guard in \cref{fig:code-bt} while
it is in the between thresholds check in $\mathit{BT}$. After the inner
loop, if a between-thresholds query was found then the index in $\mathit{ans}$
must be non-zero, so the algorithm records the noisy answer and index, and
resets the counter $i$ to pick up after the last answered query.  The inner
loops in this version of the algorithm can be analyzed synchronously.

\todo{Use $go$ as guard in last conditional?}

\begin{theorem} \label{thm:bt-priv-aprhl}
  Let $\varepsilon, \delta \in (0, 1)$, let $q_1, \dots, q_N : \cD \to \ZZ$ be a list of
  $1$-sensitive queries, and let the logical variables $D_1, D_2$ represent
  two adjacent databases. If we set
  \[
    \varepsilon' \triangleq \frac{\varepsilon}{6 \sqrt{2 C \ln(2/\delta)}}
  \]
  in $\mathit{BT}$, and the thresholds $A, B$ are equal across both runs and
  satisfy
  \[
    B - A \geq \frac{6}{\varepsilon'} \ln(4/\varepsilon') 
          + \frac{2}{\varepsilon'} \ln (2/\delta C) ,
  \]
  then the following judgment holds:
  \[
    \vdash \aprhl{\mathit{BT}}{\mathit{BT}}
    {d\sidel = D_1 \land d\sider = D_2}{\mathit{out}\sidel = \mathit{out}\sider}{(\varepsilon, \delta)} .
  \]
\end{theorem}
\begin{proof}
  The \Saprhl proof follows the approximate coupling proof in
  \cref{thm:bt-priv-coupling} closely. There are two main technicalities.
  First, we must take care to apply the rules that affect the parameter $\delta$
  in the proper order. For instance, \nameref{rule:aprhl-pw-eq} should be
  applied to pointwise judgments that are $(\varepsilon, 0)$-approximate
  couplings---if the pointwise judgment has $\delta > 0$, then
  \nameref{rule:aprhl-pw-eq} will sum $\delta$ over all possible outputs. Since
  \nameref{rule:aprhl-lapacc-l}/\nameref{rule:aprhl-lapacc-r} and
  \nameref{rule:aprhl-while-ac} increase the $\delta$ parameter, we apply these
  rules below \nameref{rule:aprhl-pw-eq} in the proof tree. Second, we need to
  make sure that the outer loop invariant is of the correct form so we can
  convert to a symmetric judgment and apply \nameref{rule:aprhl-while-ac}.

  At a high level, we apply \nameref{rule:aprhl-pw-eq} on the inner loop
  assuming in the pre-condition that the threshold noise are coupled
  appropriately, and not too large. Then, we apply the accuracy bound
  \nameref{rule:aprhl-lapacc-l} and threshold coupling
  \nameref{rule:aprhl-lapgen} for the first part of the outer loop body.
  Finally, we convert the standard \Saprhl judgment for the loop body to a
  symmetric judgment, applying \nameref{rule:aprhl-while-ac} on the outer loop
  to conclude the proof.

  Let's see this plan in action.  We begin with the inner loop, $\mathit{in}$.
  %
  We prove a pointwise judgment for the following, equivalent version of
  $\mathit{in}$, split into three stages:
  \[
    \begin{array}{l}
      \WWhile{i \leq N \land i < j}{} \\
      \quad \Rand{v}{\Lap{\varepsilon'/3}(\evalQ(i, d))}; \\
      \quad \Condt{a < v < b \land \mathit{go}}{} \\
      \quad\quad \Rand{\mathit{noisy}}{\Lap{\varepsilon'}(\evalQ(i, d))}; \\
      \quad\quad \Ass{\mathit{ans}}{(i, \mathit{noisy})}; \\
      \quad\quad \Ass{\mathit{go}}{\kwfalse}; \\
      \quad \Ass{i}{i + 1}; \\
      \WWhile{i \leq N \land i = j}{} \\
      \quad \Rand{v}{\Lap{\varepsilon'/3}(\evalQ(i, d))}; \\
      \quad \Condt{a < v < b \land \mathit{go}}{} \\
      \quad\quad \Rand{\mathit{noisy}}{\Lap{\varepsilon'}(\evalQ(i, d))}; \\
      \quad\quad \Ass{\mathit{ans}}{(i, \mathit{noisy})}; \\
      \quad\quad \Ass{\mathit{go}}{\kwfalse}; \\
      \quad \Ass{i}{i + 1}; \\
      \WWhile{i \leq N}{} \\
      \quad \Rand{v}{\Lap{\varepsilon'/3}(\evalQ(i, d))}; \\
      \quad \Condt{a < v < b \land \mathit{go}}{} \\
      \quad\quad \Rand{\mathit{noisy}}{\Lap{\varepsilon'}(\evalQ(i, d))}; \\
      \quad\quad \Ass{\mathit{ans}}{(i, \mathit{noisy})}; \\
      \quad\quad \Ass{\mathit{go}}{\kwfalse}; \\
      \quad \Ass{i}{i + 1}
    \end{array}
  \]
  We call this program $\mathit{in}'$, the three loops $w_<$, $w_=$, and $w_>$,
  and the common loop body $\mathit{body}_{\mathit{in}}$. 
  We implicitly maintain the invariant
  $d\sidel = D_1 \land d\sider = D_2$ in all judgments and take the 
  following global invariant:
  \begin{align*}
    \Xi \triangleq \begin{cases}
      i\sidel = i\sider \\
      a\sidel + 1 = a\sider \land b\sidel = b\sider + 1
      \land b\sider - a\sider \geq \frac{6}{\varepsilon'} \ln(4/\varepsilon') - 2 \\
      [a\sidel, b\sidel] = [A - u\sidel, B + u\sidel]
      \land [a\sider, b\sider] = [A - u\sider, B + u\sider]
    \end{cases}
  \end{align*}
  Reading from top to bottom, this ensures
    (i) the loops are synchronized,
    (ii) the noisy thresholds are coupled and not too close, and
    (iii) the noisy thresholds share the noise $u$.
  Since $\mathit{in}'$ does not modify the variables $a$, $b$ and $u$, this
  assertion is preserved by the loops.  Now let $(j, y) \in \NN \times \ZZ$ be
  a possible value of $\mathit{ans}$.  We define the following invariants for
  the three loops:
  \begin{align*}
    \Theta_< &\triangleq \Xi
    \land \mathit{go}\sidel \to \mathit{go}\sider
    \land \neg(i\sidel \leq N \land i\sidel < j) \to i\sidel = j \\
    \Theta_= &\triangleq \Xi \land \begin{cases}
      \mathit{go}\sidel \to \mathit{go}\sider \\
      \mathit{ans}\sidel = (j, y) \to \mathit{ans}\sider = (j, y) \\
      \neg(i\sidel \leq N \land i\sidel = j) \to i\sidel = j + 1
    \end{cases}
    \\
    \Theta_> &\triangleq \Xi
    \land i\sidel > j
    \land \mathit{ans}\sidel = (j, y) \to \mathit{ans}\sider = (j, y)
  \end{align*}
  Now we proceed one loop at a time. First, we have
  \[
    \vdash \aprhl{\mathit{body}_{\mathit{in}}}{\mathit{body}_{\mathit{in}}}{\Theta_<}{\Theta_<}{(0, 0)}
  \]
  by coupling the sampling for $v$ with \nameref{rule:aprhl-lapnull} and using
  \nameref{rule:aprhl-lap-l}, \nameref{rule:aprhl-lap-r}, or
  \nameref{rule:aprhl-lapnull} to couple the samples for $\mathit{noisy}$.  This
  ensures $|v\sidel - v\sider| \leq 1$; combined with the threshold coupling, we
  know that if the first side doesn't find a between-threshold query then
  neither does the second side, so $\mathit{go}\sidel \to \mathit{go}\sider$. We
  get a coupling for the first loop by \nameref{rule:aprhl-while}:
  \[
    \vdash \aprhl{w_<}{w_<}
    {\Theta_<}{\Theta_< \land \neg(i \leq N \land i < j)}{(0, 0)} .
  \]
  For the second loop, we prove
  \[
    \vdash \aprhl{\mathit{body}_{\mathit{in}}}{\mathit{body}_{\mathit{in}}}{\Theta_=}{\Theta_=}{(2\varepsilon', 0)} .
  \]
  We couple the samplings for $v$ with the subset coupling
  \nameref{rule:aprhl-lapint}, ensuring the two processes take the same path in
  the conditional. Since the thresholds are sufficiently apart (by $\Xi$) and
  the queries are $1$-sensitive, \nameref{rule:aprhl-lapint} is an
  $(\varepsilon', 0)$-approximate coupling by \cref{lem:lapint-simple}.

  If both processes find between-threshold queries, then we couple the samplings
  for $\mathit{noisy}$ with the standard Laplace rule \nameref{rule:aprhl-lap}
  so $\mathit{noisy}\sidel = \mathit{noisy}\sider$; this is an
  $(\varepsilon', 0)$-approximate coupling since the queries are $1$-sensitive.
  Otherwise if both sides are outside the interval, we do not sample
  $\mathit{noisy}$. Thus, we have a $(2\varepsilon', 0)$-approximate coupling
  where if $\mathit{ans}\sidel = (j, y)$, then $\mathit{ans}\sider = (j, y)$
  too. Since the loop $w_=$ executes for exactly one iteration,
  \nameref{rule:aprhl-while} gives
  \[
    \vdash \aprhl{w_=}{w_=}
    {\Theta_=}{\Theta_= \land \neg(i\sidel \leq N \land i\sidel = j)}{(2\varepsilon', 0)} .
  \]
  For the last loop we simply couple the samplings for $v$ with the null
  coupling \nameref{rule:aprhl-lapnull} and use any zero-cost coupling for
  $\mathit{noisy}$ (\nameref{rule:aprhl-lap-l}, \nameref{rule:aprhl-lap-r}, or
  \nameref{rule:aprhl-lapnull}), giving
  \[
    \vdash \aprhl{w_>}{w_>}
    {\Theta_>}{\Theta_> \land \neg(i \leq N)}{(0, 0)} .
  \]
  Applying the rule of consequence with the implications
  \begin{align*}
    &\models \Theta_< \land \neg(i\sidel \leq N \land i\sidel < j) \to \Theta_= \\
    &\models \Theta_= \land \neg(i\sidel \leq N \land i\sidel = j) \to \Theta_> ,
  \end{align*}
  we combine the loop judgments while summing the approximation parameters
  with \nameref{rule:aprhl-seq} to get
  \[
    \vdash \aprhl{\mathit{in}'}{\mathit{in}'}
    {\Xi}{\mathit{ans}\sidel = (j, y) \to \mathit{ans}\sider = (j, y)}{(2\varepsilon', 0)} .
  \]
  Pointwise equality \nameref{rule:aprhl-pw-eq} completes the proof for the
  inner loop:
  \[
    \vdash \aprhl{\mathit{in}'}{\mathit{in}'}
    {\Xi}{\mathit{ans}\sidel = \mathit{ans}\sider}{(2\varepsilon', 0)} .
  \]
  Now let the outer loop by $\mathit{w}_{\mathit{out}}$, with body
  $\mathit{body}_{\mathit{out}}$. We ensure $\Xi$ after the threshold samplings
  by applying \nameref{rule:aprhl-lapgen} and the accuracy bound
  \nameref{rule:aprhl-lapacc-l}, using an $(\varepsilon',
  \delta/2C)$-approximate coupling for the threshold samplings and showing
  \[
    \vdash \aprhl{\mathit{body}_{\mathit{out}}}{\mathit{body}_{\mathit{out}}}
    {(i, \mathit{out})\sidel = (i, \mathit{out})\sider}
    {(i, \mathit{out})\sidel = (i, \mathit{out})\sider}
    {(3\varepsilon', \delta/2C)} .
  \]
  Continuing to keep the adjacency condition $d_1\sidel = D_1 \land d_2\sider =
  D_2$ implicit, we can apply \nameref{rule:aprhl-symintro} to get the symmetric
  judgment
  \[
    \vdash \symaprhl{\mathit{body}_{\mathit{out}}}{\mathit{body}_{\mathit{out}}}
    {(i, \mathit{out})\sidel = (i, \mathit{out})\sider}
    {(i, \mathit{out})\sidel = (i, \mathit{out})\sider}
    {(3\varepsilon', \delta/2C)} .
  \]
  Taking the loop invariant $\Psi \triangleq (i, \mathit{out})\sidel = (i,
  \mathit{out})\sider \land d\sidel = D_1 \land d\sider = D_2$, the advanced
  composition rule \nameref{rule:aprhl-while-ac} gives
  \[
    \vdash \symaprhl{w_{\mathit{out}}}{w_{\mathit{out}}}{\Psi}{\Psi}{(\varepsilon, \delta)}
  \]
  using the setting of $\varepsilon'$ from \cref{fn:ac-setting}. Converting
  back to a standard judgment by \nameref{rule:aprhl-symelim-l} and handling the
  initial assignments, we conclude differential privacy:
  \[
    \vdash \aprhl{\mathit{BT}}{\mathit{BT}}
    {d\sidel = D_1 \land d\sider = D_2}{\mathit{out}\sidel = \mathit{out}\sider}{(\varepsilon, \delta)} .
  \]
\end{proof}

\section{Comparison to other approximate liftings} \label{sec:alift-rw}

The notion of approximate lifting has been formulated numerous times. We compare
with several prior definitions in the discrete case. Research on the continuous
case is ongoing; we summarize recent developments in the next chapter
(\cref{sec:concurrent}).

\subsection{Symmetric approximate liftings}

While symmetric approximate liftings are less general than their asymmetric
counterparts, they are interesting in their own right. In fact, our symmetric
approximate liftings are equivalent to the approximate liftings proposed by
\citet{BKOZ13-toplas} in the original work on proving differential privacy via
relational program logics. Unlike our definitions, which use two witnesses,
their notion is based on a single witness.

\begin{definition} \label{def:symalift-one}
  Let $\mu_1, \mu_2$ be sub-distributions over $\cA_1$ and $\cA_2$, and let $\cR
  \subseteq \cA_1 \times \cA_2$ be a relation. A sub-distribution $\mu$ over
  pair s$\cA_1 \times \cA_2$ is a \emph{witness} for the \emph{one-witness
  $(\varepsilon, \delta)$-approximate $\cR$-lifting} of $(\mu_1, \mu_2)$ if:
  \begin{enumerate}
    \item $\pi_1(\mu) \leq \mu_1$ and $\pi_2(\mu) \leq \mu_2$;
    \item $\supp(\mu) \subseteq \cR$; and
    \item $\epsdist{\varepsilon}{\mu_1}{\pi_1(\mu)} \leq \delta$ and
      $\epsdist{\varepsilon}{\mu_2}{\pi_2(\mu)} \leq \delta$.\footnote{%
        The original definition by \citet{BKOZ13-toplas} involved a symmetric
        notion of $\varepsilon$-distance, and flipped the direction of both distances
        in this point. To keep notation uniform, we present their definition in
        terms of our (asymmetric) notion of $\varepsilon$-distance from
        \cref{def:epsdist-privacy}.}
  \end{enumerate}
\end{definition}

This definition is arguably closer to the spirit of probabilistic couplings: a
single joint sub-distribution approximately modeling two given distributions as
marginals.

\begin{theorem} \label{thm:onewitness-sym}
  Let $\mu_1, \mu_2$ be sub-distributions over $\cA_1$ and $\cA_2$,
  and let $\cR \subseteq \cA_1 \times \cA_2$ be a relation. There is a
  one-witness $(\varepsilon, \delta)$-approximate lifting of $\cR$ in the sense
  of \cref{def:symalift-one} if and only if there is a symmetric approximate
  lifting
  \[
    \mu_1 \symalift{\cR}{(\varepsilon, \delta)} \mu_2 .
  \]
\end{theorem}
\begin{proof}
  For the reverse direction, let $(\mu_L, \mu_R)$ witness the symmetric
  approximate lifting and define $\eta \in \SDist(\cA_1 \times \cA_2)$ as the
  pointwise minimum: $\eta(a_1, a_2) \triangleq \min(\mu_L(a_1, a_2), \mu_R(a_1,
  a_2))$. We check that $\eta$ is a witness to an approximate lifting in the
  sense of \cref{def:symalift-one}.

  The support condition follows from the support condition for $(\mu_L, \mu_R)$.
  The marginal conditions $\pi_1(\eta) \leq \mu_1$ and $\pi_2(\eta) \leq \mu_2$
  also follow by the marginal conditions for $(\mu_L, \mu_R)$. The only thing to
  check is the distance condition. Define non-negative constants
  \[
    \delta(a_1, a_2) \triangleq
    \max(\mu_L(a_1, a_2) - \exp(\varepsilon) \cdot \mu_R(a_1, a_2), 0) .
  \]
  By the distance condition on $(\mu_L, \mu_R)$,
  \[
    \mu_L(a_1, a_2) \leq \exp(\varepsilon) \cdot \mu_R(a_1, a_2) + \delta(a_1, a_2)
  \]
  with equality when $\delta(a_1, a_2) > 0$, and $\sum_{a_1, a_2}
  \delta(a_1,a_2) \leq \delta$. We claim
  \[
    \min(\mu_L(a_1, a_2), \mu_R(a_1, a_2))
    \geq \exp(-\varepsilon) (\mu_L(a_1, a_2) - \delta(a_1, a_2)) .
  \]
  If $\delta(a_1, a_2) = 0$ then $\mu_R(a_1, a_2) \geq \exp(-\varepsilon)
  \mu_L(a_1, a_2)$. Otherwise if $\delta(a_1, a_2) > 0$, then
  \[
    \mu_R(a_1, a_2) = \exp(-\varepsilon) (\mu_L(a_1, a_2) - \delta(a_1, a_2)) \leq \mu_L(a_1, a_2)
  \]
  and the claim is again clear. Similarly, define
  \[
    \delta'(a_1, a_2) \triangleq
    \max(\mu_R(a_1, a_2) - \exp(\varepsilon) \cdot \mu_L(a_1, a_2), 0) .
  \]
  We have
  \[
    \mu_R(a_1, a_2) \leq \exp(\varepsilon) \cdot \mu_L(a_1, a_2) + \delta'(a_1, a_2)
  \]
  with equality when $\delta'(a_1, a_2) = 0$, and $\sum_{a_1, a_2}
  \delta'(a_1,a_2) \leq \delta$. By analogous reasoning, we have
  \[
    \min(\mu_L(a_1, a_2), \mu_R(a_1, a_2))
    \geq \exp(-\varepsilon) (\mu_R(a_1, a_2) - \delta'(a_1, a_2)) .
  \]
  Now let $\cS_1 \subseteq \cA_1$ be any subset. Then:
  \begin{align*}
    \mu_1(\cS_1) - \exp(\varepsilon) \cdot \pi_1(\eta)(\cS_1)
    &= \sum_{a_1 \in \cS_1} \left( \mu_1(a_1) - \exp(\varepsilon) \sum_{a_2 \in \cA_2} \min(\mu_L(a_1, a_2), \mu_R(a_1, a_2)) \right) \\
    &\leq \sum_{a_1 \in \cS_1} \left( \mu_1(a_1) - \exp(\varepsilon) \sum_{a_2 \in \cA_2} \exp(-\varepsilon) (\mu_L(a_1, a_2) - \delta(a_1, a_2)) \right) \\
    &= \sum_{a_1 \in \cS_1, a_2 \in \cA_2} \delta(a_1, a_2) \leq \delta .
  \end{align*}
  The other marginal is similar: for any subset $\cS_2 \subseteq \cA_2$, we have
  \begin{align*}
    \mu_2(\cS_2) - \exp(\varepsilon) \cdot \pi_2(\eta)(\cS_2)
    &= \sum_{a_2 \in \cS_2} \left( \mu_2(a_2) - \exp(\varepsilon) \sum_{a_1 \in \cA_1} \min(\mu_L(a_1, a_2), \mu_R(a_1, a_2)) \right) \\
    &\leq \sum_{a_2 \in \cS_2} \left( \mu_2(a_2) - \exp(\varepsilon) \sum_{a_1 \in \cA_1} \exp(-\varepsilon) (\mu_R(a_1, a_2) - \delta'(a_1, a_2)) \right) \\
    &= \sum_{a_2 \in \cS_2, a_1 \in \cA_1} \delta'(a_1, a_2) \leq \delta .
  \end{align*}
  Thus, $\eta$ witnesses the one-witness $(\varepsilon, \delta)$-approximate
  lifting of $\cR$.

  The forward direction is more interesting. Let $\eta \in \SDist(\cA_1 \times
  \cA_2)$ be the single witness and define
  \begin{align*}
    \delta(a_1) &\triangleq \max(\mu_1(a_1) - \exp(\varepsilon) \cdot \pi_1(\eta)(a_1), 0) \\
    \delta'(a_2) &\triangleq \max(\mu_2(a_2) - \exp(\varepsilon) \cdot \pi_2(\eta)(a_2), 0) .
  \end{align*}
  By the distance conditions
  $\epsdist{\varepsilon}{\mu_1}{\pi_1(\eta)} \leq \delta$ and
  $\epsdist{\varepsilon}{\mu_2}{\pi_2(\eta)} \leq \delta$, we have
  $\delta(a_1), \delta'(a_2) \geq 0$ and
  \begin{align*}
    \mu_1(a_1) &\leq \exp(\varepsilon) \cdot \pi_1(\eta)(a_1) + \delta(a_1) \\
    \mu_2(a_2) &\leq \exp(\varepsilon) \cdot \pi_2(\eta)(a_2) + \delta'(a_2) ,
  \end{align*}
  with equality when $\delta(a_1)$ or $\delta'(a_2)$ are strictly positive.
  Furthermore, $\sum_{a_1 \in \cA_1} \delta(a_1)$ and $\sum_{a_2 \in \cA_2}
  \delta'(a_2)$ are at most $\delta$. Define witnesses $\mu_L, \mu_R \in
  \SDist(\cA_1^\star \times \cA_2^\star)$ as follows:
  \begin{align*}
    \mu_L(a_1, a_2) &\triangleq
    \begin{cases}
      \eta(a_1, a_2) \cdot \frac{\mu_1(a_1) - \delta(a_1)}{\pi_1(\eta)(a_1)}
      &: a_1 \neq \star, a_2 \neq \star \\
      \mu_1(a_1) - \sum_{a_2' \in \cA_2} \mu_L(a_1, a_2')
      &: a_1 \neq \star, a_2 = \star \\
      0
      &: \text{otherwise}
    \end{cases}
    \\
    \mu_R(a_1, a_2) &\triangleq
    \begin{cases}
      \eta(a_1, a_2) \cdot \frac{\mu_2(a_2) - \delta'(a_2)}{\pi_2(\eta)(a_2)}
      &: a_1 \neq \star, a_2 \neq \star \\
      \mu_2(a_2) - \sum_{a_1' \in \cA_1} \mu_R(a_1', a_2)
      &: a_1 = \star, a_2 \neq \star \\
      0
      &: \text{otherwise} .
    \end{cases}
  \end{align*}
  As usual, if any denominator is zero we take the whole term to be zero as
  well.
  
  The support condition follows from the support condition of $\eta$; the
  marginal conditions hold by definition. All probabilities are
  non-negative---for instance in $\mu_L$, if $\delta(a_1) > 0$ then $\mu_1(a_1)
  - \delta(a_1) = \exp(\varepsilon) \cdot \pi_1(\eta)(a_1) \geq 0$ and
  \[
    \mu_L(a_1, \star)
    = \mu_1(a_1) - \exp(\varepsilon) \cdot \pi_1(\eta)(a_1)
    = \delta(a_1)
    \geq 0
  \]
  when $\pi_1(\eta)(a_1) > 0$; if $\pi_1(\eta)(a_1) = 0$ then $\mu_L(a_1,
  \star) = \mu_1(a_1) = 0$. If $\delta(a_1) = 0$ then we can check $\eta(a_1,
  \star) \geq 0$. A similar argument shows that $\mu_R$ is non-negative.
  
  So, it remains to check the distance bounds.  We first claim
  \[
    \mu_L(a_1, a_2) \leq \exp(\varepsilon) \cdot \mu_R(a_1, a_2)
    \qquad\text{and}\qquad
    \mu_R(a_1, a_2) \leq \exp(\varepsilon) \cdot \mu_L(a_1, a_2) .
  \]
  When $a_1, a_2 \neq \star$, by definition $\mu_L(a_1, a_2)$ and $\mu_R(a_1,
  a_2)$ are both positive or both zero depending on whether $\eta(a_1, a_2)$ is
  positive or zero. The zero case is immediate. In the positive case,
  \[
    \frac{\mu_L(a_1, a_2)}{\eta(a_1, a_2)}
    = \frac{\mu_1(a_1) - \delta(a_1)}{\pi_1(\eta)(a_1)}
    \leq \exp(\varepsilon)
    \qquad\text{and}\qquad
    \frac{\mu_R(a_1, a_2)}{\eta(a_1, a_2)}
    = \frac{\mu_2(a_2) - \delta'(a_2)}{\pi_2(\eta)(a_2)}
    \leq \exp(\varepsilon) .
  \]
  We can also lower bound the ratios:
  \[
    \frac{\mu_L(a_1, a_2)}{\eta(a_1, a_2)}
    = \frac{\mu_1(a_1) - \delta(a_1)}{\pi_1(\eta)(a_1)}
    \geq 1
    \qquad\text{and}\qquad
    \frac{\mu_R(a_1, a_2)}{\eta(a_1, a_2)}
    = \frac{\mu_2(a_2) - \delta'(a_2)}{\pi_2(\eta)(a_2)}
    \geq 1 ;
  \]
  for instance when $\delta(a_1) > 0$ the ratio is exactly equal to
  $\exp(\varepsilon) \geq 1$, and when $\delta(a_1) = 0$ the ratio is at least
  $1$ by the marginal property $\pi_1(\eta) \leq \mu_1$.  So, $\mu_L(a_1,
  a_2)/\eta(a_1, a_2)$ and $\mu_R(a_1, a_2)/\eta(a_1, a_2)$ are in $[1,
  \exp(\varepsilon)]$ when all distributions are positive, establishing the
  claim.
  
  Finally, we bound the mass on points $(a_1, \star)$. Letting $\cS_1 \subseteq
  \cA_1$ be any subset, $0 = \mu_R(\cS_1 \times \{ \star \}) \leq
  \exp(\varepsilon)\cdot \mu_L(\cS_1 \times \{ \star \}) + \delta$ is clear. For
  the other direction,
  \begin{align*}
    \mu_L(\cS_1 \times \{ \star \}) &= \sum_{a_1 \in \cS_1}
    \left( \mu_1(a_1)
    - \mu_1(a_1) \sum_{a_2 \in \cA_2} \frac{\eta(a_1, a_2)}{\pi_1(\eta)(a_1)}
    + \delta(a_1) \sum_{a_2 \in \cA_2} \frac{\eta(a_1, a_2)}{\pi_1(\eta)(a_1)} \right)
    \\
    &= \mu_1(\cS_1) - \mu_1(\cS_1) + \delta(\cS_1)
    \leq \exp(\varepsilon) \cdot \mu_R(\cS_1 \times \{ \star \}) + \delta .
  \end{align*}
  The mass at points $(\star, a_2)$ can be bounded in a similar way. Let $\cS_2
  \subseteq \cA_2$ be any subset. Then $0 = \mu_L(\{ \star \} \times \cS_2) \leq
  \exp(\varepsilon) \cdot \mu_R(\{ \star \} \times \cS_2) + \delta$ is clear.
  For the other direction,
  \begin{align*}
    \mu_R(\{ \star \} \times \cS_2) &= \sum_{a_2 \in \cS_2}
    \left( \mu_2(a_2)
    - \mu_2(a_2) \sum_{a_1 \in \cA_1} \frac{\eta(a_1, a_2)}{\pi_2(\eta)(a_2)}
    + \delta'(a_2) \sum_{a_1 \in \cA_1} \frac{\eta(a_1, a_2)}{\pi_2(\eta)(a_2)} \right)
    \\
    &= \mu_2(\cS_2) - \mu_2(\cS_2) + \delta'(\cS_2)
    \leq \exp(\varepsilon) \cdot \mu_L(\{ \star \} \times \cS_2) + \delta .
  \end{align*}
  So $\epsdist{\varepsilon}{\mu_L}{\mu_R} \leq \delta$ and
  $\epsdist{\varepsilon}{\mu_R}{\mu_L} \leq \delta$, and we have a
  symmetric approximate lifting.
\end{proof}

\subsection{Asymmetric approximate liftings, alternative definition}

After introducing their symmetric notion of lifting (\cref{def:symalift-one}),
\citet{BKOZ13-toplas} also considered \emph{asymmetric} approximate liftings
with a single witness distribution.

\begin{definition} \label{def:alift-one}
  Let $\mu_1, \mu_2$ be sub-distributions over $\cA_1$ and $\cA_2$, and let $\cR
  \subseteq \cA_1 \times \cA_2$ be a relation. A sub-distribution $\mu$ over
  pairs $\cA_1 \times \cA_2$ is a \emph{witness} for the \emph{one-witness
  asymmetric $(\varepsilon, \delta)$-approximate $\cR$-lifting} of $(\mu_1,
  \mu_2)$ if:
  \begin{enumerate}
    \item $\pi_1(\mu) \leq \mu_1$ and $\pi_2(\mu) \leq \mu_2$;
    \item $\supp(\mu) \subseteq \cR$; and
    \item $\epsdist{\varepsilon}{\mu_1}{\pi_1(\mu)} \leq \delta$.\footnote{%
        The original definition by \citet{BKOZ13-toplas} used the same notion of
        $\varepsilon$-distance that we use (\cref{def:epsdist-privacy}), but
        incorrectly flipped the direction of the distance bound. It is also
        possible to define a version involving the second marginal
        instead of the first.}
  \end{enumerate}
\end{definition}

Note the key difference compared to the symmetric version: the distance bound is
only required to hold between the first distribution and the first marginal. We
can show \cref{def:alift-one} coincides with our asymmetric notion of
approximate lifting.

\begin{theorem} \label{thm:onewitness-asym}
  Let $\mu_1, \mu_2$ be sub-distributions over $\cA_1$ and $\cA_2$,
  and let $\cR \subseteq \cA_1 \times \cA_2$ be a relation. Then there is a
  (one-witness) asymmetric $(\varepsilon, \delta)$-approximate lifting of $\cR$
  in the sense of \cref{def:alift-one} if and only if there is an approximate
  lifting:
  \[
    \mu_1 \alift{\cR}{(\varepsilon, \delta)} \mu_2 .
  \]
\end{theorem}
\begin{proof}
  For the reverse direction, let $(\mu_L, \mu_R)$ witness the approximate
  lifting and define $\eta \in \SDist(\cA_1 \times \cA_2)$ as the pointwise
  minimum: $\eta(a_1, a_2) \triangleq \min(\mu_L(a_1, a_2), \mu_R(a_1, a_2))$.
  We claim that $\eta$ witnesses an asymmetric approximate lifting in the
  sense of \cref{def:alift-one}.

  The support condition follows from the support condition for $(\mu_L, \mu_R)$;
  the marginal conditions $\pi_1(\eta) \leq \mu_1$ and $\pi_2(\eta) \leq \mu_2$
  also follow by the marginal conditions for $(\mu_L, \mu_R)$. To check the
  distance condition, define
  \[
    \delta(a_1, a_2) \triangleq \max(\mu_L(a_1, a_2) - \exp(\varepsilon) \cdot \mu_R(a_1, a_2), 0) .
  \]
  By the distance condition on $(\mu_L, \mu_R)$, we have
  \[
    \mu_L(a_1, a_2) \leq \exp(\varepsilon) \cdot \mu_R(a_1, a_2) + \delta(a_1, a_2)
  \]
  with equality when $\delta(a_1, a_2) > 0$, and $\sum_{a_1, a_2}
  \delta(a_1,a_2) \leq \delta$. Like in the proof of \cref{thm:onewitness-sym},
  we have
  \[
    \min(\mu_L(a_1, a_2), \mu_R(a_1, a_2))
    \geq \exp(-\varepsilon) (\mu_L(a_1, a_2) - \delta(a_1, a_2)) .
  \]
  To conclude the distance bound, let $\cS_1 \subseteq \cA_1$ be a subset. Then:
  \begin{align*}
    \mu_1(\cS_1) - \exp(\varepsilon) \cdot \pi_1(\eta)(\cS_1)
    &= \sum_{a_1 \in \cS_1} \left( \mu_1(a_1) - \exp(\varepsilon) \sum_{a_2 \in \cA_2} \min(\mu_L(a_1, a_2), \mu_R(a_1, a_2)) \right)\\
    &\leq \sum_{a_1 \in \cS_1} \left( \mu_1(a_1) - \exp(\varepsilon)
    \sum_{a_2 \in \cA_2} \exp(-\varepsilon) (\mu_L(a_1, a_2) - \delta(a_1, a_2)) \right) \\
    &= \sum_{a_1 \in \cS_1, a_2 \in \cA_2} \delta(a_1, a_2) \leq \delta .
  \end{align*}
  Thus, $\eta$ witnesses the (one-witness) asymmetric $(\varepsilon,
  \delta)$-approximate lifting of $\cR$.

  The forward direction is more interesting. Let $\eta \in \SDist(\cA_1 \times
  \cA_2)$ be the single witness and define
  \[
    \delta(a_1) \triangleq \mu_1(a_1) - \exp(\varepsilon) \cdot \pi_1(\eta)(a_1) .
  \]
  By the distance condition $\epsdist{\varepsilon}{\mu_1}{\pi_1(\eta)} \leq
  \delta$, we know $\delta(a_1)$ is non-negative. Furthermore,
  \[
    \mu_1(a_1) \leq \exp(\varepsilon) \cdot \pi_1(\eta)(a_1) + \delta(a_1)
  \]
  with equality when $\delta(a_1)$ is strictly positive, and $\sum_{a_1 \in \cA_1}
  \delta(a_1) \leq \delta$. Define two witnesses $\mu_L, \mu_R \in \SDist(\cA_1^\star
  \times \cA_2^\star)$ as follows:
  \begin{align*}
    \mu_L(a_1, a_2) &\triangleq
    \begin{cases}
      \eta(a_1, a_2) \cdot \frac{\mu_1(a_1) - \delta(a_1)}{\pi_1(\eta)(a_1)}
      &: a_1 \neq \star, a_2 \neq \star \\
      \mu_1(a_1) - \sum_{a_2' \in \cA_2} \mu_L(a_1, a_2')
      &: a_1 \neq \star, a_2 = \star \\
      0
      &: \text{otherwise}
    \end{cases}
    \\
    \mu_R(a_1, a_2) &\triangleq
    \begin{cases}
      \eta(a_1, a_2) &: a_1 \neq \star, a_2 \neq \star \\
      \mu_2(a_2) - \sum_{a_1' \in \cA_1} \mu_R(a_1', a_2)
      &: a_1 = \star, a_2 \neq \star \\
      0
      &: \text{otherwise} .
    \end{cases}
  \end{align*}
  If any denominator is zero, we take the probability to be zero as well.

  The support condition follows from the support condition of $\eta$; the
  marginal conditions hold by definition. To show all probabilities are
  non-negative, for $\mu_L$ note that if $\delta(a_1) > 0$ then $\mu_1(a_1) -
  \delta(a_1) = \exp(\varepsilon) \cdot \pi_1(\eta)(a_1) \geq 0$ and hence
  \[
    \mu_L(a_1, \star) = \mu_1(a_1) - \delta(a_1) \geq 0
  \]
  assuming $\pi_1(\eta)(a_1) > 0$; if $\pi_1(\eta)(a_1) = 0$ then $\mu_L(a_1, \star) = 0$.
  For $\mu_R$, non-negativity holds by $\pi_2(\eta) \leq \mu_2$.
  
  We just need to show the distance bound. When $a_1, a_2 \neq \star$, 
  we claim
  \[
    \mu_L(a_1, a_2) \leq \exp(\varepsilon) \cdot \eta(a_1, a_2) = \exp(\varepsilon) \cdot \mu_R(a_1, a_2) .
  \]
  By definition $\mu_L(a_1, a_2)$, $\mu_R(a_1, a_2)$, and $\eta(a_1, a_2)$ are
  all positive or all zero. The zero case is immediate. In the positive case,
  \[
    \frac{\mu_L(a_1, a_2)}{\eta(a_1, a_2)}
    = \frac{\mu_1(a_1) - \delta(a_1)}{\pi_1(\eta)(a_1)} \leq \exp(\varepsilon)
  \]
  establishes the claim. To bound the mass on points $(a_1, \star)$, let $\cS_1
  \subseteq \cA_1$ be any subset. Then:
  \begin{align*}
    \mu_L(\cS_1 \times \{ \star \}) &=
    \sum_{a_1 \in \cS_1} \left( \mu_1(a_1)
      - \mu_1(a_1) \sum_{a_2 \in \cA_2} \frac{\eta(a_1, a_2)}{\pi_1(\eta)(a_1)}
      + \delta(a_1) \sum_{a_2 \in \cA_2} \frac{\eta(a_1, a_2)}{\pi_1(\eta)(a_1)}
    \right) \\
    &= \mu_1(\cS_1) - \mu_1(\cS_1) + \delta(\cS_1)
    \leq \exp(\varepsilon) \cdot \mu_R(\cS_1 \times \{ \star \}) + \delta
  \end{align*}
  so $\epsdist{\varepsilon}{\mu_L}{\mu_R} \leq \delta$ as desired, and we have
  witnesses to an approximate lifting.
\end{proof}

\subsection{Prior two-witness approximate liftings}

Our notion of approximate lifting is strongly inspired by a prior definition.
\begin{definition}[\citet{BartheO13} and \citet{OlmedoThesis}]
  Let $\mu_1, \mu_2$ be sub-distributions over $\cA_1$ and $\cA_2$, and let $\cR
  \subseteq \cA_1 \times \cA_2$ be a relation. Two sub-distributions $\mu_L,
  \mu_R$ over pairs $\cA_1 \times \cA_2$ are said to be \emph{witnesses} for the
  \emph{$(\varepsilon, \delta)$-approximate $\cR$-lifting} of $(\mu_1, \mu_2)$
  if:
  \begin{enumerate}
    \item $\pi_1(\mu_L) = \mu_1$ and $\pi_2(\mu_R) = \mu_2$;
    \item $\supp(\mu_L) \cup \supp(\mu_R) \subseteq \cR$; and
    \item $\epsdist{\varepsilon}{\mu_L}{\mu_R} \leq \delta$.
  \end{enumerate}
\end{definition}
There are several positive features of this definition. First, it generalizes to
other notions of distance on distribution; the distance $d_\varepsilon$ can be
replaced by an $f$-divergence. Furthermore, the witness distributions are
related by a distance that looks like the distance from differential privacy, so
composition theorems from differential privacy generalize to these liftings. 

However, there are several notable drawbacks. Perhaps the biggest flaw is this
definition does not support approximate lifting when $\cR$ does not contain the
supports $\supp(\mu_1) \times \supp(\mu_2)$. This limitation rules out up-to-bad
couplings and accuracy bounds. There are also several annoying technical
issues---the mapping property in \cref{thm:alift-extend} only holds for
surjective maps, the support property \cref{lem:alift-supp} fails, the subset
coupling in \cref{thm:opt-subset-math} does not work if the larger subset
$\cS_1$ is the whole domain $\cA_1$, etc.  These flaws are remedied in our
definition.

\subsection{Other notions of approximate equivalence}

Approximate notions of lifting have also appeared in the literature on
probabilistic bisimulation.  \citet{Tschantz201161} introduced the
\emph{$\delta$-lifting} of a relation $\cR$ to relate two distributions $\mu_1,
\mu_2$ when there is a bijection $f$ on the supports matching elements with
probabilities within a multiplicative factor:
\[
  \left| \ln \frac{\mu_1(x)}{\mu_2(f(x))} \right| \leq \delta
\]
and $(x, f(x)) \in \cR$. \citet{Tschantz201161} used this notion of lifting to
prove a variant of differential privacy for probabilistic labeled transition
systems, with a proof technique based on an unwinding family of relations.

Prior researchers largely focused on additive notions of approximate
equivalence; probably the first was due to \citet{giacalone1990algebraic}.
\citet{SegalaTurrini07} proposed $\varepsilon$-\emph{lifting}, equivalent to
$(0, \varepsilon)$-approximate lifting in our terminology. More recently,
\citet{desharnaisLT08} and \citet{tracolDZ11} investigated approximate notions
of probabilistic simulation and bisimulation, again similar to our $(0,
\delta)$-approximate liftings. \citet{desharnaisLT08} noted the connection
between their approximate liftings and maximum flows in a graph, extending the
connection by \citet[Theorem 7.3.4]{desharnais1999labelled} for exact liftings;
we use a similar observation to prove our approximate version of Strassen's
theorem.

\chapter{Emerging directions} \label{chap:future}

While we have limited this thesis to core connections between probabilistic
couplings and program logics, several lines of work---recently completed or
currently in progress---have already leveraged our results. We briefly survey
these extensions (\cref{sec:concurrent}), and then discuss promising technical
directions for further investigation (\cref{sec:future}). We conclude by
considering possible future connections between the theory of formal
verification and the theory of randomized algorithms (\cref{sec:bridge}).

\section{Concurrent developments} \label{sec:concurrent}

\subsection{Couplings for non-relational properties: Independence and uniformity}

As we have seen, couplings are a natural fit for probabilistic relational
properties. Properties describing a single program can also be viewed
relationally in some cases, enabling cleaner proofs by coupling.
\citet{BEGHS17} develop this idea to prove \emph{uniformity},
\emph{probabilistic independence}, and \emph{conditional independence}, examples
of probabilistic non-relational properties. We briefly sketch their main
reductions.

A uniform distribution places equal probability on every value in some range.
Given a distribution $\mu$ over $\Mem$ and an expression $e$ with finite range
$\cS$ (say, the booleans), $e$ is \emph{uniform} in $\mu$ if for all $a$ and
$a'$ in $\cS$, we have
\[
  \Pr_{m \sim \mu} [ \denot{e}m = a ] = \Pr_{m \sim \mu} [ \denot{e}m = a' ] .
\]
When $\mu$ is the output distribution of a program $c$, uniformity follows from
the \Sprhl judgment
\[
  \forall a, a' \in \cS,\; \vdash \prhl{c}{c}{(=)}{e\sidel = a \leftrightarrow e\sider = a'} .
\]
This reduction is a direct consequence of \cref{fact:imp-couple}. Moreover, the
resulting judgment is ideally suited to relational verification since it relates
two copies of the same program $c$.

Handling independence is only a bit more involved. Given a distribution $\mu$
and expressions $e, e'$ with ranges $\cS$ and $\cS'$, we say $e$ and $e'$ are
\emph{probabilistically independent} if for all $a \in \cS$ and $a' \in \cS'$, we
have
\[
  \Pr_{m \sim \mu} [ \denot{e}m = a \land \denot{e'}m = a' ]
  = \Pr_{m \sim \mu} [ \denot{e}m = a ] \cdot \Pr_{m \sim \mu} [ \denot{e'}m = a' ] .
\]
This useful property roughly implies that properties involving $e$ and $e'$ can
be analyzed by focusing on $e$ and $e'$ separately. When $e$ and $e'$ are
uniformly distributed, independence follows from uniformity of the tuple $(e,
e')$ over the product set $\cS \times \cS'$ so the previous reduction applies. In
general, we can compare the distributions of $e$ and $e'$ in two experiments:
when both are drawn from the output distribution of a single execution, and when
they are drawn from two independent executions composed sequentially. If the
expressions are independent, these two experiments should look the same.
Concretely, independence follows from the relational judgment
\[
  \forall a \in \cS, a' \in \cS',\;
  \vdash \prhl{c}{c^{(1)}; c^{(2)}}{\Phi}
  {e\sidel = a \land e'\sidel = a' \leftrightarrow e^{(1)}\sider = a \land {e'}^{(2)}\sider = a'} ,
\]
where $c^{(1)}$ and $c^{(2)}$ are copies of $c$ with variables $x$ renamed to
$x^{(1)}$ and $x^{(2)}$ respectively; this construction is also called
\emph{self-composition} since it sequentially composes $c$ with itself
\citep{BartheDR04}. The pre-condition $\Phi$ states that the three copies of
each variable are initially equal: $x\sidel = x^{(1)}\sider = x^{(2)}\sider$.
Handling conditional independence requires a slightly more complex encoding, but
the general pattern remains the same: encode products of probabilities by
self-composition and equalities by lifted equivalence
$\lift{(\leftrightarrow)}$.

These reductions give a simple method to prove uniformity and independence.
Other non-relational properties could benefit from a similar approach,
especially in conjunction with more sophisticated program transformations in
\Sprhl to relate different copies of the same sampling instruction.

\subsection{Variable approximate couplings}

As we saw in \cref{chap:approx,chap:combining}, approximate couplings are a
powerful tool for proving differential privacy. To further enhance the proof
technique, we can consider more precise ways of reasoning about the
$\varepsilon$ and $\delta$ parameters. To keep things simple, \Saprhl opts for
the most straightforward approach: $\varepsilon$ and $\delta$ are constants or
logical variables, independent of the program state. This choice is reflected in
the form and interpretation of the judgments:
\[
  \aprhl{c_1}{c_2}{\Phi}{\Psi}{(\varepsilon, \delta)} ,
\]
where $\varepsilon$ and $\delta$ are treated as as mathematical constants.  This
approach supports clean composition---we can simply add up $\varepsilon$ and
$\delta$ parameters without regard to which variables are changed by the
program---but it can be more convenient to think of $\varepsilon$ and $\delta$
as depending on the current state. For example, we may want to assert
$\varepsilon \leq n$ for a program variable $n$, representing some kind of
counter.

However, it is not immediately clear what a state-dependent privacy parameter
should mean, especially when the state is randomized.  To give a suitable
interpretation, we can look to the notion of a \emph{privacy loss random
variable} from the privacy literature.  Roughly, the privacy parameter
$\varepsilon$ may be viewed as a function mapping outputs to costs:
\[
  \Pr_{x \sim \mu_1} [ x = \xi ]
  \leq \exp(\varepsilon(\xi)) \cdot \Pr_{x \sim \mu_2} [ x = \xi ]
\]
for every $\xi$ in the support of $\mu_1$ and $\mu_2$. Then, $\mu_1$ induces a
distribution $\liftf{\varepsilon}(\mu_1)$ over privacy costs. If every cost in
the support of this distribution is bounded by a constant $\varepsilon_0$, the
output distributions $\mu_1, \mu_2$ satisfy the condition required for
$\varepsilon_0$-differential privacy. (See the textbook by \citet{DR14} for a
more thorough exposition.)

\citet{AH17} take inspiration from this idea and define an extension of
approximate couplings called \emph{variable approximate couplings}. Unlike
approximate couplings, which require a distance bound between witnesses that is
constant in $\varepsilon$ over all pairs of samples, a variable approximate
couplings allows $\varepsilon$ to vary:
\[
  \forall (a_1, a_2) \in \cA_1 \times \cA_2,\;
  \mu_L(a_1, a_2) \leq \exp(\varepsilon(a_1, a_2)) \cdot \mu_R(a_1, a_2)
\]
where $\varepsilon : \cA_1 \times \cA_2 \to \RR$ is now a function. The result
is a refinement of $(\varepsilon, 0)$-approximate coupling supporting a more
precise, randomized notion of privacy cost.

We can broadly compare reasoning in terms of variable approximate couplings with
reasoning in terms of approximate couplings (e.g., using a system like \Saprhl).
The main difficulty with variable approximate couplings is analyzing sequential
composition: now that each coupling has multiple costs associated with different
samples, the cost after composing couplings may become quite complicated---we
can't simply add the costs together.  Furthermore, it isn't clear how to handle
the additive parameter $\delta$ for proving $(\varepsilon, \delta)$-privacy. At
the same time, variable approximate couplings allow intuitive reasoning closer
to the cost-based interpretation of privacy, where the privacy level
$\varepsilon$ is regarded as a dynamic, possibly randomized quantity that
accumulates as the program executes. Rather than bounding the cost by a constant
at each stage of composition, we only need to bound the cost at the end of the
computation; this flexibility can support significantly simpler proofs.

\citet{AH17} use these richer couplings to support fully automated proofs of
$(\varepsilon, 0)$-differential privacy for challenging examples, including the
Report-noisy-max and Sparse Vector mechanisms we saw in \cref{chap:approx}.
Roughly speaking, they encode valid approximate coupling proofs with standard
Horn clauses and a new kind of coupling constraint, and then solve the
constraint systems with automated program verification and synthesis techniques.
Variable approximate couplings simplify their proofs in two ways. First, by
allow the privacy cost to be randomized during the analysis, there is no need to
separate deterministic and randomized parts of the state.  Second, their proofs
can leverage more sophisticated approximate couplings like the variable version
of the choice coupling from \cref{sec:aprhl-rw}, making their invariants easier
to discover automatically.

\subsection{Expectation couplings}

Probabilistic couplings and approximate couplings relate distributions over
plain sets with no additional structure. Many sets come with a notion of
distance, like the Euclidean distance on real vectors or the Hamming distance on
finite sets. If $\mathfrak{d} : \cA \times \cA \to \RR^+$ is a distance function
on a set $\cA$, the \emph{Kantorovich distance} on distributions $\Dist(\cA)$ is
defined as
\[
  \liftf{\mathfrak{d}}(\mu_1, \mu_2) \triangleq \min_{\mu \in \Omega(\mu_1, \mu_2)}
  \Ex_{(a_1, a_2) \sim \mu} [\mathfrak{d}(a_1, a_2)] ,
\]
where the minimum is taken over all couplings $\mu$ of $(\mu_1, \mu_2)$.  This
is a well-studied notion in probability theory and the theory of optimal
transport, increasingly seeing applications in computer science and beyond
(e.g., \citet{DESHARNAIS2004323,vanBreugel2001a,vanBreugel2001b} consider
logical aspects, \citet{deng2009kantorovich} survey applications in computer
science, and \citet{Villani08} explores the mathematical theory). Intuitively,
the Kantorovich distance lifts a distance $\mathfrak{d}$ on the ground set to a
distance $\liftf{\mathfrak{d}}$ on distributions, much like how probabilistic
liftings lift a relation $\cR$ on the ground set to a relation $\lift{\cR}$ on
distributions. Varying the ground distance recovers common distances on
distributions as special cases.

\citet{BEGHS16} use the Kantorovich distance to define \emph{expectation
coupling}, a quantitative extension of probabilistic coupling. Given two
distributions $\mu_1$ and $\mu_2$ on a set $\cA$ equipped with a distance
$\mathfrak{d}$, a coupling $\mu$ is a $(\mathfrak{d}, \delta)$-expectation
coupling if the expected value of $\mathfrak{d}$ on $\mu$ is at most $\delta$.
To construct and reason about these couplings, \citet{BEGHS16} develop a
relational program logic \Seprhl by augmenting the pre- and post-conditions in
\Sprhl judgments with \emph{pre-} and \emph{post-distances}:
\[
  \eprhl{c_1}{c_2}{\Phi}{\Psi}{\mathfrak{d}}{\mathfrak{d}'}{f} .
\]
The function $f : \RR \to \RR$ describes how the lifted post-distance can be
bounded as a function of the pre-distance. Judgments are valid when for any two
input memories $(m_1, m_2)$ satisfying the pre-condition $\Phi$, there is an
$(\mathfrak{d}', f(\mathfrak{d}(m_1, m_2)))$-expectation coupling of the two
output distributions with support in $\Psi$.  Intuitively, valid judgments model
Lipschitz-continuity or sensitivity, where the distance on input memories is
$\mathfrak{d}$ and the distance on output distributions is the Kantorovich
distance $\liftf{\mathfrak{d}'}$.

\Seprhl judgments can be combined in various ways, reflecting the clean
compositional properties of expectation couplings.  For instance, when $f$ is a
non-decreasing affine function (i.e., $f(z) = \alpha \cdot z + \beta$ with
$\alpha, \beta \geq 0$), judgments compose sequentially:
\[
  \inferrule*[Left=Seq]
  {
    \vdash \eprhl{c_1}{c_2}{\Phi}{\Psi}{\mathfrak{d}}{\mathfrak{d}'}{f} \\
    \vdash \eprhl{c_1'}{c_2'}{\Psi}{\Theta}{\mathfrak{d}'}{\mathfrak{d}''}{f'}
  }
  { 
    \vdash \eprhl{c_1; c_1'}{c_2; c_2'}{\Phi}{\Theta}{\mathfrak{d}}{\mathfrak{d}''}{f' \circ f}
  }
\]
The transitivity rule, which combines two judgments relating $c_1 \sim c_2$ and
$c_2 \sim c_3$ into a judgment relating $c_1 \sim c_3$, fully internalizes the
path coupling principle \citep{bubley1997path} we saw in \cref{chap:products}.

\Seprhl is particularly useful for proving quantitative relational properties.
In \Sprhl, as we noted in \cref{sec:prhl-ex}, there is no way to reason about
the probability of an event in the coupling. Our logic \Sxprhl from
\cref{chap:products} makes the coupling more explicit, but the logic can only
construct the product program, not reason about it. In contrast, \Seprhl
judgments can directly express quantitative properties of the coupling with the
pre- and post-distances.

To demonstrate, \citet{BEGHS16} use \Seprhl to verify convergence for a Markov
chain from population dynamics, and for the Glauber dynamics. In contrast to our
proof from \cref{sec:xprhl-ex-path}, which required reasoning about the product
program and applying path coupling externally, the \Seprhl proof can be carried
out almost entirely within the logic. Adding to the properties that can be
handled, \Seprhl can also verify that the Stochastic Gradient Descent algorithm
is \emph{uniformly stable}, a quantitative property comparing a learning
algorithm's expected error on two training sets \citep{BousquetE02}; this
recently-proposed property is rapidly gaining currency in the machine learning
community as a way to prevent overfitting \citep{HardtRS16}.

\subsection{Couplings in the continuous case}

To simplify our presentation, in this thesis we have focused on discrete
distributions. However, programs sampling from continuous distributions are
quite common in the algorithms literature; many private algorithms, for
instance, use samples from real-valued distributions like the Gaussian
distribution and the standard Laplace distribution.  Though most of our results
should carry over, the continuous case introduces additional measure-theoretic
technicalities. Designing a verification system supporting
continuous distributions---say, a program logic where programs can sample from
the Gaussian distribution---requires carefully handling these details.

While research historically evolved from exact liftings in \Sprhl to approximate
liftings in \Saprhl, current work on the continuous case has jumped directly to
approximate liftings. As we discussed in \cref{sec:apx-strassen}, \citet{Sato16}
introduced a novel definition of approximate lifting without witness
distributions in the continuous case, developing a continuous version of
\Saprhl. Sato derived his approximate lifting using a categorical construction
called \emph{codensity lifting of monads} (also called $\top\top$-lifting),
proposed by \citet{katsumata2015codensity}. Roughly speaking, this operation
turns a monad on a base category $D$ into a (possibly \emph{indexed} or
\emph{graded}) monad on another category $C$, along a functor $C \to D$. This
approach gives a highly generic way to lift monads to new categories,
abstracting away many details about the specific categories.  Codensity lifting
also gives a more principled construction in some sense, as the lifting
satisfies certain universal properties. However, the high level of abstraction
can make it difficult to construct and manipulate these liftings; the current,
clean form of Sato's lifting is followed significant simplifications after
applying codensity lifting.

More recent work generalizes witness-based approximate liftings to the
continuous case, giving an alternative, more flexible construction of
approximate liftings that is easier to work with. \citet{SBGHK17} introduce
\emph{span-based liftings}, generalizing binary relations to categorical spans
and supporting a broad class of divergences beyond $\varepsilon$-distance with
good composition properties.  Roughly speaking, maps between spans carry
additional information needed for smooth composition in the continuous case.
Sato and his collaborators develop span-based liftings and a relational program
logic to verify differential privacy and various relaxations, including
\emph{R\'enyi differential privacy} \citep{MironovRDP} and
\emph{zero-concentrated differential privacy} \citep{BunS2016}. When specialized
to $\varepsilon$-distance, span-based liftings are equivalent to Sato's
witness-free liftings, giving an approximate version of Strassen's theorem in
the continuous case.\footnote{Tetsuya Sato, personal communication.}

\section{Promising directions} \label{sec:future}

We envision further investigation along three broad axes: extending the theory,
exploring new applications, and automating the proof technique.

\subsection{Theoretical directions}

While the theory of probabilistic couplings has been well-developed in
mathematics, our work suggests several natural directions for further
theoretical study.

\paragraph*{Defining approximate couplings.}

Our definition of approximate lifting satisfies many clean theoretical
properties, but it is not yet clear whether we have arrived at the right
definition. More evidence is needed, possibly in the form of other natural
properties satisfied by approximate liftings, equivalences with other
well-established notions, or logical and categorical characterizations of
approximate coupling; analogous results for probabilistic liftings may provide a
useful
guide~\citep{LarsenS89,DESHARNAIS2002163,DESHARNAIS2003160,fijalkow_et_al:LIPIcs:2017:7368}.

Furthermore, \citet{BartheO13} and \citet{OlmedoThesis} consider approximate
liftings where the differential privacy distance
$\epsdist{\varepsilon}{\mu_1}{\mu_2}$ is generalized to any
$f$-\emph{divergence}, a broader class of distance-like measures between
distributions. While it is straightforward to adapt our approximate liftings to
$f$-divergences, there is currently little evidence this yields a good
definition; for instance, a universal version of approximate lifting (similar to
Sato's definition) for $f$-divergences is not known.

\paragraph*{Completeness of the proof systems.}

While the proof systems of \Sxprhl and \Saprhl are sound, we did not establish
completeness: valid judgments should be provable by applying the rules.  Much
like standard Hoare logic, the best we can hope for is \emph{relative}
completeness. Assuming an oracle for formulas in the assertion logic, can the
proof system prove all valid judgments?

On this fundamental question, very little is known. For \Sxprhl, relative
completeness of standard Hoare logic combined with some basic program
transformations give relative completeness for terminating, deterministic
programs.  However, the rules for random sampling are likely to be highly
incomplete; for instance, there are many couplings beyond bijection couplings.
Furthermore, there may be more fundamental obstacles to relative completeness:
\citet{DBLP:journals/rsa/KumarR01} give an example of a Markov chain that is
rapidly mixing but where no \emph{causal} coupling can establish this fact; all
couplings encoded by \Sxprhl are causal couplings. This negative result doesn't
directly rule out relative completeness since rapid mixing is not expressible in
the logic, but it does suggest that the underlying coupling proof technique may
be incomplete.

The situation is similar for \Saprhl. Our privacy proofs often use program
transformations to compensate for the incompleteness of the loop rules; these
transformations could potentially be avoided given more advanced loop rules or
richer reasoning about the privacy parameters $\varepsilon$ and $\delta$.
However, it is not clear what role the various structural rules (e.g.,
\nameref{rule:aprhl-pw-eq}) should play when proving completeness.

Enhancing our proof systems and identifying complete fragments for randomized
programs---or even more fundamentally, coming up with sensible notions of
completeness for coupling proofs---are intriguing and challenging directions for
future theoretical work.

\paragraph*{Connecting back to probabilistic bisimulation.}

Probabilistic liftings were first developed in the context of probabilistic
bisimulation~\citep{LarsenS89}; it would be interesting to revisit this rich
theory in light of our connections. Approximate couplings, which support a
multiplicative notion of approximation, appear to be new to the probabilistic
bisimulation literature.

\subsection{New applications}

The examples we have seen are drawn from classical coupling proofs in
mathematics. While these case studies concisely demonstrate various advanced
features of the proof technique, they are perhaps less well-motivated from the
perspective of program verification. However, now that formal verification can
leverage couplings, we can search for applications to typical verification
properties.

At the same time, there remains plenty of room to push the limits of coupling
proofs on more theoretical examples, especially using approximate couplings.
For example, we only applied approximate couplings for proving $(\varepsilon,
\delta)$-differential privacy; variants of approximate couplings for reasoning
about $f$-divergences, like KL-divergence, Hellinger distance, and $\chi^2$
divergence, currently lack concrete applications. Other natural targets include
relaxations of differential privacy like \emph{random differential privacy}
(RDP)~\citep{HallWassermanRinaldo}. For exact couplings, advanced constructions
like \emph{coupling from the past}~\citep{DBLP:journals/rsa/ProppW96} and
\emph{variable length path couplings}~\citep{DBLP:journals/rsa/HayesV07} may
suggest interesting ways to enrich relational reasoning. We expect theoretically
sophisticated examples will guide the development of formal verification for
probabilistic relational properties.

\subsection{Proof automation}

Throughout, we have presented program logic proofs on paper. Such proofs can be
formalized in existing prototype implementations of \Sprhl and \Saprhl in the
\sname{EasyCrypt} system~\citep{BartheDGKSS13}, an interactive proof assistant.
To make the proof technique more practical, however, more investigation is
needed into automating coupling proofs. By eliminating much of the probabilistic
reasoning, which pose significant challenges for automated solvers today,
coupling proofs may enable automated proofs for programs and properties where
even manual, interactive proofs would previously have been quite challenging.
Realizing these gains in practice is a natural direction for further
investigation.

\section{Bridging two theories} \label{sec:bridge}

This thesis represents a confluence of ideas from two theories: coupling proofs
from the theory of algorithms, and program logics from the theory of formal
verification. While mathematical rigor is a hallmark of both areas, the two
fields currently proceed on separate tracks. The theory of algorithms and
complexity investigates quantitative aspects of computation, like running time,
space usage, and degree of approximation, while the theory of semantics and
formal verification explores the compositional structure of programs and how to
reason about them. That there should be two distinct theoretical branches is
perhaps unsurprising; in many ways, the situation mirrors traditional divisions
between analysis and algebra in mathematics. However, what is more surprising is
the wide gulf between the two communities today. In many parts of the world, for
instance, semantics and verification don't fall under the umbrella term
Theoretical Computer Science (TCS).

Our results give a glimpse of the fruitful terrain that lies in between, and the
potential gains in applying perspectives and tools from both worlds. Formal
verification stands to benefit from understanding how humans reason about
algorithms, while algorithms and complexity theory could achieve simpler proofs
by generalizing properties and focusing on composition. The time is ripe to
bring these theories back into contact, and to see where the conversation leads.

\appendix

\chapter{Soundness of \texorpdfstring{\Sxprhl}{x\Sprhl}} \label{app:sound-xprhl}

\setcounter{section}{1}

We prove soundness of the logic \Sxprhl presented in \cref{chap:products},
consisting of the logical rules in
\cref{fig:xprhl-two-sided,fig:xprhl-one-sided,fig:xprhl-structural} and the
asynchronous loop rule in \cref{fig:xprhl-while}.

We will need a pair of technical lemmas. First, distribution bind commutes with
projections.

\begin{lemma} \label{lem:proj-bind}
  Let $i \in \{ 1, 2 \}$. Given $\mu \in \SDist(\cA_1 \times \cA_2)$ and $f :
  \cA_1 \times \cA_2 \to \SDist(\cB_1 \times \cB_2)$, suppose $g_i : \cA_i
  \to \SDist(\cB_i)$ is such that for all $(a_1, a_2) \in \supp(\mu)$, we have
  $\pi_i(f(a_1, a_2)) \leq g_i(a_i)$. Then
  \[
    \pi_i(\dbind( \mu, f )) \leq \dbind( \pi_i(\mu), g_i ) .
  \]
  Similarly, if for all $(a_1, a_2) \in \supp(\mu)$ we have $\pi_i(f(a_1, a_2))
  \geq g_i(a_i)$, then
  \[
    \pi_i(\dbind( \mu, f )) \geq \dbind( \pi_i(\mu), g_i ) .
  \]
\end{lemma}
\begin{proof}
  We consider the $\leq$ case with $i = 1$; the case $i = 2$ and the $\geq$
  cases are similar. Let $\eta \triangleq \pi_1(\dbind( \mu, f ))$. For any
  element $h \in \cB_1$,
  \begin{align*}
    \eta(h) &= \sum_{t \in \cB_2} \sum_{(r, s) \in \cA_1 \times \cA_2} \mu(r, s) \cdot f(r, s)(h, t)
    \\
    &= \sum_{(r, s) \in \supp(\mu)} \mu(r, s) \sum_{t \in \cB_2} f(r, s)(h, t)
    \\
    &\leq \sum_{(r, s) \in \supp(\mu)} \mu(r, s) \cdot g_1(r)(h)
    \\
    &= \sum_{(r, s) \in \cA_1 \times \cA_2} \mu(r, s) \cdot g_1(r)(h)
    \\
    &= \sum_{r \in \cA_1} \pi_1(\mu)(r) \cdot g_1(r)(h)
    \\
    &= \dbind(\pi_1(\mu), g_1)(h) . 
  \qedhere
  \end{align*}
\end{proof}

Second, projections commute with monotone limits.

\begin{lemma} \label{lem:proj-lim}
  Let $\{ \mu^{(i)} \}_i$ be a monotonically increasing sequence  in
  $\SDist(\cA_1 \times \cA_2)$ converging to a sub-distribution $\mu$. Then
  projections commute with limits:
  \[
    \pi_j\left(\lim_{i \to \infty} \mu^{(i)}\right)
    = \lim_{i \to \infty} \pi_j\left(\mu^{(i)}\right)
  \]
  for $j \in \{ 1, 2 \}$, and all limits exist.
\end{lemma}
\begin{proof}
  By unfolding definitions and applying the monotone convergence theorem, taking
  the discrete (counting) measure over $\Mem$ (see, e.g., \citet[Theorem
  11.28]{RudinPMA}).
\end{proof}

Now we can show soundness of \Sxprhl.

\xprhlsound*

\begin{proof}
  By induction on the height of the proof derivation. In the base case the
  derivation consists of a single rule with no \Sxprhl premises; this rule must
  be one of the axiom rules: \nameref{rule:xprhl-skip},
  \nameref{rule:xprhl-assn}, \nameref{rule:xprhl-sample}, or the one-sided
  variants. In the inductive case, the derivation ends in one of the other
  rules. By performing a case analysis on the last rule in the derivation, we
  handle the base and inductive cases together.
  
  We consider the two-sided rules first (\cref{fig:xprhl-two-sided}), followed
  by the one-sided rules (\cref{fig:xprhl-one-sided}), the structural rules
  (\cref{fig:xprhl-structural}), and finally the asynchronous loop rule
  (\cref{fig:xprhl-while}). Given soundness for the premises, we show the
  product program in the conclusion satisfies the support condition and the
  marginal conditions in \cref{def:xprhl-valid}. In all cases let $m_1, m_2$ be
  two memories that satisfy the pre-condition of the conclusion, let
  $\mu_\times$ be the output distribution of the product program with input
  $(m_1, m_2)$, and let $\mu_1 \triangleq \pi_1(\mu_\times)$ and $\mu_2
  \triangleq \pi_2(\mu_\times)$ be the two projections of the output
  distribution. We will leave the logical context $\rho$ implicit when taking
  the semantics $\denot{-}$; the logical variables play no role in the proof.

  For the loop rules, recall from \cref{fig:pwhile-sem} that the semantics of a
  loop $\WWhile{e}{c}$ on initial memory $m$ is defined as the limit of its
  finite approximants:
  \[
    \mu^{(i)}(m) \triangleq
    \begin{cases}
      \bot &: i = 0 \land \denot{e} m = \kwtrue \\
      \dunit(m) &: i = 0 \land \denot{e} m = \kwfalse \\
      \dbind(\denot{\Condt{e}{c}} m, \mu^{(i - 1)}) &: i > 0 .
    \end{cases}
  \]
  \begin{description}
    \item[Case \nameref{rule:xprhl-skip}]
      Trivial.
    \item[Case \nameref{rule:xprhl-assn}]
      The support condition is clear since all program variables in $e_1\sidel,
      e_2\sider$ are tagged with $\sidel, \sider$ respectively. The marginal
      conditions are clear as well: given any two input memories satisfying the
      pre-condition, the two output memories from $c_1$ and $c_2$ are point
      distributions where $x_1$ is updated to $e_1$ and $x_2$ is updated to
      $e_2$.
    \item[Case \nameref{rule:xprhl-sample}]
      The support of $\mu_\times$ lies in
      \[
        \{ (m_1', m_2') \mid \exists v,\ m_1'(x_1) = v \land m_2'(x_2) = f(v) \} .
      \]
      Since all output memories $(m_1', m_2')$ in the support are equal to the
      input memories $(m_1, m_2)$ on all variables besides $x_1$ and $x_2$, the
      support condition is clear.

      Now recall that all primitive distributions $d_1, d_2$ are uniform over
      finite sets. Hence $\supp(d_1)$ and $\supp(d_2)$ are finite, and since
      there is a bijection $f : \supp(d_1) \to \supp(d_2)$, the supports have
      the same size $n$. For every $v \in \supp(d_1)$, we have
      \[
        \mu_1(m_1[x_1 \mapsto v]) = 1/n
      \]
      and $\mu_1(m') = 0$ otherwise. By the semantics of the sampling command,
      $\mu_1 = \denot{c_1} m_1$ so the first marginal condition is satisfied.
      Since $f$ is injective, for every $v \in \supp(d_1)$ we have
      \[
        \mu_2(m_2[x_2 \mapsto f(v)]) = 1/n .
      \]
      and $\mu_2(m') = 0$ otherwise. Since $f$ is surjective, for every $v \in
      \supp(d_2)$ we have
      \[
        \mu_2(m_2[x_2 \mapsto v]) = 1/n ,
      \]
      giving $\mu_2 = \denot{c_2} m_2$ and the second marginal condition.
    \item[Case \nameref{rule:xprhl-seq}]
      Let the product programs in the premises be $c_\times$ and $c_\times'$. By
      induction, these product programs satisfy the support and marginal
      conditions for their respective judgments. To establish the conclusion,
      the support condition is clear: by induction, the support of
      $\denot{c_\times} (m_1, m_2)$ lies in $\denot{\Psi}$ and for any $(m_1',
      m_2)' \in \denot{\Psi}$, the support of $\denot{c_\times'} (m_1', m_2')$
      lies in $\denot{\Theta}$.

      It remains to show the marginal conditions. For $i \in \{ 1, 2 \}$,
      \begin{align}
        \mu_i &= \pi_i(\denot{c_\times; c_\times'} (m_1, m_2))
        \notag \\
        &=
        \pi_i(\dbind(\denot{c_\times} (m_1, m_2), \denot{c_\times'}))
        \tag{semantics} \\
        &=
        \dbind(\pi_i(\denot{c_\times} (m_1, m_2)), \denot{c_i'})
        \tag{\cref{lem:proj-bind} and induction} \\
        &=
        \dbind(\denot{c_i} m_i, \denot{c_i'})
        \tag{induction} \\
        &= \denot{c_i; c_i'} m_i
        \tag{semantics} .
      \end{align}
    \item[Case \nameref{rule:xprhl-cond}]
      Let $c_\times, c_\times'$ be the two product programs for the two
      premises. There are two cases: either $e_1$ is true in the first initial
      memory $m_1$, or not. (Since $(m_1, m_2)$ satisfy the pre-condition
      $\Phi$, these two cases correspond to $e_2$ being true and false in the
      second initial memory $m_2$.)

      Suppose $e_1$ is true in $m_1$. Then $e_1\sidel$ is true in $(m_1, m_2)$
      and the product program is equivalent to simply executing $c_\times$ on
      $(m_1, m_2)$. Since the two initial memories $(m_1, m_2)$ satisfy $\Phi$,
      by induction on the first premise, the support of the product program lies
      in $\denot{\Psi}$ and the marginals satisfy
      \[
        \mu_1 = \pi_1(\denot{c_1} m_1)
        \quad \text{and} \quad
        \mu_2 = \pi_2(\denot{c_2} m_2) .
      \]
      Since $e_1 \sidel$ and $e_2 \sider$ are both true in $(m_1, m_2)$, we
      also have
      \[
        \mu_1 = \pi_1(\denot{\Cond{e_1}{c_1}{c_1'}} m_1)
        \quad \text{and} \quad
        \mu_2 = \pi_2(\denot{\Cond{e_2}{c_2}{c_2'}} m_2) .
      \]
      Hence, both the marginal and support conditions hold when $e_1$ is true in
      $m_1$.

      The other case, where $e_1\sidel$ and $e_2\sider$ are false in $(m_1,
      m_2)$, follows by the second premise.
    \item[Case \nameref{rule:xprhl-while}]
      Let the product program in the conclusion be
      $\WWhile{e_1\sidel}{c_\times}$ and let $\mu^{(i)}(m_1, m_2)$ be its
      $i$-th approximants. Define $\mu^{(i)}_1 \triangleq \pi_1 \circ \mu^{(i)},
      \mu^{(i)}_2 \triangleq \pi_2 \circ \mu^{(i)}$ to be the first and second
      marginals of the approximants, and $\eta^{(i)}_1, \eta^{(i)}_2$ to be the
      $i$-th approximants of the loops $\WWhile{e_1}{c_1}$ and
      $\WWhile{e_2}{c_2}$, respectively.

      Let's consider the support condition first. We prove if $(m_1, m_2)$
      satisfies $\Phi$, then $\mu^{(i)}(m_1, m_2)$ has support contained in
      $\denot{\Phi \land \neg e_1\sidel}$ for every $i$ by induction.  The base
      case $i = 0$ is clear.  For the inductive step $i > 0$ there are two
      cases. If $e_1\sidel$ is false in $(m_1, m_2)$, then $\mu^{(i)}(m_1, m_2)
      = \dunit(m_1, m_2)$. Otherwise if $e_1\sidel$ is true, then
      \[
        \mu^{(i)}(m_1, m_2)
        = \dbind(\denot{c_\times} (m_1, m_2), \mu^{(i-1)}) .
      \]
      By the outer induction hypothesis applied to the premise, the support of
      $\denot{c_\times}(m_1, m_2)$ lies in $\denot{\Phi \land \neg e_1\sidel}$.
      The inner induction hypothesis applied to $\mu^{(i-1)}$ shows
      $\mu^{(i)}(m_1, m_2)$ also has support in $\denot{\Phi \land \neg
      e_1\sidel}$, completing the inner induction. Since this holds for all $i$,
      the limit sub-distribution
      \[
        \lim_{i \to \infty} \mu^{(i)}(m_1, m_2)
        = \denot{\WWhile{e_1\sidel}{c_\times}} (m_1, m_2)
      \]
      also has support in $\denot{\Phi \land \neg e_1\sidel}$ as desired.

      Next, we turn to the marginal conditions. We first show the
      projections of the approximants of the product program are equal to the
      approximants for the individual programs, concluding the marginal
      conditions in the limit. Let $(m_1, m_2)$ be memories satisfying $\Phi$.
      We claim $\pi_1( \mu^{(i)}(m_1, m_2)) = \eta_1^{(i)}(m_1)$ and
      $\pi_2(\mu^{(i)}(m_1, m_2)) = \eta_2^{(i)}(m_2)$ for every $i$.

      The claim follows by induction on $i$. The base case $i = 0$ is
      immediate---since $(m_1, m_2)$ satisfy $\Phi$, either $e_1\sidel =
      e_2\sider = \kwtrue$ or $e_1\sidel = e_2\sider = \kwfalse$.  The inductive
      step $i > 0$ is more interesting. Unrolling the approximants one step, we
      have
      \begin{align*}
        \mu^{(i)}(m_1, m_2)
        &= \dbind(\denot{\Condt{e_1\sidel}{c_\times}} (m_1, m_2), \mu^{(i - 1)}) \\
        \eta^{(i)}_1(m_1)
        &= \dbind(\denot{\Condt{e_1}{c_1}} m_1, \eta^{(i - 1)}_1) \\
        \eta^{(i)}_2(m_2)
        &= \dbind(\denot{\Condt{e_2}{c_2}} m_2, \eta^{(i - 1)}_2) .
      \end{align*}
      If $e_1\sidel$ is false in $(m_1, m_2)$, then all three conditionals
      are equivalent to $\Skip$ so
      \[
        \mu^{(i)} = \mu^{(i - 1)}
        \quad \text{and} \quad
        \eta^{(i)}_1 = \eta^{(i - 1)}_1
        \quad \text{and} \quad
        \eta^{(i)}_2 = \eta^{(i - 1)}_2 ;
      \]
      we conclude by the inductive hypothesis. Otherwise, $e_1\sidel$ and
      $e_2\sider$ are true in $(m_1, m_2)$ so the same branch is taken in all
      three approximants:
      \begin{align*}
        \mu^{(i)}(m_1, m_2)
        &= \dbind(\denot{c_\times} (m_1, m_2), \mu^{(i - 1)}) \\
        \eta^{(i)}_1(m_1)
        &= \dbind(\denot{c_1} m_1, \eta^{(i - 1)}_1) \\
        \eta^{(i)}_2(m_2)
        &= \dbind(\denot{c_2} m_2, \eta^{(i - 1)}_2) .
      \end{align*}
      By the outer induction hypothesis on the premise of the rule (noting that
      $e_1$ is true in $m_1$),
      \[
        \pi_1(\denot{c_\times} (m_1, m_2)) = \denot{c_1} m_1
        \qquad\text{and}\qquad
        \pi_2(\denot{c_\times} (m_1, m_2)) = \denot{c_2} m_2 .
      \]
      By the inner induction hypothesis, $\mu^{(i - 1)}(m_1, m_2)$ has
      projections $\eta^{(i -1)}_1(m_1)$ and $\eta^{(i - 1)}_2(m_2)$, so
      \cref{lem:proj-bind} gives
      \[
        \pi_1(\mu^{(i)}(m_1, m_2)) = \eta^{(i)}_1(m_1)
        \quad \text{and} \quad
        \pi_2(\mu^{(i)}(m_1, m_2)) = \eta^{(i)}_2(m_2)
      \]
      for every $i$. Taking the limit as $i$ tends to $\infty$, we have the
      marginal conditions
      \begin{align*}
        \denot{\WWhile{e_j}{c_j}}(m_j)
        &\triangleq \lim_{i \to \infty} \eta^{(i)}_j(m_j) \\
        &= \lim_{i \to \infty} \pi_j(\mu^{(i)}(m_j)) \\
        &= \pi_j\left(\lim_{i \to \infty} \mu^{(i)}(m_j)\right) \\
        &\triangleq \pi_j(\denot{\WWhile{e_1\sidel}{c_\times}}(m_1, m_2))
      \end{align*}
      for $j = \{ 1, 2 \}$. (We may interchange marginals and limits by
      \cref{lem:proj-lim} since $\{ \mu^{(i)}(m_j) \}_i$ is monotonically
      increasing by definition.)
    \item[Case \nameref{rule:xprhl-assn-l} (\nameref{rule:xprhl-assn-r} similar)]
      Trivial.
    \item[Case \nameref{rule:xprhl-sample-l} (\nameref{rule:xprhl-sample-r} similar)]
      Let $d_1$ have support with size $n$. The support condition is clear. For
      the marginal condition, note
      \[
        \mu_\times(m_1[x_1 \mapsto v], m_2) = 1/n
      \]
      for every $v \in \supp(d_1)$, and zero otherwise. Hence,
      \[
        \mu_1(m_1([x_1 \mapsto v]) = 1/n
      \]
      for every $v \in \supp(d_1)$, and zero otherwise, while $\mu_2$ is the
      point distribution at $m_2$. The semantics of $\Rand{x_1}{d_1}$ and
      $\Skip$ gives the marginal conditions.
    \item[Case \nameref{rule:xprhl-cond-l} (\nameref{rule:xprhl-cond-r} similar)]
      There are two cases: either $e_1\sidel$ is true in $(m_1, m_2)$, or not.
      On input $(m_1, m_2)$, the product program has the same semantics as $c$
      and $c'$ in the respective cases, hence the support condition follows by
      induction using the support condition in the first and second premises
      respectively.

      The marginal conditions are similar. If $e_1 \sidel$ is true in $(m_1,
      m_2)$, then the product program has the same semantics as $c$, and the
      first program $\Cond{e_1}{c_1}{c_1'}$ has the same semantics as $c_1$.
      Hence, the marginal conditions follow by induction using the marginal
      condition from the first premise.

      In the other case, $e_1 \sidel$ is false in $(m_1, m_2)$ and the product
      program has the same semantics as $c'$, and on $m_1$ the first program
      $\Cond{e_1}{c_1}{c_1'}$ has the same semantics as $c_1'$. Hence, the
      marginal conditions follow by induction using the marginal condition from
      the second premise.
    \item[Case \nameref{rule:xprhl-while-l} (\nameref{rule:xprhl-while-r} similar)]
      Let the final product program be $\WWhile{e_1\sidel}{c_\times}$ with
      $i$-th approximants $\mu^{(i)}(m_1, m_2)$. Define $\mu^{(i)}_1 \triangleq
      \pi_1 \circ \mu^{(i)}, \mu^{(i)}_2 \triangleq \pi_2 \circ \mu^{(i)}$ to be
      the first and second marginals of the approximants, and $\eta^{(i)}$ to be
      the $i$-th approximants of the loop $\WWhile{e_1}{c_1}$.

      For the support condition, we show if $(m_1, m_2)$ satisfies $\Phi$ then
      $\mu^{(i)}(m_1, m_2)$ has support contained in $\denot{\Phi \land \neg
      e_1\sidel}$ for every $i$ by induction on $i$.  The base case $i = 0$ is
      clear.  For the inductive step $i > 0$, there are two cases. If
      $e_1\sidel$ is false in $(m_1, m_2)$, then $\mu^{(i)}(m_1, m_2) =
      \dunit(m_1, m_2)$ and we are done. Otherwise if $e_1\sidel$ is true, then
      \[
        \mu^{(i)}(m_1, m_2)
        = \dbind(\denot{c_\times} (m_1, m_2), \mu^{(i-1)}) .
      \]
      By the outer induction hypothesis applied to the premise, the support of
      $\denot{c_\times}(m_1, m_2)$ lies in $\denot{\Phi \land \neg e_1\sidel}$.
      The inner induction hypothesis applied to $\mu^{(i-1)}$ implies
      $\mu^{(i)}(m_1, m_2)$ also has support contained in $\denot{\Phi \land
      \neg e_1\sidel}$, completing the inner induction. Since this is true for
      all $i$, the limit sub-distribution
      \[
        \lim_{i \to \infty} \mu^{(i)}(m_1, m_2)
        = \denot{\WWhile{e_1\sidel}{c_\times}} (m_1, m_2)
      \]
      also has support in $\denot{\Phi \land \neg e_1\sidel}$ as desired.

      Now we turn to the marginal conditions. We show the projections
      of the approximant of the product program are equal to the approximants
      for the individual programs, concluding the marginal conditions in the
      limit. Let $(m_1, m_2)$ be any memories satisfying $\Phi$. We claim
      $\pi_1(\mu^{(i)}(m_1, m_2)) = \eta_1^{(i)}(m_1)$ and $\pi_2(\mu^{(i)}(m_1,
      m_2))$ is a point sub-distribution with all mass on $m_2$, for every $i$.

      The claim follows by induction on $i$. The base case $i = 0$ is
      clear---$e_1\sidel$ is either true or false. If $e_1\sidel$ is true,
      $\mu^{(i)} = \dunit(m_1, m_2)$, $\eta^{(i)}_1 = \dunit(m_1)$, and
      $\pi_2(\mu^{(i)}) = \dunit(m_2)$. If $e_1\sidel$ is false, then all
      approximants are the zero sub-distribution $\bot$.

      The inductive step $i > 0$ is more interesting. Unrolling the approximants
      one step, we have
      \begin{align*}
        \mu^{(i)}(m_1, m_2)
        &= \dbind(\denot{\Condt{e_1\sidel}{c_\times}} (m_1, m_2), \mu^{(i - 1)}) \\
        \eta^{(i)}_1(m_1)
        &= \dbind(\denot{\Condt{e_1}{c_1}} m_1, \eta^{(i - 1)}_1) .
      \end{align*}
      If $e_1\sidel$ is false in $(m_1, m_2)$, then both conditionals are
      equivalent to $\Skip$. Hence
      \[
        \mu^{(i)} = \mu^{(i - 1)}
        \quad \text{and} \quad
        \eta^{(i)}_1 = \eta^{(i - 1)}_1 ,
      \]
      and we conclude by the induction on $i$. Otherwise, $e_1\sidel$ is true in
      $(m_1, m_2)$. In this case, the conditional branch is taken in both
      programs, so
      \begin{align*}
        \mu^{(i)}(m_1, m_2)
        &= \dbind(\denot{c_\times} (m_1, m_2), \mu^{(i - 1)}) \\
        \eta^{(i)}_1(m_1)
        &= \dbind(\denot{c_1} m_1, \eta^{(i - 1)}_1) .
      \end{align*}
      By the outer induction hypothesis on the premise of the rule (noting that
      $e_1\sidel$ is true),
      \begin{align*}
        \pi_1(\denot{c_\times} (m_1, m_2)) &= \denot{c_1} m_1 \\
        \pi_2(\denot{c_\times} (m_1, m_2)) &= \denot{\Skip} m_2 = \dunit(m_2) .
      \end{align*}
      The inner induction hypothesis shows the first marginal of $\mu^{(i
      - 1)}(m_1, m_2)$ is $\eta^{(i -1)}_1(m_1)$; \cref{lem:proj-bind}
      establishes $\pi_1(\mu^{(i)}(m_1, m_2)) = \eta_1^{(i)}(m_1)$. Similarly,
      by the inner induction hypothesis showing the second marginal of $\mu^{(i
      - 1)}(m_1, m_2)$ is a point mass at $m_2$, we establish the same for the
      second marginal of $\mu^{(i)}(m_1, m_2)$.  Furthermore, since the weight
      of a sub-distribution is preserved under projections, we also know
      $\pi_2(\mu^{(i)}(m_1, m_2))$ is a point sub-distribution at $m_2$ with
      weight $|\mu^{(i)}(m_1, m_2)|$.

      Now we take limits to obtain the first marginal condition:
      \[
        \eta_1(m_1)
        = \lim_{i \to \infty} \eta^{(i)}_1(m_j)
        = \lim_{i \to \infty} \pi_1(\mu^{(i)}(m_1, m_2))
        = \pi_1\left(\lim_{i \to \infty} \mu^{(i)}(m_1, m_2)\right)
        = \pi_1(\mu_\times) ,
      \]
      interchanging limits and projections by \cref{lem:proj-lim}, since $\{
      \mu^{(i)}(m_1, m_2) \}_i$ is monotonically increasing.

      For the second marginal we have
      \[
        \pi_2(\mu_\times)
        = \dunit(m_2) \cdot |\mu_\times| .
      \]
      By the premise, the loop $\WWhile{e_1}{c_1}$ is lossless. Hence,
      \[
        1 = |\eta_1(m_1)| = |\mu_\times|
      \]
      and the second projection of $\mu_\times$ is simply $\dunit(m_2) =
      \denot{\Skip} m_2$ as claimed.
    \item[Case \nameref{rule:xprhl-conseq}]
      Trivial.
    \item[Case \nameref{rule:xprhl-equiv}]
      Trivial.
    \item[Case \nameref{rule:xprhl-case}]
      By case analysis on whether $e$ is true in $(m_1, m_2)$, using essentially
      the same reasoning as in \nameref{rule:xprhl-case},
      \nameref{rule:xprhl-cond-l}, or \nameref{rule:xprhl-cond-r}.
    \item[Case \nameref{rule:xprhl-frame}]
      The marginal conditions are clear by induction. Let $V$ be the set of
      variables that are not in $\MV(c)$. Since $\Theta$ has free variables in
      $V$, we can interpret $\Theta$ as a predicate on memories restricted to
      $V$. Then initially $(m_1[V], m_2[V]) \in \denot{\Theta}$. Since $c$ does
      not modify variables in $V$, the support of $\mu_\times$ is contained in
      \[
        \{ (m_1', m_2') \mid m_1'[V] = m_1[V] \land m_2'[V] = m_2[V] \}
        \subseteq \denot{\Theta} .
      \]
      Hence the support condition is satisfied as well.
    \item[Case \nameref{rule:xprhl-while-gen}]
      We label the premises for easy reference:
      \begin{align}
        &\models \Phi \to (e_1\sidel \lor e_2\sider) = e
        \label{prem:while-gen:guard}
        \\
        &\models \Phi \land e \to p_0 \oplus p_1 \oplus p_2
        \label{prem:while-gen:xor}
        \\
        &\models \Phi \land p_0 \land e \to e_1\sidel = e_2\sider
        \label{prem:while-gen:consist:sync}
        \\
        &\models \Phi \land p_1 \land e \to e_1\sidel \land \Phi_1\sidel
        \label{prem:while-gen:consist:left}
        \\
        &\models \Phi \land p_2 \land e \to e_2\sider \land \Phi_2\sider
        \label{prem:while-gen:consist:right}
        \\
        &\lless{\Phi_1}{\WWhile{e_1 \land p_1}{c_1}} 
        \label{prem:while-gen:ll:left}
        \\
        &\lless{\Phi_2}{\WWhile{e_2 \land p_2}{c_2}}
        \label{prem:while-gen:ll:right}
        \\
        &\vdash \xprhl{(\Condt{e_1}{c_1})^{K_1}}
        {(\Condt{e_2}{c_2})^{K_2}}
        {\Phi \land e \land p_0}{\Phi}{c_0'}
        \label{prem:while-gen:xp:sync}
        \\
        &\vdash \xprhl{c_1}{\Skip}{\Phi \land e_1\land p_1}{\Phi}{c_1'}
        \label{prem:while-gen:xp:left}
        \\
        &\vdash \xprhl{\Skip}{c_2}{\Phi \land e_2\land p_2}{\Phi}{c_2'}
        \label{prem:while-gen:xp:right}
      \end{align}
      Let $\theta_\times$ be the semantics of the product program in the
      conclusion and let $\theta^{(i)}$ be its $i$-th approximants. For the
      support condition, we first show
      \[
        \supp(\theta^{(i)}(a_1, a_2))
        \subseteq \denot{\Phi \land \neg e_1\sidel \land \neg e_2\sider}
      \]
      for every $i$ and $(a_1, a_2)$ satisfying $\Phi$. The proof is by
      induction on $i$. The base case $i = 0$ is clear: if $e$ is false in
      $(a_1, a_2)$ then $\theta^{(0)}(a_1, a_2) = \dunit(a_1, a_2)$, otherwise
      if $e$ is true then $\theta^{(0)}(a_1, a_2) = \bot$, so in both cases we
      have the desired support.

      For the inductive step $i > 0$, if $e$ is false in $(a_1, a_2)$ then
      $\theta^{(i)}(a_1, a_2) = \dunit(a_1, a_2)$ and the support condition is
      satisfied. Otherwise, we unfold the product program one step giving
      three cases:
      \[
        \theta^{(i)}(a_1, a_2) =
        \begin{cases}
          \dbind( \denot{c_0'} (a_1, a_2), \theta^{(i - 1)} ) &:
          \denot{p_0} (a_1, a_2) = \kwtrue \\
          \dbind( \denot{c_1'} (a_1, a_2), \theta^{(i - 1)} ) &:
          \denot{p_1} (a_1, a_2) = \kwtrue \\
          \dbind( \denot{c_2'} (a_1, a_2), \theta^{(i - 1)} ) &:
          \denot{p_2} (a_1, a_2) = \kwtrue .
        \end{cases}
      \]
      Exactly one of the three cases holds, by \cref{prem:while-gen:xor}.  By
      the outer induction hypothesis, the premises of the rule
      (\cref{prem:while-gen:xp:sync,prem:while-gen:xp:left,prem:while-gen:xp:right})
      show that in the three cases, the corresponding product program $c_0',
      c_1', c_2'$ on input memory $(a_1, a_2)$ produces a sub-distribution with
      support in $\denot{\Phi}$. Hence $\theta^{(i)}(a_1, a_2)$ has the desired
      support using the inner induction hypothesis on $\theta^{(i - 1)}$.
      Passing to the limit, we conclude the support condition:
      \[
        \supp(\theta_\times(m_1, m_2))
        = \supp\left( \lim_{i \to \infty} \theta^{(i)}(m_1, m_2) \right)
        \subseteq \denot{\Phi \land \neg e_1\sidel \land \neg e_2\sider} .
      \]

      Next, we turn to the marginal conditions. Let $\eta_1, \eta_2 : \Mem \to
      \SDist(\Mem)$ be the semantics of the loops $\WWhile{e_1}{c_1}$ and
      $\WWhile{e_2}{c_2}$, and let $\eta^{(i)}_1, \eta^{(i)}_2 : \Mem \to
      \SDist(\Mem)$ be their respective $i$-th approximants. We show for
      every $i$ and every $(a_1, a_2)$ satisfying the invariant $\Phi$, we have
      \begin{align}
        \eta^{(i)}_1(a_1) &\leq \pi_1(\theta_\times(a_1, a_2))
        \label{eq:whilegen-marg1a}
        \\
        \pi_1(\theta^{(i)}(a_1, a_2)) &\leq \eta_1(a_1)
        \label{eq:whilegen-marg1b}
        \\
        \eta^{(i)}_2(a_2) &\leq \pi_2(\theta_\times(a_1, a_2))
        \label{eq:whilegen-marg2a}
        \\
        \pi_2(\theta^{(i)}(a_1, a_2)) &\leq \eta_2(a_2) .
        \label{eq:whilegen-marg2b}
      \end{align}
      Taking limits as $i$ tends to infinity will give the desired marginal
      conditions.

      We begin with \cref{eq:whilegen-marg1a} by induction on $i$. For the base
      case $i = 0$, if $e$ is false in $(a_1, a_2)$ then both sides are equal to
      $\dunit(a_1)$. Otherwise, if $e$ and $e_1\sidel$ are true, then both sides
      are equal to $\bot$. Finally, if $e$ is true and $e_1\sidel$ is false,
      then $\eta^{(0)}_1(a_1) = \dunit(a_1)$ by \cref{prem:while-gen:guard}. In
      this case, $e_2\sider$ must be true. By
      \cref{prem:while-gen:consist:right}, we are in case $p_2$ and the product
      program executes $c_2'$. By the marginal condition from premise
      \cref{prem:while-gen:xp:right}, $c_2'$ preserves $a_1$ so $e_1\sidel$
      remains false. Hence,
      \[
        \theta_\times(a_1, a_2) = \denot{\WWhile{e_2\sider \land p_2}{c_2'}} (
        a_1, a_2 ) .
      \]
      By reasoning analogous to the case \nameref{rule:xprhl-while-r} with
      \cref{prem:while-gen:ll:right} and the outer inductive hypothesis on
      \cref{prem:while-gen:xp:right}, we have the marginal condition
      \[
        \pi_1(\theta_\times(a_1, a_2)) = \dunit(a_1) = \eta^{(0)}_1(a_1)
      \]
      establishing the base case.

      Next, we consider the inductive case $i > 0$. If $e$ is false in $(a_1 ,
      a_2)$, then $e_1\sidel$ is false in $a_1$ and hence $\eta^{(i)}_1(a_1) =
      \pi_1(\theta_\times(a_1, a_2)) = \dunit(a_1)$. Otherwise, $e$ is true and
      there are three subcases.

      \begin{description}
        \item[Subcase for \cref{eq:whilegen-marg1a}: $p_0$ is true.]
          If $p_0$ is true in $(a_1, a_2)$, then $e_1\sidel = e_2\sider$ are
          also true by \cref{prem:while-gen:consist:sync}. First, suppose $i =
          K_1$. Unrolling the loops gives
          \begin{align*}
            \theta_\times(a_1, a_2)
            &= \dbind( \denot{c_0'} ( a_1, a_2 ), \theta_\times ) \\
            \eta^{(K_1)}_1(a_1)
            &= \dbind( \denot{(\Condt{e_1}{c_1})^{K_1}} a_1, \eta^{(0)}_1 ) .
          \end{align*}
          The marginal condition from the induction hypothesis on
          premise \cref{prem:while-gen:xp:sync} gives
          \[
            \pi_1(\denot{c_0'} ( a_1, a_2)) = \denot{(\Condt{e_1}{c_1})^{K_1}}( a_1 ) ;
          \]
          by the support condition, $\supp(\denot{c_0'}( a_1, a_2 )) \subseteq
          \denot{\Phi}$.  Furthermore, the base case for the inner induction
          yields $\eta^{(0)}_1(b_1) \leq \pi_1(\theta_\times(b_1, b_2))$ for
          every $(b_1, b_2) \in \denot{\Phi}$, so \cref{lem:proj-bind} gives
          \[
            \eta^{(K_1)}(a_1) \leq \pi_1(\theta_\times(a_1, a_2)) .
          \]
          Now suppose $i < K_1$. From the previous case and monotonicity, we
          have
          \[
            \eta^{(i)}_1(a_1)
            \leq \eta^{(K_1)}_1(a_1)
            \leq \pi_1(\theta_\times(a_1, a_2)) .
          \]
          With the cases $i \leq K_1$ covered, we turn to the remaining cases $i
          > K_1$. Unrolling the loops:
          \begin{align*}
            \theta_\times(a_1, a_2)
            &= \dbind( \denot{c_0'}( a_1, a_2 ), \theta_\times ) \\
            \eta^{(i)}_1(a_1)
            &= \dbind( \denot{(\Condt{e_1}{c_1})^{K_1}} a_1, \eta^{(i - K_1)}_1) .
          \end{align*}
          The marginal condition from the induction hypothesis on premise
          \cref{prem:while-gen:xp:sync} gives
          \[
            \pi_1(\denot{c_0'} ( a_1, a_2)) = \denot{(\Condt{e_1}{c_1})^{K_1}} a_1 ;
          \]
          by the support condition, we have $\denot{c_0'} ( a_1, a_2) \subseteq
          \denot{\Phi}$.  Furthermore, by the inner inductive hypothesis for
          $\eta^{(i - K_1)}$ we have $\eta^{(i - K_1)}_1(b_1) \leq
          \pi_1(\theta_\times(b_1, b_2))$ for every $(b_1 , b_2) \in
          \denot{\Phi}$, so \cref{lem:proj-bind} shows
          \[
            \eta^{(i)}(a_1) \leq \pi_1(\theta_\times(a_1, a_2))
          \]
          as desired.

        \item[Subcase for \cref{eq:whilegen-marg1a}: $p_1$ is true.]
          If $p_1$ is true in $(a_1, a_2)$, then $e_1\sidel$ is also true by
          \cref{prem:while-gen:consist:left}.  Unrolling the loops:
          \begin{align*}
            \theta_\times(a_1, a_2)
            &= \dbind( \denot{c_1'} ( a_1, a_2 ), \theta_\times ) \\
            \eta^{(i)}_1(a_1)
            &= \dbind( \denot{c_1} a_1, \eta^{(i - 1)}_1 ) .
          \end{align*}
          The induction hypothesis on premise \cref{prem:while-gen:xp:left}
          gives $\pi_1(\denot{c_1'} ( a_1, a_2 )) = \denot{c_1} a_1$; by the
          support condition, we also have $\denot{c_1'} ( a_1, a_2 ) \subseteq
          \denot{\Phi}$.  Furthermore, by the inner induction hypothesis we have
          $\eta^{(i-1)}_1(b_1) \leq \pi_1(\theta_\times(b_1, b_2))$ for every
          $(b_1, b_2) \in \denot{\Phi}$, so \cref{lem:proj-bind} yields
          \[
            \eta^{(i)}(a_1) \leq \pi_1(\theta_\times(a_1, a_2)) .
          \]

        \item[Subcase for \cref{eq:whilegen-marg1a}: $p_2$ is true.]
          If $p_2$ is true in $(a_1, a_2)$, then $e_2\sider$ is true by
          \cref{prem:while-gen:consist:right}. Define
          \[
            \eta \triangleq \denot{\WWhile{e \land p_2}{c_2'}} ( a_1, a_2 ) .
          \]
          We show the equivalence
          \begin{equation} \label{clm:e-to-e2}
            \denot{\WWhile{e \land p_2}{c_2'}} ( a_1, a_2 )
            = \denot{\WWhile{e_2\sider \land p_2}{c_2'}} ( a_1, a_2 )
          \end{equation}
          by taking the approximants $\sigma^{(i)}$ and $\tau^{(i)}$ of the
          left and right sides and proving
          \[
            \sigma^{(i)}( b_1, b_2 ) = \tau^{(i)}( b_1, b_2 )
          \]
          for every $(b_1, b_2) \in \denot{\Phi}$. The proof is by induction on
          $i$, using $\supp(\tau^{(i)})( b_1, b_2 ) \subseteq \denot{\Phi}$ from
          the support condition from premise \cref{prem:while-gen:xp:right}, and
          \[
            \models \Phi \to (e \land p_2 \leftrightarrow e_2\sider \land p_2)
          \]
          from \cref{prem:while-gen:guard,prem:while-gen:consist:right}.
          
          Using the equivalence \cref{clm:e-to-e2}, we can transform $\eta$ and
          show the following:
          \begin{align}
            \supp(\eta) &\subseteq \denot{\Phi \land \neg(e \land p_2)}
            \label{clm:eta-supp} \\
            \pi_1(\eta) &= \dunit(a_1)
            \label{clm:eta-proj-unit}
          \end{align}
          Both points follow by reasoning similar to the case for
          \nameref{rule:xprhl-while-r}, using premise
          \cref{prem:while-gen:xp:right} and the lossless condition
          \cref{prem:while-gen:ll:right}. The first point also uses $\supp(\eta)
          \subseteq \denot{\neg (e \land p_2)}$, by definition of $\eta$. 
          
          Returning to the sub-case, if $e_1\sidel$ is true, unrolling the
          product program gives
          \[
            \theta_\times(a_1, a_2)
            = \dbind( \denot{\WWhile{e \land p_2}{c_2'}} ( a_1, a_2 ), \theta_\times ) .
          \]
          Since the guard $e \land p_2$ is false in $\supp(\eta)$ and $e_1$ is
          true in the first initial memory $a_1$, \cref{clm:eta-supp} gives
          \[
            \supp(\eta)
            \subseteq \denot{\Phi \land e_1\sidel \land \neg(e \land p_2)}
            \subseteq \denot{\Phi \land e \land \neg p_2}
          \]
          where the second inclusion is because $e_1$ implies $e$ (by
          \cref{prem:while-gen:guard}), so $p_2$ must be false. By
          \cref{prem:while-gen:xor} either $p_0$ or $p_1$ must be true in the
          support of $\eta$. Using \cref{clm:eta-proj-unit} to show $\eta(r_1,
          r_2) > 0$ only when $r_1 = a_1$, we compute:
          \begin{align*}
            \pi_1(\theta_\times(a_1, a_2))
            &= \pi_1 \left(\sum_{(r_1, r_2) \in \denot{\Phi \land e \land \neg p_2}}
            \eta(r_1, r_2) \cdot \theta_\times(r_1, r_2)\right) \\
            &= \pi_1 \left(\sum_{r_2 : (a_1, r_2) \in \denot{\Phi \land e \land \neg p_2}}
            \eta(a_1, r_2) \cdot \theta_\times(a_1, r_2)\right) \\
            &= \sum_{r_2 : (a_1, r_2) \in \denot{\Phi \land e \land \neg p_2}}
            \eta(a_1, r_2) \cdot \pi_1(\theta_\times(a_1, r_2)) \\
            &= \sum_{r_2 : (a_1, r_2) \in \denot{\Phi \land e \land p_0}}
            \eta(a_1, r_2) \cdot \pi_1(\theta_\times(a_1, r_2)) \\
            &+ \sum_{r_2 : (a_1, r_2) \in \denot{\Phi \land e \land p_1}}
            \eta(a_1, r_2) \cdot \pi_1(\theta_\times(a_1, r_2)) \\
            &\geq \sum_{r_2 : (a_1, r_2) \in \denot{\Phi \land e \land p_0}}
            \eta(a_1, r_2) \cdot \eta^{(i)}_1(a_1)
            + \sum_{r_2 : (a_1, r_2) \in \denot{\Phi \land e \land p_1}}
            \eta(a_1, r_2) \cdot \eta^{(i)}_1(a_1) \\
            &= \left( \sum_{r_2 : (a_1, r_2) \in \denot{\Phi}}
            \eta(a_1, r_2) \right) \cdot \eta^{(i)}_1(a_1)
          \end{align*}
          where on the third line we interchange projection and the sum by
          \cref{lem:proj-lim}, and the inequality is from the cases where $p_0$
          and $p_1$ are true. \Cref{clm:eta-supp,clm:eta-proj-unit} show
          \[
            \sum_{r_2 : (a_1, r_2) \in \denot{\Phi}} \eta(a_1, r_2)
            = |\pi_1(\eta)|
            = 1
          \]
          and so $\pi_1(\theta_\times(a_1, a_2)) \geq \eta^{(i)}_1(a_1)$ as desired.

          If $e_1\sidel$ is false, unrolling the loops gives
          \begin{align*}
            \theta_\times(a_1, a_2)
            &= \dbind( \denot{\WWhile{e \land p_2}{c_2'}} ( a_1, a_2 ), \theta_\times ) \\
            \eta^{(i)}_1(a_1)
            &= \dbind( \denot{\Skip} a_1, \eta^{(i - 1)}_1 ) .
          \end{align*}
          \Cref{clm:eta-proj-unit} implies
          \[
            \pi_1(\denot{\WWhile{e \land p_2}{c_2'}} ( a_1, a_2 ))
            = \pi_1(\eta)
            = \dunit(a_1)
            = \denot{\Skip} a_1 .
          \]
          Furthermore \cref{clm:eta-supp} implies $\supp(\eta) \subseteq
          \denot{\Phi}$, so we apply \cref{lem:proj-bind} with the induction
          hypothesis $\eta^{(i-1)}_1(b_1) \leq \theta_\times(b_1, b_2)$ for all
          $(b_1, b_2) \in \denot{\Phi}$ to conclude
          \[
            \eta^{(i)}_1(a_1) \leq \pi_1(\theta_\times(a_1, a_2)) .
          \]
      \end{description}
      This completes the inductive case $i > 0$, establishing \cref{eq:whilegen-marg1a}.

      Next, we establish \cref{eq:whilegen-marg1b} by induction on $i$. For the
      base case $i = 0$, if $e$ is true in $(a_1, a_2)$ then
      \[
        \theta^{(0)}(a_1, a_2) = \bot \leq \eta_1(a_1) .
      \]
      Otherwise if $e$ is false, then $e_1$ must be false in $a_1$ as well and
      so
      \[
        \pi_1(\theta^{(0)}(a_1, a_2)) = a_1 = \eta_1(a_1) .
      \]

      Now we consider the inductive step $i > 0$. Again if $e$ is false in
      $(a_1, a_2)$ then both sides are $\bot$ and the claim is clear. Otherwise
      if $e$ is true, there are three cases.

      \begin{description}
        \item[Subcase for \cref{eq:whilegen-marg1b}: $p_0$ is true.]
          If $p_0$ is true, then we unfold the loops:
          \begin{align*}
            \theta^{(i)}(a_1, a_2)
            &= \dbind( \denot{c_0'} ( a_1, a_2 ), \theta^{(i-1)} ) \\
            \eta_1(a_1)
            &= \dbind( \denot{(\Condt{e_1}{c_1})^{K_1}} ( a_1 ), \eta_1 ) .
          \end{align*}
          By the induction hypothesis, for every $(b_1, b_2) \in \denot{\Phi}$
          we have $\pi_1(\theta^{(i-1)}(b_1, b_2)) \leq \eta_1(b_1)$. By the
          marginal condition from the outer induction hypothesis for the premise
          \cref{prem:while-gen:xp:sync}, we also have
          \[
            \pi_1(\denot{c_0'} ( a_1, a_2 ))
            = \denot{(\Condt{e_1}{c_1})^{K_1}} a_1 .
          \]
          The support condition from the same induction hypothesis shows
          \[
            \denot{c_0'} ( a_1, a_2 ) \subseteq \denot{\Phi}
          \]
          so by \cref{lem:proj-bind}, we conclude
          \[
            \pi_1(\theta^{(i)}(a_1, a_2)) \leq \eta_1(a_1) .
          \]

        \item[Subcase for \cref{eq:whilegen-marg1b}: $p_1$ is true.]
          If $p_1$ is true, then $e_1$ is true in $a_1$ by
          \cref{prem:while-gen:consist:left}. Unfolding:
          \begin{align*}
            \theta^{(i)}(a_1, a_2)
            &= \dbind( \denot{c_1'} ( a_1, a_2 ), \theta^{(i-1)} ) \\
            \eta_1(a_1)
            &= \dbind( \denot{\Condt{e_1}{c_1}} a_1, \eta_1 )
            = \dbind( \denot{c_1} a_1, \eta_1 ) .
          \end{align*}
          By the induction hypothesis, for every $(b_1, b_2) \in \denot{\Phi}$
          we have $\pi_1(\theta^{(i-1)}(b_1, b_2)) \leq \eta_1(b_1)$. By the
          marginal condition from the outer induction hypothesis for the premise
          \cref{prem:while-gen:xp:left}, we get
          \[
            \pi_1(\denot{c_1'} ( a_1, a_2 )) = \denot{c_1} a_1 .
          \]
          \Cref{lem:proj-bind} establishes
          \[
            \pi_1(\theta^{(i)}(a_1, a_2)) \leq \eta_1(a_1) .
          \]

        \item[Subcase for \cref{eq:whilegen-marg1b}: $p_2$ is true.]
          If $p_2$ is true, then $e_2$ is true in $a_2$ by
          \cref{prem:while-gen:consist:right}. Unfolding:
          \begin{align*}
            \theta^{(i)}(a_1, a_2)
            &= \dbind( \denot{c_2'} ( a_1, a_2 ), \theta^{(i-1)} ) \\
            \eta_1(a_1)
            &= \dbind( \denot{\Skip} a_1, \eta_1 ) .
          \end{align*}
          By the induction hypothesis, for every $(b_1, b_2) \in \denot{\Phi}$
          we have $\pi_1(\theta^{(i-1)}(b_1, b_2)) \leq \eta_1(b_1)$. By the
          marginal condition from the outer induction on the premise
          \cref{prem:while-gen:xp:right}, we get
          \[
            \pi_1(\denot{c_2'} ( a_1, a_2 )) = \denot{\Skip} a_1 .
          \]
          \Cref{lem:proj-bind} establishes
          \[
            \pi_1(\theta^{(i)}(a_1, a_2)) \leq \eta_1(a_1) .
          \]
      \end{description}
      This completes the inductive case $i > 0$, establishing
      \cref{eq:whilegen-marg1b}. By taking limits in
      \cref{eq:whilegen-marg1a,eq:whilegen-marg1b} and interchanging limits and
      projections (\cref{lem:proj-lim}), we have:
      \[
        \pi_1(\theta_\times(a_1, a_2)) \leq \eta_1(a_1) \leq \pi_1(\theta_\times(a_1, a_2))
      \]
      and hence equality holds, showing the first marginal condition.

      The remaining equations \cref{eq:whilegen-marg2a,eq:whilegen-marg2b} for
      the second marginal condition follow by a symmetric argument, proving
      soundness of the rule.
  \end{description}
  This completes the induction, establishing soundness of \Sxprhl.
\end{proof}

\chapter{Soundness of \Saprhl} \label{app:sound-aprhl}

\setcounter{section}{1}

The version of the logic \Saprhl we saw is similar to existing presentations of
\Saprhl (cf. \citet{BartheO13,BKOZ13-toplas,OlmedoThesis}). The main differences
are our definition of approximate lifting (\cref{def:alift}), which is a variant
of the approximate lifting introduced by \citet{BartheO13} and
\citet{OlmedoThesis} with better theoretical properties, and the new proof rules
introduced in \cref{chap:approx,chap:combining}.

We prove soundness of this version of \Saprhl, consisting of
\cref{fig:aprhl-two-sided,fig:aprhl-one-sided,fig:aprhl-structural,fig:aprhl-lap,fig:aprhl-pw-eq,fig:aprhl-utb,fig:aprhl-lapint,fig:aprhl-ac}.

\aprhlsound*

\begin{proof}
  By induction on the height of the proof derivation.  We consider the
  two-sided rules first (\cref{fig:aprhl-two-sided}), followed by the one-sided
  rules (\cref{fig:aprhl-one-sided}), and the structural rules
  (\cref{fig:aprhl-structural}). The new rules
  (\cref{fig:aprhl-lap,fig:aprhl-pw-eq,fig:aprhl-utb,fig:aprhl-lapint,fig:aprhl-ac})
  were proved sound in \cref{chap:approx,chap:combining}; we give pointers to
  the relevant lemmas.

  If the premises are valid and we have two inputs $m_1, m_2$ that satisfy the
  pre-condition, we must construct witnesses $\mu_L, \mu_R$ of the approximate
  lifting; namely, they must satisfy the support condition, the marginal
  conditions, and the distance condition in \cref{def:alift}. Let $\mu_1$ and
  $\mu_2$ be the output distributions from inputs $m_1$ and $m_2$ respectively.
  Throughout, we will leave the logical context $\rho$ implicit when taking the
  semantics $\denot{-}$; these constants play no role in the proof.
  \begin{description}
    \item[Case \nameref{rule:aprhl-skip}]
      Trivial; take $\mu_L = \mu_R = \dunit(m_1, m_2)$.
    \item[Case \nameref{rule:aprhl-assn}]
      Trivial; take $\mu_L = \mu_R = \dunit(m_1[x_1 \mapsto v_1], m_2[x_2
      \mapsto v_2])$ with $v_i \triangleq \denot{e_i} m_i$.
    \item[Case \nameref{rule:aprhl-lap}]
      Consequence of soundness for \nameref{rule:aprhl-lapgen}
      (\cref{thm:lapgen-sound})---in \nameref{rule:aprhl-lapgen}, take $k \triangleq 0$ and
      $k' \triangleq k$ in \nameref{rule:aprhl-lap}.
    \item[Case \nameref{rule:aprhl-seq}]
      By induction, we have two maps
      \[
        \eta_L, \eta_R : \Mem \times \Mem \to \Dist(\Mem^\star \times \Mem^\star)
      \]
      such that for any memories $a_1, a_2$ satisfying $\Phi$, the distributions
      $\eta_L(a_1, a_2), \eta_R(a_1, a_2)$ witness the $(\varepsilon,
      \delta)$-approximate lifting with support $\Psi$, and we have maps
      $\eta_L', \eta_R' : \Mem \times \Mem \to \Dist(\Mem^\star \times
      \Mem^\star)$ such that for any memories $a_1', a_2'$ satisfying $\Psi$,
      the distributions $\eta_L'(a_1', a_2'), \eta_R'(a_1', a_2')$ witness the
      $(\varepsilon', \delta')$-approximate lifting with support $\Theta$.
      
      To construct the witnesses for the conclusion, we would like to combine
      the witnesses for the premises in sequence. There is a slight mismatch, as
      $\eta_L(a_1, a_2)$ and $\eta_R(a_1, a_2)$ may place probability on pairs
      $(m, \star)$ or $(\star, m)$. Accordingly, we first extend the domain of
      the second maps $\eta_L', \eta_R'$. We define
      \begin{align*}
        \widehat{\eta_L}(a_1',a_2')(x, y)      &\triangleq \eta_L'(a_1', a_2')(x, y)
        \quad \text{ if } \quad a_1', a_2' \neq \star \\
        \widehat{\eta_L}(a_1',\star)(x, \star) &\triangleq (\denot{c_1'} a_1')(x) \\
        \widehat{\eta_L}(\star,a_2')(\star, y) &\triangleq (\denot{c_2'} a_2')(y) \\
        \widehat{\eta_R}(a_1',a_2')(x, y)      &\triangleq \eta_R'(a_1', a_2')(x, y)
        \quad \text{ if } \quad a_1', a_2' \neq \star \\
        \widehat{\eta_R}(a_1',\star)(x, \star) &\triangleq (\denot{c_1'} a_1')(x) \\
        \widehat{\eta_R}(\star,a_2')(\star, y) &\triangleq (\denot{c_2'} a_2')(y)
      \end{align*}
      and zero otherwise. We now define the witnesses for the conclusion:
      \[
        \mu_L \triangleq \dbind( \eta_L(m_1, m_2), \widehat{\eta_L} )
        \qquad\text{and}\qquad
        \mu_R \triangleq \dbind( \eta_R(m_1, m_2), \widehat{\eta_R} ) .
      \]
      The support condition is clear, as
      \[
        \supp(\eta_L), \supp(\eta_R) \subseteq \denot{\Psi}^\star
        \quad\text{and}\quad
        \supp(\widehat{\eta_L}(a_1, a_2)), \supp(\widehat{\eta_R}(a_1, a_2)) \subseteq \denot{\Theta}^\star
      \]
      for all $a_1, a_2 \in \denot{\Psi}^\star$, by induction and by definition
      of $\widehat{\eta_L}, \widehat{\eta_R}$.
      The marginal conditions are also clear: by the marginal condition on
      $\eta_L$ and $\eta_R$, we have $\eta_L(\star, a_2) = \eta_R(a_1, \star) =
      0$ for all $(a_1, a_2)$.  Also note that for $a_1' \neq \star$ we have
      \[
        \pi_1(\widehat{\eta_L}(a_1', a_2')) = \denot{c_1'} a_1' ,
      \]
      and for $a_2' \neq \star$ we have
      \[
        \pi_2(\widehat{\eta_R}(a_1', a_2')) = \denot{c_2'} a_2' .
      \]
      Therefore,
      \begin{align*}
        \pi_1(\mu_L) &= \pi_1(\dbind( \eta_L(m_1, m_2), \widehat{\eta_L} ))
        \\
        &= \dbind( \pi_1(\eta_L(m_1, m_2)), \denot{c_1'} ) \\
        &= \dbind( \denot{c_1} m_1 , \denot{c_1'} ) \\
        &= \denot{c_1; c_1'} m_1
      \end{align*}
      where the first equality is by \cref{lem:proj-bind} and the marginal
      condition from the second premise, and the second equality is by the
      marginal condition from the first premise. For the second marginal,
      \begin{align*}
        \pi_2(\mu_R) &= \pi_2(\dbind( \eta_R(m_1, m_2), \widehat{\eta_R} ))
        \\
        &= \dbind( \pi_2(\eta_R(m_1, m_2)), \denot{c_2'} ) \\
        &= \dbind( \denot{c_2} m_2 , \denot{c_2'} ) \\
        &= \denot{c_2; c_2'} m_2 .
      \end{align*}
      Thus, it only remains to check the distance condition
      $\epsdist{\varepsilon + \varepsilon'}{\mu_L}{\mu_R} \leq \delta +
      \delta'$.  Let $\cS \subseteq \Mem^\star \times \Mem^\star$ be any set of
      pairs memories, possibly including $\star$. We need to bound $\mu_L(\cS)
      \leq \exp(\varepsilon) \cdot \mu_R(\cS) + \delta$. Since
      $\epsdist{\varepsilon}{\eta_L(m_1, m_2)}{\eta_R(m_1, m_2)} \leq \delta$,
      there exist constants $\zeta(x_1, x_2) \geq 0$ (possibly depending on
      $m_1, m_2$) for $x_1, x_2 \in \Mem^\star$ such that
      \[
        \eta_L(m_1, m_2)(x_1, x_2) \leq \exp(\varepsilon) \cdot \eta_R(m_1, m_2)(x_1, x_2) + \zeta(x_1, x_2)
      \]
      and
      \[
        \sum_{(x_1, x_2) \in \Mem^\star \times \Mem^\star} \zeta(x_1, x_2) \leq \delta .
      \]
      By definition, for all $a_1', a_2' \in \Mem^\star \times \Mem^\star$ we
      have
      \[
        \epsdist{\varepsilon'}{\widehat{\eta_L}(a_1', a_2')}{\widehat{\eta_R}(a_1', a_2')} \leq \delta' .
      \]
      Thus, we can directly compute (with all sums over $\Mem^\star \times
      \Mem^\star$):
      \begin{align*}
        \mu_L(\cS) &= \sum_{(x_1, x_2)}
        \eta_L(m_1, m_2)(x_1, x_2) \cdot \widehat{\eta}_L(x_1, x_2)(\cS) \\
        &\leq \sum_{(x_1, x_2)}
        \eta_L(m_1, m_2)(x_1, x_2) \cdot \min(\exp(\varepsilon') \widehat{\eta}_R(x_1, x_2)(\cS) + \delta', 1) \\
        &= \sum_{(x_1, x_2)}
        \eta_L(m_1, m_2)(x_1, x_2) \cdot (\min(\exp(\varepsilon') \widehat{\eta}_R(x_1, x_2)(\cS), 1 - \delta') + \delta') \\
        &= \delta' + \sum_{(x_1, x_2)}
        \eta_L(m_1, m_2)(x_1, x_2) \cdot \min(\exp(\varepsilon') \widehat{\eta}_R(x_1, x_2)(\cS), 1 - \delta') \\
        &\leq \delta' + \sum_{(x_1, x_2)}
        (\exp(\varepsilon) \cdot \eta_R(m_1, m_2)(x_1, x_2) + \zeta(x_1, x_2))
        \cdot \min(\exp(\varepsilon') \widehat{\eta}_R(x_1, x_2)(\cS), 1 - \delta') \\
        &\leq \delta' + \sum_{(x_1, x_2)}
        \exp(\varepsilon) \cdot \eta_R(m_1, m_2)(x_1, x_2) \cdot \exp(\varepsilon') \widehat{\eta}_R(x_1, x_2)(\cS)
        + \sum_{(x_1, x_2)} \zeta(x_1, x_2) \cdot (1 - \delta') \\
        &\leq \delta' + \sum_{(x_1, x_2)}
        \exp(\varepsilon) \cdot \eta_R(m_1, m_2)(x_1, x_2) \cdot \exp(\varepsilon') \widehat{\eta}_R(x_1, x_2)(\cS)
        + (1 - \delta') \sum_{(x_1, x_2)} \zeta(x_1, x_2) \\
        &\leq \delta + \delta' +
        \exp(\varepsilon + \varepsilon') \sum_{(x_1, x_2)} \eta_R(m_1, m_2)(x_1, x_2) \cdot \widehat{\eta}_R(x_1, x_2)(\cS) \\
        &= \delta + \delta' + \exp(\varepsilon + \varepsilon') \mu_R(\cS) .
      \end{align*}
      This establishes the distance condition $\epsdist{\varepsilon +
      \varepsilon'}{\mu_L}{\mu_R} \leq \delta + \delta'$. Thus, $\mu_L, \mu_R$
      are witnesses to the desired approximate lifting.
    \item[Case \nameref{rule:aprhl-cond}]
      There are two cases. If $e_1$ is true in $m_1$, then $e_2$ is also true in
      $m_2$ by the pre-condition. Hence, $\denot{\Cond{e_1}{c_1}{c_1'}} m_1 =
      \denot{c_1} m_1$ and $\denot{\Cond{e_2}{c_2}{c_2'}} m_2 = \denot{c_2}
      m_2$, and we can take $\mu_L, \mu_R$ to be the witnesses from the first
      inductive premise. Otherwise, if $e_1$ is false in $m_1$ then $e_2$ is
      false in $m_2$ and we take $\mu_L, \mu_R$ to be the witnesses from the
      second inductive premise.
    \item[Case \nameref{rule:aprhl-while}]
      We prove that for every two memories $(a_1, a_2) \in \denot{\Phi}$, if
      $\denot{e_v} a_1 = k$ then we have
      \[
        \denot{\WWhile{e_1}{c_1}} a_1
        \alift{(\Phi \land \neg e_1\sidel)}{(k \cdot \varepsilon, k \cdot \delta)}
        \denot{\WWhile{e_2}{c_2}} a_2 .
      \]
      The proof is by induction on $k$. In the base case $k = 0$, by the
      premises $e_1$ is false in $a_1$ and hence $e_2$ is false in $a_2$.
      Therefore, we have
      \[
        \denot{\Skip} a_1
        \alift{(\Phi \land \neg e_1\sidel)}{(0,0)}
        \denot{\Skip} a_2
      \]
      by taking witnesses $\eta_L = \eta_R \triangleq \dunit(a_1, a_2)$.

      For the inductive step $k > 0$, if $e_1$ is false in $a_1$ then $e_2$ is
      false in $a_2$, both loops are equivalent to $\Skip$ and we take the
      witnesses as in the base case. Otherwise, $e_1$ and $e_2$ are both true
      and we need to show
      \[
        \denot{c_1; \WWhile{e_1}{c_1}} a_1
        \alift{\Phi}{(k \cdot \varepsilon, k \cdot \delta)}
        \denot{c_2; \WWhile{e_2}{c_2}} a_2 .
      \]
      From the premise, for every two memories $(a_1, a_2) \in \denot{\Phi}$
      with $e_1$ true in $a_1$, we have
      \[
        \denot{c_1} a_1
        \alift{(\Phi \land e_v\sidel < k)}{(\varepsilon, \delta)}
        \denot{c_2} a_2 .
      \]
      For every pair of memories $b_1, b_2$ satisfying $\Phi$ with $e_v < k$ in
      $b_1$, the induction hypothesis gives
      \[
        \denot{\WWhile{e_1}{c_1}} b_1
        \alift{(\Phi \land \neg e_1\sidel)}{((k-1) \cdot \varepsilon, (k-1) \cdot \delta)}
        \denot{\WWhile{e_2}{c_2}} b_2 .
      \]
      Combining these two witnesses with the reasoning from the case for
      \nameref{rule:aprhl-seq}, we have
      \[
        \denot{c_1; \WWhile{e_1}{c_1}} a_1
        \alift{(\Phi \land \neg e_1\sidel)}{(k \cdot \varepsilon, k \cdot \delta)}
        \denot{c_2; \WWhile{e_2}{c_2}} a_2
      \]
      as desired. Applying this claim for $a_1 \triangleq m_1, a_2
      \triangleq m_2$ and $k \triangleq N$ establishes soundness of the rule.
    \item[Case \nameref{rule:aprhl-assn-l} (\nameref{rule:aprhl-assn-r} similar)]
      Trivial; take $\mu_L = \mu_R = \dunit(m_1[x_1 \mapsto v_1], m_2)$ with
      $v_1 \triangleq \denot{e_1} m_1$.
    \item[Case \nameref{rule:aprhl-lap-l} (\nameref{rule:aprhl-lap-r} similar)]
      Let $\lambda \in \Dist(\ZZ)$ be the distribution
      $\denot{\Lap{\varepsilon}(e)} m_1$. We define the witnesses
      \[
        \mu_L(m_1[x_1 \mapsto v_1], m_2) 
        = \mu_R(m_1[x_1 \mapsto v_1], m_2) 
        \triangleq \lambda(v_1)
      \]
      for every $v_1 \in \ZZ$, and zero otherwise. The support, marginal, and
      distance conditions are easy to check.
    \item[Case \nameref{rule:aprhl-cond-l} (\nameref{rule:aprhl-cond-r} similar)]
      There are two cases. If $e_1$ is true in $m_1$, then
      \[
        \denot{\Cond{e_1}{c_1}{c_1'}} m_1 = \denot{c_1} m_1 .
      \]
      We let $\mu_L, \mu_R$ be the witnesses from the first premise by
      induction.  Otherwise if $e_1$ is false in $m_1$, we let $\mu_L, \mu_R$ be
      the witnesses from the second premise by induction.
    \item[Case \nameref{rule:aprhl-while-l} (\nameref{rule:aprhl-while-r} similar)]
      Trivial; by soundness of the \Sprhl version using \cref{prop:alift-plift}.
    \item[Case \nameref{rule:aprhl-conseq}]
      Trivial; take the witnesses from the premise by induction.
    \item[Case \nameref{rule:aprhl-equiv}]
      Trivial; take the witnesses from the premise by induction.
    \item[Case \nameref{rule:aprhl-case}]
      There are two cases. If $(m_1, m_2) \in \denot{\Theta}$, then the input
      memories satisfy the pre-condition in the first premise. Otherwise if
      $(m_1, m_2) \in \denot{\neg \Theta}$, then the input memories satisfy the
      pre-condition in the second premise.  In either case, by induction we take
      the witnesses from the respective premise as the witnesses for the
      conclusion.
    \item[Case \nameref{rule:aprhl-trans}]
      By \cref{lem:alift-trans}.
    \item[Case \nameref{rule:aprhl-frame}]
      By the induction hypothesis, there are witnesses $\mu_L', \mu_R'$ to an
      $(\varepsilon, \delta)$-approximate lifting of the two output
      distributions $\mu_1, \mu_2$ on inputs $m_1, m_2$. Let $V = \FV(\Theta)$
      be the free variables in $\Theta$ and suppose $m_1[V] = a_1$ and $m_2[V] =
      a_2$, where $m[V] : V \to \Val$ is the restriction of $m$ to $V$, and
      $a_1, a_2$ are maps $V \to \Val$. Since $c_1$ and $c_2$ do not modify
      variables in $V$, memories $m_1'$ in the support of $\mu_1$ satisfy
      $m_1'[V] = a_1$ and memories $m_2'$ in the support of $\mu_2$ satisfy
      $m_2'[V] = a_2$.
      
      By \cref{lem:alift-supp} and the inductive hypothesis, we can find
      witnesses $\mu_L, \mu_R$ to an $(\varepsilon, \delta)$-approximate lifting
      of $\mu_1, \mu_2$ such that
      \[
        \supp(\mu_L) \cup \supp(\mu_R)
        \subseteq \denot{\Psi}
        \cap \{ (m_1', m_2') \mid m_1'[V] = a_1, m_2'[V] = a_2 \} 
        \subseteq \denot{\Psi \land \Theta} ,
      \]
      where the last inclusion holds because $m_1, m_2$ restricted to $V$
      satisfy $\Theta$ by assumption. Hence, $\mu_L, \mu_R$ witness the desired
      approximate lifting.
    \item[Case \nameref{rule:aprhl-lapnull}]
      By \cref{thm:lapnull-sound}.
    \item[Case \nameref{rule:aprhl-lapgen}]
      By \cref{thm:lapgen-sound}.
    \item[Case \nameref{rule:aprhl-pw-eq}]
      By \cref{thm:pw-eq-sound}.
    \item[Cases \nameref{rule:aprhl-utb-l} and \nameref{rule:aprhl-utb-r}]
      By \cref{thm:utb-sound}.
    \item[Case \nameref{rule:aprhl-lapint}]
      By \cref{thm:lapint-sound}.
    \item[Case \nameref{rule:aprhl-while-ac}]
      By \cref{thm:while-ac-sound}. \qedhere
  \end{description}
\end{proof}

\bibliographystyle{ACM-Reference-Format}
\raggedright
\bibliography{../header,../myrefs}


\newcommand{\SortNoop}[1]{}
\begin{thebibliography}{110}


\ifx \showCODEN    \undefined \def \showCODEN     #1{\unskip}     \fi
\ifx \showDOI      \undefined \def \showDOI       #1{#1}\fi
\ifx \showISBNx    \undefined \def \showISBNx     #1{\unskip}     \fi
\ifx \showISBNxiii \undefined \def \showISBNxiii  #1{\unskip}     \fi
\ifx \showISSN     \undefined \def \showISSN      #1{\unskip}     \fi
\ifx \showLCCN     \undefined \def \showLCCN      #1{\unskip}     \fi
\ifx \shownote     \undefined \def \shownote      #1{#1}          \fi
\ifx \showarticletitle \undefined \def \showarticletitle #1{#1}   \fi
\ifx \showURL      \undefined \def \showURL       {\relax}        \fi
\providecommand\bibfield[2]{#2}
\providecommand\bibinfo[2]{#2}
\providecommand\natexlab[1]{#1}
\providecommand\showeprint[2][]{arXiv:#2}

\bibitem[\protect\citeauthoryear{Aharoni, Berger, Georgakopoulos, Perlstein,
  and Spr{\"{u}}ssel}{Aharoni et~al\mbox{.}}{2011}]%
        {DBLP:journals/jct/AharoniBGPS11}
\bibfield{author}{\bibinfo{person}{Ron Aharoni}, \bibinfo{person}{Eli Berger},
  \bibinfo{person}{Agelos Georgakopoulos}, \bibinfo{person}{Amitai Perlstein},
  {and} \bibinfo{person}{Philipp Spr{\"{u}}ssel}.}
  \bibinfo{year}{2011}\natexlab{}.
\newblock \showarticletitle{The Max-Flow Min-Cut Theorem for Countable
  Networks}.
\newblock \bibinfo{journal}{\emph{Journal of Combinatorial Theory, Series B}}
  \bibinfo{volume}{101}, \bibinfo{number}{1} (\bibinfo{year}{2011}),
  \bibinfo{pages}{1--17}.
\newblock
\urldef\tempurl%
\url{https://doi.org/10.1016/j.jctb.2010.08.002}
\showDOI{\tempurl}


\bibitem[\protect\citeauthoryear{Albarghouthi and Hsu}{Albarghouthi and
  Hsu}{2018}]%
        {AH17}
\bibfield{author}{\bibinfo{person}{Aws Albarghouthi} {and}
  \bibinfo{person}{Justin Hsu}.} \bibinfo{year}{2018}\natexlab{}.
\newblock \showarticletitle{Synthesizing Coupling Proofs of Differential
  Privacy}.
\newblock \bibinfo{journal}{\emph{Proceedings of the {ACM} on Programming
  Languages}} \bibinfo{volume}{1}, \bibinfo{number}{POPL}
  (\bibinfo{year}{2018}).
\newblock
\showeprint[arxiv]{cs.PL/1709.05361}
\urldef\tempurl%
\url{http://arxiv.org/abs/1709.05361}
\showURL{%
\tempurl}
\newblock
\shownote{To appear at {ACM} {SIGPLAN--SIGACT} {S}ymposium on {P}rinciples of
  {P}rogramming {L}anguages ({POPL}), Los Angeles, California.}


\bibitem[\protect\citeauthoryear{Aldous}{Aldous}{1983}]%
        {aldous1983random}
\bibfield{author}{\bibinfo{person}{David Aldous}.}
  \bibinfo{year}{1983}\natexlab{}.
\newblock \showarticletitle{Random Walks on Finite Groups and Rapidly Mixing
  {Markov} Chains}. In \bibinfo{booktitle}{\emph{S{\'e}minaire de
  Probabilit{\'e}s {XVII} 1981/82}} \emph{(\bibinfo{series}{Lecture Notes in
  Mathematics})}, Vol.~\bibinfo{volume}{986}.
  \bibinfo{publisher}{Springer-Verlag}, \bibinfo{pages}{243--297}.
\newblock
\urldef\tempurl%
\url{https://eudml.org/doc/113445}
\showURL{%
\tempurl}


\bibitem[\protect\citeauthoryear{Apt}{Apt}{1981}]%
        {Apt:1981:TYH:357146.357150}
\bibfield{author}{\bibinfo{person}{Krzysztof~R. Apt}.}
  \bibinfo{year}{1981}\natexlab{}.
\newblock \showarticletitle{Ten Years of {Hoare}'s Logic: {A} Survey--Part
  {I}}.
\newblock \bibinfo{journal}{\emph{ACM Transactions on Programming Languages and
  Systems}} \bibinfo{volume}{3}, \bibinfo{number}{4} (\bibinfo{year}{1981}),
  \bibinfo{pages}{431--483}.
\newblock
\showISSN{0164-0925}
\urldef\tempurl%
\url{https://doi.org/10.1145/357146.357150}
\showDOI{\tempurl}


\bibitem[\protect\citeauthoryear{Apt}{Apt}{1983}]%
        {APT198383}
\bibfield{author}{\bibinfo{person}{Krzysztof~R. Apt}.}
  \bibinfo{year}{1983}\natexlab{}.
\newblock \showarticletitle{Ten years of {Hoare}'s logic: {A} survey--Part
  {II}: Nondeterminism}.
\newblock \bibinfo{journal}{\emph{Theoretical Computer Science}}
  \bibinfo{volume}{28}, \bibinfo{number}{1} (\bibinfo{year}{1983}),
  \bibinfo{pages}{83--109}.
\newblock
\showISSN{0304-3975}
\urldef\tempurl%
\url{https://doi.org/10.1016/0304-3975(83)90066-X}
\showDOI{\tempurl}


\bibitem[\protect\citeauthoryear{Azevedo~de Amorim, Gaboardi, Gallego~Arias,
  and Hsu}{Azevedo~de Amorim et~al\mbox{.}}{2014}]%
        {AGGH14}
\bibfield{author}{\bibinfo{person}{Arthur Azevedo~de Amorim},
  \bibinfo{person}{Marco Gaboardi}, \bibinfo{person}{Emilio~Jes{\'u}s
  Gallego~Arias}, {and} \bibinfo{person}{Justin Hsu}.}
  \bibinfo{year}{2014}\natexlab{}.
\newblock \showarticletitle{Really Natural Linear Indexed Type-Checking}. In
  \bibinfo{booktitle}{\emph{Symposium on Implementation and Application of
  Functional Programming Languages (IFL), Boston, Massachusetts}}.
  \bibinfo{publisher}{ACM Press}, \bibinfo{pages}{5:1--5:12}.
\newblock
\urldef\tempurl%
\url{https://doi.org/10.1145/2746325.2746335}
\showDOI{\tempurl}
\showeprint[arxiv]{cs.LO/1503.04522}


\bibitem[\protect\citeauthoryear{Azevedo~de Amorim, Gaboardi, Hsu, Katsumata,
  and Cherigui}{Azevedo~de Amorim et~al\mbox{.}}{2017}]%
        {ACGHK16}
\bibfield{author}{\bibinfo{person}{Arthur Azevedo~de Amorim},
  \bibinfo{person}{Marco Gaboardi}, \bibinfo{person}{Justin Hsu},
  \bibinfo{person}{{Shin-ya} Katsumata}, {and} \bibinfo{person}{Ikram
  Cherigui}.} \bibinfo{year}{2017}\natexlab{}.
\newblock \showarticletitle{A Semantic Account of Metric Preservation}. In
  \bibinfo{booktitle}{\emph{{ACM} {SIGPLAN--SIGACT} {S}ymposium on {P}rinciples
  of {P}rogramming {L}anguages ({POPL}), Paris, France}}.
  \bibinfo{pages}{545--556}.
\newblock
\urldef\tempurl%
\url{https://doi.org/10.1145/3009837.3009890}
\showDOI{\tempurl}
\showeprint[arxiv]{cs.PL/1702.00374}


\bibitem[\protect\citeauthoryear{Barthe, Crespo, and Kunz}{Barthe
  et~al\mbox{.}}{2011a}]%
        {BartheCK11}
\bibfield{author}{\bibinfo{person}{Gilles Barthe}, \bibinfo{person}{Juan~Manuel
  Crespo}, {and} \bibinfo{person}{C{\'e}sar Kunz}.}
  \bibinfo{year}{2011}\natexlab{a}.
\newblock \showarticletitle{Relational Verification Using Product Programs}. In
  \bibinfo{booktitle}{\emph{International Symposium on Formal Methods (FM),
  Limerick, Ireland}} \emph{(\bibinfo{series}{Lecture Notes in Computer
  Science})}, Vol.~\bibinfo{volume}{6664}.
  \bibinfo{publisher}{Springer-Verlag}, \bibinfo{pages}{200--214}.
\newblock
\urldef\tempurl%
\url{https://doi.org/10.1007/978-3-642-21437-0_17}
\showDOI{\tempurl}


\bibitem[\protect\citeauthoryear{Barthe, Crespo, and Kunz}{Barthe
  et~al\mbox{.}}{2013a}]%
        {BartheCK13}
\bibfield{author}{\bibinfo{person}{Gilles Barthe}, \bibinfo{person}{Juan~Manuel
  Crespo}, {and} \bibinfo{person}{C{\'e}sar Kunz}.}
  \bibinfo{year}{2013}\natexlab{a}.
\newblock \showarticletitle{Beyond 2-Safety: {Asymmetric} Product Programs for
  Relational Program Verification}. In \bibinfo{booktitle}{\emph{International
  Symposium on Logical Foundations of Computer Science (LFCS), San Diego,
  California}} \emph{(\bibinfo{series}{Lecture Notes in Computer Science})},
  Vol.~\bibinfo{volume}{7734}. \bibinfo{publisher}{Springer-Verlag},
  \bibinfo{pages}{29--43}.
\newblock
\urldef\tempurl%
\url{https://doi.org/10.1007/978--3-642--35722--0_3}
\showDOI{\tempurl}


\bibitem[\protect\citeauthoryear{Barthe, D'Argenio, and Rezk}{Barthe
  et~al\mbox{.}}{2011b}]%
        {BartheDR04}
\bibfield{author}{\bibinfo{person}{Gilles Barthe}, \bibinfo{person}{Pedro~R.
  D'Argenio}, {and} \bibinfo{person}{Tamara Rezk}.}
  \bibinfo{year}{2011}\natexlab{b}.
\newblock \showarticletitle{Secure Information Flow by Self-Composition}.
\newblock \bibinfo{journal}{\emph{Mathematical Structures in Computer Science}}
  \bibinfo{volume}{21}, \bibinfo{number}{06} (\bibinfo{year}{2011}),
  \bibinfo{pages}{1207--1252}.
\newblock
\urldef\tempurl%
\url{https://doi.org/10.1017/S0960129511000193}
\showDOI{\tempurl}


\bibitem[\protect\citeauthoryear{Barthe, Dupressoir, Gr{\'{e}}goire, Kunz,
  Schmidt, and Strub}{Barthe et~al\mbox{.}}{2013b}]%
        {BartheDGKSS13}
\bibfield{author}{\bibinfo{person}{Gilles Barthe},
  \bibinfo{person}{Fran{\c{c}}ois Dupressoir}, \bibinfo{person}{Benjamin
  Gr{\'{e}}goire}, \bibinfo{person}{C{\'{e}}sar Kunz},
  \bibinfo{person}{Benedikt Schmidt}, {and} \bibinfo{person}{Pierre{-}Yves
  Strub}.} \bibinfo{year}{2013}\natexlab{b}.
\newblock \showarticletitle{EasyCrypt: {A} Tutorial}. In
  \bibinfo{booktitle}{\emph{Foundations of Security Analysis and Design
  ({FOSAD})}} \emph{(\bibinfo{series}{Lecture Notes in Computer Science})},
  Vol.~\bibinfo{volume}{8604}. \bibinfo{publisher}{Springer-Verlag},
  \bibinfo{pages}{146--166}.
\newblock
\urldef\tempurl%
\url{https://doi.org/10.1007/978-3-319-10082-1_6}
\showDOI{\tempurl}
\newblock
\shownote{Tutorial Lectures.}


\bibitem[\protect\citeauthoryear{Barthe, Espitau, Gaboardi, Gr{\'e}goire, Hsu,
  and Strub}{Barthe et~al\mbox{.}}{2017a}]%
        {BEGGHS16}
\bibfield{author}{\bibinfo{person}{Gilles Barthe}, \bibinfo{person}{Thomas
  Espitau}, \bibinfo{person}{Marco Gaboardi}, \bibinfo{person}{Benjamin
  Gr{\'e}goire}, \bibinfo{person}{Justin Hsu}, {and}
  \bibinfo{person}{{Pierre}-{Yves} Strub}.} \bibinfo{year}{2017}\natexlab{a}.
\newblock \bibinfo{title}{Formal Certification of Randomized Algorithms}.
  (\bibinfo{year}{2017}).
\newblock
\urldef\tempurl%
\url{http://justinh.su/files/papers/ellora.pdf}
\showURL{%
\tempurl}


\bibitem[\protect\citeauthoryear{Barthe, Espitau, Gr{\'e}goire, Hsu,
  Stefanesco, and Strub}{Barthe et~al\mbox{.}}{2015a}]%
        {BEGHSS15}
\bibfield{author}{\bibinfo{person}{Gilles Barthe}, \bibinfo{person}{Thomas
  Espitau}, \bibinfo{person}{Benjamin Gr{\'e}goire}, \bibinfo{person}{Justin
  Hsu}, \bibinfo{person}{L{\'e}o Stefanesco}, {and}
  \bibinfo{person}{{Pierre}-{Yves} Strub}.} \bibinfo{year}{2015}\natexlab{a}.
\newblock \showarticletitle{Relational Reasoning via Probabilistic Coupling}.
  In \bibinfo{booktitle}{\emph{International Conference on Logic for
  Programming, Artificial Intelligence and Reasoning (LPAR), Suva, Fiji}}
  \emph{(\bibinfo{series}{Lecture Notes in Computer Science})},
  Vol.~\bibinfo{volume}{9450}. \bibinfo{publisher}{Springer-Verlag},
  \bibinfo{pages}{387--401}.
\newblock
\urldef\tempurl%
\url{https://doi.org/10.1007/978-3-662-48899-7_27}
\showDOI{\tempurl}
\showeprint[arxiv]{cs.LO/1509.03476}


\bibitem[\protect\citeauthoryear{Barthe, Espitau, Gr{\'e}goire, Hsu, and
  Strub}{Barthe et~al\mbox{.}}{2017b}]%
        {BEGHS17}
\bibfield{author}{\bibinfo{person}{Gilles Barthe}, \bibinfo{person}{Thomas
  Espitau}, \bibinfo{person}{Benjamin Gr{\'e}goire}, \bibinfo{person}{Justin
  Hsu}, {and} \bibinfo{person}{{Pierre}-{Yves} Strub}.}
  \bibinfo{year}{2017}\natexlab{b}.
\newblock \showarticletitle{Proving Uniformity and Independence by
  Self-Composition and Coupling}. In \bibinfo{booktitle}{\emph{International
  Conference on Logic for Programming, Artificial Intelligence and Reasoning
  (LPAR), Maun, Botswana}} \emph{(\bibinfo{series}{EPiC Series in Computing})},
  Vol.~\bibinfo{volume}{46}. \bibinfo{pages}{385--403}.
\newblock
\showeprint[arxiv]{cs.PL/1701.06477}
\urldef\tempurl%
\url{https://arxiv.org/abs/1701.06477}
\showURL{%
\tempurl}


\bibitem[\protect\citeauthoryear{Barthe, Espitau, Gr{\'e}goire, Hsu, and
  Strub}{Barthe et~al\mbox{.}}{2018}]%
        {BEGHS16}
\bibfield{author}{\bibinfo{person}{Gilles Barthe}, \bibinfo{person}{Thomas
  Espitau}, \bibinfo{person}{Benjamin Gr{\'e}goire}, \bibinfo{person}{Justin
  Hsu}, {and} \bibinfo{person}{{Pierre}-{Yves} Strub}.}
  \bibinfo{year}{2018}\natexlab{}.
\newblock \showarticletitle{Proving Expected Sensitivity of Probabilistic
  Programs}.
\newblock \bibinfo{journal}{\emph{Proceedings of the {ACM} on Programming
  Languages}} \bibinfo{volume}{1}, \bibinfo{number}{POPL}
  (\bibinfo{year}{2018}).
\newblock
\showeprint[arxiv]{cs.PL/1708.02537}
\urldef\tempurl%
\url{http://arxiv.org/abs/1708.02537}
\showURL{%
\tempurl}
\newblock
\shownote{To appear at {ACM} {SIGPLAN--SIGACT} {S}ymposium on {P}rinciples of
  {P}rogramming {L}anguages ({POPL}), Los Angeles, California.}


\bibitem[\protect\citeauthoryear{Barthe, Espitau, Hsu, Sato, and Strub}{Barthe
  et~al\mbox{.}}{2017c}]%
        {BEHSS17}
\bibfield{author}{\bibinfo{person}{Gilles Barthe}, \bibinfo{person}{Thomas
  Espitau}, \bibinfo{person}{Justin Hsu}, \bibinfo{person}{Tetsuya Sato}, {and}
  \bibinfo{person}{{Pierre}-{Yves} Strub}.} \bibinfo{year}{2017}\natexlab{c}.
\newblock \showarticletitle{$\star$-Liftings for Differential Privacy}. In
  \bibinfo{booktitle}{\emph{International Colloquium on Automata, Languages and
  Programming (ICALP), Warsaw, Poland}} \emph{(\bibinfo{series}{Leibniz
  International Proceedings in Informatics})}, Vol.~\bibinfo{volume}{80}.
  \bibinfo{publisher}{Schloss Dagstuhl--Leibniz Center for Informatics},
  \bibinfo{pages}{102:1--102:12}.
\newblock
\urldef\tempurl%
\url{https://doi.org/10.4230/LIPIcs.ICALP.2017.102}
\showDOI{\tempurl}
\showeprint[arxiv]{cs.LO/1705.00133}


\bibitem[\protect\citeauthoryear{Barthe, Fong, Gaboardi, Gr{\'e}goire, Hsu, and
  Strub}{Barthe et~al\mbox{.}}{2016a}]%
        {BGGHS16c}
\bibfield{author}{\bibinfo{person}{Gilles Barthe}, \bibinfo{person}{No{\'e}mie
  Fong}, \bibinfo{person}{Marco Gaboardi}, \bibinfo{person}{Benjamin
  Gr{\'e}goire}, \bibinfo{person}{Justin Hsu}, {and}
  \bibinfo{person}{{Pierre}-{Yves} Strub}.} \bibinfo{year}{2016}\natexlab{a}.
\newblock \showarticletitle{Advanced Probabilistic Couplings for Differential
  Privacy}. In \bibinfo{booktitle}{\emph{{ACM} {SIGSAC} Conference on Computer
  and Communications Security ({CCS}), Vienna, Austria}}.
  \bibinfo{pages}{55--67}.
\newblock
\urldef\tempurl%
\url{https://doi.org/10.1145/2976749.2978391}
\showDOI{\tempurl}
\showeprint[arxiv]{cs.LO/1606.07143}


\bibitem[\protect\citeauthoryear{Barthe, Fournet, Gr{\'e}goire, Strub, Swamy,
  and Zanella{-}B{\'e}guelin}{Barthe et~al\mbox{.}}{2014a}]%
        {rfstar}
\bibfield{author}{\bibinfo{person}{Gilles Barthe}, \bibinfo{person}{C{\'e}dric
  Fournet}, \bibinfo{person}{Benjamin Gr{\'e}goire},
  \bibinfo{person}{{Pierre}-{Yves} Strub}, \bibinfo{person}{Nikhil Swamy},
  {and} \bibinfo{person}{Santiago Zanella{-}B{\'e}guelin}.}
  \bibinfo{year}{2014}\natexlab{a}.
\newblock \showarticletitle{Probabilistic Relational Verification for
  Cryptographic Implementations}. In \bibinfo{booktitle}{\emph{{ACM}
  {SIGPLAN--SIGACT} {S}ymposium on {P}rinciples of {P}rogramming {L}anguages
  ({POPL}), San Diego, California}}. \bibinfo{pages}{193--206}.
\newblock
\urldef\tempurl%
\url{https://doi.org/10.1145/2535838.2535847}
\showDOI{\tempurl}


\bibitem[\protect\citeauthoryear{Barthe, Gaboardi, Gallego~Arias, Hsu, Kunz,
  and Strub}{Barthe et~al\mbox{.}}{2014b}]%
        {BGGHKS14}
\bibfield{author}{\bibinfo{person}{Gilles Barthe}, \bibinfo{person}{Marco
  Gaboardi}, \bibinfo{person}{Emilio~Jes{\'u}s Gallego~Arias},
  \bibinfo{person}{Justin Hsu}, \bibinfo{person}{C\'esar Kunz}, {and}
  \bibinfo{person}{Pierre-Yves Strub}.} \bibinfo{year}{2014}\natexlab{b}.
\newblock \showarticletitle{Proving Differential Privacy in {Hoare} Logic}. In
  \bibinfo{booktitle}{\emph{{IEEE} {C}omputer {S}ecurity {F}oundations
  {S}ymposium ({CSF}), Vienna, Austria}}. \bibinfo{pages}{411--424}.
\newblock
\urldef\tempurl%
\url{https://doi.org/10.1109/CSF.2014.36}
\showDOI{\tempurl}
\showeprint[arxiv]{cs.LO/1407.2988}


\bibitem[\protect\citeauthoryear{Barthe, Gaboardi, Gallego~Arias, Hsu, Roth,
  and Strub}{Barthe et~al\mbox{.}}{2015b}]%
        {BGGHRS15}
\bibfield{author}{\bibinfo{person}{Gilles Barthe}, \bibinfo{person}{Marco
  Gaboardi}, \bibinfo{person}{Emilio~Jes{\'u}s Gallego~Arias},
  \bibinfo{person}{Justin Hsu}, \bibinfo{person}{Aaron Roth}, {and}
  \bibinfo{person}{Pierre-Yves Strub}.} \bibinfo{year}{2015}\natexlab{b}.
\newblock \showarticletitle{Higher-Order Approximate Relational Refinement
  Types for Mechanism Design and Differential Privacy}. In
  \bibinfo{booktitle}{\emph{{ACM} {SIGPLAN--SIGACT} {S}ymposium on {P}rinciples
  of {P}rogramming {L}anguages ({POPL}), Mumbai, India}}.
  \bibinfo{pages}{55--68}.
\newblock
\urldef\tempurl%
\url{https://doi.org/10.1145/2676726.2677000}
\showDOI{\tempurl}
\showeprint[arxiv]{cs.PL/1407.6845}


\bibitem[\protect\citeauthoryear{Barthe, Gaboardi, Gallego~Arias, Hsu, Roth,
  and Strub}{Barthe et~al\mbox{.}}{2016b}]%
        {HKM-verif16}
\bibfield{author}{\bibinfo{person}{Gilles Barthe}, \bibinfo{person}{Marco
  Gaboardi}, \bibinfo{person}{Emilio~Jes{\'u}s Gallego~Arias},
  \bibinfo{person}{Justin Hsu}, \bibinfo{person}{Aaron Roth}, {and}
  \bibinfo{person}{Pierre-Yves Strub}.} \bibinfo{year}{2016}\natexlab{b}.
\newblock \showarticletitle{Computer-Aided Verification in Mechanism Design}.
  In \bibinfo{booktitle}{\emph{Conference on Web and Internet Economics (WINE),
  Montr{\'e}al, Qu{\'e}bec}} \emph{(\bibinfo{series}{Lecture Notes in Computer
  Science})}, Vol.~\bibinfo{volume}{10123}.
  \bibinfo{publisher}{Springer-Verlag}, \bibinfo{pages}{273--293}.
\newblock
\urldef\tempurl%
\url{https://doi.org/10.1007/978-3-662-54110-4_20}
\showDOI{\tempurl}
\showeprint[arxiv]{cs.GT/1502.04052}


\bibitem[\protect\citeauthoryear{Barthe, Gaboardi, Gr{\'e}goire, Hsu, and
  Strub}{Barthe et~al\mbox{.}}{2016c}]%
        {BGGHS16}
\bibfield{author}{\bibinfo{person}{Gilles Barthe}, \bibinfo{person}{Marco
  Gaboardi}, \bibinfo{person}{Benjamin Gr{\'e}goire}, \bibinfo{person}{Justin
  Hsu}, {and} \bibinfo{person}{{Pierre}-{Yves} Strub}.}
  \bibinfo{year}{2016}\natexlab{c}.
\newblock \showarticletitle{Proving Differential Privacy via Probabilistic
  Couplings}. In \bibinfo{booktitle}{\emph{{IEEE} {S}ymposium on {L}ogic in
  {C}omputer {S}cience ({LICS}), New York, New York}}.
  \bibinfo{pages}{749--758}.
\newblock
\urldef\tempurl%
\url{https://doi.org/10.1145/2933575.2934554}
\showDOI{\tempurl}
\showeprint[arxiv]{cs.LO/1601.05047}


\bibitem[\protect\citeauthoryear{Barthe, Gaboardi, Hsu, and Pierce}{Barthe
  et~al\mbox{.}}{2016d}]%
        {Murawski:2016:2893582}
\bibfield{author}{\bibinfo{person}{Gilles Barthe}, \bibinfo{person}{Marco
  Gaboardi}, \bibinfo{person}{Justin Hsu}, {and} \bibinfo{person}{Benjamin~C.
  Pierce}.} \bibinfo{year}{2016}\natexlab{d}.
\newblock \showarticletitle{Programming Language Techniques for Differential
  Privacy}.
\newblock \bibinfo{journal}{\emph{ACM SIGLOG News}} \bibinfo{volume}{3},
  \bibinfo{number}{1} (\bibinfo{date}{Jan.} \bibinfo{year}{2016}),
  \bibinfo{pages}{34--53}.
\newblock
\urldef\tempurl%
\url{https://doi.org/10.1145/2893582.2893591}
\showDOI{\tempurl}


\bibitem[\protect\citeauthoryear{Barthe, Gr{\'e}goire, Hsu, and Strub}{Barthe
  et~al\mbox{.}}{2017d}]%
        {BGHS16}
\bibfield{author}{\bibinfo{person}{Gilles Barthe}, \bibinfo{person}{Benjamin
  Gr{\'e}goire}, \bibinfo{person}{Justin Hsu}, {and}
  \bibinfo{person}{{Pierre}-{Yves} Strub}.} \bibinfo{year}{2017}\natexlab{d}.
\newblock \showarticletitle{Coupling Proofs Are Probabilistic Product
  Programs}. In \bibinfo{booktitle}{\emph{{ACM} {SIGPLAN--SIGACT} {S}ymposium
  on {P}rinciples of {P}rogramming {L}anguages ({POPL}), Paris, France}}.
  \bibinfo{pages}{161--174}.
\newblock
\urldef\tempurl%
\url{https://doi.org/10.1145/3009837.3009896}
\showDOI{\tempurl}
\showeprint[arxiv]{cs.PL/1607.03455}


\bibitem[\protect\citeauthoryear{Barthe, Gr{\'e}goire, and
  Zanella{-}B{\'e}guelin}{Barthe et~al\mbox{.}}{2009}]%
        {BartheGZ09}
\bibfield{author}{\bibinfo{person}{Gilles Barthe}, \bibinfo{person}{Benjamin
  Gr{\'e}goire}, {and} \bibinfo{person}{Santiago Zanella{-}B{\'e}guelin}.}
  \bibinfo{year}{2009}\natexlab{}.
\newblock \showarticletitle{Formal Certification of Code-Based Cryptographic
  Proofs}. In \bibinfo{booktitle}{\emph{{ACM} {SIGPLAN--SIGACT} {S}ymposium on
  {P}rinciples of {P}rogramming {L}anguages ({POPL}), Savannah, Georgia}}.
  \bibinfo{pages}{90--101}.
\newblock
\urldef\tempurl%
\url{https://doi.org/10.1145/1480881.1480894}
\showDOI{\tempurl}


\bibitem[\protect\citeauthoryear{Barthe, K{\"{o}}pf, Olmedo, and
  Zanella{-}B{\'{e}}guelin}{Barthe et~al\mbox{.}}{2013c}]%
        {BKOZ13-toplas}
\bibfield{author}{\bibinfo{person}{Gilles Barthe}, \bibinfo{person}{Boris
  K{\"{o}}pf}, \bibinfo{person}{Federico Olmedo}, {and}
  \bibinfo{person}{Santiago Zanella{-}B{\'{e}}guelin}.}
  \bibinfo{year}{2013}\natexlab{c}.
\newblock \showarticletitle{Probabilistic Relational Reasoning for Differential
  Privacy}.
\newblock \bibinfo{journal}{\emph{ACM Transactions on Programming Languages and
  Systems}} \bibinfo{volume}{35}, \bibinfo{number}{3} (\bibinfo{date}{Nov.}
  \bibinfo{year}{2013}), \bibinfo{pages}{9:1--9:49}.
\newblock
\urldef\tempurl%
\url{https://doi.org/10.1145/2492061}
\showDOI{\tempurl}


\bibitem[\protect\citeauthoryear{Barthe and Olmedo}{Barthe and Olmedo}{2013}]%
        {BartheO13}
\bibfield{author}{\bibinfo{person}{Gilles Barthe} {and}
  \bibinfo{person}{Federico Olmedo}.} \bibinfo{year}{2013}\natexlab{}.
\newblock \showarticletitle{Beyond Differential Privacy: Composition Theorems
  and Relational Logic for $f$-Divergences between Probabilistic Programs}. In
  \bibinfo{booktitle}{\emph{International Colloquium on Automata, Languages and
  Programming (ICALP), Riga, Latvia}} \emph{(\bibinfo{series}{Lecture Notes in
  Computer Science})}, Vol.~\bibinfo{volume}{7966}.
  \bibinfo{publisher}{Springer-Verlag}, \bibinfo{pages}{49--60}.
\newblock
\urldef\tempurl%
\url{https://doi.org/10.1007/978-3-642-39212-2_8}
\showDOI{\tempurl}


\bibitem[\protect\citeauthoryear{Bellare and Rogaway}{Bellare and
  Rogaway}{2006}]%
        {Bellare:2006}
\bibfield{author}{\bibinfo{person}{Mihir Bellare} {and}
  \bibinfo{person}{Phillip Rogaway}.} \bibinfo{year}{2006}\natexlab{}.
\newblock \showarticletitle{The Security of Triple Encryption and a Framework
  for Code-Based Game-Playing Proofs}. In \bibinfo{booktitle}{\emph{{IACR}
  {I}nternational {C}onference on the {T}heory and {A}pplications of
  {C}ryptographic {T}echniques (EUROCRYPT), Saint Petersburg, Russia}}
  \emph{(\bibinfo{series}{Lecture Notes in Computer Science})},
  Vol.~\bibinfo{volume}{4004}. \bibinfo{publisher}{Springer-Verlag},
  \bibinfo{pages}{409--426}.
\newblock
\urldef\tempurl%
\url{https://doi.org/10.1007/11761679_25}
\showDOI{\tempurl}


\bibitem[\protect\citeauthoryear{Benton}{Benton}{2004}]%
        {Benton04}
\bibfield{author}{\bibinfo{person}{Nick Benton}.}
  \bibinfo{year}{2004}\natexlab{}.
\newblock \showarticletitle{Simple Relational Correctness Proofs for Static
  Analyses and Program Transformations}. In \bibinfo{booktitle}{\emph{{ACM}
  {SIGPLAN--SIGACT} {S}ymposium on {P}rinciples of {P}rogramming {L}anguages
  ({POPL}), Venice, Italy}}. \bibinfo{pages}{14--25}.
\newblock
\urldef\tempurl%
\url{https://doi.org/10.1145/964001.964003}
\showDOI{\tempurl}


\bibitem[\protect\citeauthoryear{Bousquet and Elisseeff}{Bousquet and
  Elisseeff}{2002}]%
        {BousquetE02}
\bibfield{author}{\bibinfo{person}{Olivier Bousquet} {and}
  \bibinfo{person}{Andr{\'{e}} Elisseeff}.} \bibinfo{year}{2002}\natexlab{}.
\newblock \showarticletitle{Stability and Generalization}.
\newblock \bibinfo{journal}{\emph{Journal of Machine Learning Research}}
  \bibinfo{volume}{2} (\bibinfo{year}{2002}), \bibinfo{pages}{499--526}.
\newblock
\urldef\tempurl%
\url{https://doi.org/10.1162/153244302760200704}
\showDOI{\tempurl}


\bibitem[\protect\citeauthoryear{Bubley and Dyer}{Bubley and Dyer}{1997}]%
        {bubley1997path}
\bibfield{author}{\bibinfo{person}{Russ Bubley} {and} \bibinfo{person}{Martin
  Dyer}.} \bibinfo{year}{1997}\natexlab{}.
\newblock \showarticletitle{Path Coupling: {A} Technique for Proving Rapid
  Mixing in {Markov} Chains}. In \bibinfo{booktitle}{\emph{{IEEE} {S}ymposium
  on {F}oundations of {C}omputer {S}cience (FOCS), Miami Beach, Florida}}.
  \bibinfo{pages}{223--231}.
\newblock
\urldef\tempurl%
\url{https://doi.org/10.1109/SFCS.1997.646111}
\showDOI{\tempurl}


\bibitem[\protect\citeauthoryear{Buch}{Buch}{2017}]%
        {buch:thesis}
\bibfield{author}{\bibinfo{person}{Mads Buch}.}
  \bibinfo{year}{2017}\natexlab{}.
\newblock \emph{\bibinfo{title}{Formalizing Differential Privacy}}.
\newblock \bibinfo{thesistype}{Master's\ thesis}. \bibinfo{school}{Aarhus
  University}.
\newblock
\urldef\tempurl%
\url{http://madsbuch.com/thesis}
\showURL{%
\tempurl}


\bibitem[\protect\citeauthoryear{Bun and Steinke}{Bun and Steinke}{2016}]%
        {BunS2016}
\bibfield{author}{\bibinfo{person}{Mark Bun} {and} \bibinfo{person}{Thomas
  Steinke}.} \bibinfo{year}{2016}\natexlab{}.
\newblock \showarticletitle{Concentrated Differential Privacy: Simplifications,
  Extensions, and Lower Bounds}. In \bibinfo{booktitle}{\emph{{IACR} {T}heory
  of {C}ryptography {C}onference (TCC), Beijing, China}}
  \emph{(\bibinfo{series}{Lecture Notes in Computer Science})},
  Vol.~\bibinfo{volume}{9985}. \bibinfo{publisher}{Springer-Verlag},
  \bibinfo{pages}{635--658}.
\newblock
\urldef\tempurl%
\url{https://doi.org/10.1007/978-3-662-53641-4_24}
\showDOI{\tempurl}
\showeprint[arxiv]{cs.CR/1605.02065}


\bibitem[\protect\citeauthoryear{Bun, Steinke, and Ullman}{Bun
  et~al\mbox{.}}{2017}]%
        {BunSU16}
\bibfield{author}{\bibinfo{person}{Mark Bun}, \bibinfo{person}{Thomas Steinke},
  {and} \bibinfo{person}{Jonathan Ullman}.} \bibinfo{year}{2017}\natexlab{}.
\newblock \showarticletitle{Make Up Your Mind: The Price of Online Queries in
  Differential Privacy}. In \bibinfo{booktitle}{\emph{{ACM--SIAM} {S}ymposium
  on {D}iscrete {A}lgorithms (SODA), Barcelona, Spain}}.
  \bibinfo{pages}{1306--1325}.
\newblock
\urldef\tempurl%
\url{https://doi.org/10.1137/1.9781611974782.85}
\showDOI{\tempurl}


\bibitem[\protect\citeauthoryear{Callahan and Kennedy}{Callahan and
  Kennedy}{1988}]%
        {Callahan1988}
\bibfield{author}{\bibinfo{person}{David Callahan} {and} \bibinfo{person}{Ken
  Kennedy}.} \bibinfo{year}{1988}\natexlab{}.
\newblock \showarticletitle{Compiling Programs for Distributed-Memory
  Multiprocessors}.
\newblock \bibinfo{journal}{\emph{The Journal of Supercomputing}}
  \bibinfo{volume}{2}, \bibinfo{number}{2} (\bibinfo{date}{Oct.}
  \bibinfo{year}{1988}), \bibinfo{pages}{151--169}.
\newblock
\showISSN{1573-0484}
\urldef\tempurl%
\url{https://doi.org/10.1007/BF00128175}
\showDOI{\tempurl}


\bibitem[\protect\citeauthoryear{Chatterjee, Fu, and Goharshady}{Chatterjee
  et~al\mbox{.}}{2016a}]%
        {Chatterjee2016}
\bibfield{author}{\bibinfo{person}{Krishnendu Chatterjee},
  \bibinfo{person}{Hongfei Fu}, {and} \bibinfo{person}{Amir~Kafshdar
  Goharshady}.} \bibinfo{year}{2016}\natexlab{a}.
\newblock \showarticletitle{Termination Analysis of Probabilistic Programs
  through {Positivstellensatz's}}. In \bibinfo{booktitle}{\emph{International
  Conference on Computer Aided Verification (CAV), Toronto, Ontario}}
  \emph{(\bibinfo{series}{Lecture Notes in Computer Science})},
  Vol.~\bibinfo{volume}{9779}. \bibinfo{publisher}{Springer-Verlag},
  \bibinfo{pages}{3--22}.
\newblock
\showISBNx{978-3-319-41528-4}
\urldef\tempurl%
\url{https://doi.org/10.1007/978-3-319-41528-4_1}
\showDOI{\tempurl}


\bibitem[\protect\citeauthoryear{Chatterjee, Fu, Novotn\'{y}, and
  Hasheminezhad}{Chatterjee et~al\mbox{.}}{2016b}]%
        {Chatterjee:2016:AAQ:2837614.2837639}
\bibfield{author}{\bibinfo{person}{Krishnendu Chatterjee},
  \bibinfo{person}{Hongfei Fu}, \bibinfo{person}{Petr Novotn\'{y}}, {and}
  \bibinfo{person}{Rouzbeh Hasheminezhad}.} \bibinfo{year}{2016}\natexlab{b}.
\newblock \showarticletitle{Algorithmic Analysis of Qualitative and
  Quantitative Termination Problems for Affine Probabilistic Programs}. In
  \bibinfo{booktitle}{\emph{{ACM} {SIGPLAN--SIGACT} {S}ymposium on {P}rinciples
  of {P}rogramming {L}anguages ({POPL}), Saint Petersburg, Florida}}.
  \bibinfo{pages}{327--342}.
\newblock
\showISBNx{978-1-4503-3549-2}
\urldef\tempurl%
\url{https://doi.org/10.1145/2837614.2837639}
\showDOI{\tempurl}


\bibitem[\protect\citeauthoryear{Chatterjee, Novotn\'{y}, and
  \v{Z}ikeli\'{c}}{Chatterjee et~al\mbox{.}}{2017}]%
        {Chatterjee:2017:SIP:3009837.3009873}
\bibfield{author}{\bibinfo{person}{Krishnendu Chatterjee},
  \bibinfo{person}{Petr Novotn\'{y}}, {and} \bibinfo{person}{{\DH}or{\dj}e
  \v{Z}ikeli\'{c}}.} \bibinfo{year}{2017}\natexlab{}.
\newblock \showarticletitle{Stochastic Invariants for Probabilistic
  Termination}. In \bibinfo{booktitle}{\emph{{ACM} {SIGPLAN--SIGACT}
  {S}ymposium on {P}rinciples of {P}rogramming {L}anguages ({POPL}), Paris,
  France}}. \bibinfo{pages}{145--160}.
\newblock
\showISBNx{978-1-4503-4660-3}
\urldef\tempurl%
\url{https://doi.org/10.1145/3009837.3009873}
\showDOI{\tempurl}


\bibitem[\protect\citeauthoryear{Deng and Du}{Deng and Du}{2009}]%
        {deng2009kantorovich}
\bibfield{author}{\bibinfo{person}{Yuxin Deng} {and} \bibinfo{person}{Wenjie
  Du}.} \bibinfo{year}{2009}\natexlab{}.
\newblock \showarticletitle{The {Kantorovich} Metric in Computer Science: {A}
  Brief Survey}.
\newblock \bibinfo{journal}{\emph{Electronic Notes in Theoretical Computer
  Science}} \bibinfo{volume}{253}, \bibinfo{number}{3} (\bibinfo{year}{2009}),
  \bibinfo{pages}{73--82}.
\newblock
\urldef\tempurl%
\url{https://doi.org/10.1016/j.entcs.2009.10.006}
\showDOI{\tempurl}


\bibitem[\protect\citeauthoryear{Deng and Du}{Deng and Du}{2011}]%
        {DengD11}
\bibfield{author}{\bibinfo{person}{Yuxin Deng} {and} \bibinfo{person}{Wenjie
  Du}.} \bibinfo{year}{2011}\natexlab{}.
\newblock \bibinfo{booktitle}{\emph{Logical, Metric, and Algorithmic
  Characterisations of Probabilistic Bisimulation}}.
\newblock \bibinfo{type}{{T}echnical {R}eport} CMU-CS-11-110.
  \bibinfo{institution}{Carnegie Mellon University}.
\newblock
\urldef\tempurl%
\url{http://arxiv.org/abs/1103.4577}
\showURL{%
\tempurl}


\bibitem[\protect\citeauthoryear{Desharnais}{Desharnais}{1999}]%
        {desharnais1999labelled}
\bibfield{author}{\bibinfo{person}{Jos{\'e}e Desharnais}.}
  \bibinfo{year}{1999}\natexlab{}.
\newblock \emph{\bibinfo{title}{Labelled {Markov} Processes}}.
\newblock \bibinfo{thesistype}{Ph.D. Dissertation}. \bibinfo{school}{{McGill}
  University}.
\newblock
\urldef\tempurl%
\url{www.collectionscanada.gc.ca/obj/s4/f2/dsk1/tape3/PQDD_0031/NQ64546.pdf}
\showURL{%
\tempurl}


\bibitem[\protect\citeauthoryear{Desharnais, Edalat, and Panangaden}{Desharnais
  et~al\mbox{.}}{2002}]%
        {DESHARNAIS2002163}
\bibfield{author}{\bibinfo{person}{Jos{\'e}e Desharnais},
  \bibinfo{person}{Abbas Edalat}, {and} \bibinfo{person}{Prakash Panangaden}.}
  \bibinfo{year}{2002}\natexlab{}.
\newblock \showarticletitle{Bisimulation for Labelled {Markov} Processes}.
\newblock \bibinfo{journal}{\emph{Information and Computation}}
  \bibinfo{volume}{179}, \bibinfo{number}{2} (\bibinfo{year}{2002}),
  \bibinfo{pages}{163--193}.
\newblock
\showISSN{0890-5401}
\urldef\tempurl%
\url{https://doi.org/10.1006/inco.2001.2962}
\showDOI{\tempurl}


\bibitem[\protect\citeauthoryear{Desharnais, Gupta, Jagadeesan, and
  Panangaden}{Desharnais et~al\mbox{.}}{2003}]%
        {DESHARNAIS2003160}
\bibfield{author}{\bibinfo{person}{Jos{\'e}e Desharnais},
  \bibinfo{person}{Vineet Gupta}, \bibinfo{person}{Radha Jagadeesan}, {and}
  \bibinfo{person}{Prakash Panangaden}.} \bibinfo{year}{2003}\natexlab{}.
\newblock \showarticletitle{Approximating Labelled {Markov} Processes}.
\newblock \bibinfo{journal}{\emph{Information and Computation}}
  \bibinfo{volume}{184}, \bibinfo{number}{1} (\bibinfo{year}{2003}),
  \bibinfo{pages}{160--200}.
\newblock
\showISSN{0890-5401}
\urldef\tempurl%
\url{https://doi.org/10.1016/S0890-5401(03)00051-8}
\showDOI{\tempurl}


\bibitem[\protect\citeauthoryear{Desharnais, Gupta, Jagadeesan, and
  Panangaden}{Desharnais et~al\mbox{.}}{2004}]%
        {DESHARNAIS2004323}
\bibfield{author}{\bibinfo{person}{Jos{\'e}e Desharnais},
  \bibinfo{person}{Vineet Gupta}, \bibinfo{person}{Radha Jagadeesan}, {and}
  \bibinfo{person}{Prakash Panangaden}.} \bibinfo{year}{2004}\natexlab{}.
\newblock \showarticletitle{Metrics for Labelled {Markov} Processes}.
\newblock \bibinfo{journal}{\emph{Theoretical Computer Science}}
  \bibinfo{volume}{318}, \bibinfo{number}{3} (\bibinfo{year}{2004}),
  \bibinfo{pages}{323--354}.
\newblock
\showISSN{0304-3975}
\urldef\tempurl%
\url{https://doi.org/10.1016/j.tcs.2003.09.013}
\showDOI{\tempurl}


\bibitem[\protect\citeauthoryear{Desharnais, Laviolette, and Tracol}{Desharnais
  et~al\mbox{.}}{2008}]%
        {desharnaisLT08}
\bibfield{author}{\bibinfo{person}{Jos{\'e}e Desharnais},
  \bibinfo{person}{Fran{\c{c}}ois Laviolette}, {and} \bibinfo{person}{Mathieu
  Tracol}.} \bibinfo{year}{2008}\natexlab{}.
\newblock \showarticletitle{Approximate Analysis of Probabilistic Processes:
  Logic, Simulation and Games}. In \bibinfo{booktitle}{\emph{International
  Conference on Quantitative Evaluation of Systems ({QEST}), Saint Malo,
  France}}. \bibinfo{publisher}{IEEE}, \bibinfo{pages}{264--273}.
\newblock
\urldef\tempurl%
\url{https://doi.org/10.1109/QEST.2008.42}
\showDOI{\tempurl}


\bibitem[\protect\citeauthoryear{Dijkstra}{Dijkstra}{1976}]%
        {dijkstra1976discipline}
\bibfield{author}{\bibinfo{person}{Edsger~W. Dijkstra}.}
  \bibinfo{year}{1976}\natexlab{}.
\newblock \bibinfo{booktitle}{\emph{A Discipline of Programming}}.
\newblock \bibinfo{publisher}{Prentice Hall}.
\newblock


\bibitem[\protect\citeauthoryear{Dwork, {McSherry}, Nissim, and Smith}{Dwork
  et~al\mbox{.}}{2006}]%
        {DMNS06}
\bibfield{author}{\bibinfo{person}{Cynthia Dwork}, \bibinfo{person}{Frank
  {McSherry}}, \bibinfo{person}{Kobbi Nissim}, {and} \bibinfo{person}{Adam~D.
  Smith}.} \bibinfo{year}{2006}\natexlab{}.
\newblock \showarticletitle{Calibrating Noise to Sensitivity in Private Data
  Analysis}. In \bibinfo{booktitle}{\emph{{IACR} {T}heory of {C}ryptography
  {C}onference (TCC), New York, New York}} \emph{(\bibinfo{series}{Lecture
  Notes in Computer Science})}, Vol.~\bibinfo{volume}{3876}.
  \bibinfo{publisher}{Springer-Verlag}, \bibinfo{pages}{265--284}.
\newblock
\urldef\tempurl%
\url{https://doi.org/10.1007/11681878_14}
\showDOI{\tempurl}


\bibitem[\protect\citeauthoryear{Dwork and Roth}{Dwork and Roth}{2014}]%
        {DR14}
\bibfield{author}{\bibinfo{person}{Cynthia Dwork} {and} \bibinfo{person}{Aaron
  Roth}.} \bibinfo{year}{2014}\natexlab{}.
\newblock \showarticletitle{The Algorithmic Foundations of Differential
  Privacy}.
\newblock \bibinfo{journal}{\emph{Foundations and Trends in Theoretical
  Computer Science}} \bibinfo{volume}{9}, \bibinfo{number}{3--4}
  (\bibinfo{year}{2014}), \bibinfo{pages}{211--407}.
\newblock
\urldef\tempurl%
\url{https://doi.org/10.1561/0400000042}
\showDOI{\tempurl}


\bibitem[\protect\citeauthoryear{Dwork, Rothblum, and Vadhan}{Dwork
  et~al\mbox{.}}{2010}]%
        {DRV10}
\bibfield{author}{\bibinfo{person}{Cynthia Dwork}, \bibinfo{person}{Guy~N.
  Rothblum}, {and} \bibinfo{person}{Salil Vadhan}.}
  \bibinfo{year}{2010}\natexlab{}.
\newblock \showarticletitle{Boosting and Differential Privacy}. In
  \bibinfo{booktitle}{\emph{{IEEE} {S}ymposium on {F}oundations of {C}omputer
  {S}cience (FOCS), Las Vegas, Nevada}}. \bibinfo{pages}{51--60}.
\newblock
\urldef\tempurl%
\url{https://doi.org/10.1109/FOCS.2010.12}
\showDOI{\tempurl}


\bibitem[\protect\citeauthoryear{Ebadi, Antignac, and Sands}{Ebadi
  et~al\mbox{.}}{2016}]%
        {DBLP:conf/pst/EbadiAS16}
\bibfield{author}{\bibinfo{person}{Hamid Ebadi}, \bibinfo{person}{Thibaud
  Antignac}, {and} \bibinfo{person}{David Sands}.}
  \bibinfo{year}{2016}\natexlab{}.
\newblock \showarticletitle{Sampling and Partitioning for Differential
  Privacy}. In \bibinfo{booktitle}{\emph{Conference on Privacy, Security and
  Trust ({PST}), Auckland, New Zealand}}. \bibinfo{publisher}{{IEEE}},
  \bibinfo{pages}{664--673}.
\newblock
\urldef\tempurl%
\url{https://doi.org/10.1109/PST.2016.7906954}
\showURL{%
\tempurl}


\bibitem[\protect\citeauthoryear{Ebadi and Sands}{Ebadi and Sands}{2016}]%
        {ebadi2016featherweight}
\bibfield{author}{\bibinfo{person}{Hamid Ebadi} {and} \bibinfo{person}{David
  Sands}.} \bibinfo{year}{2016}\natexlab{}.
\newblock \showarticletitle{Featherweight {PINQ}}.
\newblock \bibinfo{journal}{\emph{Journal of Privacy and Confidentiality}}
  \bibinfo{volume}{7}, \bibinfo{number}{2} (\bibinfo{year}{2016}),
  \bibinfo{pages}{159--184}.
\newblock
\urldef\tempurl%
\url{https://arxiv.org/abs/1505.02642}
\showURL{%
\tempurl}


\bibitem[\protect\citeauthoryear{Ebadi, Sands, and Schneider}{Ebadi
  et~al\mbox{.}}{2015}]%
        {DBLP:conf/popl/EbadiSS15}
\bibfield{author}{\bibinfo{person}{Hamid Ebadi}, \bibinfo{person}{David Sands},
  {and} \bibinfo{person}{Gerardo Schneider}.} \bibinfo{year}{2015}\natexlab{}.
\newblock \showarticletitle{Differential Privacy: Now It's Getting Personal}.
  In \bibinfo{booktitle}{\emph{{ACM} {SIGPLAN--SIGACT} {S}ymposium on
  {P}rinciples of {P}rogramming {L}anguages ({POPL}), Mumbai, India}}.
  \bibinfo{pages}{69--81}.
\newblock
\urldef\tempurl%
\url{https://doi.org/10.1145/2676726.2677005}
\showDOI{\tempurl}


\bibitem[\protect\citeauthoryear{Ferrer~Fioriti and Hermanns}{Ferrer~Fioriti
  and Hermanns}{2015}]%
        {FerrerHermanns15}
\bibfield{author}{\bibinfo{person}{Luis~Mar{\'\i}a Ferrer~Fioriti} {and}
  \bibinfo{person}{Holger Hermanns}.} \bibinfo{year}{2015}\natexlab{}.
\newblock \showarticletitle{Probabilistic Termination: Soundness, Completeness,
  and Compositionality}. In \bibinfo{booktitle}{\emph{{ACM} {SIGPLAN--SIGACT}
  {S}ymposium on {P}rinciples of {P}rogramming {L}anguages ({POPL}), Mumbai,
  India}}. \bibinfo{pages}{489--501}.
\newblock
\urldef\tempurl%
\url{https://doi.org/10.1145/2676726.2677001}
\showDOI{\tempurl}


\bibitem[\protect\citeauthoryear{Fijalkow, Klin, and Panangaden}{Fijalkow
  et~al\mbox{.}}{2017}]%
        {fijalkow_et_al:LIPIcs:2017:7368}
\bibfield{author}{\bibinfo{person}{Nathana{\"e}l Fijalkow},
  \bibinfo{person}{Bartek Klin}, {and} \bibinfo{person}{Prakash Panangaden}.}
  \bibinfo{year}{2017}\natexlab{}.
\newblock \showarticletitle{Expressiveness of Probabilistic Modal Logics,
  Revisited}. In \bibinfo{booktitle}{\emph{International Colloquium on
  Automata, Languages and Programming (ICALP), Warsaw, Poland}}
  \emph{(\bibinfo{series}{Leibniz International Proceedings in Informatics})},
  Vol.~\bibinfo{volume}{80}. \bibinfo{publisher}{Schloss Dagstuhl--Leibniz
  Center for Informatics}, \bibinfo{address}{Dagstuhl, Germany},
  \bibinfo{pages}{105:1--105:12}.
\newblock
\showISBNx{978-3-95977-041-5}
\showISSN{1868-8969}
\urldef\tempurl%
\url{https://doi.org/10.4230/LIPIcs.ICALP.2017.105}
\showDOI{\tempurl}


\bibitem[\protect\citeauthoryear{Floyd}{Floyd}{1967}]%
        {Floyd67}
\bibfield{author}{\bibinfo{person}{Robert~W. Floyd}.}
  \bibinfo{year}{1967}\natexlab{}.
\newblock \showarticletitle{Assigning Meanings to Programs}.
\newblock In \bibinfo{booktitle}{\emph{Symposium on Applied Mathematics}}.
  \bibinfo{publisher}{American Mathematical Society}.
\newblock
\urldef\tempurl%
\url{https://doi.org/10.1007/978-94-011-1793-7_4}
\showDOI{\tempurl}


\bibitem[\protect\citeauthoryear{Gaboardi, Haeberlen, Hsu, Narayan, and
  Pierce}{Gaboardi et~al\mbox{.}}{2013}]%
        {GHHNP13}
\bibfield{author}{\bibinfo{person}{Marco Gaboardi}, \bibinfo{person}{Andreas
  Haeberlen}, \bibinfo{person}{Justin Hsu}, \bibinfo{person}{Arjun Narayan},
  {and} \bibinfo{person}{Benjamin~C. Pierce}.} \bibinfo{year}{2013}\natexlab{}.
\newblock \showarticletitle{Linear Dependent Types for Differential Privacy}.
  In \bibinfo{booktitle}{\emph{{ACM} {SIGPLAN--SIGACT} {S}ymposium on
  {P}rinciples of {P}rogramming {L}anguages ({POPL}), Rome, Italy}}.
  \bibinfo{pages}{357--370}.
\newblock
\urldef\tempurl%
\url{https://doi.org/10.1145/2429069.2429113}
\showDOI{\tempurl}


\bibitem[\protect\citeauthoryear{Ghosh, Roughgarden, and Sundararajan}{Ghosh
  et~al\mbox{.}}{2012}]%
        {ghosh2012universally}
\bibfield{author}{\bibinfo{person}{Arpita Ghosh}, \bibinfo{person}{Tim
  Roughgarden}, {and} \bibinfo{person}{Mukund Sundararajan}.}
  \bibinfo{year}{2012}\natexlab{}.
\newblock \showarticletitle{Universally Utility-Maximizing Privacy Mechanisms}.
\newblock \bibinfo{journal}{\emph{SIAM Journal on Computing}}
  \bibinfo{volume}{41}, \bibinfo{number}{6} (\bibinfo{year}{2012}),
  \bibinfo{pages}{1673--1693}.
\newblock
\urldef\tempurl%
\url{https://doi.org/10.1137/09076828X}
\showDOI{\tempurl}


\bibitem[\protect\citeauthoryear{Giacalone, Jou, and Smolka}{Giacalone
  et~al\mbox{.}}{1990}]%
        {giacalone1990algebraic}
\bibfield{author}{\bibinfo{person}{Alessandro Giacalone},
  \bibinfo{person}{Chi-Chang Jou}, {and} \bibinfo{person}{Scott~A. Smolka}.}
  \bibinfo{year}{1990}\natexlab{}.
\newblock \showarticletitle{Algebraic Reasoning for Probabilistic Concurrent
  Systems}. In \bibinfo{booktitle}{\emph{{IFIP} {TC2} Working Conference on
  Programming Concepts and Methods}}.
\newblock
\urldef\tempurl%
\url{http://citeseerx.ist.psu.edu/viewdoc/summary?doi=10.1.1.56.3664}
\showURL{%
\tempurl}


\bibitem[\protect\citeauthoryear{Hall, Wasserman, and Rinaldo}{Hall
  et~al\mbox{.}}{2013}]%
        {HallWassermanRinaldo}
\bibfield{author}{\bibinfo{person}{Robert Hall}, \bibinfo{person}{Larry
  Wasserman}, {and} \bibinfo{person}{Alessandro Rinaldo}.}
  \bibinfo{year}{2013}\natexlab{}.
\newblock \showarticletitle{Random Differential Privacy}.
\newblock \bibinfo{journal}{\emph{Journal of Privacy and Confidentiality}}
  \bibinfo{volume}{4}, \bibinfo{number}{3} (\bibinfo{year}{2013}).
\newblock
Issue 2.
\urldef\tempurl%
\url{http://repository.cmu.edu/jpc/vol4/iss2/3}
\showURL{%
\tempurl}


\bibitem[\protect\citeauthoryear{Hardt, Recht, and Singer}{Hardt
  et~al\mbox{.}}{2016}]%
        {HardtRS16}
\bibfield{author}{\bibinfo{person}{Moritz Hardt}, \bibinfo{person}{Ben Recht},
  {and} \bibinfo{person}{Yoram Singer}.} \bibinfo{year}{2016}\natexlab{}.
\newblock \showarticletitle{Train Faster, Generalize Better: {Stability} of
  Stochastic Gradient Descent}. In \bibinfo{booktitle}{\emph{{I}nternational
  {C}onference on {M}achine {L}earning (ICML), New York, NY}}
  \emph{(\bibinfo{series}{Proceedings of Machine Learning Research})},
  Vol.~\bibinfo{volume}{48}. \bibinfo{pages}{1225--1234}.
\newblock
\showeprint[arxiv]{cs.LG/1509.01240}
\urldef\tempurl%
\url{http://arxiv.org/abs/1603.01445}
\showURL{%
\tempurl}


\bibitem[\protect\citeauthoryear{Hayes and Vigoda}{Hayes and Vigoda}{2007}]%
        {DBLP:journals/rsa/HayesV07}
\bibfield{author}{\bibinfo{person}{Thomas~P. Hayes} {and} \bibinfo{person}{Eric
  Vigoda}.} \bibinfo{year}{2007}\natexlab{}.
\newblock \showarticletitle{Variable Length Path Coupling}.
\newblock \bibinfo{journal}{\emph{Random Structures and Algorithms}}
  \bibinfo{volume}{31}, \bibinfo{number}{3} (\bibinfo{year}{2007}),
  \bibinfo{pages}{251--272}.
\newblock
\urldef\tempurl%
\url{https://doi.org/10.1002/rsa.20166}
\showDOI{\tempurl}


\bibitem[\protect\citeauthoryear{Hoare}{Hoare}{1969}]%
        {hoare1969axiomatic}
\bibfield{author}{\bibinfo{person}{Charles A.~R. Hoare}.}
  \bibinfo{year}{1969}\natexlab{}.
\newblock \showarticletitle{An Axiomatic Basis for Computer Programming}.
\newblock \bibinfo{journal}{\emph{Communications of the {ACM}}}
  \bibinfo{volume}{12}, \bibinfo{number}{10} (\bibinfo{year}{1969}),
  \bibinfo{pages}{576--580}.
\newblock
\urldef\tempurl%
\url{https://doi.org/10.1145/363235.363259}
\showDOI{\tempurl}


\bibitem[\protect\citeauthoryear{Jerrum}{Jerrum}{1995}]%
        {DBLP:journals/rsa/Jerrum95}
\bibfield{author}{\bibinfo{person}{Mark Jerrum}.}
  \bibinfo{year}{1995}\natexlab{}.
\newblock \showarticletitle{A Very Simple Algorithm for Estimating the Number
  of $k$-Colorings of a Low-Degree Graph}.
\newblock \bibinfo{journal}{\emph{Random Structures and Algorithms}}
  \bibinfo{volume}{7}, \bibinfo{number}{2} (\bibinfo{year}{1995}),
  \bibinfo{pages}{157--166}.
\newblock
\urldef\tempurl%
\url{https://doi.org/10.1002/rsa.3240070205}
\showDOI{\tempurl}


\bibitem[\protect\citeauthoryear{Jones}{Jones}{2003}]%
        {Jones:2003:EST:858595.858602}
\bibfield{author}{\bibinfo{person}{Cliff~B. Jones}.}
  \bibinfo{year}{2003}\natexlab{}.
\newblock \showarticletitle{The Early Search for Tractable Ways of Reasoning
  about Programs}.
\newblock \bibinfo{journal}{\emph{Annals of the History of Computing}}
  \bibinfo{volume}{25}, \bibinfo{number}{2} (\bibinfo{date}{April}
  \bibinfo{year}{2003}), \bibinfo{pages}{26--49}.
\newblock
\showISSN{1058-6180}
\urldef\tempurl%
\url{https://doi.org/10.1109/MAHC.2003.1203057}
\showDOI{\tempurl}


\bibitem[\protect\citeauthoryear{Jonsson and Larsen}{Jonsson and
  Larsen}{1991}]%
        {jonsson1991specification}
\bibfield{author}{\bibinfo{person}{Bengt Jonsson} {and}
  \bibinfo{person}{Kim~Guldstrand Larsen}.} \bibinfo{year}{1991}\natexlab{}.
\newblock \showarticletitle{Specification and Refinement of Probabilistic
  Processes}. In \bibinfo{booktitle}{\emph{{IEEE} {S}ymposium on {L}ogic in
  {C}omputer {S}cience ({LICS}), Amsterdam, The Netherlands}}.
  \bibinfo{pages}{266--277}.
\newblock
\urldef\tempurl%
\url{https://doi.org/10.1109/LICS.1991.151651}
\showDOI{\tempurl}


\bibitem[\protect\citeauthoryear{Kairouz, Oh, and Viswanath}{Kairouz
  et~al\mbox{.}}{2017}]%
        {DBLP:journals/corr/OhV13}
\bibfield{author}{\bibinfo{person}{P. Kairouz}, \bibinfo{person}{S. Oh}, {and}
  \bibinfo{person}{P. Viswanath}.} \bibinfo{year}{2017}\natexlab{}.
\newblock \showarticletitle{The Composition Theorem for Differential Privacy}.
\newblock \bibinfo{journal}{\emph{IEEE Transactions on Information Theory}}
  \bibinfo{volume}{63}, \bibinfo{number}{6} (\bibinfo{date}{June}
  \bibinfo{year}{2017}), \bibinfo{pages}{4037--4049}.
\newblock
\showISSN{0018-9448}
\urldef\tempurl%
\url{https://doi.org/10.1109/TIT.2017.2685505}
\showDOI{\tempurl}


\bibitem[\protect\citeauthoryear{Katsumata and Sato}{Katsumata and
  Sato}{2015}]%
        {katsumata2015codensity}
\bibfield{author}{\bibinfo{person}{Shin-ya Katsumata} {and}
  \bibinfo{person}{Tetsuya Sato}.} \bibinfo{year}{2015}\natexlab{}.
\newblock \showarticletitle{Codensity Liftings of Monads}. In
  \bibinfo{booktitle}{\emph{International Conference on Concurrency Theory
  (CONCUR), Nijmegen, The Netherlands}} \emph{(\bibinfo{series}{Leibniz
  International Proceedings in Informatics})}, Vol.~\bibinfo{volume}{35}.
  \bibinfo{publisher}{Schloss Dagstuhl--Leibniz Center for Informatics},
  \bibinfo{pages}{156--170}.
\newblock
\urldef\tempurl%
\url{https://doi.org/10.4230/LIPIcs.CALCO.2015.156}
\showDOI{\tempurl}


\bibitem[\protect\citeauthoryear{Kleinberg and Tardos}{Kleinberg and
  Tardos}{2005}]%
        {Kleinberg:2005:AD:1051910}
\bibfield{author}{\bibinfo{person}{Jon Kleinberg} {and} \bibinfo{person}{Eva
  Tardos}.} \bibinfo{year}{2005}\natexlab{}.
\newblock \bibinfo{booktitle}{\emph{Algorithm Design}}.
\newblock \bibinfo{publisher}{Addison-Wesley}.
\newblock
\showISBNx{0321295358}


\bibitem[\protect\citeauthoryear{Kozen}{Kozen}{1981}]%
        {Kozen81}
\bibfield{author}{\bibinfo{person}{Dexter Kozen}.}
  \bibinfo{year}{1981}\natexlab{}.
\newblock \showarticletitle{Semantics of Probabilistic Programs}.
\newblock \bibinfo{journal}{\emph{Journal of Computer and System Sciences}}
  \bibinfo{volume}{22}, \bibinfo{number}{3} (\bibinfo{year}{1981}),
  \bibinfo{pages}{328--350}.
\newblock
\urldef\tempurl%
\url{https://doi.org/10.1016/0022-0000(81)90036-2}
\showDOI{\tempurl}


\bibitem[\protect\citeauthoryear{Kozen}{Kozen}{1985}]%
        {Kozen:1985}
\bibfield{author}{\bibinfo{person}{Dexter Kozen}.}
  \bibinfo{year}{1985}\natexlab{}.
\newblock \showarticletitle{A Probabilistic {PDL}}.
\newblock \bibinfo{journal}{\emph{Journal of Computer and System Sciences}}
  \bibinfo{volume}{30}, \bibinfo{number}{2} (\bibinfo{year}{1985}).
\newblock
\urldef\tempurl%
\url{https://doi.org/10.1016/0022-0000(85)90012-1}
\showDOI{\tempurl}


\bibitem[\protect\citeauthoryear{Kumar and Ramesh}{Kumar and Ramesh}{2001}]%
        {DBLP:journals/rsa/KumarR01}
\bibfield{author}{\bibinfo{person}{V.~S.~Anil Kumar} {and} \bibinfo{person}{H.
  Ramesh}.} \bibinfo{year}{2001}\natexlab{}.
\newblock \showarticletitle{Coupling vs. Conductance for the {Jerrum-Sinclair}
  Chain}.
\newblock \bibinfo{journal}{\emph{Random Structures and Algorithms}}
  \bibinfo{volume}{18}, \bibinfo{number}{1} (\bibinfo{year}{2001}),
  \bibinfo{pages}{1--17}.
\newblock
\urldef\tempurl%
\url{http://hariharan-ramesh.com/papers/mcmc.pdf}
\showURL{%
\tempurl}


\bibitem[\protect\citeauthoryear{Larsen and Skou}{Larsen and Skou}{1991}]%
        {LarsenS89}
\bibfield{author}{\bibinfo{person}{Kim~Guldstrand Larsen} {and}
  \bibinfo{person}{Arne Skou}.} \bibinfo{year}{1991}\natexlab{}.
\newblock \showarticletitle{Bisimulation through Probabilistic Testing}.
\newblock \bibinfo{journal}{\emph{Information and Computation}}
  \bibinfo{volume}{94}, \bibinfo{number}{1} (\bibinfo{year}{1991}),
  \bibinfo{pages}{1--28}.
\newblock
\urldef\tempurl%
\url{https://doi.org/10.1016/0890-5401(91)90030-6}
\showDOI{\tempurl}


\bibitem[\protect\citeauthoryear{Levin, Peres, and Wilmer}{Levin
  et~al\mbox{.}}{2009}]%
        {LevinPW09}
\bibfield{author}{\bibinfo{person}{David~A. Levin}, \bibinfo{person}{Yuval
  Peres}, {and} \bibinfo{person}{Elizabeth~L. Wilmer}.}
  \bibinfo{year}{2009}\natexlab{}.
\newblock \bibinfo{booktitle}{\emph{Markov Chains and Mixing Times}}.
\newblock \bibinfo{publisher}{American Mathematical Society}.
\newblock
\urldef\tempurl%
\url{http://pages.uoregon.edu/dlevin/MARKOV/markovmixing.pdf}
\showURL{%
\tempurl}


\bibitem[\protect\citeauthoryear{Lindvall}{Lindvall}{2002}]%
        {Lindvall02}
\bibfield{author}{\bibinfo{person}{Torgny Lindvall}.}
  \bibinfo{year}{2002}\natexlab{}.
\newblock \bibinfo{booktitle}{\emph{Lectures on the Coupling Method}}.
\newblock \bibinfo{publisher}{Courier Corporation}.
\newblock


\bibitem[\protect\citeauthoryear{Lyu, Su, and Li}{Lyu et~al\mbox{.}}{2017}]%
        {lyu2016understanding}
\bibfield{author}{\bibinfo{person}{Min Lyu}, \bibinfo{person}{Dong Su}, {and}
  \bibinfo{person}{Ninghui Li}.} \bibinfo{year}{2017}\natexlab{}.
\newblock \showarticletitle{Understanding the {Sparse Vector Technique} for
  Differential Privacy}.
\newblock \bibinfo{journal}{\emph{Proceedings of the {VLDB} Endowment}}
  \bibinfo{volume}{10}, \bibinfo{number}{6} (\bibinfo{year}{2017}),
  \bibinfo{pages}{637--648}.
\newblock
\urldef\tempurl%
\url{https://doi.org/10.14778/3055330.3055331}
\showDOI{\tempurl}
\newblock
\shownote{Appeared at the {I}nternational {C}onference on {V}ery {L}arge {D}ata
  {B}ases (VLDB), Munich, Germany.}


\bibitem[\protect\citeauthoryear{McIver, Morgan, Kaminski, and Katoen}{McIver
  et~al\mbox{.}}{2018}]%
        {mciver2016new}
\bibfield{author}{\bibinfo{person}{Annabelle McIver}, \bibinfo{person}{Carroll
  Morgan}, \bibinfo{person}{Benjamin~Lucien Kaminski}, {and}
  \bibinfo{person}{{Joost}-{Pieter} Katoen}.} \bibinfo{year}{2018}\natexlab{}.
\newblock \showarticletitle{A New Rule for Almost-Certain Termination}.
\newblock \bibinfo{journal}{\emph{Proceedings of the {ACM} on Programming
  Languages}} \bibinfo{volume}{1}, \bibinfo{number}{POPL}
  (\bibinfo{year}{2018}).
\newblock
\urldef\tempurl%
\url{https://arxiv.org/abs/1612.01091}
\showURL{%
\tempurl}
\newblock
\shownote{To appear at {ACM} {SIGPLAN--SIGACT} {S}ymposium on {P}rinciples of
  {P}rogramming {L}anguages ({POPL}), Los Angeles, California.}


\bibitem[\protect\citeauthoryear{{McSherry}}{{McSherry}}{2009}]%
        {pinq}
\bibfield{author}{\bibinfo{person}{Frank {McSherry}}.}
  \bibinfo{year}{2009}\natexlab{}.
\newblock \showarticletitle{Privacy Integrated Queries}. In
  \bibinfo{booktitle}{\emph{{ACM} {SIGMOD} {I}nternational {C}onference on
  {M}anagement of {D}ata (SIGMOD), Providence, Rhode Island}}.
  \bibinfo{pages}{19--30}.
\newblock
\urldef\tempurl%
\url{https://doi.org/10.1145/1559845.1559850}
\showDOI{\tempurl}


\bibitem[\protect\citeauthoryear{{McSherry} and Talwar}{{McSherry} and
  Talwar}{2007}]%
        {MT07}
\bibfield{author}{\bibinfo{person}{Frank {McSherry}} {and}
  \bibinfo{person}{Kunal Talwar}.} \bibinfo{year}{2007}\natexlab{}.
\newblock \showarticletitle{Mechanism Design via Differential Privacy}. In
  \bibinfo{booktitle}{\emph{{IEEE} {S}ymposium on {F}oundations of {C}omputer
  {S}cience (FOCS), Providence, Rhode Island}}. \bibinfo{pages}{94--103}.
\newblock
\urldef\tempurl%
\url{https://doi.org/10.1109/FOCS.2007.66}
\showDOI{\tempurl}


\bibitem[\protect\citeauthoryear{Mironov}{Mironov}{2017}]%
        {MironovRDP}
\bibfield{author}{\bibinfo{person}{Ilya Mironov}.}
  \bibinfo{year}{2017}\natexlab{}.
\newblock \showarticletitle{R\'enyi Differential Privacy}. In
  \bibinfo{booktitle}{\emph{{IEEE} {C}omputer {S}ecurity {F}oundations
  {S}ymposium ({CSF}), Santa Barbara, California}}. \bibinfo{pages}{263--275}.
\newblock
\urldef\tempurl%
\url{https://doi.org/10.1109/CSF.2017.11}
\showDOI{\tempurl}


\bibitem[\protect\citeauthoryear{Morgan, McIver, and Seidel}{Morgan
  et~al\mbox{.}}{1996}]%
        {Morgan:1996}
\bibfield{author}{\bibinfo{person}{Carroll Morgan}, \bibinfo{person}{Annabelle
  McIver}, {and} \bibinfo{person}{Karen Seidel}.}
  \bibinfo{year}{1996}\natexlab{}.
\newblock \showarticletitle{Probabilistic Predicate Transformers}.
\newblock \bibinfo{journal}{\emph{ACM Transactions on Programming Languages and
  Systems}} \bibinfo{volume}{18}, \bibinfo{number}{3} (\bibinfo{year}{1996}),
  \bibinfo{pages}{325--353}.
\newblock
\urldef\tempurl%
\url{https://doi.org/10.1145/229542.229547}
\showDOI{\tempurl}


\bibitem[\protect\citeauthoryear{O'Hearn}{O'Hearn}{2007}]%
        {OHEARN2007271}
\bibfield{author}{\bibinfo{person}{Peter~W. O'Hearn}.}
  \bibinfo{year}{2007}\natexlab{}.
\newblock \showarticletitle{Resources, Concurrency, and Local Reasoning}.
\newblock \bibinfo{journal}{\emph{Theoretical Computer Science}}
  \bibinfo{volume}{375}, \bibinfo{number}{1} (\bibinfo{year}{2007}),
  \bibinfo{pages}{271--307}.
\newblock
\showISSN{0304-3975}
\urldef\tempurl%
\url{https://doi.org/10.1016/j.tcs.2006.12.035}
\showDOI{\tempurl}
\newblock
\shownote{Festschrift for John C. Reynolds's 70th birthday.}


\bibitem[\protect\citeauthoryear{O'Hearn, Reynolds, and Yang}{O'Hearn
  et~al\mbox{.}}{2001}]%
        {OhRY01}
\bibfield{author}{\bibinfo{person}{Peter~W. O'Hearn}, \bibinfo{person}{John~C.
  Reynolds}, {and} \bibinfo{person}{Hongseok Yang}.}
  \bibinfo{year}{2001}\natexlab{}.
\newblock \showarticletitle{Local Reasoning about Programs That Alter Data
  Structures}. In \bibinfo{booktitle}{\emph{International Workshop on Computer
  Science Logic (CSL), Paris, France}} \emph{(\bibinfo{series}{Lecture Notes in
  Computer Science})}, Vol.~\bibinfo{volume}{2142}.
  \bibinfo{publisher}{Springer-Verlag}, \bibinfo{pages}{1--19}.
\newblock
\urldef\tempurl%
\url{https://doi.org/10.1007/3-540-44802-0_1}
\showDOI{\tempurl}


\bibitem[\protect\citeauthoryear{Olmedo}{Olmedo}{2014}]%
        {OlmedoThesis}
\bibfield{author}{\bibinfo{person}{Federico Olmedo}.}
  \bibinfo{year}{2014}\natexlab{}.
\newblock \emph{\bibinfo{title}{Approximate Relational Reasoning for
  Probabilistic Programs}}.
\newblock \bibinfo{thesistype}{Ph.D. Dissertation}.
  \bibinfo{school}{Universidad Polit\'ecnica de Madrid}.
\newblock
\urldef\tempurl%
\url{http://oa.upm.es/23088/1/FEDERICO_OLMEDO.pdf}
\showURL{%
\tempurl}


\bibitem[\protect\citeauthoryear{Palamidessi and Stronati}{Palamidessi and
  Stronati}{2012}]%
        {palamidessi:hal-00760688}
\bibfield{author}{\bibinfo{person}{Catuscia Palamidessi} {and}
  \bibinfo{person}{Marco Stronati}.} \bibinfo{year}{2012}\natexlab{}.
\newblock \showarticletitle{Differential Privacy for Relational Algebra:
  {Improving} the Sensitivity Bounds via Constraint Systems}. In
  \bibinfo{booktitle}{\emph{Workshop on Quantitative Aspects of Programming
  Languages ({QAPL}), Tallin, Estonia}} \emph{(\bibinfo{series}{Electronic
  Proceedings in Theoretical Computer Science})}, Vol.~\bibinfo{volume}{85}.
  \bibinfo{publisher}{Open Publishing Association}, \bibinfo{pages}{92--105}.
\newblock
\urldef\tempurl%
\url{https://doi.org/10.4204/EPTCS.85.7}
\showDOI{\tempurl}


\bibitem[\protect\citeauthoryear{Propp and Wilson}{Propp and Wilson}{1996}]%
        {DBLP:journals/rsa/ProppW96}
\bibfield{author}{\bibinfo{person}{James~Gary Propp} {and}
  \bibinfo{person}{David~Bruce Wilson}.} \bibinfo{year}{1996}\natexlab{}.
\newblock \showarticletitle{Exact Sampling with Coupled {Markov} Chains and
  Applications to Statistical Mechanics}.
\newblock \bibinfo{journal}{\emph{Random Structures and Algorithms}}
  \bibinfo{volume}{9}, \bibinfo{number}{1-2} (\bibinfo{year}{1996}),
  \bibinfo{pages}{223--252}.
\newblock
\urldef\tempurl%
\url{https://doi.org/10.1002/(SICI)1098-2418(199608/09)9:1/2<223::AID-RSA14>3.0.CO;2-O}
\showDOI{\tempurl}


\bibitem[\protect\citeauthoryear{Proserpio, Goldberg, and {McSherry}}{Proserpio
  et~al\mbox{.}}{2014}]%
        {proserpio2014calibrating}
\bibfield{author}{\bibinfo{person}{Davide Proserpio}, \bibinfo{person}{Sharon
  Goldberg}, {and} \bibinfo{person}{Frank {McSherry}}.}
  \bibinfo{year}{2014}\natexlab{}.
\newblock \showarticletitle{Calibrating Data to Sensitivity in Private Data
  Analysis: {A} Platform for Differentially-Private Analysis of Weighted
  Datasets}.
\newblock \bibinfo{journal}{\emph{Proceedings of the {VLDB} Endowment}}
  \bibinfo{volume}{7}, \bibinfo{number}{8} (\bibinfo{year}{2014}),
  \bibinfo{pages}{637--648}.
\newblock
\urldef\tempurl%
\url{https://doi.org/10.14778/2732296.2732300}
\showDOI{\tempurl}
\newblock
\shownote{Appeared at the {I}nternational {C}onference on {V}ery {L}arge {D}ata
  {B}ases (VLDB), Hangzhou, China.}


\bibitem[\protect\citeauthoryear{Reed and Pierce}{Reed and Pierce}{2010}]%
        {ReedPierce10}
\bibfield{author}{\bibinfo{person}{Jason Reed} {and}
  \bibinfo{person}{Benjamin~C. Pierce}.} \bibinfo{year}{2010}\natexlab{}.
\newblock \showarticletitle{Distance Makes the Types Grow Stronger: {A}
  Calculus for Differential Privacy}. In \bibinfo{booktitle}{\emph{{ACM}
  {SIGPLAN} {I}nternational {C}onference on {F}unctional {P}rogramming
  ({ICFP}), Baltimore, Maryland}}. \bibinfo{pages}{157--168}.
\newblock
\urldef\tempurl%
\url{https://doi.org/10.1145/1863543.1863568}
\showDOI{\tempurl}


\bibitem[\protect\citeauthoryear{Reynolds}{Reynolds}{2001}]%
        {Reynolds01}
\bibfield{author}{\bibinfo{person}{John~C. Reynolds}.}
  \bibinfo{year}{2001}\natexlab{}.
\newblock \showarticletitle{Intuitionistic Reasoning about Shared Mutable Data
  Structure}.
\newblock \bibinfo{journal}{\emph{Millennial Perspectives in Computer Science}}
  \bibinfo{volume}{2}, \bibinfo{number}{1} (\bibinfo{year}{2001}),
  \bibinfo{pages}{303--321}.
\newblock
\urldef\tempurl%
\url{http://citeseerx.ist.psu.edu/viewdoc/download?doi=10.1.1.11.5999&rep=rep1&type=pdf}
\showURL{%
\tempurl}


\bibitem[\protect\citeauthoryear{Reynolds}{Reynolds}{2002}]%
        {Reynolds02}
\bibfield{author}{\bibinfo{person}{John~C. Reynolds}.}
  \bibinfo{year}{2002}\natexlab{}.
\newblock \showarticletitle{Separation Logic: {A} Logic for Shared Mutable Data
  Structures}. In \bibinfo{booktitle}{\emph{{IEEE} {S}ymposium on {L}ogic in
  {C}omputer {S}cience ({LICS}), Copenhagen, Denmark}}.
  \bibinfo{pages}{55--74}.
\newblock
\urldef\tempurl%
\url{https://doi.org/10.1109/LICS.2002.1029817}
\showDOI{\tempurl}


\bibitem[\protect\citeauthoryear{Rogers, Vadhan, Roth, and Ullman}{Rogers
  et~al\mbox{.}}{2016}]%
        {rogers2016privacy}
\bibfield{author}{\bibinfo{person}{Ryan Rogers}, \bibinfo{person}{Salil
  Vadhan}, \bibinfo{person}{Aaron Roth}, {and} \bibinfo{person}{Jonathan
  Ullman}.} \bibinfo{year}{2016}\natexlab{}.
\newblock \showarticletitle{Privacy Odometers and Filters: Pay-as-you-go
  Composition}. In \bibinfo{booktitle}{\emph{{C}onference on {N}eural
  {I}nformation {P}rocessing {S}ystems (NIPS), Barcelona, Spain}}.
  \bibinfo{pages}{1921--1929}.
\newblock
\urldef\tempurl%
\url{http://arxiv.org/abs/1605.08294}
\showURL{%
\tempurl}


\bibitem[\protect\citeauthoryear{Royden and Fitzpatrick}{Royden and
  Fitzpatrick}{2010}]%
        {royden-analysis}
\bibfield{author}{\bibinfo{person}{Halsey~L. Royden} {and}
  \bibinfo{person}{Patrick Fitzpatrick}.} \bibinfo{year}{2010}\natexlab{}.
\newblock \bibinfo{booktitle}{\emph{Real Analysis} (\bibinfo{edition}{fourth}
  ed.)}.
\newblock \bibinfo{publisher}{Prentice Hall}.
\newblock
\showISBNx{9780131437470}
\showLCCN{2009048692}


\bibitem[\protect\citeauthoryear{Rudin}{Rudin}{1976}]%
        {RudinPMA}
\bibfield{author}{\bibinfo{person}{Walter Rudin}.}
  \bibinfo{year}{1976}\natexlab{}.
\newblock \bibinfo{booktitle}{\emph{Principles of Mathematical Analysis}
  (\bibinfo{edition}{third} ed.)}.
\newblock \bibinfo{publisher}{McGraw-Hill}.
\newblock


\bibitem[\protect\citeauthoryear{Sato}{Sato}{2016}]%
        {Sato16}
\bibfield{author}{\bibinfo{person}{Tetsuya Sato}.}
  \bibinfo{year}{2016}\natexlab{}.
\newblock \showarticletitle{Approximate Relational {Hoare} Logic for Continuous
  Random Samplings}. In \bibinfo{booktitle}{\emph{Conference on the
  Mathematical Foundations of Programming Semantics (MFPS), Pittsburgh,
  Pennsylvania}} \emph{(\bibinfo{series}{Electronic Notes in Theoretical
  Computer Science})}, Vol.~\bibinfo{volume}{325}.
  \bibinfo{publisher}{Elsevier}, \bibinfo{pages}{277--298}.
\newblock
\urldef\tempurl%
\url{https://doi.org/10.1016/j.entcs.2016.09.043}
\showDOI{\tempurl}


\bibitem[\protect\citeauthoryear{Sato, Barthe, Gaboradi, Hsu, and
  Katsumata}{Sato et~al\mbox{.}}{2017}]%
        {SBGHK17}
\bibfield{author}{\bibinfo{person}{Tetsuya Sato}, \bibinfo{person}{Gilles
  Barthe}, \bibinfo{person}{Marco Gaboradi}, \bibinfo{person}{Justin Hsu},
  {and} \bibinfo{person}{{Shin-ya} Katsumata}.}
  \bibinfo{year}{2017}\natexlab{}.
\newblock \bibinfo{title}{Reasoning about Divergences for Relaxations of
  Differential Privacy}.  (\bibinfo{year}{2017}).
\newblock
\showeprint[arxiv]{cs.PL/1710.09010}
\urldef\tempurl%
\url{http://arxiv.org/abs/1710.09010}
\showURL{%
\tempurl}


\bibitem[\protect\citeauthoryear{Segala and Lynch}{Segala and Lynch}{1995}]%
        {DBLP:journals/njc/SegalaL95}
\bibfield{author}{\bibinfo{person}{Roberto Segala} {and}
  \bibinfo{person}{Nancy~A. Lynch}.} \bibinfo{year}{1995}\natexlab{}.
\newblock \showarticletitle{Probabilistic Simulations for Probabilistic
  Processes}.
\newblock \bibinfo{journal}{\emph{Nordic Journal of Computing}}
  \bibinfo{volume}{2}, \bibinfo{number}{2} (\bibinfo{year}{1995}),
  \bibinfo{pages}{250--273}.
\newblock
\urldef\tempurl%
\url{https://doi.org/doi.org/10.1007/BFb0015027}
\showDOI{\tempurl}


\bibitem[\protect\citeauthoryear{Segala and Turrini}{Segala and
  Turrini}{2007}]%
        {SegalaTurrini07}
\bibfield{author}{\bibinfo{person}{Roberto Segala} {and}
  \bibinfo{person}{Andrea Turrini}.} \bibinfo{year}{2007}\natexlab{}.
\newblock \showarticletitle{Approximated Computationally Bounded Simulation
  Relations for Probabilistic Automata}. In \bibinfo{booktitle}{\emph{{IEEE}
  {C}omputer {S}ecurity {F}oundations {S}ymposium ({CSF}), Venice, Italy}}.
  \bibinfo{pages}{140--156}.
\newblock
\showISSN{1063-6900}
\urldef\tempurl%
\url{https://doi.org/10.1109/CSF.2007.8}
\showDOI{\tempurl}


\bibitem[\protect\citeauthoryear{Shoup}{Shoup}{2004}]%
        {Shoup:2004}
\bibfield{author}{\bibinfo{person}{Victor Shoup}.}
  \bibinfo{year}{2004}\natexlab{}.
\newblock \bibinfo{title}{Sequences of Games: {A} Tool for Taming Complexity in
  Security Proofs}.
\newblock \bibinfo{howpublished}{Cryptology ePrint Archive, Report 2004/332}.
  (\bibinfo{year}{2004}).
\newblock
\urldef\tempurl%
\url{http://eprint.iacr.org/2004/332}
\showURL{%
\tempurl}


\bibitem[\protect\citeauthoryear{Sousa and Dillig}{Sousa and Dillig}{2016}]%
        {sousa2016cartesian}
\bibfield{author}{\bibinfo{person}{Marcelo Sousa} {and} \bibinfo{person}{Isil
  Dillig}.} \bibinfo{year}{2016}\natexlab{}.
\newblock \showarticletitle{Cartesian {Hoare} Logic for Verifying $k$-Safety
  Properties}. In \bibinfo{booktitle}{\emph{{ACM SIGPLAN Conference on
  Programming Language Design and Implementation (PLDI)}, Santa Barbara,
  California}}. \bibinfo{pages}{57--69}.
\newblock
\urldef\tempurl%
\url{https://doi.org/10.1145/2908080.2908092}
\showDOI{\tempurl}


\bibitem[\protect\citeauthoryear{Strassen}{Strassen}{1965}]%
        {strassen1965existence}
\bibfield{author}{\bibinfo{person}{Volker Strassen}.}
  \bibinfo{year}{1965}\natexlab{}.
\newblock \showarticletitle{The Existence of Probability Measures with Given
  Marginals}.
\newblock \bibinfo{journal}{\emph{The Annals of Mathematical Statistics}}
  (\bibinfo{year}{1965}), \bibinfo{pages}{423--439}.
\newblock
\urldef\tempurl%
\url{https://doi.org/10.1214/aoms/1177700153}
\showDOI{\tempurl}


\bibitem[\protect\citeauthoryear{Terauchi and Aiken}{Terauchi and
  Aiken}{2005}]%
        {TerauchiA05}
\bibfield{author}{\bibinfo{person}{Tachio Terauchi} {and} \bibinfo{person}{Alex
  Aiken}.} \bibinfo{year}{2005}\natexlab{}.
\newblock \showarticletitle{Secure Information Flow as a Safety Problem}. In
  \bibinfo{booktitle}{\emph{International Symposium on Static Analysis (SAS),
  London, England}} \emph{(\bibinfo{series}{Lecture Notes in Computer
  Science})}, Vol.~\bibinfo{volume}{3672}.
  \bibinfo{publisher}{Springer-Verlag}, \bibinfo{pages}{352--367}.
\newblock
\urldef\tempurl%
\url{https://doi.org/10.1007/11547662_24}
\showDOI{\tempurl}


\bibitem[\protect\citeauthoryear{Thorisson}{Thorisson}{2000}]%
        {Thorisson00}
\bibfield{author}{\bibinfo{person}{Hermann Thorisson}.}
  \bibinfo{year}{2000}\natexlab{}.
\newblock \bibinfo{booktitle}{\emph{Coupling, Stationarity, and Regeneration}}.
\newblock \bibinfo{publisher}{Springer-Verlag}.
\newblock


\bibitem[\protect\citeauthoryear{Tracol, Desharnais, and Zhioua}{Tracol
  et~al\mbox{.}}{2011}]%
        {tracolDZ11}
\bibfield{author}{\bibinfo{person}{Mathieu Tracol}, \bibinfo{person}{Jos{\'e}e
  Desharnais}, {and} \bibinfo{person}{Abir Zhioua}.}
  \bibinfo{year}{2011}\natexlab{}.
\newblock \showarticletitle{Computing Distances between Probabilistic
  Automata}. In \bibinfo{booktitle}{\emph{Workshop on Quantitative Aspects of
  Programming Languages ({QAPL}), Saarbr{\"{u}}cken, Germany}}
  \emph{(\bibinfo{series}{Electronic Proceedings in Theoretical Computer
  Science})}, Vol.~\bibinfo{volume}{57}. \bibinfo{publisher}{Open Publishing
  Association}, \bibinfo{pages}{148--162}.
\newblock
\urldef\tempurl%
\url{https://doi.org/10.4204/EPTCS.57.11}
\showDOI{\tempurl}


\bibitem[\protect\citeauthoryear{Tschantz, Kaynar, and Datta}{Tschantz
  et~al\mbox{.}}{2011}]%
        {Tschantz201161}
\bibfield{author}{\bibinfo{person}{Michael~Carl Tschantz},
  \bibinfo{person}{Dilsun Kaynar}, {and} \bibinfo{person}{Anupam Datta}.}
  \bibinfo{year}{2011}\natexlab{}.
\newblock \showarticletitle{Formal Verification of Differential Privacy for
  Interactive Systems}. In \bibinfo{booktitle}{\emph{Conference on the
  Mathematical Foundations of Programming Semantics (MFPS), Pittsburgh,
  Pennsylvania}} \emph{(\bibinfo{series}{Electronic Notes in Theoretical
  Computer Science})}, Vol.~\bibinfo{volume}{276}.
  \bibinfo{publisher}{Elsevier}, \bibinfo{pages}{61--79}.
\newblock
\urldef\tempurl%
\url{https://doi.org/10.1016/j.entcs.2011.09.015}
\showDOI{\tempurl}


\bibitem[\protect\citeauthoryear{van Breugel and Worrell}{van Breugel and
  Worrell}{2001a}]%
        {vanBreugel2001a}
\bibfield{author}{\bibinfo{person}{Franck van Breugel} {and}
  \bibinfo{person}{James Worrell}.} \bibinfo{year}{2001}\natexlab{a}.
\newblock \showarticletitle{An Algorithm for Quantitative Verification of
  Probabilistic Transition Systems}. In \bibinfo{booktitle}{\emph{International
  Conference on Concurrency Theory (CONCUR), Aalborg, Denmark}}
  \emph{(\bibinfo{series}{Lecture Notes in Computer Science})},
  Vol.~\bibinfo{volume}{2154}. \bibinfo{publisher}{Springer-Verlag},
  \bibinfo{pages}{336--350}.
\newblock
\showISBNx{978-3-540-44685-9}
\urldef\tempurl%
\url{https://doi.org/10.1007/3-540-44685-0_23}
\showDOI{\tempurl}


\bibitem[\protect\citeauthoryear{van Breugel and Worrell}{van Breugel and
  Worrell}{2001b}]%
        {vanBreugel2001b}
\bibfield{author}{\bibinfo{person}{Franck van Breugel} {and}
  \bibinfo{person}{James Worrell}.} \bibinfo{year}{2001}\natexlab{b}.
\newblock \showarticletitle{Towards Quantitative Verification of Probabilistic
  Transition Systems}. In \bibinfo{booktitle}{\emph{International Colloquium on
  Automata, Languages and Programming (ICALP), Crete, Greece}}
  \emph{(\bibinfo{series}{Lecture Notes in Computer Science})},
  Vol.~\bibinfo{volume}{2076}. \bibinfo{publisher}{Springer-Verlag},
  \bibinfo{pages}{421--432}.
\newblock
\showISBNx{978-3-540-48224-6}
\urldef\tempurl%
\url{https://doi.org/10.1007/3-540-48224-5_35}
\showDOI{\tempurl}


\bibitem[\protect\citeauthoryear{Villani}{Villani}{2008}]%
        {Villani08}
\bibfield{author}{\bibinfo{person}{C{\'e}dric Villani}.}
  \bibinfo{year}{2008}\natexlab{}.
\newblock \bibinfo{booktitle}{\emph{Optimal Transport: {Old} and New}}.
\newblock \bibinfo{publisher}{Springer-Verlag}.
\newblock
\urldef\tempurl%
\url{https://doi.org/10.1007/978-3-540-71050-9}
\showDOI{\tempurl}


\bibitem[\protect\citeauthoryear{Winograd-Cort, Haeberlen, Roth, and
  Pierce}{Winograd-Cort et~al\mbox{.}}{2017}]%
        {adafuzz}
\bibfield{author}{\bibinfo{person}{Daniel Winograd-Cort},
  \bibinfo{person}{Andreas Haeberlen}, \bibinfo{person}{Aaron Roth}, {and}
  \bibinfo{person}{Benjamin~C. Pierce}.} \bibinfo{year}{2017}\natexlab{}.
\newblock \showarticletitle{A Framework for Adaptive Differential Privacy}. In
  \bibinfo{booktitle}{\emph{{ACM} {SIGPLAN} {I}nternational {C}onference on
  {F}unctional {P}rogramming ({ICFP}), Oxford, England}}.
  \bibinfo{pages}{10:1--10:29}.
\newblock
\urldef\tempurl%
\url{https://doi.org/10.1145/3110254}
\showDOI{\tempurl}


\bibitem[\protect\citeauthoryear{Yang}{Yang}{2007}]%
        {YANG2007308}
\bibfield{author}{\bibinfo{person}{Hongseok Yang}.}
  \bibinfo{year}{2007}\natexlab{}.
\newblock \showarticletitle{Relational Separation Logic}.
\newblock \bibinfo{journal}{\emph{Theoretical Computer Science}}
  \bibinfo{volume}{375}, \bibinfo{number}{1} (\bibinfo{year}{2007}),
  \bibinfo{pages}{308--334}.
\newblock
\showISSN{0304-3975}
\urldef\tempurl%
\url{https://doi.org/10.1016/j.tcs.2006.12.036}
\showDOI{\tempurl}
\newblock
\shownote{Festschrift for John C. Reynolds's 70th birthday.}


\bibitem[\protect\citeauthoryear{Zaks and Pnueli}{Zaks and Pnueli}{2008}]%
        {ZaksP08}
\bibfield{author}{\bibinfo{person}{Anna Zaks} {and} \bibinfo{person}{Amir
  Pnueli}.} \bibinfo{year}{2008}\natexlab{}.
\newblock \showarticletitle{{CoVaC}: Compiler Validation by Program Analysis of
  the Cross-Product}. In \bibinfo{booktitle}{\emph{International Symposium on
  Formal Methods (FM), Turku, Finland}} \emph{(\bibinfo{series}{Lecture Notes
  in Computer Science})}, Vol.~\bibinfo{volume}{5014}.
  \bibinfo{publisher}{Springer-Verlag}, \bibinfo{pages}{35--51}.
\newblock
\urldef\tempurl%
\url{https://doi.org/10.1007/978-3-540-68237-0_5}
\showDOI{\tempurl}


\bibitem[\protect\citeauthoryear{Zhang and Kifer}{Zhang and Kifer}{2017}]%
        {zhang2016autopriv}
\bibfield{author}{\bibinfo{person}{Danfeng Zhang} {and} \bibinfo{person}{Daniel
  Kifer}.} \bibinfo{year}{2017}\natexlab{}.
\newblock \showarticletitle{{LightDP}: Towards Automating Differential Privacy
  Proofs}. In \bibinfo{booktitle}{\emph{{ACM} {SIGPLAN--SIGACT} {S}ymposium on
  {P}rinciples of {P}rogramming {L}anguages ({POPL}), Paris, France}}.
  \bibinfo{pages}{888--901}.
\newblock
\urldef\tempurl%
\url{https://doi.org/10.1145/3009837.3009884}
\showDOI{\tempurl}


\end{thebibliography}

\end{document}